\documentclass[11pt]{article}
\usepackage{setspace}
\usepackage{mathtools}
\allowdisplaybreaks
\usepackage{comment}
\usepackage{natbib}     
\setcitestyle{year,open={(},close={)}} %Citation-related commands
\usepackage[psamsfonts]{amssymb}
\usepackage{amsthm,amsmath,amsbsy,mathabx}
\usepackage{bbm, bm}
\usepackage[dvipsnames]{xcolor}
\usepackage{subcaption}
\usepackage[shortlabels]{enumitem}

\RequirePackage[colorlinks,citecolor=blue,urlcolor=blue]{hyperref}
\usepackage{doi}
\usepackage{url}

\usepackage{graphicx,placeins}

\usepackage{inputenc}
\usepackage{graphics,multirow}
\usepackage{supertabular}
\usepackage{amscd}
\usepackage{fancyhdr}
\usepackage{verbatim}

\usepackage{pdflscape}
\usepackage{array}
\usepackage{float}
\usepackage{colortbl}
\usepackage{rotating}
\usepackage{lscape}

\usepackage{booktabs,caption}
\usepackage[flushleft]{threeparttable}
\usepackage{amsthm}
\usepackage{algorithm}
\usepackage{algpseudocode}

\usepackage{dcolumn}
\newcolumntype{L}{D{.}{.}{2,5}}

\usepackage{chngpage} % allows for temporary adjustment of side margins
\newcommand\tr{\mathrm{tr}}
\newcommand\prox{\mathrm{prox}}
\newcommand\acc{\mathrm{acc}}
\DeclareMathOperator*{\argmin}{\arg\min}
\usepackage{footnote}
\usepackage{adjustbox}

\newtheorem{theorem}{Theorem}[section]
\newtheorem{lemma}[theorem]{Lemma}

\newtheorem{proposition}[theorem]{Proposition}

\newtheorem{assumption}{Assumption}
\newtheorem{remark}{Remark}[section]

\def\Nc{{\mathcal N}}

\long\def\symbolfootnote[#1]#2{\begingroup%
\def\thefootnote{\fnsymbol{footnote}}\footnotetext[#1]{#2}\footnotemark[#1]\endgroup}

\def\ba#1\ea{\begin{align*}#1\end{align*}} %\ba = \begin{algin*}, \ea = \end{align*}
\def\banum#1\eanum{\begin{align}#1\end{align}} %\banum = \begin{algin}, \eanum

\makeatletter
\@addtoreset{equation}{section}

\makeatother

\newcounter{Fig}[figure]

\newcounter{Tab}[table]

  {\refstepcounter{Tab}
   \addtocounter{table}{1}
   \samepage\vspace{0.2cm}
   \centerline{#1} \nobreak
   \begin{verse}{Table \thesection.\arabic{Tab}:}
  }%
  {\end{verse} \vspace{0.2cm}}

\relax

\def \Eb{{\mathbb E}}

\def \Gb{{\mathbb G}}

\def \Lb{{\mathbb L}}

\def \Pb{{\mathbb P}}
\def \Rb{{\mathbb R}}

\def \Bc{{\mathcal B}}

\def \Ic{{\mathcal I}}

\def \Uc{{\mathcal U}}
\def \Ec{{\mathcal E}}

\def \Sc{{\mathcal S}}

\def \Nc{{\mathcal N}}

\def \Tc{{\mathcal T}}

\def \Xc{{\mathcal X}}

\DeclareMathOperator*{\Sym}{\mathrm{Sym}}

\newcommand{\bqa}{\begin{align*}}
\newcommand{\eqa}{\end{align*}}
\newcommand{\bqan}{\begin{align}}
\newcommand{\eqan}{\end{align}}
\newcommand{\bqt}{\begin{quote}}
\newcommand{\eqt}{\end{quote}}
\newcommand{\bt}{\begin{tabbing}}
\newcommand{\et}{\end{tabbing}}
\newcommand{\bit}{\begin{itemize}}
\newcommand{\eit}{\end{itemize}}
\newcommand{\ben}{\begin{enumerate}}
\newcommand{\een}{\end{enumerate}}
\newcommand{\beq}{\begin{equation}}
\newcommand{\eeq}{\end{equation}}
\newcommand{\beqw}{\begin{equation*}}
\newcommand{\eeqw}{\end{equation*}}

\newcommand{\eps}{\epsilon}

\def\1{{\mathbf 1}}
\def\0{{\mathbf 0}}

\usepackage{authblk}

\usepackage{xr}
\makeatletter
\newcommand*{\addFileDependency}[1]{% argument=file name and extension
  \typeout{(#1)}
  \@addtofilelist{#1}
  \IfFileExists{#1}{}{\typeout{No file #1.}}
}
\makeatother

\makeatletter
\def\@seccntformat#1{\@ifundefined{#1@cntformat}%
   {\csname the#1\endcsname\quad}%      default
   {\csname #1@cntformat\endcsname}%    enable individual control
}
\makeatother

\begin{document}

\title{{\Large \bf Change-point detection in variance-covariance matrix}}

\medskip

\author[a]{Ying Lin}
\author[b]{Benjamin Poignard\footnote{Corresponding author}}

\affil[a]{\it \footnotesize Department of Data and Systems Engineering, The University of Hong Kong, People's Republic of China. E-mail address: ylin95@hku.hk}
\affil[b]{\it \footnotesize Faculty of Science and Technology, Keio University and Riken-AIP, Japan. E-mail address: bpoignard@math.keio.ac.jp}

\date{}
\maketitle

\begin{abstract}
We consider the joint estimation of change-point locations and the sparsity pattern of the variance–covariance matrix, which is assumed to evolve in a piecewise constant manner.
By applying Group Fused LASSO and LASSO penalties to the squared Frobenius norm, we estimate both the covariance structure and the change points.
Adaptive weights are incorporated into the penalty terms to enhance change-point detection and covariance estimation accuracy.
We establish the conditions under which the estimated change points and the sparse estimators within each segment are consistent.
To solve the resulting optimization problem efficiently, we develop an alternating direction method of multipliers (ADMM) whose updates reduce to computationally tractable subproblems.
The performance of the proposed method is illustrated through synthetic and real-data experiments, including comparisons with several competing procedures.

\medskip

\noindent\textbf{Key words}: Alternating direction method of multipliers, Fused lasso, High-dimensional time series, Penalized least squares, Sparsity, Structural breaks.
\medskip
\noindent

\end{abstract}

\section{Introduction}

Variance–covariance specification and estimation, particularly in high-dimensional settings, have attracted substantial attention due to the importance of accurately capturing second-order moments and their wide range of applications, including finance and social networks. Sparsity-based procedures have been proposed to simplify the covariance structure and enhance estimation accuracy; see, for example, \citep{bickel2008,lam2009,ledoit2022}. However, these works typically assume that the variance–covariance matrix is constant, thereby excluding non-stationary time series settings.

The primary focus of this paper is the detection of change points in the variance–covariance matrix, along with its estimation within each regime. A large body of literature has addressed the detection of multiple change points in time series distributions, often through changes in the mean or variance, which are assumed to be piecewise constant but may exhibit ``jumps'' or ``change points'' over time. One line of research is based on binary segmentation (BS) and CUSUM techniques. The BS method, introduced by \citet{vostrikova1981}, scans for a change point over the entire sample; if one is detected, the sample is split into two sub-samples, and the same procedure is applied recursively until no further change points are identified. The BS method has since been extended in several directions; in particular, \citet{fryzlewicz2014} introduced the Wild Binary Segmentation (WBS) method to handle closely spaced change points and small jump magnitudes. These techniques have also been employed to test for change points in covariance structures: see, e.g., \citep{cai2013}. \citet{Avanesov2018} proposed a matrix-based CUSUM statistic using a de-sparsified estimator of the precision matrix. To detect multiple change points in high-dimensional variance–covariance matrices, \citet{wang2021} proposed the Wild Binary Segmentation through Independent Projection (WBSIP), a variant of WBS, and showed that its localization error is minimax-rate optimal, although under the restriction that the dimension diverges at most at a slow polynomial rate. See \citep{liu2022} for a detailed review of hypothesis testing methods for change-point detection.

Another line of work on change-point detection relies on penalized likelihood–based techniques applied to model parameters. A widely used approach is the fused LASSO, which identifies change points by penalizing the sum of absolute successive differences in the parameters. For example, \citet{Harchaoui2010} studied multiple change points in one-dimensional piecewise constant signals, a framework extended by \citet{Bleakley2011} to grouped parameters via the $\ell_2$ norm on parameter differences. Change-point detection via parameter regularization has also been applied to precision matrices. While \citet{Kolar2012} focused on detecting multiple change points at the node level through the Group Fused LASSO, \citet{Roy2017} considered a single change point affecting the global network structure. \citet{Hallac2017} and \citet{Gibberd2017} studied the Group Fused Graphical LASSO (GFGL) penalization to detect change points in precision matrices. The latter approach combines the Group Fused LASSO regularization of \citet{Bleakley2011} with LASSO regularization applied to the Gaussian likelihood. As an alternative to Gaussian-based estimators, \citet{LPPT24GFDtL} proposed a D-trace loss estimator under GFGL filtering to detect change points and estimate sparse structures.

Our work builds on this line of regularization-based techniques to detect multiple change points and estimate covariance matrices. More precisely, we consider a combination of the Group Fused LASSO, applied to successive differences of covariance parameters to detect change points, and the LASSO, applied entrywise to the covariance parameters. Both regularization terms are incorporated into a squared Frobenius norm loss, which serves as a dissimilarity measure between the sample covariance and the covariance parameter. Unlike the Gaussian likelihood, this loss does not involve the inverse covariance matrix, which simplifies both the theoretical analysis and the design of efficient algorithms. To improve estimation accuracy for both change points and covariance matrices, we introduce adaptive weights in both regularization terms. Adaptive penalization, originally proposed by \citet{zou2006adaptive} for the LASSO, has been applied, for example, to mixtures of LASSO and Group LASSO for general M-estimators in \citet{poignard2020aism}, and to the Group Fused LASSO for quantile estimation in \citet{Ciuperca2020}. We study the asymptotic properties of the resulting estimator, referred to hereafter as the adaptive Group Fused least squares LASSO (GFlsL). 

To solve the corresponding optimization problem, we employ the alternating direction method of multipliers (ADMM), introduced independently by \citet{GM1975} and \citet{GM1976}, with convergence theory established by \citet{FG1983} and extended to the semi-proximal setting by \citet{FPST2013}.

Our contributions can be summarized as follows.
First, we propose the adaptive GFlsL estimator for change-point detection in variance–covariance matrices.
Second, we establish conditions under which both the change points and the covariance matrices can be consistently estimated.
Third, we develop an alternating direction method of multipliers (ADMM) algorithm to solve the resulting optimization problem, with computationally tractable subproblems.
Finally, we assess the performance of the proposed method against several competing approaches through simulation studies and a real-data application.

The rest of the paper is organized as follows. Section \ref{framework} details the framework and estimation procedure. Section \ref{asymptotic_properties} contains the asymptotic properties. Section \ref{sec:opt} is devoted to the optimization and implementation aspects of the estimation procedure. Section \ref{sec:simulations} contains some simulation experiments, a sensitivity and computational complexity analysis. A real data experiment is provided in Section \ref{sec:real_data}. All proofs are deferred to the appendices.

\textbf{\textit{Notation:}} Throughout this paper, we use $ V_k $ and $ A_{kl} $ to denote the $ k $-th element of a vector $ V \in \Rb^d $ and the $ (k, l) $-th element of a matrix $ A \in \Rb^{m \times n} $, respectively.
We write $A^\top$ (resp. $V^\top$) to denote the transpose of the matrix $A$ (resp. the vector $V$).
For any square matrix $ A \in \Rb^{n \times n} $, we write $ \Sym(A) \coloneqq (A + A^{\top}) / 2 $ to denote the symmetrization of $ A $.
The $ p $-th order identity matrix is denoted by $ I_p $.
We denote by $\textbf{0}_{k \times l} \in \Rb^{k \times l}$ (resp. $\mathbf{1}_{k \times l} \in \Rb^{k \times l}$) the $k \times l$ zero matrix (resp. $ k \times l $ matrix of ones).
We write $\text{vec}(A)$ to denote the vectorization operator that stacks the columns of $A$ on top of one another into a vector.
For two matrices $A$ and $B$, $A \otimes B$ is the Kronecker product.
The set of symmetric matrices is denoted by $ \mathcal{S}^n $.
A symmetric matrix $ S \in \mathcal{S}^n $ is said to be positive semi-definite (resp. positive definite) and written as $ S \succeq 0 $ (resp. $ S \succ 0 $) if all its eigenvalues are non-negative (resp. positive).
The expression $ A \succeq B $ (resp. $ A \succ B $) for $ A $, $ B \in \mathcal{S}^n $ means $ A - B \succeq 0 $ (resp. $ A - B \succ 0 $).
We use $ \mathcal{S}_+^n $ and $ \mathcal{S}_{++}^n $ to denote the sets of positive semi-definite matrices and positive definite matrices, respectively.
For a positive semi-definite matrix $ S $, $ S^{\frac{1}{2}} $ is the unique positive semi-definite matrix such that $ S = S^{\frac{1}{2}} S^{\frac{1}{2}} $.
The $\ell_p$ norm for $V \in \Rb^{d}$ is denoted by $\|V\|_p = \big(\sum^{d}_{k=1} |V_k|^p \big)^{1/p}$, $p \geq 1$.
The Frobenius norm and the off-diagonal $ \ell_1 $ (semi)norm of a matrix $ A \in \Rb^{m \times n} $ are denoted by $ \| A \|_F = \sqrt{\sum_{k=1}^m \sum_{l=1}^n | A_{kl} |^2 } $ and $ \| A \|_{1, \mathrm{off}} = \sum_{k\neq l} | A_{kl} | $, respectively.
%The spectral norm, i.e., the maximum singular value of a matrix $ A $ is written as $ \| A \|_s $. 
For a symmetric matrix A, we use $\lambda_{\max}(A)$ (resp. $\lambda_{\min}(A)$) to denote its largest (resp. smallest) eigenvalue.
The maximum absolute value of all entries of a matrix $ A $ is denoted by $ \| A \|_{\max} $.
We use $ \mathcal{S}_{\mathrm{off}}^n $ to denote the set of $ n \times n $ symmetric matrices whose diagonal elements are $ 0 $.
By $ \{Y_{t, \mathrm{off}}\}_{t=1}^T $ we denote a sequence of symmetric matrices whose diagonal elements are $ 0 $, i.e., $ Y_{t, \mathrm{off}} \in \mathcal{S}_{\mathrm{off}}^n $.
For a sequence of symmetric matrices $ \{\Theta_t\}_{t=1}^T $, $ \Theta_{uv, t} $ refers to the $ (u, v) $-th element of $ \Theta_t $, and $ \Theta_{t, \mathrm{off}} $ is the copy of $ \Theta_t $ with diagonal elements set to $ 0 $; in particular, $ \| \Theta_t \|_{1, \mathrm{off}} = \| \Theta_{t, \mathrm{off}} \|_1 $.
Given $ \epsilon \geq 0 $, we use $ \mathrm{Proj}_{\cdot \succeq \epsilon I_p} $ to denote the projection onto $ \{S \,:\, S \succeq \epsilon I_p\} $.
The proximal operator of a function $ f $ at $ x $ is defined as $\prox_f(x) = \argmin_y \{ f(y) + \frac{1}{2} \| y - x \|^2 \}$; for more details about the proximal operator, we refer the interested readers to Sections 12.4, 14.1, 14.2 in \citet{BC2011}.

\section{Framework}\label{framework}

For a sequence of a $p$-dimensional random vector $(X_t)$ observed at $t=1,\ldots,T$, we consider the estimation of the corresponding variance-covariance matrix $\text{Var}(X_t)$.
The latter is assumed to evolve over time and the task is to recover the locations of the $m$ change points. 
More formally, we denote by $\{\Bc_j\}_{1 \leq j \leq m+1}$ a disjoint partitioning of the set $\{1,\ldots,T\}$ such that $\Bc_j \cap \Bc_{j'}=\emptyset $ whenever $ j \neq j'$, $\cup_j \Bc_j=\{1,\ldots,T\}$, and $\Bc_j = \{T_{j-1},T_{j-1}+1,\ldots,T_j-1\}$.
The partition of the locations of the change points (the change-point dates) is denoted by $\Tc_m =\{T_1<T_2<\ldots<T_m\}$ with the convention $T_0=1, T_{m+1}=T+1$.
Then, we assume $\Eb[X_t]=0$ and $\text{Var}(X_t) = \Sigma_j$ for $t\in \Bc_j$, such that the observations indexed by elements in $\Bc_j$ are $p$-dimensional realizations of a centered random variable with variance-covariance $\Sigma_j$ with entries $\Sigma_{uv,j}, 1 \leq u,v\leq p$.
In practice, we consider the sequence of variance-covariance matrices $\{\Theta_1,\ldots,\Theta_T\}$ such that the total number of distinct matrices in the set is $m+1$ and $\Theta_t = \Sigma_j, \; t \in \Bc_j, \; j = 1,\ldots,m+1$.
We are interested in estimating the unknown true number $m^*$ of unknown change points, the true partition $\Tc^*_{m^*}=\{T^*_1<T^*_2<\ldots<T^*_{m^*}\}$ and the true unknown variance-covariance matrices $\Sigma^*_j$ which are assumed sparse.
As a consequence, the true data generating process is assumed to be $\Eb[X_t] = 0$, $\text{Var}(X_t) = \Sigma^*_j$, $\Theta^*_t = \Sigma^*_j$ when $t =T^*_{j-1},T^*_{j-1}+1,\ldots,T^*_{j}-1$,
for $1 \leq j \leq m^*+1$, where $T^*_0=1, T^*_{m^*+1}=T+1$ with blocks $\Bc^*_j = \{T^*_{j-1},\ldots,T^*_j-1\}$.
We define $s^*_j = |\Sc^*_j|$, with $\Sc^*_j=\{(u,v): u\neq v, \Sigma^*_{uv,j} \neq 0\}$ the true sparse support and $s^*_{\max}=\max_{1 \leq j \leq m^*+1}s^*_j < p(p-1)$. 
While $m^*$ and the change points are unknown, $m^*$ is typically much smaller than $T$ and, assuming the variance-covariance structure may exhibit some sparsity, we consider the sparse estimation of $\Theta_t$'s and the estimation of $\Tc_{m}$ via a mixture of LASSO and Group Fused LASSO applied to the least square loss for variance-covariance estimation, which we will hereafter refer as the Group Fused least squares LASSO (GFlsL).
Our approach builds upon the adaptive version of the Group Fused LASSO: in the first stage, we run the Group Fused LASSO to filter the change points and estimate the variance-covariance matrices in each regime, a framework where the estimator will be denoted by $\widetilde{\Theta}_t$; in the second stage, we propose to detect the change points and estimate the variance-covariance matrices under the sparse constraint via the adaptive GFlsL, where $\widetilde{\Theta}_t$ will be used in the adaptive weights, and the resulting estimator will be denoted by $\widehat{\Theta}_t$. The procedure can be outlined as follows:
\begin{itemize}
    \item[(i)] In the first stage, $\{\Sigma^*_j\}_{j=1}^{m^*+1}$ are estimated by minimizing:
\begin{equation}\label{stat_crit_1}
\{\widetilde{\Theta}_t\}^T_{t=1} = \underset{\Theta_t\succ 0, 1 \leq t \leq T}{\arg\;\min} \left\{\Lb(\{\Theta_t\}^T_{t=1},\mathcal{X}_T) + \lambda \overset{T}{\underset{t=2}{\sum}} \|\Theta_t-\Theta_{t-1}\|_F\right\},
\end{equation}
where $\mathcal{X}_T=(X_1,\ldots,X_T)$ is the sample of observations, $\lambda$ is the tuning parameter and the least squares loss is defined as
\begin{equation*}
\Lb(\{\Theta_t\}^T_{t=1},\mathcal{X}_T) = \frac{1}{2T}\overset{T}{\underset{t=1}{\sum}}\text{tr}\big(\Theta^\top_t\Theta_t - (X_tX^\top_t)^\top\Theta_t - \Theta^\top_t(X_tX^\top_t)\big).
\end{equation*}
For a solution $\widetilde{\Theta}_t,t=1,\ldots,T$ of (\ref{stat_crit_1}), there is a block partition $\{\widetilde{\mathcal{B}}_1,\ldots,\widetilde{\mathcal{B}}_{\widetilde{m}+1}\}$ of $\{1,\ldots,T\}$ such that $\widetilde{\Theta}_t = \widetilde{\Theta}_{s}$ for any $t,s \in \widetilde{\mathcal{B}}_j=[\widetilde{T}_{j-1},\widetilde{T}_j-1]$ and $\widetilde{\Theta}_{\widetilde{T}_j} \neq \widetilde{\Theta}_{\widetilde{T}_{j}-1}$, $j=1,\ldots,\widetilde{m}$, where $\widetilde{T}_0=1, \widetilde{T}_{\widetilde{m}+1}=T+1$. So $\widetilde{m}$ and $\widetilde{\Tc}_{\widetilde{m}}\coloneqq \{\widetilde{T}_1<\widetilde{T}_2<\ldots<\widetilde{T}_{\widetilde{m}}\}$ are the estimated number of change points and set of change points. We define $\widetilde{\Sigma}_{j}=\widetilde{\Theta}_{\widetilde{T}_{j-1}}$ as the estimator of $\Sigma_j$ for $j=1,\ldots,\widetilde{m}+1$. Since the penalization does not involve any adaptive weights, we refer to the procedure as non-adaptive change-point estimation. 
\item[(ii)]  In the second stage, we estimate the change points and the variance-covariance matrices under sparsity via the adaptive version of the Group Fused LASSO and LASSO penalization. We consider the adaptive GFlsL estimator defined by:
\begin{equation}\label{stat_crit_adaptive}
\{\widehat{\Theta}_t\}^T_{t=1} = \underset{\Theta_t\succ 0, 1 \leq t \leq T}{\arg\;\min} \left\{\Lb(\{\Theta_t\}^T_{t=1},\mathcal{X}_T) + \lambda_1\overset{T}{\underset{t=1}{\sum}}\sum_{u\neq v}\xi_{uv,1t}|\Theta_{uv,t}|+\lambda_2\overset{T}{\underset{t=2}{\sum}}\xi_{2t} \|\Theta_t-\Theta_{t-1}\|_F\right\},
\end{equation}
with $\lambda_1,\lambda_2$ the tuning parameters, $\xi_{uv,1t}, \xi_{2t}$ are the adaptive weights defined by
\begin{equation}
  \label{eq:adaptive-weight}
\xi_{uv,1t} = \max(|\widetilde{\Theta}_{uv,t}|,a_T)^{-\mu_1}, \; \xi_{2t} = \max(\|\widetilde{\Theta}_{t}-\widetilde{\Theta}_{t-1}\|_F,a_T)^{-\mu_2},
\end{equation}
where $a_T$ is some deterministic sequence tending to zero as $T \rightarrow \infty$, $\mu_1,\mu_2>0$. In the same spirit as the estimator defined in \eqref{stat_crit_1}, there is a partition $\{\widehat{\mathcal{B}}_1,\ldots,\widehat{\mathcal{B}}_{\widehat{m}+1}\}$ of $\{1,\ldots,T\}$ such that $\widehat{\Theta}_t = \widehat{\Theta}_{s}$ for any $t,s \in \widehat{\mathcal{B}}_j=[\widehat{T}_{j-1},\widehat{T}_j-1]$ and $\widehat{\Theta}_{\widehat{T}_j} \neq \widehat{\Theta}_{\widehat{T}_{j}-1}$, $j=1,\ldots,\widehat{m}$, where $\widehat{T}_0=1, \widehat{T}_{\widehat{m}+1}=T+1$. So $\widehat{m}$ and $\widehat{\Tc}_{\widehat{m}}\coloneqq \{\widehat{T}_1<\widehat{T}_2<\ldots<\widehat{T}_{\widehat{m}}\}$ are the estimated number of change points and the estimated set of change points. We define $\widehat{\Sigma}_{j}=\widehat{\Theta}_{\widehat{T}_{j-1}}$ as the estimator of $\Sigma_j$ for $j=1,\ldots,\widehat{m}+1$. This procedure will be referred to as adaptive change-point estimation.
In Lemma \ref{optimality_cond_adaptive} of Appendix \ref{appendix:intermediary}, we provide the optimality conditions satisfied by the adaptive GFlsL: we show if $t=\widehat{T}_j$ with $j\in \{1,\ldots,\widehat{m}\}$, i.e., $t$ corresponds to one of the estimated change points, then $\widehat{\Theta}_t-\widehat{\Theta}_{t-1} \neq \mathbf{0}_{p \times p}$, and we obtain the corresponding optimality conditions.

\end{itemize}

\section{Asymptotic properties}\label{asymptotic_properties}

First, we define some notations and detail the assumptions. Let $\mathcal{I}^*_j=T^*_j-T^*_{j-1}$ 
and $\mathcal{I}_{\min}=\underset{1 \leq j \leq m^*+1}{\min}|\mathcal{I}^*_j|$, $\eta_{\min}=\underset{1 \leq j \leq m^*}{\min}\|\Sigma^*_{j+1}-\Sigma^*_j\|_F, \; \eta_{\max}=\underset{1 \leq j \leq m^*}{\max}\|\Sigma^*_{j+1}-\Sigma^*_j\|_F$. 
Assume the following conditions.
\begin{assumption}\label{assumption_dgp}
\begin{itemize}
    \item[(i)] $(X_t)$ is a centered strong mixing process, that is, $\exists 0 < \rho < 1$ such that for all $t \in \mathbb{Z}^+$, $\alpha(t) \leq c_{\alpha}\rho^{t}$, with $c_{\alpha}>0$ and $\alpha(\cdot)$ the mixing coefficient $\alpha(T) = \sup_{A \in \mathcal{F}^0_{-\infty}, B\in \mathcal{F}^{\infty}_T } |\Pb(A \cap B)-\Pb(A)\,\Pb(B)|$, where $\mathcal{F}^0_{-\infty},\mathcal{F}^{\infty}_T$ are the filtrations generated by $\{(X_t): - \infty \leq t \leq 0\}$ and $\{(X_t): T \leq t \leq \infty\}$.
    \item[(ii)] $\exists \gamma, b>0$ such that $\forall s>0$, $\forall 1 \leq k,l \leq p$, $\sup_{t\geq 1}\,\Pb(|X_{k,t}X_{l,t}|>s)\leq \exp(1-(s/b)^\gamma)$.
    \end{itemize}
\end{assumption}
\begin{assumption}\label{assumption_regularity}
$\exists \underline{\mu}, \overline{\mu}$:
    $0 < \underline{\mu} \leq \underset{1 \leq j \leq m^*+1}{\min}\lambda_{\min}\big(\Sigma^*_j\big)\leq \underset{1 \leq j \leq m^*+1}{\max}\lambda_{\max}\big(\Sigma^*_j\big)\leq \overline{\mu}<\infty$.
    %\item[(ii)] $T \delta_T (p^2 \log(T))^{-1} \rightarrow 0$ as $T \rightarrow \infty$.
\end{assumption}
\begin{assumption}\label{assumption_rates}
Let $(\delta_T)$ be a non-increasing positive sequence converging to zero. The following conditions hold:
\begin{itemize} 
    \item[(i)] $T \delta_T \geq c_v \log(pT)^{(2+\gamma)/\gamma}$ for some $c_v >0$. 
    \item[(ii)] $m^*=O(\log(T))$ and $\mathcal{I}_{\min}/(T\delta_T)\rightarrow \infty$ as $T\rightarrow \infty$.
    \item[(iii)] $\eta_{\max} =O(1)$, $p\sqrt{\log(pT)/(T\delta_T)}\rightarrow 0$ and $(\sqrt{T\delta_T}\eta_{\min})^{-1}p\sqrt{\log(pT)}\rightarrow 0$.
    \item[(iv)] $\lambda/(\eta_{\min}\delta_T) \rightarrow 0$ as $T \rightarrow \infty$.
    \item[(v)] $T\delta_T=O(\Ic^{1/2}_{\min})$ and $\frac{Tm^*}{\Ic_{\min}\eta^2_{\min}}\Big(p\sqrt{\frac{\log(pT)}{T}}p\sqrt{\frac{\log(pT)}{\Ic_{\min}}}+p^2\frac{\log(pT)}{\Ic_{\min}}+p\sqrt{\frac{\log(pT)}{T}}\Big)=o(1)$.
\end{itemize}
\end{assumption}

Assumption \ref{assumption_dgp}(i) relates to the properties of $(X_t)$. %Condition (ii) of the latter assumption
Assumption \ref{assumption_dgp}(ii) is a tail condition and will allow us to apply exponential inequalities for dependent processes. Assumption \ref{assumption_regularity} ensures the identification of the model: it is similar to Assumption A.1 of \citet{Kolar2012} or Assumption A.2 of \citet{Qian2016}. 
Assumption \ref{assumption_rates} provides conditions on $\delta_T$, $m^*$, $\mathcal{I}_{\min}$, $\eta_{\min}$ and the tuning parameter $\lambda$. Condition (i) concerns the convergence rate of $\delta_T$ to $0$. In condition (ii), the sample size in each regime may diverge with rate $T\delta_T$, but at a slower rate than $T$, and the number of change points $m^*$ may diverge slowly: this is similar to Assumption A.3-(i) of \citet{Qian2016} or Assumption H3 of \citet{Chan2014}. It also sets the slowest rate at which $\delta_T$ may shrink to zero: $\delta_T=o(\mathcal{I}_{\min}/T)$. Conditions (iii) and (iv) specify the fastest rate at which $\delta_T$ may shrink to zero, which is $\delta_T \gg \max(\lambda/\eta_{\min},p^2\log(pT)/(T\eta_{\min}^2))$. 
Conditions (ii)-(iii) imply $p^2\log(pT) =o(\mathcal{I}_{\min}\eta^2_{\min})$ and conditions (ii) and (iv) imply that $\lambda T = o(\eta_{\min}\mathcal{I}_{\min})$. Condition (iv) is used for the analysis of the non-adaptive estimator. Part (v) will ensure the $\sqrt{I^*_j/p^2\log(pT)}$-consistency and that the number of change points estimated by both criteria (\ref{stat_crit_1}) and (\ref{stat_crit_adaptive}) is larger than the true number of change points with probability tending to one. 

\subsection{Non-adaptive case}\label{subsec:asymptotic_nonada}

The consistency of change-point locations $\widetilde{T}_j$ and of the variance covariance $\widetilde{\Sigma}_j$, given $\widetilde{m}=m^*$, is provided in the next result. 

\begin{theorem}\label{theorem_consistency}
Suppose Assumptions \ref{assumption_dgp}-\ref{assumption_regularity} and Assumption \ref{assumption_rates}(i)-(iv) are satisfied. Under $\widetilde{m}=m^*$, then:
\begin{itemize}
    \item[(i)] $\Pb\Big(\underset{1 \leq j \leq m^*}{\max}|\widetilde{T}_j-T^*_j|\leq T \delta_T\Big) \rightarrow 1$ as $T \rightarrow \infty$.
    \item[(ii)] $\|\widetilde{\Sigma}_j-\Sigma^*_j\|_F = O_p\left(\frac{\lambda\,T}{I^*_j}+\frac{T\delta_T}{I^*_j}+p\sqrt{\frac{\log(pT)}{I^*_j}}\right)$, for $j = 1,\ldots,m^*+1$.
\end{itemize}
\end{theorem}
Result (i) implies $\max_{1 \leq j \leq m^*}T^{-1}|\widetilde{T}_j-T^*_j|=o_p(\delta_T)$. Since $\delta_T=o(1)$, this means $T^{-1}|\widetilde{T}_j-T^*_j|=o_p(1)$. Here, $\delta_T$ is a key quantity to control for the rate at which $\widetilde{T}_j/T$ converges to $T^*_j/T$. Note that $\delta_T \gg \max(\frac{\lambda}{\eta_{\min}},p^2\log(pT)/(T\eta_{\min}^2))$, implies that the fastest convergence rate for the change-point ratio estimator depends on the regularization parameter $\lambda/\eta_{\min}$ and $p^2\log(pT)/(T\eta_{\min}^2)$. Result (ii) relates to the consistency of the variance-covariance matrix in each regime. The proof is provided in Appendix \ref{appendix:proofs}.

% \mds

The true number of change points $m^*$ is unknown in practice. Following the common approach in the change-point literature, we assume that $m^*$ is bounded by a known conservative upper bound $m_{\max}$ with $m_{\max} \leq C\log(T)$, $C>0$ large enough. Next, we define $h(A,B)\coloneqq \sup_{b \in B}\inf_{a\in A}|a-b|$ for any two sets $A$ and $B$. The next result establishes that all true change points in $\Tc^*_{m^*}$ can be consistently estimated by some points in $\widetilde{\Tc}_{\widetilde{m}}$.
\begin{theorem}\label{theorem_date_recov}
Suppose Assumptions \ref{assumption_dgp}-\ref{assumption_regularity} and Assumption \ref{assumption_rates}(i)-(iv) are satisfied. If $m^* \leq \widetilde{m} \leq m_{\max}$, then $\Pb\left(h(\widetilde{\Tc}_{\widetilde{m}},\Tc^*_{m^*})\leq T\delta_T\right) \rightarrow 1 \; \text{as} \; T \rightarrow \infty$.
\end{theorem}
The proof of Theorem \ref{theorem_date_recov}, provided in Appendix \ref{appendix:proofs}, is done by contradiction and follows similar arguments as in the proof of Theorem \ref{theorem_consistency}. Theorem \ref{theorem_date_recov} ensures that even if the number of blocks is overestimated, there will be an estimated change point close to each unknown true change point. The next result establishes that the estimated number of change points is not underestimated, that is $\widetilde{m}\geq m^*$ when $T \rightarrow \infty$.

\begin{theorem}\label{theorem_understimation_neg}
Suppose Assumptions \ref{assumption_dgp}-\ref{assumption_rates} are satisfied. Then $\Pb\left(\widetilde{m}<m^*\right) \rightarrow 0$ as $T \rightarrow \infty$.
\end{theorem}
The proof of Theorem \ref{theorem_understimation_neg} is provided in Appendix \ref{appendix:proofs}.
The adaptive LASSO, originally proposed by \citet{zou2006adaptive}, aims to ``adapt'' to the values of the least squares variance-covariance estimator to reduce the bias inherent to LASSO and improve the selection accuracy. Adaptive weights were also proposed in \citet{poignard2020aism} to improve the performances of the Sparse Group LASSO estimator. This motivates the adaptive estimation proposed in (\ref{stat_crit_adaptive}).

\subsection{Adaptive case}\label{subsec:adaptive_change_point_detection}

Under Assumptions \ref{assumption_dgp}-\ref{assumption_rates}, if $m^* \leq \widetilde{m} \leq m_{\max}$ holds, then by Theorem \ref{theorem_consistency}, Theorem \ref{theorem_date_recov} and Theorem \ref{theorem_understimation_neg}, $\forall k =1,\ldots,m^*+1,\; \underset{1 \leq j \leq \widetilde{m}+1,|\widetilde{T}_j-T^*_k|\leq T \delta_T}{\sup} \|\widetilde{\Sigma}_j-\Sigma^*_k\|_F = O_p\left(p\sqrt{\frac{\log(pT)}{\Ic_{\min}}}\right)$.
The adaptive GFlsL defined in (\ref{stat_crit_adaptive}) relies on $\{\widetilde{\Theta}_t\}^T_{t=1}$ and its theoretical properties require that the non-adaptive procedure does not underestimate the number of change points.
Hereafter, we make the following assumption. Define $\alpha_{p,T}=p\sqrt{\log(pT)/\Ic_{\min}}$.
Hereafter, we will work under the following condition.
\begin{assumption}\label{assumption_rates_adaptive}
$\lambda_2\max(\alpha_{p,T},a_T)^{-\mu_2}/(\eta_{\min}\delta_T) \rightarrow 0$ and $\lambda_1Tp\max(\alpha_{p,T},a_T)^{-\mu_1}/\eta_{\min}\rightarrow 0$.
% \begin{itemize}
%     \item[(i)] $\lambda_2\max(\alpha_{p,T},a_T)^{-\mu_2}/(\eta_{\min}\delta_T) \rightarrow 0$ and $\lambda_1Tp\max(\alpha_{p,T},a_T)^{-\mu_1}/\eta_{\min}\rightarrow 0$.
%     \item[(ii)] $\lambda_2T\max(\alpha_{p,T},a_T)^{-\mu_2}=O(p\sqrt{\Ic_{\min}\log(pT)})$ and $\lambda_1T\max(\alpha_{p,T},a_T)^{-\mu_1}=O\Big(\sqrt{\frac{\log(pT)}{\Ic_{\min}}}\Big)$. 
%     \item[(iii)] $m^* s^*_{\max}/(p\Ic_{\min}\eta_{\min}) \rightarrow 0$, and $m^*T\delta_T/(\Ic_{\min}\eta_{\min})\rightarrow 0$.
% \end{itemize}
\end{assumption}
Assumption \ref{assumption_rates_adaptive} modifies Assumption \ref{assumption_rates} due to the introduction of the adaptive weights and the entry-wise penalization of the variance-covariance matrix. 
The LASSO penalty primarily affects the estimation of the covariance structure within each regime and does not directly influence the detection of change points, which explains the absence of link between $\lambda_1$ and $\delta_T$. Assume $\mu_1=\mu_2=1$. Assumption \ref{assumption_rates_adaptive} implies that $\lambda_1 = o(\eta_{\min}(Tp)^{-1}\max(\alpha_{p,T},a_T))$ and $\lambda_2 = o(\eta_{\min}\delta_T \max(\alpha_{p,T},a_T))$. Suppose we set $a_T=\alpha_{p,T}$, then $\lambda_1 = o(\eta_{\min}T^{-1}\sqrt{\log(pT)/\Ic_{\min}})$ and $\lambda_2 = o(\eta_{\min}\delta_Tp \sqrt{\log(pT)/\Ic_{\min}})$: the adaptive weights alter the calibration of $\lambda_2$ compared with $\lambda$ specified in criterion (\ref{stat_crit_1}), which satisfies $\lambda = o(\eta_{\min}\delta_T)$. 
In the next result, we establish the consistency of $\widehat{T}_j$ and $\widehat{\Sigma}_j$ when $\widehat{m}=m^*$.
\begin{theorem}\label{theorem_consistency_adaptive}
Suppose Assumptions \ref{assumption_dgp}-\ref{assumption_rates_adaptive} are satisfied. Under $\widehat{m}=m^*$, then:
\begin{itemize}
    \item[(i)] $\Pb\Big(\underset{1 \leq j \leq m^*}{\max}|\widehat{T}_j-T^*_j|\leq T \delta_T\Big) \rightarrow 1$ as $T \rightarrow \infty$.
    \item[(ii)] $\|\widehat{\Sigma}_j-\Sigma^*_j\|_F = O_p\left(\frac{\lambda_2 \, T }{\max(\alpha_{p,T},a_T)^{\mu_2}I^*_{j}}+\frac{\lambda_1Tp}{\max(\alpha_{p,T},a_T)^{\mu_1}}\Big(1+\frac{T\delta_T}{I^*_{j}}\Big)+\frac{T\delta_T}{I^*_{j}} + p\sqrt{ \frac{\log(pT)}{I^*_{j}}}\right)$, for $j = 1,\ldots,m^*+1$.
\end{itemize}
\end{theorem}
The proof of Theorem \ref{theorem_consistency_adaptive} can be found in Appendix \ref{appendix:proofs}.
As in Theorem \ref{theorem_consistency}, we have $T^{-1}|\widehat{T}_j-T^*_j|=o_p(1)$ as $\delta_T \rightarrow 0$. However, the introduction of the adaptive weights implies $\delta_T \gg \max(\lambda_2\max(\alpha_{p,T},a_T)^{-\mu_2}/\eta_{\min},p^2\log(pT)/(T\eta_{\min}^2))$.
The $\ell_1$-regularization and the presence of adaptive weights introducing additional scaling factors through $\max(\alpha_{p,T},a_T)$ alters the convergence rate obtained in Theorem \ref{theorem_consistency}.
In the same spirit as in Theorem \ref{theorem_date_recov}, the following result shows that the true change points can be consistently recovered by some points in $\widehat{\Tc}_{\widehat{m}}$ by the adaptive estimation, assuming $m^* \leq C\log(T)$, $C>0$.
\begin{theorem}\label{theorem_date_recov_adaptive}
Suppose Assumptions \ref{assumption_dgp}-\ref{assumption_rates_adaptive} hold. Under $m^* \leq \widehat{m} \leq m_{\max}$, then, as $T\rightarrow \infty$, $\Pb\left(h(\widehat{\Tc}_{\widehat{m}},\Tc^*_{m^*})\leq T\delta_T\right) \rightarrow 1$.
\end{theorem}

The next result guarantees that the number of change points is not underestimated by the adaptive estimation. Let $r_{p,T} := \frac{\lambda_2 \, T }{\max(\alpha_{p,T},a_T)^{\mu_2}\Ic_{\min}}+\frac{\lambda_1Tp}{\max(\alpha_{p,T},a_T)^{\mu_1}}\frac{T\delta_T}{\Ic_{\min}}+\frac{T\delta_T}{\Ic_{\min}} + p\sqrt{ \frac{\log(pT)}{\Ic_{\min}}}$.
\begin{theorem}\label{theorem_understimation_neg_adaptive}
Suppose Assumptions \ref{assumption_dgp}-\ref{assumption_rates_adaptive} hold. Under {\small{$m^* s^*_{\max}/(p\Ic_{\min}\eta_{\min}) \rightarrow 0, m^*T\delta_T/(\Ic_{\min}\eta_{\min})\rightarrow 0, \frac{Tm^*}{\Ic_{\min}\eta^2_{\min}}\Big(p\sqrt{\frac{\log(pT)}{T}}r_{p,T}+r^2_{p,T}\Big)=o(1)$}}.
Then $\Pb\left(\widehat{m}<m^*\right) \rightarrow 0 \; \text{as} \; T \rightarrow \infty$.
\end{theorem}
The condition $\frac{Tm^*}{\Ic_{\min}\eta^2_{\min}}\Big[p\sqrt{\frac{\log(pT)}{T}}r_{p,T}+r^2_{p,T}\Big]=o(1)$ is similar to  Assumption \ref{assumption_rates}-(v), but now accounts for the consistency of $\widehat{\Sigma}_j$ derived in Theorem \ref{theorem_consistency_adaptive}. Finally, we consider the correct estimation of the non-zero entries of the
true variance-covariance matrix. In the same spirit as \citet{Kolar2012}, we
show that, for a sufficiently large estimated block $\widehat{\Bc}_j$, the
estimated support $\widehat{\Sc}_j$ coincides with $\Sc(\Theta_k^*)$, where
$\Sc(\Theta_k^*)$ denotes the support of the true parameter $\Theta_k^*$ on
the true block $\Bc_k^*$ that has the largest overlap with
$\widehat{\Bc}_j$. Let $\mathcal{S}^*_j:=\mathcal{S}(\Theta^*_j)$.
\begin{theorem}\label{theorem_recovery}
Assume that the conditions of Theorem \ref{theorem_consistency_adaptive} are satisfied. Moreover, assume $\frac{1}{T}\Big(\sqrt{\frac{\log(pT)}{\Ic_{\min}}}+\frac{T\delta_T}{\Ic_{\min}}\Big)=o(\lambda_1\max(\alpha_{p,T},a_T)^{-\mu_1})$ and $\Ic^{-1}_{\min}\lambda_2\max_t\xi_{2,t}/(\lambda_1\max(\alpha_{p,T},a_T)^{-\mu_1}) \rightarrow 0$.
Assume 
$\min_j \min_{(u,v)\in \Sc^*_j}|\Sigma^*_{uv,j}| \geq L\left(p\sqrt{\frac{\log(pT)}{\Ic_{\min}}}+\frac{T\delta_T}{\Ic_{\min}}+T\lambda_1+\frac{T}{\Ic_{\min}}\lambda_2\max_t\xi_{2,t}\right)$,
with $L>0$ a finite constant.
Then, if $\widehat{m}=m^*$, $\forall j=1,\ldots,m^*+1,\; \Pb\left(\widehat{\Sc}_j = \Sc^*_j\right) \rightarrow 1$.
\end{theorem}
Theorem \ref{theorem_recovery} establishes support recovery on every sufficiently large estimated block of the partition, where each block satisfies $|\widehat{\Bc}_j|\geq \Ic_{\min}-2T\delta_T$. Since $T\delta_T=o(\Ic_{\min})$, this guarantees that each estimated block contains enough observations to recover the sparse support. Theorem \ref{theorem_recovery} provides the trade-off between $\lambda_2$-regularization for the detection of the change points and $\lambda_1$-regularization for sparse covariance estimation: $\lambda_2\max_t\xi_{2,t}=o(\lambda_1\Ic_{\min}\max(\alpha_{p,T},a_T)^{-\mu_1})$. Moreover, the lower bound condition on $\min_j \min_{(u,v)\in \Sc^*_j}|\Sigma^*_{uv,j}|$ is similar to Assumption A.5 in~\citet{Kolar2012}. The proof of Theorem \ref{theorem_recovery} is detailed in Appendix \ref{appendix:proofs}.

\section{Optimization and implementation}\label{sec:opt}
In this section, we develop an alternating direction method of multipliers (ADMM) for solving a practical variant of~\eqref{stat_crit_adaptive} and discuss its implementation details.

\subsection{An alternating direction method of multipliers}\label{subsec:admm}
We consider the following optimization problem
\begin{equation}
	\label{eq:opt-prob-two-stage}
	\min_{\substack{\Theta_t \succeq \epsilon I_p, \\[0.1cm] 1 \leq t \leq T}} \left\{ \sum_{t=1}^T \frac{1}{2T} \| X_t X_t^{\top} - \Theta_t \|_F^2 +\, \lambda_1\sum_{t=1}^T \sum_{u \neq v} \xi_{uv, 1t} | \Theta_{uv, t} | + \lambda_2 \sum_{t=2}^T \xi_{2t} \| \Theta_t - \Theta_{t-1} \|_F \right\},
\end{equation}
where the constraint \( \Theta_t \succeq \epsilon I_p \) for a small \( \epsilon > 0 \) ensures positive definiteness of the solution without affecting the theoretical properties of the estimator.
Setting \( \lambda_1 = 0 \) and \( \lambda_2 \xi_{2t} = \lambda \) for all \( u, v, t \) reduces~\eqref{eq:opt-prob-two-stage} to~\eqref{stat_crit_1} with a relaxed positive semidefiniteness constraint, so the same algorithm applies to both formulations.

We introduce auxiliary variables \( V_t = \Theta_t \), \( \Upsilon_{t, \mathrm{off}} = \Theta_{t, \mathrm{off}} \), and \( D_t = \Theta_t - \Theta_{t-1} \) to decouple the objective terms, yielding the following constrained problem equivalent to~\eqref{eq:opt-prob-two-stage}
\begin{align}
  \min_{\mathbf{X}} \quad & \left\{ \sum_{t=1}^T \left[ \frac{1}{2T} \| X_t X_t^{\top} - \Theta_t \|_F^2 + \delta_{\cdot \succeq \epsilon I_p}(V_t) \right] + \lambda_1 \sum_{t=1}^T \sum_{u \neq v} \xi_{uv, 1t} | \Upsilon_{uv, t} | + \lambda_2 \sum_{t=2}^T \xi_{2t} \| D_t \|_F \right\}, \notag\\[0.1cm]
  \text{s.t.} \quad & V_t = \Theta_t, \Upsilon_{t, \mathrm{off}} = \Theta_{t, \mathrm{off}} \,\, \forall \, t = 1, \dots , T; \quad D_t = \Theta_t - \Theta_{t-1}\,\, \forall \, t = 2, \dots ,T, \label{eq:opt-prob-two-stage-constrained}
\end{align}
where \( \mathbf{X} \coloneqq \left\{ \{\Theta_t\}_{t=1}^T, \{V_t\}_{t=1}^T, \{\Upsilon_{t, \mathrm{off}} \}_{t=1}^T, \{D_t\}_{t=2}^T \right\} \); \( \Theta_{t, \mathrm{off}} \) denotes the copy of \( \Theta_t \) with diagonal entries set to zero, and \( \Upsilon_{t, \mathrm{off}} \in \mathcal{S}_{\mathrm{off}}^p \).

We next develop an ADMM to solve~\eqref{eq:opt-prob-two-stage-constrained}.
Let \( f(\mathbf{X}) \) denote the objective function in~\eqref{eq:opt-prob-two-stage-constrained}.
The augmented Lagrangian with block penalty parameters \( \boldsymbol{\beta} \coloneqq (\beta_V, \beta_\Upsilon, \beta_D) \) is
\begin{align}
  & \mathcal{L}_{\boldsymbol{\beta}}(\mathbf{X}; \boldsymbol{\Lambda}) \coloneqq  f(\mathbf{X}) - \sum_{t=1}^T \langle A_t, \Theta_t - V_t \rangle - \sum_{t=1}^T \langle Y_{t, \mathrm{off}}, \Theta_{t, \mathrm{off}} - \Upsilon_{t, \mathrm{off}} \rangle - \sum_{t=2}^T \langle Z_t, \Theta_t - \Theta_{t-1} - D_t \rangle \notag\\
  & \qquad + \frac{\beta_V}{2} \sum_{t=1}^T \| \Theta_t - V_t \|_F^2 + \frac{\beta_\Upsilon}{2} \sum_{t=1}^T \| \Theta_{t, \mathrm{off}} - \Upsilon_{t, \mathrm{off}} \|_F^2 + \frac{\beta_D}{2} \sum_{t=2}^T \| \Theta_t - \Theta_{t-1} - D_t \|_F^2, \label{eq:aug-lag}
\end{align}
where \( \boldsymbol{\Lambda} \coloneqq \left\{ \{A_t\}_{t=1}^T, \{Y_{t, \mathrm{off}}\}_{t=1}^T, \{Z_t\}_{t=2}^T \right\} \) collects the multipliers associated with the constraints \( V_t = \Theta_t \), \( \Upsilon_{t, \mathrm{off}} = \Theta_{t, \mathrm{off}} \), and \( D_t = \Theta_t - \Theta_{t-1} \), respectively; \( \beta_V, \beta_\Upsilon, \beta_D > 0 \) weight the quadratic augmentation of the three constraint blocks.
Setting \( \beta_V = \beta_\Upsilon = \beta_D = \beta \) recovers the standard augmented Lagrangian with a single penalty parameter; the general triplet preconditions the iteration without affecting the estimator or its convergence properties, as discussed in the remark following Algorithm~\ref{algo:ADMM}.
Let \( (\mathbf{X}^k, \boldsymbol{\Lambda}^k) \) denote the \( k \)-th primal-dual ADMM iterate, with \( k=0 \) corresponding to the initialization. For notational convenience, we set \( Z_1^k = Z_{T+1}^k = D_1^k = D_{T+1}^k = \mathbf{0}_{p\times p} \) for all \( k \).
The ADMM is initialized with \( \Theta_t^0 = X_tX_t^{\top} \), \( V_t^0 = \mathrm{Proj}_{\cdot \succeq \epsilon I_p}(\Theta_t^0) \), \( \Upsilon_{t, \mathrm{off}}^0 = \Theta_{t, \mathrm{off}}^0 \), and \( D_t^0 = \Theta_t^0 - \Theta_{t-1}^0 \) for \( t = 2, \dots, T \), while \( A_t^0 = Y_{t, \mathrm{off}}^0 = \mathbf{0}_{p \times p} \) and \( Z_t^0 = \mathbf{0}_{p \times p} \).
This choice is natural because \( \Theta_t^0 = X_tX_t^{\top} \) is the unpenalized least-squares fit, \( V_t^0 \) is its projected feasible version, \( \Upsilon_{t, \mathrm{off}}^0 \) and \( D_t^0 \) are the induced copy and first-difference variables, and zero multipliers provide a neutral starting point before the equality constraints are enforced by the iterations.
With these initial values in hand, iteration \( k+1 \) consists of three successive updates:
\begin{enumerate}[\normalfont(i)]
	\item \textbf{\( \{\Theta_t\}_{t=1}^T \) update:} minimize~\eqref{eq:aug-lag} jointly over \( \{\Theta_t\}_{t=1}^T \) with all other primal and dual variables fixed at their current values, yielding a linear system;
	\item \textbf{\( \{V_t, \Upsilon_{t, \mathrm{off}}, D_t\} \) update:} minimize~\eqref{eq:aug-lag} over \( (V_t, \Upsilon_{t, \mathrm{off}}, D_t) \) for each \( t \) with \( \{\Theta_t^{k+1}\} \) and all dual variables fixed; this subproblem decomposes into three independent optimization problems that admit closed-form solutions;
	\item \textbf{Dual ascent:} update the dual variables with stepsize \( \gamma = 1.61 \in (0, (\sqrt{5}+1)/2) \), a range permitted by the convergence theory in~\cite[Appendix~B]{FPST2013}:
    \begin{align}
            A_t ^{k+1} & = A_t ^k - \gamma \beta_V (\Theta_t^{k+1} - V_t^{k+1}),                                                                        \notag\\
            Y_{t, \mathrm{off}} ^{k+1} & = Y_{t, \mathrm{off}} ^k - \gamma \beta_\Upsilon (\Theta_{t, \mathrm{off}}^{k+1} - \Upsilon_{t, \mathrm{off}}^{k+1}), \label{eq:dual-update}\\
            Z_t ^{k+1} & = Z_t ^k - \gamma \beta_D (\Theta_t^{k+1} - \Theta_{t-1}^{k+1} - D_t^{k+1}). \notag
    \end{align}
\end{enumerate}

We first discuss Step~(i).
Setting to zero the derivative of the augmented Lagrangian with respect to each \( \Theta_t \) yields
\[
  \begin{aligned}
	\tfrac{\Theta_t - X_tX_t^{\top}}{T} + \beta_V \Theta_t + & \beta_\Upsilon \Theta_{t, \mathrm{off}} + \beta_D (2\Theta_t - \Theta_{t+1} - \Theta_{t-1}) = \\
    & A_t^k + Y_{t, \mathrm{off}}^k - Z_{t+1}^k + Z_t^k + \beta_V V_t^k + \beta_\Upsilon \Upsilon_{t, \mathrm{off}}^k - \beta_D D_{t+1}^k + \beta_D D_t^k
  \end{aligned}
\]
for all \( t = 1, \dots, T \). By denoting the right-hand side by
\[
	\Psi_t^k \coloneqq \frac{X_tX_t^{\top}}{T} + A_t^k + Y_{t, \mathrm{off}}^k - Z_{t+1}^k + Z_t^k + \beta_V V_t^k + \beta_\Upsilon \Upsilon_{t, \mathrm{off}}^k - \beta_D D_{t+1}^k + \beta_D D_t^k,
\]
the system reduces to the following linear system
\begin{align*}
	(\beta_V + \beta_D + \tfrac{1}{T}) \Theta_1 + \beta_\Upsilon \Theta_{1, \mathrm{off}} - \beta_D \Theta_2                                & = \Psi_1^k,     \\[.01cm]
		- \beta_D \Theta_1 + ( \beta_V + 2 \beta_D + \tfrac{1}{T} ) \Theta_2 + \beta_\Upsilon \Theta_{2, \mathrm{off}} - \beta_D \Theta_3           & = \Psi_2^k,     \\[.01cm]
		\dots                                                                                                  &                 \\[.01cm]
		- \beta_D \Theta_{T-2} + ( \beta_V + 2 \beta_D + \tfrac{1}{T} ) \Theta_{T-1} + \beta_\Upsilon \Theta_{T-1, \mathrm{off}} - \beta_D \Theta_T & = \Psi_{T-1}^k, \\[.01cm]
		- \beta_D \Theta_{T-1} + (\beta_V + \beta_D + \tfrac{1}{T}) \Theta_T + \beta_\Upsilon \Theta_{T, \mathrm{off}}                          & = \Psi_T^k.
\end{align*}
Notice that this system is coordinate-decomposable: exploiting the symmetry of each \( \Theta_t \), it decouples into \( \frac{p(p+1)}{2} \) independent tridiagonal subsystems.
For each pair \( (u, v) \) with \( 1 \leq u \leq v \leq p \), the subsystem reads
\begin{equation}
	\label{eq:linear-subsystem}
	\begin{bmatrix}
		b_{uv} & -\beta_D &        &        &        \\
		-\beta_D & c_{uv} & -\beta_D &        &        \\
		       &        & \dots  &        &        \\
		       &        & -\beta_D & c_{uv} & -\beta_D \\
		       &        &        & -\beta_D & b_{uv}
	\end{bmatrix}
	\begin{bmatrix}
		\Theta_{uv, 1}   \\
		\Theta_{uv, 2}   \\
		\dots            \\
		\Theta_{uv, T-1} \\
		\Theta_{uv, T}
	\end{bmatrix} =
	\begin{bmatrix}
		\Psi_{uv, 1}^k   \\
		\Psi_{uv, 2}^k   \\
		\dots            \\
		\Psi_{uv, T-1}^k \\
		\Psi_{uv, T}^k
	\end{bmatrix},
\end{equation}
where \( b_{uv} = \beta_V + \beta_\Upsilon + \beta_D + 1 / T \) and \( c_{uv} = \beta_V + \beta_\Upsilon + 2\beta_D + 1 / T \) if \( u \neq v \), otherwise \( b_{uu} = \beta_V + \beta_D + 1 / T \), \( c_{uu} = \beta_V + 2\beta_D + 1 / T \).
Each tridiagonal system can be solved in \( O(T) \) operations via the Thomas algorithm, while we solve it through sparse diagonal matrices and backslash (\textsf{\(\backslash\)}) in \textsc{Matlab}.

We next turn to Step~(ii).
Since the augmented Lagrangian is separable across \( V_t \), \( \Upsilon_{t, \mathrm{off}} \), and \( D_t \) when \( \{\Theta_t^{k+1}\} \) is fixed, this step decomposes into three independent subproblems, each admitting a closed-form solution.

For \( V_t \), the subproblem is \( \min_{V_t \succeq \epsilon I_p} \frac{\beta_V}{2}\| V_t - \Theta_t^{k+1} \|_F^2 - \langle A_t^k, V_t - \Theta_t^{k+1} \rangle \).
Completing the square in \( V_t \) reduces this to the Euclidean projection onto \( \{S \in \mathcal{S}^p \mid S \succeq \epsilon I_p\} \):
\begin{equation}
	\label{eq:V-update}
	V_t^{k+1} = \mathrm{Proj}_{\cdot \succeq \epsilon I_p}\!\left( \Theta_t^{k+1} - A_t^k / \beta_V \right),
\end{equation}
which is computed via eigendecomposition and thresholding of eigenvalues at \( \epsilon \).

For \( \Upsilon_{t, \mathrm{off}} \), each off-diagonal entry \( (u, v) \) with \( u \neq v \) is updated by solving a one-dimensional \( \ell_1 \)-penalized least-squares problem, whose solution is the soft-thresholding operator \( \mathcal{T}_{\tau}(x) \coloneqq \mathrm{sgn}(x)(|x| - \tau)_+ \), where \( (\cdot)_+ \coloneqq \max \{\cdot , 0\} \):
\begin{equation}
	\label{eq:Up-update}
	\Upsilon_{uv, t}^{k+1} = \mathcal{T}_{\lambda_1 \xi_{uv, 1t} / \beta_\Upsilon}\!\left( \Theta_{uv, t}^{k+1} - Y_{uv, t}^k / \beta_\Upsilon \right).
\end{equation}

For \( D_t \), letting \( \Xi_t^k \coloneqq \Theta_t^{k+1} - \Theta_{t-1}^{k+1} - Z_t^k / \beta_D \), the subproblem amounts to minimizing \( \frac{\lambda_2 \xi_{2t}}{\beta_D} \| D_t \|_F + \frac{1}{2}\| D_t - \Xi_t^k \|_F^2 \). By the Cauchy--Schwarz inequality, the solution is aligned with \( \Xi_t^k \), reducing the problem to a one-dimensional group shrinkage:
\begin{equation}
	\label{eq:D-update}
	D_t^{k+1} = \left( 1 - \frac{\lambda_2 \xi_{2t} / \beta_D}{\| \Xi_t^k \|_F} \right)_{\!+} \Xi_t^k.
\end{equation}

The ADMM is summarized in Algorithm~\ref{algo:ADMM}.
Its convergence to a primal-dual optimal pair is guaranteed by~\cite[Theorem~B.1]{FPST2013}.
Moreover, since the least-squares loss is \(1/T\)-strongly convex in \( \Theta \), the minimizer \( \Theta^* \) of~\eqref{eq:opt-prob-two-stage} is unique.
Consequently, the primal solution of~\eqref{eq:opt-prob-two-stage-constrained}, and with it the estimator returned by Algorithm~\ref{algo:ADMM}, is uniquely determined.

\begin{remark}[Block preconditioning]
  The triplet \( \boldsymbol{\beta} \) plays the role of a fixed diagonal preconditioner.
  Although~\cite[Theorem~B.1]{FPST2013} does not explicitly cover a preconditioner of this kind, it still applies here, because the block penalties amount to a rescaling of the constraint operators \( A \) and \( B \) of~\cite[eq.~(B.1)]{FPST2013}, on which the theorem imposes no restriction.
  Specifically, an equality constraint may be multiplied by any positive constant without changing the feasible set, whereas the quadratic penalty of the rescaled constraint is scaled by the square of that constant: multiplying the \( V \)-, \( \Upsilon \)-, and \( D \)-blocks of~\eqref{eq:opt-prob-two-stage-constrained} by \( c_V = (\beta_V/\rho)^{1/2} \), \( c_\Upsilon = (\beta_\Upsilon/\rho)^{1/2} \), and \( c_D = (\beta_D/\rho)^{1/2} \), respectively, for any common \( \rho > 0 \), makes a single penalty \( \rho \) produce exactly the block penalties \( \rho c_V^2 = \beta_V \), \( \rho c_\Upsilon^2 = \beta_\Upsilon \), and \( \rho c_D^2 = \beta_D \).
  The iteration in Algorithm~\ref{algo:ADMM} is thus a special case of the update scheme of the proximal ADMM in~\cite[eq.~(B.2)]{FPST2013} in which both proximal terms are set to zero, so the step-length range \( \gamma \in (0, (\sqrt{5}+1)/2) \) and the convergence guarantee recalled above apply unchanged.
  Consequently, neither the estimator nor the convergence guarantee depends on the triplet; only the convergence speed does, which we exploit to accelerate the high-dimensional real-data experiment in Section~\ref{sec:real_data}.
\end{remark}

Algorithm~\ref{algo:ADMM} stops when a termination criterion is satisfied; the version used in our implementation is discussed in Subsection~\ref{subsec:convergence-termination}, and the two-stage adaptive procedure is described in Algorithm~\ref{algo:adaptive}.

\begin{algorithm}[t]
	\caption{ADMM for solving \eqref{eq:opt-prob-two-stage}}\label{algo:ADMM}
		\begin{algorithmic}[1]
			\onehalfspacing
			\Statex \hspace{-0.5cm}\textbf{input:} data $ \mathcal{X}_T = (X_1, \dots, X_T) $; parameters $ \epsilon > 0 $, $ \beta_V, \beta_\Upsilon, \beta_D > 0 $, \( \lambda_1 \geq 0 \), \( \lambda_2 > 0 \); weights \( \xi_{uv, 1t} \geq 0 \), \( \xi_{2t} > 0 \).
			\Statex \hspace{-0.5cm}\textbf{initialize:} for \( t = 1, \dots, T \), set \( \Theta_t^0 = X_tX_t^{\top} \), \( V_t^0 = \mathrm{Proj}_{\cdot \succeq \epsilon I_p}(\Theta_t^0) \), \( \Upsilon_{t, \mathrm{off}}^0 = \Theta_{t, \mathrm{off}}^0 \),
			\Statex \hspace{1.15cm}\( A_t^0 = Y_{t, \mathrm{off}}^0 = \mathbf{0}_{p \times p} \); for \( t = 2, \dots, T \), set \( D_t^0 = \Theta_t^0 - \Theta_{t-1}^0 \), \( Z_t^0 = \mathbf{0}_{p \times p} \).

		\For{$ k = 0, 1, 2, \dots $}
		\State Solve $\frac{p(p+1)}{2}$ tridiagonal systems~\eqref{eq:linear-subsystem} to obtain $ \{\Theta_t^{k+1}\}_{t=1}^T $.
		\State Update $ \{V_t^{k+1}\}_{t=1}^T $, $ \{\Upsilon_{t, \mathrm{off}}^{k+1}\}_{t=1}^T $, $ \{D_t^{k+1}\}_{t=2}^T $ via~\eqref{eq:V-update},~\eqref{eq:Up-update},~\eqref{eq:D-update}.
		\State Perform $\gamma$-scaled dual ascent~\eqref{eq:dual-update} to obtain $ \{A_t^{k+1}\}_{t=1}^T $, $ \{Y_{t, \mathrm{off}}^{k+1}\}_{t=1}^T $, $ \{Z_t^{k+1}\}_{t=2}^T $.
			\State \textbf{if} a termination criterion is met \textbf{then} \textbf{return} all primal and dual variables.
			\EndFor
		\end{algorithmic}
	\end{algorithm}

\begin{algorithm}[t]
	\caption{Two-stage adaptive change-point detection}\label{algo:adaptive}
	\begin{algorithmic}[1]
		\onehalfspacing
		\Statex \hspace{-0.5cm}\textbf{input:} $ \mathcal{X}_T = (X_1, \dots, X_T) $; parameters $ \epsilon > 0 $, $ \beta_V, \beta_\Upsilon, \beta_D > 0 $, \( \lambda > 0 \), \( a_T > 0 \), \( \mu_1, \mu_2 > 0 \), \( \lambda_1 > 0 \), \( \lambda_2 > 0 \).

			\State Apply Algorithm~\ref{algo:ADMM} to solve~\eqref{eq:opt-prob-two-stage} with \( \lambda_1 = 0 \), \( \lambda_2 = \lambda \), and \( \xi_{2t} = 1 \) for all \( t \), and obtain the preliminary estimate \( \{\widetilde{\Theta}_t\}_{t=1}^T \).
    \State Compute adaptive weights \( \tilde{\xi}_{uv, 1t} \) and \( \tilde{\xi}_{2t} \) via~\eqref{eq:adaptive-weight}.
    \State Apply Algorithm~\ref{algo:ADMM} to solve~\eqref{eq:opt-prob-two-stage} with \( \lambda_1 \), \( \lambda_2 \), \( \xi_{uv, 1t} = \tilde{\xi}_{uv, 1t} \), \( \xi_{2t} = \tilde{\xi}_{2t} \) and return \( \{\widehat{\Theta}_t\}_{t=1}^T \).
	\end{algorithmic}
\end{algorithm}

\subsection{Termination criteria}\label{subsec:convergence-termination}
We describe the termination criteria used in our implementation.
The algorithm is terminated by monitoring dual feasibility and the primal-dual gap derived from the dual problem in Proposition~\ref{prop:dual}.
Let \( \mathbf{Y} \coloneqq \left\{ \{W_t\}_{t=1}^T, \{Y_{t, \mathrm{off}}\}_{t=1}^T, \{Z_t\}_{t=2}^T \right\} \), where \( W_t \in \mathcal{S}^p \), \( Y_{t, \mathrm{off}} \in \mathcal{S}_{\mathrm{off}}^p \), and \( Z_t \in \mathcal{S}^p \), and define \( \Delta_t \coloneqq Z_{t+1} - Z_t + W_t - Y_{t, \mathrm{off}} \) with \( Z_1 = Z_{T+1} = \mathbf{0}_{p \times p} \).
Here \( W_t \) is the dual variable associated with the auxiliary relation \( U_t = \Theta_t - X_tX_t^{\top} \) introduced in the proof of Proposition~\ref{prop:dual}, and thus corresponds to the least-squares block of the primal problem.
The dual derivation further shows that the \( U_t \)-block satisfies \( W_t = U_t / T \); since \( U_t^k = \Theta_t^k - X_tX_t^{\top} \) at iteration \( k \), we set \( W_t^k = (\Theta_t^k - X_tX_t^{\top}) / T \) and \( \mathbf{Y}^k \coloneqq \left\{ \{W_t^k\}_{t=1}^T, \{Y_{t, \mathrm{off}}^k\}_{t=1}^T, \{Z_t^k\}_{t=2}^T \right\} \).
We then compute
\begin{align*}
	\mathsf{dfeas}_1(\mathbf{Y}) & \coloneqq \max_{1 \leq t \leq T} \frac{| \min \{ \lambda_{\min}(\Delta_t), 0 \} |}{1 + \| \Delta_t \|_F}, &
	\mathsf{dfeas}_2(\mathbf{Y}) & \coloneqq \frac{\max_{2 \leq t \leq T} ( \| Z_t \|_F - \lambda_2 \xi_{2t} )_+}{1 + \max_{2 \leq t \leq T} \| Z_t \|_F}, \\
	\mathsf{dfeas}_3(\mathbf{Y}) & \coloneqq \frac{\max_{\substack{1 \leq t \leq T \\ u \neq v}} ( | Y_{uv, t} | - \lambda_1 \xi_{uv, 1t} )_+}{1 + \max_{\substack{1 \leq t \leq T \\ u \neq v}} | Y_{uv, t} |}, &
	\mathsf{dfeas}(\mathbf{Y})   & \coloneqq \max_{i=1,2,3} \mathsf{dfeas}_i(\mathbf{Y}),
\end{align*}
and the relative duality gap $\mathsf{gap}(\mathbf{X}, \mathbf{Y}) \coloneqq \frac{| v_p(\mathbf{X}) - v_d(\mathbf{Y}) |}{1 + | v_p(\mathbf{X}) | + | v_d(\mathbf{Y}) |}$,
where \( v_p(\mathbf{X}) \) and \( v_d(\mathbf{Y}) \) are the primal~\eqref{eq:opt-prob-two-stage} and dual~\eqref{eq:dual-prob} objective values. For any matrix sequence \( \{M_t\}_{t \in \mathcal{T}} \), let
$\| \{M_t\}_{t \in \mathcal{T}} \| \coloneqq \left( \sum_{t \in \mathcal{T}} \| M_t \|_F^2 \right)^{1/2}$.
Given tolerance \( \epsilon_{\mathrm{tol}} > 0 \), Algorithm~\ref{algo:ADMM} stops at iteration \( k \) as soon as either
\begin{enumerate}[\normalfont(T1)]
	\item\label{cond:T1} the current iterate is already sufficiently close to primal-dual optimality: \\$\max \left\{ \mathsf{gap}(\mathbf{X}^k, \mathbf{Y}^k), \mathsf{dfeas}(\mathbf{Y}^k) \right\} \leq \epsilon_{\mathrm{tol}}$;
	\item\label{cond:T2} as a safeguard against numerical stagnation, the relative successive changes of all primal and dual blocks satisfy
	      \begin{align*}
		      \max \Big\{ & \frac{\| \{\Theta_t^{k+1} - \Theta_t^k\}_{t=1}^T \|}{1 + \| \{\Theta_t^{k+1}\}_{t=1}^T \| + \| \{\Theta_t^k\}_{t=1}^T \|}, \frac{\| \{A_t^{k+1} - A_t^k\}_{t=1}^T \|}{1 + \| \{A_t^{k+1}\}_{t=1}^T \| + \| \{A_t^k\}_{t=1}^T \|},                                                                                              \\
		                  & \frac{\| \{Y_{t, \mathrm{off}}^{k+1} - Y_{t, \mathrm{off}}^k\}_{t=1}^T \|}{1 + \| \{Y_{t, \mathrm{off}}^{k+1}\}_{t=1}^T \| + \| \{Y_{t, \mathrm{off}}^k\}_{t=1}^T \|}, \frac{\| \{Z_t^{k+1} - Z_t^k\}_{t=2}^T \|}{1 + \| \{Z_t^{k+1}\}_{t=2}^T \| + \| \{Z_t^k\}_{t=2}^T \|} \Big\} \leq \frac{\epsilon_{\mathrm{tol}}}{10^3}.
	      \end{align*}
\end{enumerate}

\subsection{Tuning parameter selection}\label{subsec:tuning-para}
Selecting the combination of tuning parameters \(\mu_1, \mu_2, \lambda, \lambda_1,\lambda_2\) for the adaptive GFlsL estimator is a challenging task, as the adaptive weights introduce an extra layer of complexity and are sensitive to the preliminary estimation.
While information criteria such as BIC~\citep{Monti2014} provide a natural starting point, they often fail to capture the tradeoff between fit, sparsity, and temporal smoothness in high-dimensional scenarios: see Section 4.2 and Appendix of~\citet{Gibberd2017} and Section 6.2 of~\citet{LPPT24GFDtL}.
We describe the following seven criteria for choosing the optimal combination.
Let \( \{\widehat{\Theta}_t\}_{t=1}^T \) denote an estimator throughout, and recall that \( \widehat{m} \) and \( \widehat{T}_1 < \cdots < \widehat{T}_{\widehat{m}} \) are its estimated number of change points and estimated change-point locations, respectively.
\begin{itemize}
	\item[\textbf{optimal}.] This method was proposed in~\citet{LPPT24GFDtL}.
	      Assuming the data generating process is known in advance, we select the tuning parameters that minimize or maximize the performance measures described in later sections; this requires the ground truth, but shows how the estimator performs in an ideal setting.

	\item[\textbf{lossval}.] This strategy also comes from~\citet{LPPT24GFDtL}.
	      Given \( B > 1 \) samples \( \{\mathcal{X}_{i, T}\}_{i=1}^B \), we estimate \( \{\widehat{\Theta}_t\}_{t=1}^T \) from one of them, say \( \mathcal{X}_{1, T} \), and select the tuning parameters that minimize the mean loss on the remaining samples,
	      \begin{equation}\label{eq:def-lossval}
		      \text{lossval}(\mu_1, \mu_2, \lambda, \lambda_1, \lambda_2) \coloneqq \frac{1}{B} \sum_{i=2}^B \Lb\left(\{\widehat{\Theta}_t\}_{t=1}^T, \mathcal{X}_{i, T} \right).
	      \end{equation}

	\item[\textbf{BIC}.]
        We adopt a standard BIC-style criterion following the approach of~\citet{Gibberd2017}.
        Specifically, we select the combination of tuning parameters that minimizes
	      \begin{equation}\label{eq:def-BIC}
		      \text{BIC}(\mu_1, \mu_2, \lambda, \lambda_1, \lambda_2) \coloneqq \Lb(\{\widehat{\Theta}_t\}_{t=1}^T, \mathcal{X}_T) + p\log(T)K,
	      \end{equation}
	      where $K$ denotes the model dimension.
	      We use the same notion of model dimension across all information criteria considered below.
	      Writing $\tau_j$ for the first time index of the $j$-th estimated segment, i.e., $\tau_1 = 1$ and $\tau_j = \widehat{T}_{j-1} + 1$ for $j \geq 2$, let
	      \begin{equation}\label{eq:def-K}
		      K \coloneqq \sum_{j=1}^{\widehat{m}+1} K_j, \qquad K_j \coloneqq \mathrm{card}\big(\mathbf{1}(\widehat\Theta_{uv,\tau_j} \neq 0), \forall u \leq v\big),
	      \end{equation}
	      so that $K_j$ is the number of active parameters in the $j$-th segment.

	\item[\textbf{HBIC}.]
	      The BIC in~\eqref{eq:def-BIC} penalizes the model complexity \( K \) using only \( \log (T) \), assuming the dimension of variable is fixed or growing slowly.
	      However, in our cases, the dimension of variable \( \frac{Tp(p+1)}{2} \) grows polynomially in \( T \), then BIC tends to under-penalize the degrees of freedom, causing overfitting or inconsistent detection of change points.
	      %This behavior is also confirmed via empirical analysis in~\cite[Section 6.2]{LPPT24GFDtL}.
	      To handle this issue, we employ the high-dimensional BIC (\textbf{HBIC})~\citep{WKL2013} that scales the penalty depending on the variable dimension.
	      In particular, our HBIC contains three components: the least-squares loss, a high-dimensional penalty on the within-segment model dimension, and a BIC-type penalty on the number of estimated change points. The two penalty terms are motivated by existing HBIC-type criteria for high-dimensional model selection; see, for example,~\citep{CC2008,WKL2013,ZS2012}. Our HBIC is defined as follows:
	      \begin{equation}
		      \label{eq:def-HBIC}
		      \hspace{-1.5cm}
		      \text{HBIC}( \mu_1, \mu_2, \lambda, \lambda_1, \lambda_2) \coloneqq \underbrace{\Lb(\{\widehat{\Theta}_t\}_{t=1}^T, \mathcal{X}_T)}_{\text{fit}} + \underbrace{\frac{\log(p) \log (T)}{T} K}_{\text{within-segment complexity}} + \underbrace{\frac{\log (T)p}{T} \hat{m}}_{\text{model complexity}},
	      \end{equation}
	      where \( K \) is the model dimension in~\eqref{eq:def-K}; the two penalty scales are \( C_T \coloneqq \log(p)\log(T)/T \) for the within-segment complexity and \( \alpha_T \coloneqq \log(T)p/T \) for the break count.
	      Unlike BIC, which uses a single global penalty, HBIC separately penalizes the within-segment model dimension and the number of detected change points. This separation provides additional control over overfitting in high-dimensional or highly dynamic settings.
	      More generally, for penalty weights \( \gamma, b > 0 \) we write
	      \begin{equation}\label{eq:def-HBIC-family}
		      \mathrm{HBIC}(\gamma, b) \coloneqq \Lb(\{\widehat{\Theta}_t\}_{t=1}^T, \mathcal{X}_T) + \gamma C_T K + b \alpha_T \widehat{m},
	      \end{equation}
	      so that \( \mathrm{HBIC} = \mathrm{HBIC}(1,1) \); the weights \( \gamma \) and \( b \) govern the relative penalties on the within-segment dimension and on the break count.

	\item[\textbf{HBIC$_c$}.]
	      The break-count penalty in HBIC can also be calibrated from the set of candidate solutions.
	      Among all evaluated candidates, let \( (\widehat{m}_{\min}, K_{\min}, \Lb_{\min}) \) denote the break count, model dimension, and loss of a candidate with the fewest breaks, and let \( (\widehat{m}_{*}, K_{*}, \Lb_{*}) \) denote those of the candidate with the smallest loss.
	      We estimate the marginal improvement in fit associated with an additional break, net of the within-segment penalty and measured in units of \( \alpha_T \), by \( \widehat{c}_T \coloneqq \big(\Lb_{\min} - \Lb_{*} - C_T (K_{*} - K_{\min})\big) / \big(\alpha_T (\widehat{m}_{*} - \widehat{m}_{\min})\big) \), and define the calibrated criterion \( \mathrm{HBIC}_c \), the instance of~\eqref{eq:def-HBIC-family} with \( (\gamma, b) = (1, 2\widehat{c}_T) \).
	      The factor of two makes the calibrated break penalty exceed the estimated marginal benefit of an additional break and therefore favors more parsimonious segmentations when signals are weak or diffuse.
	      If all candidates share the same break count, \( \widehat{c}_T \) is undefined and we use the baseline HBIC criterion.

	\item[\textbf{HBICG}.]
	      For approximately Gaussian data, the least-squares fit term is replaced by the Gaussian negative log-likelihood
	      \( \Lb_{\text{Gauss}}(\{\widehat{\Theta}_t\}; \mathcal{X}_T) = T^{-1}\sum_t [\log\det(\widehat{\Theta}_t) + \tr(X_tX_t^{\top}\widehat{\Theta}_t^{-1})] \).
        The resulting criterion, denoted \textbf{HBICG}, combines the Gaussian fit measure with the same two penalty components as HBIC.
        Analogously to~\eqref{eq:def-HBIC-family}, we write
	      \begin{equation}\label{eq:def-HBICG}
		      \mathrm{HBICG}(\gamma, b) \coloneqq \Lb_{\text{Gauss}}(\{\widehat{\Theta}_t\}_{t=1}^T; \mathcal{X}_T) + \gamma C_T K + b \alpha_T \widehat{m},
	      \end{equation}
	      so that \( \mathrm{HBICG} = \mathrm{HBICG}(1,1) \).

	\item[\textbf{HBICG$(0,2)$}.]
	      The simulation results indicate that HBICG can be conservative in selecting change points.
	      We therefore also report the rebalanced criterion \( \mathrm{HBICG}(0,2) \), the instance of~\eqref{eq:def-HBICG} with \( (\gamma, b) = (0,2) \), which removes the within-segment dimension penalty and doubles the break-count penalty.
	      Section~\ref{subsec:exp} shows that this specification provides a useful benchmark when the baseline HBICG criterion under-selects change points.
\end{itemize}

To save computational resources, rather than selecting \(\mu_1, \mu_2 \) via cross-validation, we fix their values, and the threshold floor \( a_T \), prior to running the algorithm.
Given \( a_T, \mu_1, \mu_2\), we select the optimal combination of tuning parameters over given \( \lambda \in \{\lambda^{(i)}\}_{i\geq 0} \), and \( (\lambda_1, \lambda_2) \in \{(\lambda_1^{(j)}, \lambda_2^{(j)})\}_{j\geq 0} \) as follows.
We first solve~\eqref{eq:opt-prob-two-stage} with \( \lambda_1 = 0 \), \( \lambda_2 = \lambda^{(i)} \), and \( \xi_{2t} = 1 \) for all \( t \) to obtain \( \{\{\widetilde{\Theta}_t^{(i)}\}_{t=1}^T\}_{i\geq 0} \), traversing the grid in increasing order of \( \lambda^{(i)} \) and applying two shortcuts.
If the smallest candidate already yields no change point, the grid is extended downward, at the mesh size of the grid, until at least one preliminary change point is found: a constant preliminary path would render the adaptive weights uninformative, since it would imply \( \tilde\xi_{2t} \equiv a_T^{-\mu_2} \) for all \( t \).
Conversely, once a grid value produces no preliminary change point, the traversal stops and all larger values are skipped.
Then for each evaluated \( \lambda^{(i)} \), we compute \( \{\xi_{uv, 1t}^{(i)}\}_{u,v=1,t=1}^{p,p,T} \) and \( \{\xi_{2t}^{(i)}\}_{t=1}^T \) using~\eqref{eq:adaptive-weight} and \( a_T, \mu_1, \mu_2 \).
Next, we solve~\eqref{eq:opt-prob-two-stage} with \( \{\xi_{uv, 1t}^{(i)}\}_{u,v=1,t=1}^{p,p,T} \), \( \{\xi_{2t}^{(i)}\}_{t=1}^T \), and \( (\lambda_1^{(j)}, \lambda_2^{(j)}) \) to obtain \( \{\{\widehat{\Theta}_t^{(ij)}\}_{t=1}^T\}_{i,j\geq 0} \).
Finally, using \( \{\{\widehat{\Theta}_t^{(ij)}\}_{t=1}^T\}_{i,j\geq 0} \) we select the optimal combination of tuning parameters via the above seven criteria.

\subsection{Implementation details}\label{subsec:implementation}
In all experiments, we set \( \epsilon = 0.01 \), \( \epsilon_{\mathrm{tol}} = 10^{-3} \), \( a_T = T^{-1/2} \), and \( (\mu_1, \mu_2) = (0.8, 1.5) \).
The block penalty parameters \( (\beta_V, \beta_\Upsilon, \beta_D) \) affect only convergence speed and are not fine-tuned.
The synthetic experiments and the complexity study use the isotropic choice \( (\beta_V, \beta_\Upsilon, \beta_D) = 0.1 \cdot (1, 1, 1) = (0.1, 0.1, 0.1) \), whereas the real-data experiment of Section~\ref{sec:real_data} uses the anisotropic preconditioner \( (\beta_V, \beta_\Upsilon, \beta_D) = 0.1 \cdot (0.25, 0.5, 2) = (0.025, 0.05, 0.2) \), which accelerates convergence on the large panel without changing the estimator, as discussed in the remark following Algorithm~\ref{algo:ADMM}.
To reduce the cost of the stopping diagnostics, the quantities in~\ref{cond:T1} and~\ref{cond:T2} are evaluated every 20 iterations.

All algorithms and numerical experiments are implemented in \textsc{Matlab} R2025a; the codes are available at \url{https://github.com/linyopt/GFlsL}.
The default exponents \( (\mu_1, \mu_2) \) are assessed in Subsection~\ref{subsec:sensitivity}, and the empirical computational complexity of GFlsL is examined in Subsection~\ref{subsec:complexity}.
All experiments were conducted on an Intel(R) Xeon(R) Gold6242R CPU (3.10\,GHz) with 128\,GB of RAM.

\section{Synthetic experiments}\label{sec:simulations}
In this section, we conduct simulation experiments to assess the performances of the GFlsL. 

\subsection{Simulation settings}\label{subsec:simulations}
We consider the following three data generating processes, where we denote by $m^*$ the true number of change points and by $\{\Sigma^*_j\}_{j=1}^{m^{*}+1}$ the true $p \times p$ variance-covariance matrices:

\noindent\textbf{Setting (i):}
For each true $\Sigma^*_j$, the off-diagonal non-zero entries are drawn uniformly on $\Uc([-2,-0.2]) \cup \Uc([0.2,2])$, so that every active entry is bounded away from zero by $0.2$ in absolute value, and the diagonal entries are drawn in $\Uc([1.5,3.5])$.
The proportion of zero entries in the lower triangular part of $\Sigma^*_j$ is set as $80\%$.
Whenever $\lambda_{\min}(\Sigma^*_j)<0.01$ we restore positive definiteness by adding $(\zeta+|\lambda_{\min}(\Sigma^*_j)|)I_p$, where $\zeta$ is the first value in $\{0.005,0.01,0.015,\ldots\}$ rendering $\lambda_{\min}(\Sigma^*_j)>0.01$.

\noindent\textbf{Setting (ii):}
The variance-covariance matrix is generated following the same spirit as in Section 5 of \citet{cai2016}.
We construct $\Sigma^*_j = D^{1/2}_j\,C_j\,D^{1/2}_j$, $j =1,\ldots,m^*+1$, where $D_j = \text{diag}(U_{1,j},\ldots,U_{p,j})$ with $U_{k,j} \in \Uc([1.5,4])$, so that the diagonal entries of $\Sigma^*_j$ and $\Theta^*_j$ differ.
We set $C_j = (c_{uv,j})_{1 \leq u,v \leq p}$ with $c_{uv,j}=a^{|v-u|}_j$.
The coefficient $a_j$ equals $0.3$ with probability $0.5$, and equals $0.8$ otherwise.
Then, we set $\Sigma^{*}_{j,uv}=0$ when $|\Sigma^{*}_{j,uv}|<0.05$ and each non-zero off-diagonal coefficient is multiplied by $1$ (resp. $-1$) with probability $0.5$ (resp. $0.5$).
Positive definiteness is restored as in \textbf{Setting~(i)}; the resulting matrix has a banded structure.

\noindent\textbf{Setting (iii):}
The variance-covariance matrix follows a sparse factor model structure.
For each regime $j$, we generate \(\Sigma^*_j = \Lambda_j \Lambda_j^\top + \Psi_j \), where $\Lambda_j \in \Rb^{p \times r_j}$ is a sparse factor loading matrix with $r_j$ latent factors and $\Psi_j$ is a $p \times p$ diagonal matrix of variable-specific variances, so that the cross-sectional dependence is driven by $r_j$ common factors.
For each regime, $r_j$ is drawn uniformly from $\{2,\ldots,6\}$.
Approximately $80\%$ of the $p\,r_j$ entries of $\Lambda_j$ are set to zero, meaning each variable loads on only a small subset of the factors; the remaining non-zero entries are drawn from $\Uc([-2,-0.5]\cup[0.5,2])$.
The diagonal entries of $\Psi_j$ are drawn from $\Uc([0.5,1])$, which makes $\Sigma^*_j$ sparse; see \citet{poignard_terada_2025} for the asymptotic properties of sparse factor models.

\subsection{Experiments}\label{subsec:exp}
We set \( T = 200 \) and consider \( p = 10 \) and \( p = 20 \).
Three cases relating to the change points are considered: (a) no change point (\( m^{*} = 0 \)); (b) a single change point (\( m^{*} = 1 \)); (c) multiple change points (\( m^{*} = 3 \)).
The change-point locations $T^*_j,j=1,\ldots,m^*$, are drawn at random subject to $\mathcal{I}_{\min}\geq \kappa T$, where $\kappa=1/(m^*+c)$ with $c=8$, so that $\kappa = 0.11$ (resp. $0.0833$) when $m^*=1$ (resp. $3$): the regimes may have different lengths but must each be long enough for reliable estimation.
This minimum-length constraint plays the role of the trimming parameter of the change-point literature, for which Section~4 of~\citet{Qian2016} reports the standard range $\kappa \in [0.05,0.25]$, within which our choices fall.
For each setting, for $t=1,\ldots,T$, we draw $X_t \sim \Nc_{\Rb^p}(0,\Theta^{*}_t)$ the $p$-dimensional Gaussian distribution, with $\Theta^*_t = \Sigma^*_j$ when $t \in \Bc^*_j$.

For each simulation setting, and for each \( p \) and \( m^{*} \), we draw one hundred batches of $T$ independent samples $\Xc_T$.
For each batch we apply Algorithm~\ref{algo:adaptive} over a \( p \)-specific grid of $(\lambda,\lambda_1,\lambda_2)$ candidates (\emph{independent} of the simulation setting and \( m^{*} \)), selecting the optimal combination according to the seven criteria of Subsection~\ref{subsec:tuning-para}.
For $p=10$, the grid is $\lambda_1 \in p\times\{2.0\times10^{-4},4.5\times10^{-4},\ldots,4.45\times10^{-3}\}$, $\lambda_2 \in p\times\{0.08,0.12,\ldots,0.56\}$ for the non-adaptive and first-stage fits, and $\lambda_2 \in p\times\{0.02,0.09,\ldots,0.58\}$ for the adaptive stage; for $p=20$, it is $\lambda_1 \in p\times\{2.0\times10^{-4},6.0\times10^{-4},\ldots,2.6\times10^{-3}\}$, $\lambda_2 \in p\times\{0.06,0.08,\ldots,0.38\}$, and $\lambda_2 \in p\times\{0.02,0.08,\ldots,0.50\}$, respectively.
This yields seven solutions to problem \eqref{stat_crit_adaptive} depending on $(\lambda,\lambda_1,\lambda_2)$.
For the sake of comparison, we also compute the non-adaptive GFlsL estimator \eqref{stat_crit_1} with $\xi_{uv,1t}=1, \xi_{2t}=1$, and select the optimal pair $(\lambda_1,\lambda_2)$ according to the same seven criteria.

We compare the change-point detection performance of GFlsL with four competing methods: Binary Segmentation through Operator Norm (BSOP) and Wild Binary Segmentation through Independent Projection (WBSIP) of~\citet{wang2021}, Divide and Conquer Dynamic Programming (DCDP) of~\citet{li2023dcdp}, and the Threshold Block Fused Lasso (TBFL) of~\citet{bai2024unified}.
BSOP recursively partitions the time series by locating the split that maximizes an operator-norm discrepancy between the sample covariances on either side.
WBSIP extends this idea by drawing random sub-intervals and projecting the data onto one-dimensional subspaces, then aggregating the resulting univariate change-point statistics to guide segmentation.
DCDP solves a penalized dynamic programming problem using a divide-and-conquer scheme tailored to high-dimensional covariance change-point detection.
TBFL is a three-step procedure that combines penalized multivariate regression with neighborhood selection to perform change-point detection in Gaussian graphical models.
Among these competitors, only TBFL estimates precision matrices; BSOP, WBSIP, and DCDP return only change-point locations.
We note that WBSIP requires two \emph{independent} samples, which limits its practical applicability.
Furthermore, Assumptions~2 and~3 in~\citet{wang2021} link \(T\) and \(p\) so as to require a sufficiently large sample relative to the dimension; at \(p=20\) and \(T=200\) these conditions are unlikely to hold, and both BSOP and WBSIP systematically detect no change point.
For BSOP and WBSIP, since explicit tuning parameter selection procedures are not provided in~\citet{wang2021}, we use the parameters following their \texttt{R} package \texttt{changepoints}.\footnote{Available at \url{https://cran.r-project.org/web/packages/changepoints/}.}
In particular, for BSOP we set \(\tau = 10\); for WBSIP we use \(\Delta = 5\), \(\tau = 1.5\sqrt{p \log(T)}\) and the number of candidate intervals \(M = 100\).
For DCDP, we set \(\zeta = 0\), \(Q = 50\), and select \(\gamma\) from \(\{100,200,300\}\) via the cross-validation procedure based on testing datasets, also following their implementation.
For TBFL, we use the default parameters from their \texttt{R} package \texttt{LinearDetect}.\footnote{Archived at \url{https://cran.r-project.org/src/contrib/Archive/LinearDetect/}.}

To provide a fairer basis for covariance comparison, we augment each change-point-detection-only competitor with a post-segmentation covariance refit: after the detector produces its change-point estimates, we refit one covariance matrix per estimated segment.
We refer to the combination of a detector and a refit as a detector-plus-refit pipeline.
Because the refit is conditional on the estimated segmentation, the detector-plus-refit pipelines are not theoretically equivalent to our joint GFlsL estimator, but they permit a direct comparison of covariance estimation accuracy across all methods.
We consider a representative set of refits spanning distinct covariance structures: the positive-definite \(\ell_1\)-penalized estimator (Xue--Ma--Zou) of~\citet{xue2012positive}, the EC2 estimator of~\citet{liu2014ec2}, and adaptive thresholding of~\citet{cai2011adaptive} with the positive-definite correction of~\citet{choi2019fspd} (AT+).\footnote{Because the adaptive thresholding of~\citet{cai2011adaptive} does not guarantee positive definiteness, we apply the fixed-support positive-definite correction of~\citet{choi2019fspd} to the resulting estimator.}
For TBFL, we also evaluate its native regime-wise precision estimates (inverted to obtain covariance matrices, denoted by TBFL native).
In addition, for \textbf{Setting~(ii)} we include convex banding~\citep{bien2016convex} to match the banded structure.

The tuning parameters for each post-segmentation refit are set as follows.
For the Xue--Ma--Zou estimator, we follow~\citet{xue2012positive} and set the \(\ell_1\) penalty to \(\lambda = \bar{\sigma}^2 \sqrt{\log(p)/n_{\mathrm{seg}}}\), where \(\bar{\sigma}^2 = p^{-1}\operatorname{tr}(\widehat{\Sigma}_{n_{\mathrm{seg}}})\) is the average diagonal entry of the segment sample covariance \(\widehat{\Sigma}_{n_{\mathrm{seg}}}\), and $n_{\mathrm{seg}}$ is the length of corresponding segment.
For EC2 in~\citet{liu2014ec2}, we apply the penalty \(\lambda = \sqrt{\log(p)/n_{\mathrm{seg}}}\) on the correlation scale and rescale the resulting sparse correlation matrix back to the covariance scale.
For AT+, the entrywise-adaptive threshold is \(\lambda_{ij} = \delta \sqrt{\widehat{\theta}_{ij} \log(p) / n_{\mathrm{seg}}}\) with \(\delta = 2\) following~\citet{cai2011adaptive}, where \(\widehat{\theta}_{ij}\) estimates the variance of the \((i,j)\)-th sample covariance entry.
The fixed-support positive-definite correction of~\citet{choi2019fspd} is then applied with eigenvalue floor \(0.01\) and the data-dependent shrinkage target recommended therein.
For convex banding, we use the penalty form \(\lambda = 2c\sqrt{\log(p)/n_{\mathrm{seg}}}\) with the hierarchical weights in~\citet{bien2016convex}; we calibrate the constant \(c\) via a small pilot simulation, obtaining \(c = 6.0\) for \(p = 10\) and \(c = 5.0\) for \(p = 20\).

For all methods, we report two change-point metrics: (i) the number of detected change points, denoted by $nb$, and (ii) the Hausdorff distance $d_h =100\times \max\big(h(\widehat{\Tc}_{\widehat{m}},\Tc^*_{m^*}),\allowbreak h(\Tc^*_{m^*},\widehat{\Tc}_{\widehat{m}})\big)/T$, which measures the accuracy of the estimated change-point locations.
When one of $\widehat{\Tc}_{\widehat{m}}$ and $\Tc^*_{m^*}$ is empty, we set $d_h$ equal to the largest time index of the change points in the other set; when both sets are empty, we set $d_h=0$.
For GFlsL and all detector-plus-refit pipelines, we additionally report (iii) the $\text{F}_1$ score $\text{F}_1=2\text{TP}/(2\text{TP}+\text{FN}+\text{FP})$, where $\text{TP}$, $\text{FN}$ and $\text{FP}$ are the numbers of correctly identified non-zero, wrongly nullified non-zero and wrongly retained zero entries, respectively, so that values closer to $1$ are better, (iv) the support accuracy $ \acc = (\text{TP} + \text{TN}) / (\text{TP} + \text{TN} + \text{FP} + \text{FN}) $, with $ \text{TN} $ the correctly identified zeros, and (v) the root mean squared error (MSE) $\sqrt{(p^2T)^{-1}\sum^{T}_{t=1}\|\widehat{\Theta}_t-\Theta^*_t\|^2_F}$.

Tables~\ref{tab:res-iid},~\ref{tab:res-banded}, and~\ref{tab:res-factor} report the above five metrics that are averaged over 100 independent replications for \textbf{Settings~(i)}, \textbf{(ii)} and \textbf{(iii)}, respectively.
Each table contains the non-adaptive and adaptive GFlsL estimators, as well as the detector-plus-refit pipelines, across all seven selection criteria of Subsection~\ref{subsec:tuning-para}, with one row per criterion or pipeline.
For BSOP, WBSIP, DCDP, and TBFL, $nb$ and $d_h$ depend only on the detector, so these two metrics are shared across its refit variants.
The ``\textbf{optimal}'' criterion of Subsection~\ref{subsec:tuning-para} selects the tuning parameters that sequentially minimize the Hausdorff distance, maximize the F$_1$ score among the minimizers, and then minimize the remaining error.
The ``optimal'' and ``lossval'' criteria, as well as WBSIP and DCDP, use information beyond the training data and are therefore not applicable in practice.
To highlight fully data-driven performance, \fbox{boxes} mark the best results obtained from training data only, namely the best outcomes among the five data-driven GFlsL criteria and those of BSOP and TBFL.

\begin{table}[htb!]
\centering
\caption{Change-point recovery, support recovery, and estimation accuracy for \textbf{Setting~(i)}, based on 100 replications.
\textbf{Bold} marks the best overall result; \fbox{boxed} marks the best among fully data-driven methods (training data only).
Abbreviations: NA = Non-adaptive, AD = Adaptive, XMZ = Xue--Ma--Zou, AT+ = adaptive thresholding with FSPD correction; HBIC$_c$ and HBICG$(0,2)$ are the calibrated and rebalanced variants of HBIC and HBICG defined in Subsection~\ref{subsec:tuning-para}.\label{tab:res-iid}}
{\setlength{\tabcolsep}{2pt}
\resizebox{\textwidth}{!}{%
\begin{tabular}{ll|ccccc|ccccc|ccccc|ccccc|ccccc|ccccc}\hline\hline
Methods & Criteria/ & \multicolumn{5}{c|}{(0,10)} & \multicolumn{5}{c|}{(1,10)} & \multicolumn{5}{c|}{(3,10)} & \multicolumn{5}{c|}{(0,20)} & \multicolumn{5}{c|}{(1,20)} & \multicolumn{5}{c}{(3,20)} \\
& Refit & $nb$ & $d_h$ & $\text{F}_1$ & $\acc$ & MSE & $nb$ & $d_h$ & $\text{F}_1$ & $\acc$ & MSE & $nb$ & $d_h$ & $\text{F}_1$ & $\acc$ & MSE & $nb$ & $d_h$ & $\text{F}_1$ & $\acc$ & MSE & $nb$ & $d_h$ & $\text{F}_1$ & $\acc$ & MSE & $nb$ & $d_h$ & $\text{F}_1$ & $\acc$ & MSE \\ \hline
\multirow{7}{*}{\shortstack{NA\\GFlsL}} & optimal & \textbf{0.00} & \textbf{0.00} & 0.88 & 0.93 & 0.32 & 3.08 & 12.91 & 0.64 & 0.79 & 0.48 & 18.28 & 25.88 & 0.51 & 0.59 & 0.50 & \textbf{0.00} & \textbf{0.00} & \textbf{0.82} & \textbf{0.92} & 0.34 & 3.67 & 10.60 & 0.55 & 0.79 & 0.49 & 23.80 & 28.16 & 0.41 & 0.61 & 0.54 \\
  & lossval & 0.17 & 3.09 & 0.79 & 0.86 & \textbf{0.22} & 14.25 & 95.41 & 0.50 & 0.44 & 0.33 & 17.43 & 40.66 & 0.47 & 0.39 & \textbf{0.42} & 8.71 & 72.15 & 0.69 & 0.77 & 0.34 & 43.05 & 115.45 & 0.41 & 0.34 & 0.40 & 49.93 & 56.76 & 0.40 & 0.32 & \textbf{0.46} \\
  & BIC & \fbox{\textbf{0.00}} & \fbox{\textbf{0.00}} & 0.54 & 0.82 & 0.52 & 0.00 & 98.96 & 0.56 & 0.82 & 0.54 & 0.00 & 147.88 & 0.54 & 0.82 & 0.54 & \fbox{\textbf{0.00}} & \fbox{\textbf{0.00}} & 0.35 & 0.79 & 0.53 & 0.00 & 100.06 & 0.35 & 0.79 & 0.54 & 0.00 & 150.34 & 0.35 & 0.79 & 0.56 \\
  & HBIC & \fbox{\textbf{0.00}} & \fbox{\textbf{0.00}} & 0.78 & 0.85 & \fbox{\textbf{0.22}} & 0.05 & 95.11 & 0.68 & 0.77 & 0.41 & 0.00 & 147.88 & 0.60 & 0.72 & 0.50 & \fbox{\textbf{0.00}} & \fbox{\textbf{0.00}} & \fbox{0.81} & 0.91 & 0.34 & 0.00 & 100.06 & 0.66 & 0.85 & 0.47 & 0.00 & 150.34 & 0.51 & 0.80 & 0.53 \\
  & HBIC$_c$ & 29.10 & 141.34 & 0.52 & 0.66 & 0.42 & 52.77 & 81.39 & 0.50 & 0.68 & 0.51 & 52.82 & 42.66 & 0.49 & 0.67 & 0.54 & 48.36 & 165.56 & 0.39 & 0.46 & 0.42 & 59.31 & 93.84 & 0.37 & 0.58 & 0.50 & 64.13 & \fbox{47.42} & 0.38 & 0.57 & 0.53 \\
  & HBICG & \fbox{\textbf{0.00}} & \fbox{\textbf{0.00}} & 0.80 & 0.88 & 0.26 & 0.00 & 98.96 & 0.65 & 0.84 & 0.49 & 0.00 & 147.88 & 0.55 & 0.82 & 0.54 & \fbox{\textbf{0.00}} & \fbox{\textbf{0.00}} & 0.60 & 0.85 & 0.43 & 0.00 & 100.06 & 0.35 & 0.80 & 0.54 & 0.00 & 150.34 & 0.35 & 0.79 & 0.56 \\
  & HBICG$(0,2)$ & \fbox{\textbf{0.00}} & \fbox{\textbf{0.00}} & 0.78 & 0.85 & \fbox{\textbf{0.22}} & 0.00 & 98.96 & 0.68 & 0.77 & 0.41 & 0.00 & 147.88 & 0.60 & 0.71 & 0.50 & \fbox{\textbf{0.00}} & \fbox{\textbf{0.00}} & \fbox{0.81} & 0.91 & 0.34 & 0.00 & 100.06 & 0.66 & 0.85 & 0.47 & 0.00 & 150.34 & 0.51 & 0.80 & 0.53 \\
\hline
\multirow{7}{*}{\shortstack{AD\\GFlsL}} & optimal & \textbf{0.00} & \textbf{0.00} & \textbf{0.91} & \textbf{0.95} & 0.24 & \textbf{1.00} & \textbf{9.24} & \textbf{0.84} & \textbf{0.91} & 0.36 & 3.77 & \textbf{19.21} & \textbf{0.72} & \textbf{0.85} & 0.47 & \textbf{0.00} & \textbf{0.00} & 0.81 & \textbf{0.92} & 0.31 & 0.99 & \textbf{4.63} & \textbf{0.71} & \textbf{0.88} & 0.40 & \textbf{3.47} & \textbf{13.92} & 0.57 & \textbf{0.82} & 0.49 \\
  & lossval & 0.28 & 21.37 & 0.87 & 0.93 & \textbf{0.22} & 1.82 & 38.08 & 0.79 & 0.87 & \textbf{0.31} & \textbf{3.69} & 33.87 & 0.69 & 0.80 & \textbf{0.42} & 0.49 & 42.29 & 0.79 & 0.91 & 0.31 & 2.96 & 63.10 & 0.69 & 0.85 & \textbf{0.38} & 5.05 & 34.93 & \textbf{0.59} & 0.79 & \textbf{0.46} \\
  & BIC & \fbox{\textbf{0.00}} & \fbox{\textbf{0.00}} & 0.54 & 0.82 & 0.52 & 0.00 & 98.96 & 0.56 & 0.83 & 0.54 & 0.00 & 147.88 & 0.55 & 0.82 & 0.55 & \fbox{\textbf{0.00}} & \fbox{\textbf{0.00}} & 0.35 & 0.81 & 0.53 & 0.00 & 100.06 & 0.36 & 0.81 & 0.55 & 0.00 & 150.34 & 0.36 & 0.81 & 0.56 \\
  & HBIC & 2.44 & 88.03 & 0.82 & 0.89 & 0.32 & 3.99 & 67.04 & 0.77 & 0.86 & 0.40 & 6.51 & \fbox{38.00} & \fbox{0.68} & 0.80 & 0.49 & 26.58 & 168.67 & 0.58 & 0.74 & 0.93 & 31.23 & 119.71 & 0.53 & 0.71 & 1.04 & 35.99 & 59.50 & 0.50 & 0.68 & 1.18 \\
  & HBIC$_c$ & 1.56 & 68.71 & 0.85 & 0.91 & 0.26 & \fbox{1.05} & \fbox{29.51} & \fbox{0.82} & \fbox{0.90} & 0.34 & 0.76 & 105.86 & 0.66 & 0.82 & 0.48 & 1.39 & 23.75 & 0.79 & 0.91 & 0.32 & \fbox{\textbf{1.00}} & 25.57 & \fbox{\textbf{0.71}} & \fbox{\textbf{0.88}} & 0.40 & \fbox{0.94} & 111.08 & 0.52 & \fbox{\textbf{0.82}} & 0.52 \\
  & HBICG & \fbox{\textbf{0.00}} & \fbox{\textbf{0.00}} & 0.88 & \fbox{0.94} & 0.23 & 0.01 & 97.99 & 0.68 & 0.86 & 0.46 & 0.00 & 147.88 & 0.56 & 0.82 & 0.54 & \fbox{\textbf{0.00}} & \fbox{\textbf{0.00}} & 0.79 & \fbox{\textbf{0.92}} & 0.31 & 0.00 & 100.06 & 0.39 & 0.82 & 0.53 & 0.00 & 150.34 & 0.35 & 0.81 & 0.56 \\
  & HBICG$(0,2)$ & \fbox{\textbf{0.00}} & \fbox{\textbf{0.00}} & \fbox{0.89} & \fbox{0.94} & \fbox{\textbf{0.22}} & 0.75 & 31.36 & \fbox{0.82} & \fbox{0.90} & \fbox{0.33} & 0.33 & 129.62 & 0.64 & 0.81 & 0.50 & \fbox{\textbf{0.00}} & \fbox{\textbf{0.00}} & \fbox{0.81} & \fbox{\textbf{0.92}} & 0.31 & 0.78 & \fbox{24.09} & \fbox{\textbf{0.71}} & \fbox{\textbf{0.88}} & \fbox{0.39} & 0.33 & 130.44 & 0.49 & 0.81 & 0.52 \\
\hline
\multirow{3}{*}{BSOP} & EC2 & \multirow{3}{*}{1.86} & \multirow{3}{*}{129.68} & 0.75 & 0.85 & 0.34 & \multirow{3}{*}{1.95} & \multirow{3}{*}{70.08} & 0.75 & 0.84 & 0.38 & \multirow{3}{*}{\fbox{1.90}} & \multirow{3}{*}{44.45} & 0.67 & 0.80 & \fbox{0.46} & \multirow{3}{*}{\fbox{\textbf{0.00}}} & \multirow{3}{*}{\fbox{\textbf{0.00}}} & 0.80 & 0.90 & \fbox{\textbf{0.27}} & \multirow{3}{*}{0.00} & \multirow{3}{*}{100.06} & 0.63 & 0.81 & 0.45 & \multirow{3}{*}{0.00} & \multirow{3}{*}{150.34} & 0.49 & 0.76 & 0.53 \\
 & XMZ &  &  & 0.76 & 0.85 & 0.36 &  &  & 0.75 & 0.84 & 0.39 &  &  & 0.67 & 0.80 & 0.47 &  &  & 0.80 & 0.90 & \fbox{\textbf{0.27}} &  &  & 0.63 & 0.81 & 0.45 &  &  & 0.49 & 0.76 & 0.53 \\
 & AT+ &  &  & 0.76 & 0.89 & 0.47 &  &  & 0.73 & 0.88 & 0.50 &  &  & 0.64 & \fbox{\textbf{0.85}} & 0.53 &  &  & 0.74 & 0.90 & 0.42 &  &  & 0.51 & 0.84 & 0.52 &  &  & 0.38 & 0.81 & 0.55 \\
\hline
\multirow{3}{*}{WBSIP} & EC2 & \multirow{3}{*}{8.71} & \multirow{3}{*}{172.68} & 0.67 & 0.81 & 0.54 & \multirow{3}{*}{7.84} & \multirow{3}{*}{110.61} & 0.67 & 0.81 & 0.54 & \multirow{3}{*}{7.27} & \multirow{3}{*}{45.18} & 0.65 & 0.79 & 0.55 & \multirow{3}{*}{\textbf{0.00}} & \multirow{3}{*}{\textbf{0.00}} & 0.80 & 0.90 & \textbf{0.27} & \multirow{3}{*}{0.00} & \multirow{3}{*}{100.06} & 0.63 & 0.81 & 0.45 & \multirow{3}{*}{0.00} & \multirow{3}{*}{150.34} & 0.49 & 0.76 & 0.53 \\
 & XMZ &  &  & 0.66 & 0.81 & 0.56 &  &  & 0.66 & 0.81 & 0.56 &  &  & 0.64 & 0.80 & 0.56 &  &  & 0.80 & 0.90 & \textbf{0.27} &  &  & 0.63 & 0.81 & 0.45 &  &  & 0.49 & 0.76 & 0.53 \\
 & AT+ &  &  & 0.59 & 0.84 & 0.61 &  &  & 0.61 & 0.84 & 0.61 &  &  & 0.59 & 0.83 & 0.61 &  &  & 0.74 & 0.90 & 0.42 &  &  & 0.51 & 0.84 & 0.52 &  &  & 0.38 & 0.81 & 0.55 \\
\hline
\multirow{3}{*}{DCDP} & EC2 & \multirow{3}{*}{46.00} & \multirow{3}{*}{196.00} & 0.52 & 0.71 & 0.88 & \multirow{3}{*}{46.00} & \multirow{3}{*}{132.68} & 0.52 & 0.71 & 0.89 & \multirow{3}{*}{46.00} & \multirow{3}{*}{63.43} & 0.51 & 0.70 & 0.89 & \multirow{3}{*}{46.00} & \multirow{3}{*}{196.00} & 0.39 & 0.74 & 0.83 & \multirow{3}{*}{46.00} & \multirow{3}{*}{134.73} & 0.38 & 0.74 & 0.84 & \multirow{3}{*}{46.00} & \multirow{3}{*}{65.10} & 0.38 & 0.74 & 0.85 \\
 & XMZ &  &  & 0.52 & 0.75 & 0.96 &  &  & 0.52 & 0.74 & 0.97 &  &  & 0.52 & 0.75 & 0.97 &  &  & 0.39 & 0.75 & 0.96 &  &  & 0.39 & 0.75 & 0.97 &  &  & 0.39 & 0.75 & 0.98 \\
 & AT+ &  &  & 0.52 & 0.79 & 0.90 &  &  & 0.52 & 0.79 & 0.91 &  &  & 0.52 & 0.79 & 0.90 &  &  & 0.36 & 0.79 & 0.85 &  &  & 0.36 & 0.79 & 0.86 &  &  & 0.36 & 0.79 & 0.87 \\
\hline
\multirow{4}{*}{TBFL} & EC2 & \multirow{4}{*}{1.09} & \multirow{4}{*}{94.35} & 0.79 & 0.87 & 0.29 & \multirow{4}{*}{1.11} & \multirow{4}{*}{37.82} & 0.76 & 0.85 & 0.35 & \multirow{4}{*}{1.11} & \multirow{4}{*}{82.51} & 0.66 & 0.79 & \fbox{0.46} & \multirow{4}{*}{1.14} & \multirow{4}{*}{73.96} & 0.74 & 0.87 & 0.35 & \multirow{4}{*}{1.12} & \multirow{4}{*}{32.09} & 0.70 & 0.86 & \fbox{0.39} & \multirow{4}{*}{0.89} & \multirow{4}{*}{92.49} & \fbox{0.55} & 0.79 & \fbox{0.50} \\
 & XMZ &  &  & 0.78 & 0.86 & 0.30 &  &  & 0.76 & 0.85 & 0.36 &  &  & 0.66 & 0.79 & 0.47 &  &  & 0.73 & 0.87 & 0.36 &  &  & 0.70 & 0.86 & 0.40 &  &  & 0.54 & 0.79 & \fbox{0.50} \\
 & AT+ &  &  & 0.81 & 0.91 & 0.42 &  &  & 0.76 & 0.89 & 0.46 &  &  & 0.64 & \fbox{\textbf{0.85}} & 0.53 &  &  & 0.59 & 0.86 & 0.49 &  &  & 0.55 & 0.85 & 0.51 &  &  & 0.41 & \fbox{\textbf{0.82}} & 0.55 \\
 & Native &  &  & 0.57 & 0.59 & 2.37 &  &  & 0.53 & 0.61 & 2.08 &  &  & 0.56 & 0.75 & 1.42 &  &  & 0.39 & 0.37 & 6.60 &  &  & 0.42 & 0.56 & 2.94 &  &  & 0.40 & 0.72 & 1.33 \\
\hline\hline
\end{tabular}}}
\end{table}
\begin{table}[htb!]
\centering
\caption{Change-point recovery, support recovery, and estimation accuracy for \textbf{Setting~(ii)}, based on 100 replications.
\textbf{Bold} marks the best overall result; \fbox{boxed} marks the best among fully data-driven methods (training data only).
Abbreviations: NA = Non-adaptive, AD = Adaptive, XMZ = Xue--Ma--Zou, AT+ = adaptive thresholding with FSPD correction, CB = convex banding; HBIC$_c$ and HBICG$(0,2)$ are the calibrated and rebalanced variants of HBIC and HBICG defined in Subsection~\ref{subsec:tuning-para}.\label{tab:res-banded}}
{\setlength{\tabcolsep}{2pt}
\resizebox{\textwidth}{!}{%
\begin{tabular}{ll|ccccc|ccccc|ccccc|ccccc|ccccc|ccccc}\hline\hline
Methods & Criteria/ & \multicolumn{5}{c|}{(0,10)} & \multicolumn{5}{c|}{(1,10)} & \multicolumn{5}{c|}{(3,10)} & \multicolumn{5}{c|}{(0,20)} & \multicolumn{5}{c|}{(1,20)} & \multicolumn{5}{c}{(3,20)} \\
& Refit & $nb$ & $d_h$ & $\text{F}_1$ & $\acc$ & MSE & $nb$ & $d_h$ & $\text{F}_1$ & $\acc$ & MSE & $nb$ & $d_h$ & $\text{F}_1$ & $\acc$ & MSE & $nb$ & $d_h$ & $\text{F}_1$ & $\acc$ & MSE & $nb$ & $d_h$ & $\text{F}_1$ & $\acc$ & MSE & $nb$ & $d_h$ & $\text{F}_1$ & $\acc$ & MSE \\ \hline
\multirow{7}{*}{\shortstack{NA\\GFlsL}} & optimal & \textbf{0.01} & \textbf{1.10} & 0.99 & \textbf{1.00} & 1.27 & 1.69 & \textbf{8.34} & 0.78 & 0.88 & 2.23 & 8.25 & 17.47 & 0.61 & 0.66 & 2.30 & \textbf{0.00} & \textbf{0.00} & 0.96 & 0.99 & 1.26 & 2.43 & \textbf{11.10} & 0.51 & 0.81 & 1.79 & 13.10 & \textbf{20.09} & 0.36 & 0.52 & 1.74 \\
  & lossval & 0.76 & 41.42 & 0.60 & 0.61 & 0.69 & 12.07 & 95.57 & 0.47 & 0.38 & 1.19 & 25.25 & 49.77 & 0.47 & 0.36 & 1.64 & 3.89 & 106.37 & 0.36 & 0.48 & 0.60 & 17.63 & 97.34 & 0.28 & 0.24 & 0.99 & 40.95 & 55.64 & 0.27 & 0.23 & \textbf{1.33} \\
  & BIC & \fbox{\textbf{0.01}} & \fbox{\textbf{1.10}} & 0.59 & 0.83 & 2.53 & 0.30 & 92.83 & 0.59 & 0.82 & 2.59 & 0.32 & 144.68 & 0.55 & 0.81 & 2.68 & \fbox{\textbf{0.00}} & \fbox{\textbf{0.00}} & 0.41 & 0.77 & 1.90 & 0.60 & 76.29 & 0.44 & 0.80 & 1.91 & 0.55 & 135.64 & 0.42 & 0.80 & 1.98 \\
  & HBIC & 87.98 & 195.86 & 0.45 & 0.31 & 2.88 & 88.03 & 133.15 & 0.45 & 0.31 & 2.86 & 90.04 & 63.20 & 0.45 & 0.31 & 2.81 & 160.66 & 197.68 & 0.26 & 0.19 & 3.12 & 161.44 & 135.10 & 0.26 & 0.19 & 3.18 & 161.82 & 68.90 & 0.26 & 0.19 & 3.15 \\
  & HBIC$_c$ & 0.10 & 7.01 & 0.54 & 0.51 & 0.72 & 2.83 & 31.41 & 0.51 & 0.46 & 1.37 & 3.55 & 59.57 & 0.50 & 0.45 & 2.06 & 0.59 & 32.01 & 0.60 & 0.74 & 0.75 & 4.79 & 54.63 & 0.41 & 0.54 & 1.26 & \fbox{\textbf{3.01}} & 97.47 & 0.39 & 0.60 & 1.75 \\
  & HBICG & 0.02 & 2.42 & 0.69 & 0.72 & 0.75 & 0.29 & 92.83 & 0.73 & 0.87 & 2.30 & 0.32 & 144.68 & 0.62 & 0.83 & 2.61 & \fbox{\textbf{0.00}} & \fbox{\textbf{0.00}} & 0.81 & 0.92 & 0.90 & 0.60 & 76.44 & 0.47 & 0.82 & 1.86 & 0.55 & 135.64 & 0.43 & 0.80 & 1.98 \\
  & HBICG$(0,2)$ & 0.19 & 14.17 & 0.52 & 0.48 & 0.72 & 2.38 & 33.78 & 0.54 & 0.51 & 1.43 & 1.33 & 107.94 & 0.58 & 0.65 & 2.32 & 0.66 & 34.55 & 0.58 & 0.73 & 0.74 & 1.71 & 64.69 & 0.46 & 0.71 & 1.61 & 0.75 & 131.42 & 0.44 & 0.76 & 1.93 \\
\hline
\multirow{7}{*}{\shortstack{AD\\GFlsL}} & optimal & 0.24 & 18.75 & \textbf{1.00} & \textbf{1.00} & 0.84 & 1.71 & 14.19 & 0.97 & 0.98 & 1.36 & 4.67 & \textbf{15.10} & 0.89 & 0.94 & 1.92 & 1.89 & 106.18 & 0.98 & 0.99 & 0.97 & 3.97 & 51.48 & 0.92 & 0.98 & 1.20 & 5.41 & 21.95 & 0.84 & 0.95 & 1.50 \\
  & lossval & 0.46 & 42.04 & 0.82 & 0.87 & \textbf{0.64} & 2.07 & 19.62 & 0.82 & 0.87 & \textbf{1.06} & 4.40 & 24.20 & 0.75 & 0.81 & \textbf{1.60} & 2.06 & 111.10 & 0.82 & 0.93 & 0.71 & 4.37 & 56.41 & 0.77 & 0.91 & 1.01 & 5.38 & 27.47 & 0.72 & 0.89 & \textbf{1.33} \\
  & BIC & 0.22 & 18.75 & \fbox{0.98} & \fbox{0.99} & 0.86 & 1.52 & 20.29 & 0.94 & 0.97 & 1.52 & 2.09 & 55.68 & 0.83 & 0.92 & 2.11 & 1.83 & 106.36 & 0.97 & 0.99 & 1.02 & 3.83 & 51.49 & 0.91 & 0.97 & 1.27 & 4.22 & \fbox{30.93} & 0.83 & 0.95 & 1.54 \\
  & HBIC & 85.77 & 195.86 & 0.46 & 0.36 & 9.41 & 85.67 & 133.11 & 0.46 & 0.36 & 9.35 & 87.16 & 63.17 & 0.46 & 0.35 & 9.15 & 160.73 & 197.70 & 0.27 & 0.25 & 11.31 & 161.35 & 135.10 & 0.27 & 0.25 & 11.47 & 161.65 & 68.91 & 0.27 & 0.25 & 11.35 \\
  & HBIC$_c$ & 0.38 & 21.18 & 0.65 & 0.68 & 0.87 & 1.67 & 19.94 & 0.65 & 0.68 & 1.28 & 2.58 & 45.95 & 0.59 & 0.62 & 1.90 & 2.00 & 108.01 & 0.55 & 0.73 & 1.04 & 4.22 & 54.33 & 0.51 & 0.69 & 1.38 & 4.96 & 31.41 & 0.46 & 0.64 & 1.70 \\
  & HBICG & 0.22 & 18.75 & 0.86 & 0.91 & \fbox{0.69} & 1.52 & 20.36 & 0.91 & 0.94 & \fbox{1.17} & 2.09 & 55.74 & 0.86 & 0.92 & 1.95 & 1.83 & 106.36 & 0.94 & 0.98 & 0.83 & 3.83 & 51.49 & 0.91 & 0.98 & 1.22 & 4.22 & \fbox{30.93} & 0.83 & 0.96 & 1.53 \\
  & HBICG$(0,2)$ & 0.33 & 27.84 & 0.65 & 0.68 & 0.78 & 1.87 & 18.83 & 0.63 & 0.66 & 1.25 & \fbox{\textbf{3.35}} & \fbox{29.91} & 0.59 & 0.61 & \fbox{1.80} & 2.04 & 111.50 & 0.73 & 0.88 & 0.78 & 4.21 & 54.67 & 0.75 & 0.89 & 1.07 & 4.72 & 31.78 & 0.65 & 0.84 & \fbox{1.40} \\
\hline
\multirow{4}{*}{BSOP} & EC2 & \multirow{4}{*}{1.89} & \multirow{4}{*}{127.61} & 0.85 & 0.90 & 1.34 & \multirow{4}{*}{1.85} & \multirow{4}{*}{65.87} & 0.84 & 0.89 & 1.50 & \multirow{4}{*}{1.90} & \multirow{4}{*}{43.78} & 0.80 & 0.87 & 1.94 & \multirow{4}{*}{\fbox{\textbf{0.00}}} & \multirow{4}{*}{\fbox{\textbf{0.00}}} & 0.80 & 0.93 & 0.62 & \multirow{4}{*}{0.00} & \multirow{4}{*}{96.88} & 0.75 & 0.91 & 1.39 & \multirow{4}{*}{0.00} & \multirow{4}{*}{151.06} & 0.70 & 0.90 & 1.75 \\
 & XMZ &  &  & 0.86 & 0.90 & 1.37 &  &  & 0.85 & 0.90 & 1.53 &  &  & 0.80 & 0.88 & 1.96 &  &  & 0.81 & 0.93 & 0.63 &  &  & 0.75 & 0.91 & 1.40 &  &  & 0.71 & 0.90 & 1.75 \\
 & AT+ &  &  & 0.94 & 0.97 & 2.14 &  &  & 0.92 & 0.96 & 2.20 &  &  & 0.83 & 0.92 & 2.39 &  &  & \fbox{\textbf{1.00}} & \fbox{\textbf{1.00}} & 1.09 &  &  & 0.86 & 0.97 & 1.59 &  &  & 0.75 & 0.94 & 1.86 \\
 & CB &  &  & 0.97 & 0.98 & 1.26 &  &  & 0.98 & 0.99 & 1.44 &  &  & 0.98 & \fbox{\textbf{0.99}} & 1.93 &  &  & 0.99 & \fbox{\textbf{1.00}} & \fbox{\textbf{0.53}} &  &  & \fbox{\textbf{0.99}} & \fbox{\textbf{1.00}} & 1.37 &  &  & \fbox{\textbf{0.99}} & \fbox{\textbf{1.00}} & 1.73 \\
\hline
\multirow{4}{*}{WBSIP} & EC2 & \multirow{4}{*}{12.83} & \multirow{4}{*}{183.41} & 0.74 & 0.84 & 2.55 & \multirow{4}{*}{12.58} & \multirow{4}{*}{119.51} & 0.74 & 0.84 & 2.52 & \multirow{4}{*}{12.44} & \multirow{4}{*}{51.42} & 0.73 & 0.83 & 2.50 & \multirow{4}{*}{\textbf{0.00}} & \multirow{4}{*}{\textbf{0.00}} & 0.80 & 0.93 & 0.62 & \multirow{4}{*}{0.00} & \multirow{4}{*}{96.88} & 0.75 & 0.91 & 1.39 & \multirow{4}{*}{0.00} & \multirow{4}{*}{151.06} & 0.70 & 0.90 & 1.75 \\
 & XMZ &  &  & 0.72 & 0.83 & 2.62 &  &  & 0.72 & 0.84 & 2.58 &  &  & 0.72 & 0.83 & 2.56 &  &  & 0.81 & 0.93 & 0.63 &  &  & 0.75 & 0.91 & 1.40 &  &  & 0.71 & 0.90 & 1.75 \\
 & AT+ &  &  & 0.59 & 0.83 & 2.97 &  &  & 0.60 & 0.84 & 2.94 &  &  & 0.59 & 0.84 & 2.90 &  &  & \textbf{1.00} & \textbf{1.00} & 1.09 &  &  & 0.86 & 0.97 & 1.59 &  &  & 0.75 & 0.94 & 1.86 \\
 & CB &  &  & 0.91 & 0.95 & 2.54 &  &  & 0.92 & 0.96 & 2.51 &  &  & 0.92 & 0.96 & 2.51 &  &  & 0.99 & \textbf{1.00} & \textbf{0.53} &  &  & \textbf{0.99} & \textbf{1.00} & 1.37 &  &  & \textbf{0.99} & \textbf{1.00} & 1.73 \\
\hline
\multirow{4}{*}{DCDP} & EC2 & \multirow{4}{*}{46.00} & \multirow{4}{*}{196.00} & 0.56 & 0.72 & 3.52 & \multirow{4}{*}{46.00} & \multirow{4}{*}{131.36} & 0.56 & 0.72 & 3.51 & \multirow{4}{*}{46.00} & \multirow{4}{*}{61.81} & 0.55 & 0.72 & 3.44 & \multirow{4}{*}{46.00} & \multirow{4}{*}{196.00} & 0.45 & 0.82 & 2.58 & \multirow{4}{*}{46.00} & \multirow{4}{*}{132.06} & 0.45 & 0.82 & 2.59 & \multirow{4}{*}{46.00} & \multirow{4}{*}{66.35} & 0.45 & 0.82 & 2.59 \\
 & XMZ &  &  & 0.57 & 0.77 & 3.71 &  &  & 0.57 & 0.77 & 3.68 &  &  & 0.57 & 0.76 & 3.62 &  &  & 0.46 & 0.83 & 2.88 &  &  & 0.46 & 0.83 & 2.91 &  &  & 0.46 & 0.83 & 2.90 \\
 & AT+ &  &  & 0.53 & 0.80 & 3.63 &  &  & 0.53 & 0.80 & 3.61 &  &  & 0.53 & 0.80 & 3.55 &  &  & 0.47 & 0.88 & 2.63 &  &  & 0.47 & 0.88 & 2.66 &  &  & 0.47 & 0.88 & 2.64 \\
 & CB &  &  & 0.65 & 0.85 & 3.56 &  &  & 0.65 & 0.85 & 3.55 &  &  & 0.64 & 0.85 & 3.48 &  &  & 0.66 & 0.92 & 2.56 &  &  & 0.67 & 0.92 & 2.58 &  &  & 0.67 & 0.92 & 2.57 \\
\hline
\multirow{5}{*}{TBFL} & EC2 & \multirow{5}{*}{1.17} & \multirow{5}{*}{107.93} & 0.86 & 0.90 & 1.12 & \multirow{5}{*}{\fbox{\textbf{1.09}}} & \multirow{5}{*}{\fbox{14.17}} & 0.86 & 0.91 & 1.22 & \multirow{5}{*}{1.30} & \multirow{5}{*}{71.29} & 0.81 & 0.88 & 1.96 & \multirow{5}{*}{1.17} & \multirow{5}{*}{97.92} & 0.80 & 0.93 & 0.90 & \multirow{5}{*}{\fbox{\textbf{1.02}}} & \multirow{5}{*}{\fbox{16.94}} & 0.80 & 0.93 & 0.95 & \multirow{5}{*}{1.04} & \multirow{5}{*}{79.06} & 0.73 & 0.91 & 1.48 \\
 & XMZ &  &  & 0.87 & 0.91 & 1.15 &  &  & 0.86 & 0.91 & 1.24 &  &  & 0.82 & 0.89 & 1.97 &  &  & 0.79 & 0.92 & 0.93 &  &  & 0.80 & 0.93 & 0.98 &  &  & 0.73 & 0.91 & 1.49 \\
 & AT+ &  &  & 0.97 & 0.98 & 1.87 &  &  & 0.96 & 0.98 & 1.89 &  &  & 0.84 & 0.92 & 2.35 &  &  & 0.94 & 0.98 & 1.51 &  &  & 0.94 & 0.98 & 1.52 &  &  & 0.80 & 0.95 & 1.78 \\
 & CB &  &  & \fbox{0.98} & \fbox{0.99} & 1.07 &  &  & \fbox{\textbf{0.99}} & \fbox{\textbf{1.00}} & 1.18 &  &  & \fbox{\textbf{0.99}} & \fbox{\textbf{0.99}} & 1.96 &  &  & 0.99 & \fbox{\textbf{1.00}} & 0.76 &  &  & \fbox{\textbf{0.99}} & \fbox{\textbf{1.00}} & \fbox{\textbf{0.83}} &  &  & \fbox{\textbf{0.99}} & \fbox{\textbf{1.00}} & 1.43 \\
 & Native &  &  & 0.44 & 0.28 & 4.32 &  &  & 0.45 & 0.35 & 5.12 &  &  & 0.50 & 0.51 & 6.85 &  &  & 0.25 & 0.15 & 7.22 &  &  & 0.28 & 0.31 & 4.38 &  &  & 0.45 & 0.68 & 2.93 \\
\hline\hline
\end{tabular}}}
\end{table}
\begin{table}[htb!]
\centering
\caption{Change-point recovery, support recovery, and estimation accuracy for \textbf{Setting~(iii)}, based on 100 replications.
\textbf{Bold} marks the best overall result; \fbox{boxed} marks the best among fully data-driven methods (training data only).
Abbreviations: NA = Non-adaptive, AD = Adaptive, XMZ = Xue--Ma--Zou, AT+ = adaptive thresholding with FSPD correction; HBIC$_c$ and HBICG$(0,2)$ are the calibrated and rebalanced variants of HBIC and HBICG defined in Subsection~\ref{subsec:tuning-para}.\label{tab:res-factor}}
{\setlength{\tabcolsep}{2pt}
\resizebox{\textwidth}{!}{%
\begin{tabular}{ll|ccccc|ccccc|ccccc|ccccc|ccccc|ccccc}\hline\hline
Methods & Criteria/ & \multicolumn{5}{c|}{(0,10)} & \multicolumn{5}{c|}{(1,10)} & \multicolumn{5}{c|}{(3,10)} & \multicolumn{5}{c|}{(0,20)} & \multicolumn{5}{c|}{(1,20)} & \multicolumn{5}{c}{(3,20)} \\
& Refit & $nb$ & $d_h$ & $\text{F}_1$ & $\acc$ & MSE & $nb$ & $d_h$ & $\text{F}_1$ & $\acc$ & MSE & $nb$ & $d_h$ & $\text{F}_1$ & $\acc$ & MSE & $nb$ & $d_h$ & $\text{F}_1$ & $\acc$ & MSE & $nb$ & $d_h$ & $\text{F}_1$ & $\acc$ & MSE & $nb$ & $d_h$ & $\text{F}_1$ & $\acc$ & MSE \\ \hline
\multirow{7}{*}{\shortstack{NA\\GFlsL}} & optimal & \textbf{0.00} & \textbf{0.00} & \textbf{0.97} & \textbf{0.99} & 0.23 & 1.62 & 5.36 & 0.77 & 0.89 & 0.55 & 11.49 & \textbf{17.29} & 0.62 & 0.78 & 0.64 & \textbf{0.00} & \textbf{0.00} & 0.93 & \textbf{0.98} & 0.29 & 2.09 & 3.83 & 0.71 & 0.89 & 0.51 & 14.55 & 14.92 & 0.57 & 0.78 & 0.57 \\
  & lossval & 0.16 & 6.05 & 0.94 & 0.97 & 0.19 & 6.07 & 51.13 & 0.67 & 0.76 & 0.30 & 8.47 & 28.92 & 0.57 & 0.66 & 0.49 & 1.92 & 39.32 & 0.87 & 0.93 & 0.29 & 9.87 & 71.34 & 0.63 & 0.78 & 0.37 & 10.99 & 29.00 & 0.60 & 0.77 & 0.49 \\
  & BIC & \fbox{\textbf{0.00}} & \fbox{\textbf{0.00}} & 0.65 & 0.88 & 0.59 & 0.00 & 103.96 & 0.65 & 0.88 & 0.70 & 0.00 & 153.70 & 0.63 & 0.88 & 0.76 & \fbox{\textbf{0.00}} & \fbox{\textbf{0.00}} & 0.45 & 0.86 & 0.65 & 0.00 & 102.77 & 0.44 & 0.86 & 0.69 & 0.00 & 149.11 & 0.44 & 0.87 & 0.73 \\
  & HBIC & 0.01 & 1.41 & 0.94 & 0.97 & 0.19 & 1.37 & 47.77 & 0.78 & 0.87 & 0.41 & 1.44 & 110.73 & 0.66 & 0.80 & 0.63 & 0.03 & 0.23 & 0.93 & 0.97 & 0.29 & 0.80 & 56.23 & 0.79 & 0.92 & 0.47 & 0.82 & 121.70 & 0.60 & 0.86 & 0.65 \\
  & HBIC$_c$ & 6.42 & 79.81 & 0.77 & 0.86 & 0.33 & 10.47 & 49.90 & 0.68 & 0.81 & 0.46 & 17.34 & 54.27 & 0.58 & 0.72 & 0.60 & 13.72 & 97.45 & 0.64 & 0.85 & 0.43 & 60.31 & 68.44 & 0.48 & 0.81 & 0.62 & 69.07 & 43.51 & 0.44 & 0.80 & 0.69 \\
  & HBICG & \fbox{\textbf{0.00}} & \fbox{\textbf{0.00}} & 0.95 & 0.97 & 0.19 & 0.02 & 101.50 & 0.79 & 0.90 & 0.57 & 0.00 & 153.70 & 0.66 & 0.86 & 0.73 & \fbox{\textbf{0.00}} & \fbox{\textbf{0.00}} & 0.93 & \fbox{\textbf{0.98}} & 0.29 & 0.01 & 101.67 & 0.69 & 0.90 & 0.58 & 0.00 & 149.11 & 0.46 & 0.87 & 0.72 \\
  & HBICG$(0,2)$ & \fbox{\textbf{0.00}} & \fbox{\textbf{0.00}} & 0.95 & 0.97 & 0.19 & 0.62 & 60.53 & 0.78 & 0.88 & 0.46 & 0.33 & 136.35 & 0.66 & 0.80 & 0.68 & \fbox{\textbf{0.00}} & \fbox{\textbf{0.00}} & 0.93 & 0.97 & 0.29 & 0.74 & 53.11 & 0.79 & 0.92 & 0.48 & 0.63 & 121.94 & 0.60 & 0.86 & 0.66 \\
\hline
\multirow{7}{*}{\shortstack{AD\\GFlsL}} & optimal & 0.04 & 2.04 & \textbf{0.97} & \textbf{0.99} & 0.19 & 1.13 & \textbf{5.24} & \textbf{0.94} & \textbf{0.98} & 0.29 & \textbf{3.45} & 18.73 & \textbf{0.88} & \textbf{0.95} & 0.44 & \textbf{0.00} & \textbf{0.00} & 0.92 & 0.97 & 0.24 & 1.04 & \textbf{1.50} & 0.87 & \textbf{0.96} & 0.32 & \textbf{3.40} & \textbf{6.52} & 0.81 & \textbf{0.94} & 0.41 \\
  & lossval & 0.19 & 14.41 & 0.96 & 0.98 & \textbf{0.18} & 1.57 & 16.15 & 0.93 & 0.97 & \textbf{0.25} & 3.48 & 26.33 & 0.87 & 0.94 & \textbf{0.39} & 0.30 & 25.66 & 0.92 & 0.97 & 0.24 & 1.90 & 31.97 & \textbf{0.88} & \textbf{0.96} & \textbf{0.29} & 4.45 & 19.18 & \textbf{0.83} & \textbf{0.94} & \textbf{0.38} \\
  & BIC & 0.03 & 2.04 & 0.67 & 0.89 & 0.55 & 0.01 & 102.71 & 0.66 & 0.88 & 0.70 & 0.00 & 153.70 & 0.65 & 0.88 & 0.77 & \fbox{\textbf{0.00}} & \fbox{\textbf{0.00}} & 0.48 & 0.87 & 0.63 & 0.00 & 102.77 & 0.45 & 0.87 & 0.70 & 0.00 & 149.11 & 0.45 & 0.87 & 0.73 \\
  & HBIC & 1.31 & 62.85 & 0.95 & \fbox{0.98} & 0.26 & 2.21 & 27.68 & \fbox{0.93} & \fbox{0.97} & 0.31 & 4.38 & \fbox{28.83} & \fbox{0.87} & \fbox{0.94} & \fbox{0.44} & 2.34 & 76.51 & 0.90 & 0.96 & 0.33 & 3.50 & 50.57 & 0.87 & 0.95 & 0.37 & 5.37 & 25.60 & \fbox{0.82} & \fbox{\textbf{0.94}} & 0.43 \\
  & HBIC$_c$ & 0.65 & 43.47 & \fbox{0.96} & \fbox{0.98} & 0.19 & 0.61 & 46.74 & 0.89 & 0.95 & 0.37 & 0.68 & 112.96 & 0.75 & 0.89 & 0.62 & 0.96 & 55.82 & 0.91 & 0.97 & 0.26 & 0.86 & 15.76 & 0.86 & \fbox{\textbf{0.96}} & 0.32 & 1.50 & 70.61 & 0.73 & 0.92 & 0.50 \\
  & HBICG & 0.03 & 2.04 & \fbox{0.96} & \fbox{0.98} & 0.19 & 0.79 & 20.84 & 0.91 & 0.96 & 0.31 & 0.87 & 100.05 & 0.73 & 0.90 & 0.61 & \fbox{\textbf{0.00}} & \fbox{\textbf{0.00}} & 0.91 & 0.97 & 0.24 & 0.62 & 43.78 & 0.83 & 0.95 & 0.37 & 0.46 & 119.11 & 0.55 & 0.88 & 0.66 \\
  & HBICG$(0,2)$ & 0.04 & 2.04 & \fbox{0.96} & \fbox{0.98} & \fbox{\textbf{0.18}} & \fbox{\textbf{0.94}} & \fbox{11.96} & \fbox{0.93} & \fbox{0.97} & \fbox{0.27} & \fbox{2.19} & 40.71 & 0.86 & \fbox{0.94} & \fbox{0.44} & \fbox{\textbf{0.00}} & \fbox{\textbf{0.00}} & 0.91 & 0.97 & 0.24 & \fbox{\textbf{0.99}} & \fbox{3.59} & \fbox{\textbf{0.88}} & \fbox{\textbf{0.96}} & \fbox{0.30} & \fbox{\textbf{2.60}} & \fbox{22.46} & 0.81 & \fbox{\textbf{0.94}} & \fbox{0.41} \\
\hline
\multirow{3}{*}{BSOP} & EC2 & \multirow{3}{*}{1.86} & \multirow{3}{*}{126.30} & 0.79 & 0.89 & 0.30 & \multirow{3}{*}{1.90} & \multirow{3}{*}{67.50} & 0.77 & 0.87 & 0.35 & \multirow{3}{*}{1.89} & \multirow{3}{*}{42.97} & 0.72 & 0.84 & 0.53 & \multirow{3}{*}{\fbox{\textbf{0.00}}} & \multirow{3}{*}{\fbox{\textbf{0.00}}} & 0.82 & 0.92 & 0.18 & \multirow{3}{*}{0.00} & \multirow{3}{*}{102.77} & 0.68 & 0.84 & 0.51 & \multirow{3}{*}{0.00} & \multirow{3}{*}{149.11} & 0.52 & 0.74 & 0.66 \\
 & XMZ &  &  & 0.85 & 0.92 & 0.28 &  &  & 0.83 & 0.91 & 0.34 &  &  & 0.76 & 0.87 & 0.52 &  &  & 0.87 & 0.95 & \fbox{\textbf{0.16}} &  &  & 0.72 & 0.86 & 0.50 &  &  & 0.54 & 0.76 & 0.66 \\
 & AT+ &  &  & 0.91 & 0.96 & 0.46 &  &  & 0.88 & 0.95 & 0.50 &  &  & 0.80 & 0.92 & 0.62 &  &  & \fbox{\textbf{0.96}} & \fbox{\textbf{0.98}} & 0.34 &  &  & 0.74 & 0.91 & 0.59 &  &  & 0.56 & 0.86 & 0.71 \\
\hline
\multirow{3}{*}{WBSIP} & EC2 & \multirow{3}{*}{8.78} & \multirow{3}{*}{171.04} & 0.72 & 0.86 & 0.50 & \multirow{3}{*}{7.72} & \multirow{3}{*}{107.24} & 0.73 & 0.86 & 0.51 & \multirow{3}{*}{7.03} & \multirow{3}{*}{47.19} & 0.70 & 0.84 & 0.56 & \multirow{3}{*}{\textbf{0.00}} & \multirow{3}{*}{\textbf{0.00}} & 0.82 & 0.92 & 0.18 & \multirow{3}{*}{0.00} & \multirow{3}{*}{102.77} & 0.68 & 0.84 & 0.51 & \multirow{3}{*}{0.00} & \multirow{3}{*}{149.11} & 0.52 & 0.74 & 0.66 \\
 & XMZ &  &  & 0.79 & 0.90 & 0.50 &  &  & 0.80 & 0.90 & 0.50 &  &  & 0.77 & 0.89 & 0.55 &  &  & 0.87 & 0.95 & \textbf{0.16} &  &  & 0.72 & 0.86 & 0.50 &  &  & 0.54 & 0.76 & 0.66 \\
 & AT+ &  &  & 0.75 & 0.90 & 0.63 &  &  & 0.75 & 0.91 & 0.64 &  &  & 0.73 & 0.90 & 0.67 &  &  & \textbf{0.96} & \textbf{0.98} & 0.34 &  &  & 0.74 & 0.91 & 0.59 &  &  & 0.56 & 0.86 & 0.71 \\
\hline
\multirow{3}{*}{DCDP} & EC2 & \multirow{3}{*}{46.00} & \multirow{3}{*}{196.00} & 0.54 & 0.75 & 0.78 & \multirow{3}{*}{46.00} & \multirow{3}{*}{130.87} & 0.54 & 0.75 & 0.80 & \multirow{3}{*}{46.00} & \multirow{3}{*}{64.58} & 0.54 & 0.74 & 0.80 & \multirow{3}{*}{46.00} & \multirow{3}{*}{196.00} & 0.46 & 0.80 & 0.74 & \multirow{3}{*}{46.00} & \multirow{3}{*}{133.98} & 0.46 & 0.79 & 0.74 & \multirow{3}{*}{46.00} & \multirow{3}{*}{65.03} & 0.45 & 0.80 & 0.74 \\
 & XMZ &  &  & 0.66 & 0.84 & 0.79 &  &  & 0.66 & 0.84 & 0.82 &  &  & 0.66 & 0.84 & 0.81 &  &  & 0.58 & 0.85 & 0.73 &  &  & 0.57 & 0.85 & 0.74 &  &  & 0.57 & 0.85 & 0.73 \\
 & AT+ &  &  & 0.60 & 0.85 & 0.80 &  &  & 0.59 & 0.84 & 0.83 &  &  & 0.59 & 0.84 & 0.83 &  &  & 0.43 & 0.84 & 0.75 &  &  & 0.43 & 0.84 & 0.76 &  &  & 0.43 & 0.84 & 0.76 \\
\hline
\multirow{4}{*}{TBFL} & EC2 & \multirow{4}{*}{1.14} & \multirow{4}{*}{101.17} & 0.80 & 0.89 & 0.25 & \multirow{4}{*}{1.08} & \multirow{4}{*}{27.92} & 0.77 & 0.87 & 0.34 & \multirow{4}{*}{1.04} & \multirow{4}{*}{82.14} & 0.68 & 0.81 & 0.59 & \multirow{4}{*}{1.17} & \multirow{4}{*}{89.74} & 0.80 & 0.92 & 0.27 & \multirow{4}{*}{\fbox{\textbf{1.01}}} & \multirow{4}{*}{27.22} & 0.79 & 0.91 & 0.34 & \multirow{4}{*}{1.09} & \multirow{4}{*}{81.71} & 0.64 & 0.83 & 0.55 \\
 & XMZ &  &  & 0.86 & 0.93 & 0.22 &  &  & 0.83 & 0.91 & 0.33 &  &  & 0.72 & 0.84 & 0.59 &  &  & 0.86 & 0.94 & 0.23 &  &  & 0.83 & 0.93 & 0.31 &  &  & 0.67 & 0.85 & 0.54 \\
 & AT+ &  &  & 0.94 & 0.97 & 0.40 &  &  & 0.90 & 0.95 & 0.47 &  &  & 0.76 & 0.90 & 0.66 &  &  & 0.88 & 0.96 & 0.47 &  &  & 0.85 & 0.95 & 0.51 &  &  & 0.65 & 0.90 & 0.65 \\
 & Native &  &  & 0.74 & 0.81 & 0.55 &  &  & 0.67 & 0.77 & 0.77 &  &  & 0.60 & 0.75 & 0.95 &  &  & 0.52 & 0.63 & 2.21 &  &  & 0.48 & 0.64 & 1.91 &  &  & 0.46 & 0.69 & 1.13 \\
\hline\hline
\end{tabular}}}
\end{table}

The four metrics $d_h$, $\text{F}_1$, $\acc$ and MSE must be read jointly, so that a method attaining low estimation error on a mis-specified segmentation, or performing well only for one covariance structure or a narrow range of \(m^*\), is not preferable.
Under this criterion, Tables~\ref{tab:res-iid},~\ref{tab:res-banded} and~\ref{tab:res-factor} confirm that GFlsL, and especially its adaptive variant, is the most consistently competitive procedure, attaining the best or near-best fully data-driven value in at least three of the four metrics in most configurations and remaining close to the leading competitor in the fourth.

Adaptive weighting is particularly beneficial in the cases with breaks.
With no break present, \(m^*=0\), both estimators are conservative under BIC, HBICG and HBICG$(0,2)$ in \textbf{Settings~(i)} and \textbf{(iii)}, selecting \(\widehat{m}=0\) and thereby yielding \(d_h=0\).
The exception is \textbf{Setting~(ii)}, where the adaptive estimator retains a small number of spurious breaks.
Once breaks exist, the non-adaptive estimators with fully data-driven criteria fail in different ways.
In \textbf{Setting~(i)} at \((m^*,p)=(1,10)\), BIC selects no break at all, while HBICG and HBICG$(0,2)$ incur the largest Hausdorff distance among these criteria, and in \textbf{Setting~(iii)} at \((m^*,p)=(1,20)\), non-adaptive HBIC detects roughly the right number of breaks but misplaces them widely.
The adaptive estimator behaves differently in these same configurations.
At \((m^*,p)=(1,10)\), adaptive HBICG$(0,2)$ and adaptive HBIC$_c$ select approximately one break on average and improve all four metrics compared to every non-adaptive data-driven criterion, most notably in localization accuracy.
The gain is amplified in \textbf{Setting~(iii)}, where adaptive HBICG$(0,2)$ attains the best fully data-driven value in all four metrics at \((1,10)\) and \((1,20)\), and in three of the four at \((3,20)\), falling short of the best \(\text{F}_1\) by \(0.01\); adaptive HBIC is the closest competitor, with the best fully data-driven localization at \((3,10)\) and \(\text{F}_1\) at \((3,20)\).
Adaptive weighting therefore improves all four metrics jointly, most consistently under the rebalanced HBICG$(0,2)$, which supports the penalty structure adopted in its definition.

Turning to the detector-plus-refit pipelines, adaptive GFlsL balances localization, support recovery and estimation accuracy: in \textbf{Setting~(iii)} adaptive HBICG$(0,2)$ improves on every such pipeline in all four metrics at each configuration with breaks, and in \textbf{Setting~(i)} the adaptive criteria match or improve on them at \(m^*=1\), so that adaptive GFlsL is often the only fully data-driven method that performs well on all four metrics at once.
In contrast, \textbf{Setting~(ii)} presents a comparison that is less favorable to our estimator, because the true covariance matrices are banded, precisely the structure that the convex banding (CB) refit of~\citet{bien2016convex} targets; accordingly, at every configuration with breaks, TBFL paired with CB attains the strongest $\text{F}_1$ and support accuracy in Table~\ref{tab:res-banded}.
This strength is qualified both by the structural deficiencies of the first-stage detectors and by the prior knowledge of the banded structure that CB requires: BSOP detects no break at \(p=20\) regardless of the true \(m^*\), TBFL estimates approximately one break across all configurations, WBSIP over-segments severely at \(p=10\) but detects no break at \(p=20\), and DCDP returns approximately 46 estimated breaks nearly uniformly, so that errors in the first stage propagate irreversibly into the second, whereas GFlsL integrates detection and estimation within a single convex criterion.
Since the banded structure is unavailable in practice, this advantage cannot be exploited, and adaptive GFlsL is the most reliable of the fully data-driven procedures.

We close with an assessment of the fully data-driven criteria.
BIC penalizes so heavily that it selects essentially no break whenever breaks exist, except for the adaptive estimator in \textbf{Setting~(ii)}, while HBIC penalizes so lightly that it over-segments severely in \textbf{Setting~(ii)} and at \(p=20\) in \textbf{Setting~(i)}.
The calibrated criterion HBIC$_c$ corrects this for the adaptive estimator, but its plug-in penalty, computed from the extremes of the candidate path, is too weak on the non-adaptive path and over-selects; we therefore treat it as a useful adjunct rather than a standalone recommendation, as the sensitivity analysis in Subsection~\ref{subsec:sensitivity} confirms.
By contrast, HBICG is appropriate for the adaptive estimator in \textbf{Setting~(ii)} but conservative whenever breaks exist in \textbf{Settings~(i)} and \textbf{(iii)}, whereas its rebalanced variant HBICG$(0,2)$ is the most consistent criterion across the three settings.
No single criterion dominates, as expected since the penalty governs both the number of regimes and the covariance sparsity; overall, the evidence in Tables~\ref{tab:res-iid},~\ref{tab:res-banded} and~\ref{tab:res-factor} supports adaptive GFlsL selected by HBICG$(0,2)$ as the recommended procedure, with HBIC$_c$ a competitive alternative for the adaptive estimator.

\subsection{Sensitivity to adaptive weight exponents}\label{subsec:sensitivity}

The adaptive weights~\eqref{eq:adaptive-weight} depend on the element-wise shrinkage exponent \( \mu_1 \), the temporal-difference shrinkage exponent \( \mu_2 \), and the threshold floor \( a_T = T^{-1/2} \).
In this subsection, we examine the sensitivity of the adaptive estimator to \( (\mu_1, \mu_2) \) and assess whether the default choice \( (\mu_1, \mu_2) = (0.8, 1.5) \) adopted throughout the paper provides robust performance.
We consider \( \mu_1 \in \{0.5, 0.8, 1.5\} \) and \( \mu_2 \in \{0.5, 0.8, 1.5\} \), yielding nine configurations per scenario, with \( T = 200 \), \( p = 10 \), and \( m^* \in \{1, 3\} \) across the three data generating processes described in Subsection~\ref{subsec:simulations}.
To isolate the effect of the two exponents, the first-stage fit is deliberately held fixed, with its fusion parameter selected so that the detected break count is the smallest one not below \( m^* \).
We note that the replications used here differ from those of the main experiment, so the metric levels are not directly comparable with those in the previous subsection.
For each configuration, we report the five metrics averaged over 100 independent experiments under the criteria in Subsection~\ref{subsec:tuning-para}, where the ``optimal'' criterion is defined as in Subsection~\ref{subsec:exp}.
Tables~\ref{tab:sensitivity-m1} and~\ref{tab:sensitivity-m3} report the results for \( m^* = 1 \) and \( m^* = 3 \), respectively, with the default row \( (\mu_1, \mu_2) = (0.8, 1.5) \) shaded.

\begin{table}[ht!]
    \centering
    \caption{Sensitivity of the adaptive estimator to $(\mu_1, \mu_2)$ for $T = 200$, $p = 10$ and $m^* = 1$, averaged over 100 replications, with the first-stage fit held fixed.
    The default row $(\mu_1, \mu_2) = (0.8, 1.5)$ is shaded; \textbf{bold} marks the best value within each column.}
    \label{tab:sensitivity-m1}
    \resizebox{\textwidth}{!}{%
        \begin{tabular}{c|ccccc|ccccc|ccccc|ccccc|ccccc|ccccc|ccccc}
            \toprule
            \multicolumn{1}{c|}{\textbf{Setting~(i)}} & \multicolumn{5}{c|}{optimal} & \multicolumn{5}{c|}{lossval} & \multicolumn{5}{c|}{BIC} & \multicolumn{5}{c|}{HBIC} & \multicolumn{5}{c|}{HBIC$_c$} & \multicolumn{5}{c|}{HBICG} & \multicolumn{5}{c}{HBICG$(0,2)$} \\
            \( (\mu_1, \mu_2) \) & $nb$ & $d_h$ & $\text{F}_1$ & $\acc$ & MSE & $nb$ & $d_h$ & $\text{F}_1$ & $\acc$ & MSE & $nb$ & $d_h$ & $\text{F}_1$ & $\acc$ & MSE & $nb$ & $d_h$ & $\text{F}_1$ & $\acc$ & MSE & $nb$ & $d_h$ & $\text{F}_1$ & $\acc$ & MSE & $nb$ & $d_h$ & $\text{F}_1$ & $\acc$ & MSE & $nb$ & $d_h$ & $\text{F}_1$ & $\acc$ & MSE \\
            \midrule
            \( (0.5, 0.5) \) & 1.15 & 19.78 & 0.75 & 0.86 & 0.41 & 5.19 & 66.42 & 0.67 & 0.75 & \textbf{0.32} & 0.00 & \textbf{95.80} & 0.55 & 0.82 & \textbf{0.53} & 0.88 & 49.71 & 0.74 & 0.84 & 0.36 & 7.56 & 84.90 & 0.66 & 0.77 & \textbf{0.38} & 0.00 & 95.80 & 0.66 & 0.85 & 0.46 & 0.22 & 73.75 & 0.75 & 0.85 & 0.38 \\
            \( (0.5, 0.8) \) & 1.09 & 17.87 & 0.78 & 0.88 & 0.38 & 1.36 & \textbf{29.56} & 0.76 & 0.85 & 0.33 & 0.00 & \textbf{95.80} & 0.55 & 0.82 & \textbf{0.53} & 1.13 & 34.58 & 0.77 & 0.86 & \textbf{0.34} & 0.57 & 88.64 & 0.73 & 0.84 & 0.40 & 0.00 & 95.80 & 0.66 & 0.85 & 0.46 & 0.45 & \textbf{52.61} & 0.77 & 0.87 & \textbf{0.35} \\
            \( (0.5, 1.5) \) & 0.88 & 23.73 & 0.79 & 0.89 & 0.38 & 1.03 & 31.54 & 0.76 & 0.86 & 0.34 & 0.00 & \textbf{95.80} & 0.55 & 0.82 & \textbf{0.53} & 1.00 & 37.08 & 0.77 & 0.86 & \textbf{0.34} & 0.21 & 85.56 & 0.73 & 0.84 & 0.41 & 0.01 & \textbf{94.99} & 0.67 & 0.85 & 0.46 & 0.44 & 54.40 & 0.77 & 0.87 & \textbf{0.35} \\
            \( (0.8, 0.5) \) & 1.24 & 18.86 & 0.76 & 0.87 & 0.39 & 3.85 & 58.08 & 0.72 & 0.82 & 0.33 & 0.00 & \textbf{95.80} & 0.55 & 0.82 & \textbf{0.53} & 0.90 & 46.02 & 0.76 & 0.86 & 0.36 & 5.36 & 80.41 & 0.69 & 0.82 & 0.39 & 0.00 & 95.80 & 0.67 & 0.85 & 0.45 & 0.26 & 68.60 & 0.76 & 0.87 & 0.38 \\
            \( (0.8, 0.8) \) & 1.18 & \textbf{17.33} & 0.79 & 0.89 & \textbf{0.37} & 1.43 & 30.16 & 0.78 & 0.87 & 0.33 & 0.00 & \textbf{95.80} & 0.55 & 0.82 & \textbf{0.53} & 1.13 & 34.01 & 0.79 & 0.88 & \textbf{0.34} & 0.38 & 88.62 & 0.74 & 0.85 & 0.41 & 0.00 & 95.80 & 0.67 & 0.85 & 0.45 & 0.41 & 54.17 & 0.78 & 0.88 & 0.36 \\
            \rowcolor{gray!15} \( (0.8, 1.5) \) & 0.88 & 23.81 & \textbf{0.80} & 0.89 & \textbf{0.37} & 1.01 & 30.45 & 0.78 & 0.88 & 0.34 & 0.00 & \textbf{95.80} & 0.55 & 0.82 & \textbf{0.53} & 0.99 & 36.34 & 0.79 & 0.88 & 0.35 & 0.18 & 85.01 & 0.73 & 0.85 & 0.41 & 0.01 & \textbf{94.99} & 0.68 & 0.85 & 0.45 & 0.44 & 54.85 & \textbf{0.78} & 0.88 & 0.36 \\
            \( (1.5, 0.5) \) & 1.22 & 20.11 & 0.77 & 0.88 & 0.38 & 2.63 & 49.47 & 0.76 & 0.87 & 0.34 & 0.00 & \textbf{95.80} & \textbf{0.56} & \textbf{0.83} & \textbf{0.53} & 1.14 & 41.28 & 0.78 & 0.89 & 0.36 & 2.45 & \textbf{78.71} & 0.73 & 0.85 & 0.40 & 0.01 & \textbf{94.99} & 0.69 & \textbf{0.86} & \textbf{0.44} & 0.22 & 73.24 & 0.76 & 0.88 & 0.39 \\
            \( (1.5, 0.8) \) & 1.25 & 18.51 & 0.79 & 0.89 & \textbf{0.37} & 1.38 & 30.47 & \textbf{0.78} & 0.89 & 0.35 & 0.00 & \textbf{95.80} & 0.55 & 0.82 & \textbf{0.53} & 1.22 & \textbf{33.31} & \textbf{0.79} & \textbf{0.89} & 0.36 & 0.32 & 90.96 & \textbf{0.74} & 0.86 & 0.42 & 0.01 & \textbf{94.99} & \textbf{0.69} & \textbf{0.86} & \textbf{0.44} & 0.34 & 63.62 & 0.77 & 0.88 & 0.38 \\
            \( (1.5, 1.5) \) & 0.87 & 24.07 & 0.79 & \textbf{0.89} & \textbf{0.37} & 1.02 & 30.54 & 0.78 & \textbf{0.89} & 0.36 & 0.00 & \textbf{95.80} & 0.55 & \textbf{0.83} & \textbf{0.53} & 0.98 & 36.34 & 0.79 & 0.89 & 0.36 & 0.16 & 86.65 & 0.74 & \textbf{0.87} & 0.42 & 0.01 & \textbf{94.99} & 0.69 & \textbf{0.86} & \textbf{0.44} & 0.40 & 59.52 & 0.77 & \textbf{0.88} & 0.38 \\
            \midrule
            \multicolumn{1}{c|}{\textbf{Setting~(ii)}} & \multicolumn{5}{c|}{optimal} & \multicolumn{5}{c|}{lossval} & \multicolumn{5}{c|}{BIC} & \multicolumn{5}{c|}{HBIC} & \multicolumn{5}{c|}{HBIC$_c$} & \multicolumn{5}{c|}{HBICG} & \multicolumn{5}{c}{HBICG$(0,2)$} \\
            \( (\mu_1, \mu_2) \) & $nb$ & $d_h$ & $\text{F}_1$ & $\acc$ & MSE & $nb$ & $d_h$ & $\text{F}_1$ & $\acc$ & MSE & $nb$ & $d_h$ & $\text{F}_1$ & $\acc$ & MSE & $nb$ & $d_h$ & $\text{F}_1$ & $\acc$ & MSE & $nb$ & $d_h$ & $\text{F}_1$ & $\acc$ & MSE & $nb$ & $d_h$ & $\text{F}_1$ & $\acc$ & MSE & $nb$ & $d_h$ & $\text{F}_1$ & $\acc$ & MSE \\
            \midrule
            \( (0.5, 0.5) \) & 1.40 & \textbf{9.72} & 0.97 & 0.98 & 1.50 & 2.17 & 23.91 & 0.72 & 0.77 & \textbf{1.07} & 1.24 & \textbf{13.56} & 0.92 & 0.96 & 1.81 & 92.09 & 130.38 & 0.46 & 0.33 & 3.10 & 1.96 & 20.79 & 0.58 & 0.58 & \textbf{1.25} & 1.24 & \textbf{13.56} & 0.91 & 0.95 & 1.25 & 1.62 & 14.94 & 0.57 & 0.58 & 1.20 \\
            \( (0.5, 0.8) \) & 1.49 & 11.34 & \textbf{0.97} & \textbf{0.98} & 1.50 & 2.71 & 20.24 & 0.74 & 0.80 & \textbf{1.07} & 1.38 & 15.18 & 0.91 & 0.96 & 1.80 & 26.18 & 119.24 & 0.47 & 0.36 & 1.78 & 2.30 & 19.07 & 0.58 & 0.58 & 1.34 & 1.38 & 15.18 & 0.91 & 0.95 & 1.24 & 1.72 & 15.30 & 0.58 & 0.59 & 1.24 \\
            \( (0.5, 1.5) \) & 1.86 & 13.93 & 0.97 & 0.98 & 1.49 & 2.09 & 15.89 & 0.79 & 0.84 & 1.12 & 1.76 & 17.01 & 0.91 & 0.95 & 1.87 & 2.81 & \textbf{19.69} & 0.57 & 0.58 & 1.52 & 1.78 & \textbf{17.01} & 0.64 & 0.66 & 1.34 & 1.76 & 17.01 & 0.91 & 0.94 & 1.29 & 2.06 & 16.52 & 0.59 & 0.60 & 1.32 \\
            \( (0.8, 0.5) \) & 1.47 & 10.20 & 0.97 & 0.98 & 1.39 & 2.12 & 23.95 & 0.76 & 0.82 & \textbf{1.07} & 1.27 & 14.04 & \textbf{0.95} & \textbf{0.97} & 1.49 & 91.13 & 130.22 & 0.47 & 0.37 & 3.02 & 1.97 & 21.28 & 0.60 & 0.62 & \textbf{1.25} & 1.27 & 14.04 & 0.92 & 0.95 & 1.21 & 1.60 & 14.68 & 0.60 & 0.63 & 1.19 \\
            \( (0.8, 0.8) \) & 1.57 & 11.36 & 0.97 & 0.98 & 1.37 & 2.58 & 19.57 & 0.78 & 0.84 & \textbf{1.07} & 1.40 & 15.20 & 0.95 & \textbf{0.97} & 1.49 & 25.77 & 118.57 & 0.49 & 0.41 & 1.74 & 2.30 & 18.97 & 0.60 & 0.63 & 1.33 & 1.40 & 15.20 & 0.92 & 0.95 & 1.21 & 1.67 & 14.96 & 0.61 & 0.63 & 1.21 \\
            \rowcolor{gray!15} \( (0.8, 1.5) \) & 1.88 & 13.93 & 0.97 & 0.98 & 1.39 & 2.08 & 15.86 & 0.82 & 0.87 & 1.12 & 1.80 & 17.02 & 0.94 & 0.97 & 1.54 & 2.81 & \textbf{19.69} & 0.60 & 0.62 & 1.51 & 1.80 & \textbf{17.01} & 0.66 & 0.69 & 1.31 & 1.79 & 17.01 & 0.92 & 0.95 & 1.23 & 1.98 & 15.43 & 0.64 & 0.68 & 1.26 \\
            \( (1.5, 0.5) \) & 1.52 & 10.85 & 0.94 & 0.96 & 1.23 & 2.02 & 21.61 & 0.82 & 0.87 & \textbf{1.07} & 1.40 & 14.88 & 0.94 & 0.96 & \textbf{1.24} & 85.20 & 129.78 & 0.52 & 0.48 & 2.82 & 1.95 & 18.37 & 0.65 & 0.69 & 1.26 & 1.39 & 14.83 & 0.91 & 0.95 & \textbf{1.17} & 1.52 & 14.59 & 0.65 & 0.70 & \textbf{1.18} \\
            \( (1.5, 0.8) \) & 1.59 & 11.76 & 0.95 & 0.97 & \textbf{1.22} & 2.70 & 22.36 & 0.83 & 0.88 & \textbf{1.07} & 1.51 & 15.64 & 0.94 & 0.97 & \textbf{1.24} & 22.96 & 114.93 & 0.54 & 0.52 & 1.68 & 2.32 & 19.06 & 0.64 & 0.69 & 1.34 & 1.50 & 15.63 & 0.92 & 0.95 & \textbf{1.17} & 1.61 & 15.33 & 0.69 & 0.73 & 1.19 \\
            \( (1.5, 1.5) \) & 1.90 & 13.94 & 0.95 & 0.97 & 1.26 & 2.05 & \textbf{15.83} & \textbf{0.86} & \textbf{0.90} & 1.12 & 1.81 & 17.02 & 0.95 & 0.97 & 1.28 & 2.80 & \textbf{19.69} & \textbf{0.64} & \textbf{0.69} & \textbf{1.49} & 1.82 & 17.03 & \textbf{0.68} & \textbf{0.73} & 1.31 & 1.81 & 17.02 & \textbf{0.93} & \textbf{0.95} & 1.20 & 1.88 & \textbf{14.15} & \textbf{0.76} & \textbf{0.81} & 1.20 \\
            \midrule
            \multicolumn{1}{c|}{\textbf{Setting~(iii)}} & \multicolumn{5}{c|}{optimal} & \multicolumn{5}{c|}{lossval} & \multicolumn{5}{c|}{BIC} & \multicolumn{5}{c|}{HBIC} & \multicolumn{5}{c|}{HBIC$_c$} & \multicolumn{5}{c|}{HBICG} & \multicolumn{5}{c}{HBICG$(0,2)$} \\
            \( (\mu_1, \mu_2) \) & $nb$ & $d_h$ & $\text{F}_1$ & $\acc$ & MSE & $nb$ & $d_h$ & $\text{F}_1$ & $\acc$ & MSE & $nb$ & $d_h$ & $\text{F}_1$ & $\acc$ & MSE & $nb$ & $d_h$ & $\text{F}_1$ & $\acc$ & MSE & $nb$ & $d_h$ & $\text{F}_1$ & $\acc$ & MSE & $nb$ & $d_h$ & $\text{F}_1$ & $\acc$ & MSE & $nb$ & $d_h$ & $\text{F}_1$ & $\acc$ & MSE \\
            \midrule
            \( (0.5, 0.5) \) & 1.12 & \textbf{5.59} & 0.92 & 0.96 & 0.35 & 2.20 & 31.20 & 0.89 & 0.94 & \textbf{0.27} & 0.02 & 105.87 & 0.66 & 0.88 & 0.71 & 1.40 & 20.39 & 0.92 & 0.96 & \textbf{0.28} & 1.12 & 67.37 & 0.85 & 0.93 & \textbf{0.44} & 0.62 & 50.70 & 0.86 & 0.94 & 0.39 & 0.95 & 21.85 & 0.92 & 0.96 & 0.30 \\
            \( (0.5, 0.8) \) & 1.03 & 6.69 & \textbf{0.93} & \textbf{0.97} & \textbf{0.32} & 1.17 & 10.43 & \textbf{0.92} & \textbf{0.96} & \textbf{0.27} & 0.15 & 94.11 & 0.66 & 0.89 & 0.70 & 1.10 & 13.68 & \textbf{0.93} & \textbf{0.97} & \textbf{0.28} & 0.23 & 88.62 & 0.83 & 0.92 & 0.52 & 0.70 & 43.21 & 0.88 & 0.95 & 0.35 & 0.95 & \textbf{17.04} & \textbf{0.93} & \textbf{0.97} & \textbf{0.29} \\
            \( (0.5, 1.5) \) & 0.94 & 12.92 & 0.93 & 0.97 & \textbf{0.32} & 1.03 & 14.12 & \textbf{0.92} & \textbf{0.96} & 0.29 & 0.37 & 73.59 & 0.66 & 0.89 & 0.68 & 0.98 & 19.70 & 0.93 & 0.97 & 0.29 & 0.48 & 65.16 & \textbf{0.85} & 0.94 & \textbf{0.44} & 0.71 & 38.20 & \textbf{0.89} & \textbf{0.96} & \textbf{0.34} & 0.90 & 19.57 & 0.93 & 0.97 & \textbf{0.29} \\
            \( (0.8, 0.5) \) & 1.10 & \textbf{5.59} & 0.92 & 0.96 & 0.34 & 2.03 & 29.82 & 0.90 & 0.95 & 0.28 & 0.04 & 103.98 & 0.67 & 0.89 & 0.71 & 1.46 & 19.77 & 0.92 & 0.96 & 0.29 & 0.48 & 70.97 & 0.85 & 0.94 & 0.45 & 0.64 & 47.04 & 0.86 & 0.94 & 0.39 & 0.93 & 22.50 & 0.91 & 0.96 & 0.31 \\
            \( (0.8, 0.8) \) & 1.03 & 6.69 & 0.93 & 0.97 & \textbf{0.32} & 1.16 & \textbf{9.76} & 0.92 & \textbf{0.96} & 0.28 & 0.15 & 94.11 & 0.67 & 0.89 & 0.69 & 1.11 & 12.42 & 0.93 & 0.97 & \textbf{0.28} & 0.22 & 89.88 & 0.83 & 0.93 & 0.52 & 0.72 & 40.14 & 0.88 & 0.95 & 0.35 & 0.94 & 18.69 & 0.92 & 0.97 & \textbf{0.29} \\
            \rowcolor{gray!15} \( (0.8, 1.5) \) & 0.94 & 12.92 & 0.93 & 0.97 & \textbf{0.32} & 1.03 & 14.12 & 0.92 & \textbf{0.96} & 0.29 & 0.37 & 73.59 & 0.66 & 0.89 & 0.67 & 0.97 & 19.70 & 0.92 & 0.97 & 0.30 & 0.47 & 66.81 & 0.85 & 0.94 & \textbf{0.44} & 0.71 & 38.77 & 0.89 & 0.95 & 0.35 & 0.91 & 19.57 & 0.92 & 0.96 & 0.30 \\
            \( (1.5, 0.5) \) & 1.14 & 6.39 & 0.90 & 0.96 & 0.35 & 1.81 & 27.13 & 0.90 & 0.96 & 0.29 & 0.07 & 101.89 & 0.69 & 0.89 & 0.67 & 1.51 & 18.51 & 0.90 & 0.96 & 0.30 & 0.41 & 76.00 & 0.84 & 0.93 & 0.46 & 0.65 & 45.54 & 0.86 & 0.94 & 0.38 & 0.92 & 22.73 & 0.90 & 0.95 & 0.32 \\
            \( (1.5, 0.8) \) & 1.04 & 6.69 & 0.91 & 0.96 & 0.33 & 1.18 & 9.77 & 0.90 & 0.96 & 0.30 & 0.17 & 93.17 & 0.70 & 0.90 & 0.66 & 1.10 & \textbf{10.87} & 0.90 & 0.96 & 0.30 & 0.25 & 88.50 & 0.82 & 0.93 & 0.51 & 0.70 & 42.35 & 0.86 & 0.95 & 0.37 & 0.91 & 21.03 & 0.90 & 0.96 & 0.31 \\
            \( (1.5, 1.5) \) & 0.95 & 12.92 & 0.91 & 0.96 & 0.33 & 1.03 & 14.12 & 0.90 & 0.96 & 0.31 & 0.40 & \textbf{68.93} & \textbf{0.70} & \textbf{0.90} & \textbf{0.62} & 0.97 & 19.70 & 0.90 & 0.96 & 0.31 & 0.49 & \textbf{63.82} & 0.85 & \textbf{0.94} & \textbf{0.44} & 0.72 & \textbf{37.28} & 0.87 & 0.95 & 0.35 & 0.90 & 21.05 & 0.90 & 0.96 & 0.31 \\
            \bottomrule
        \end{tabular}%
    }
\end{table}

\begin{table}[ht!]
    \centering
    \caption{Sensitivity of the adaptive estimator to $(\mu_1, \mu_2)$ for $T = 200$, $p = 10$ and $m^* = 3$, averaged over 100 replications, with the first-stage fit held fixed.
    The default row $(\mu_1, \mu_2) = (0.8, 1.5)$ is shaded; \textbf{bold} marks the best value within each column.}
    \label{tab:sensitivity-m3}
    \resizebox{\textwidth}{!}{%
        \begin{tabular}{c|ccccc|ccccc|ccccc|ccccc|ccccc|ccccc|ccccc}
            \toprule
            \multicolumn{1}{c|}{\textbf{Setting~(i)}} & \multicolumn{5}{c|}{optimal} & \multicolumn{5}{c|}{lossval} & \multicolumn{5}{c|}{BIC} & \multicolumn{5}{c|}{HBIC} & \multicolumn{5}{c|}{HBIC$_c$} & \multicolumn{5}{c|}{HBICG} & \multicolumn{5}{c}{HBICG$(0,2)$} \\
            \( (\mu_1, \mu_2) \) & $nb$ & $d_h$ & $\text{F}_1$ & $\acc$ & MSE & $nb$ & $d_h$ & $\text{F}_1$ & $\acc$ & MSE & $nb$ & $d_h$ & $\text{F}_1$ & $\acc$ & MSE & $nb$ & $d_h$ & $\text{F}_1$ & $\acc$ & MSE & $nb$ & $d_h$ & $\text{F}_1$ & $\acc$ & MSE & $nb$ & $d_h$ & $\text{F}_1$ & $\acc$ & MSE & $nb$ & $d_h$ & $\text{F}_1$ & $\acc$ & MSE \\
            \midrule
            \( (0.5, 0.5) \) & 5.09 & \textbf{25.31} & 0.64 & 0.78 & \textbf{0.47} & 6.62 & \textbf{32.89} & 0.61 & 0.71 & \textbf{0.43} & 0.00 & 154.20 & 0.53 & 0.82 & \textbf{0.55} & 1.78 & 83.56 & 0.64 & 0.78 & 0.47 & 5.28 & \textbf{71.41} & 0.61 & 0.74 & \textbf{0.47} & 0.00 & 154.20 & 0.54 & 0.82 & \textbf{0.54} & 0.09 & 147.93 & 0.62 & 0.79 & 0.51 \\
            \( (0.5, 0.8) \) & 2.70 & 47.41 & 0.68 & 0.82 & \textbf{0.47} & 2.67 & 50.80 & 0.66 & 0.78 & 0.44 & 0.00 & 154.20 & 0.53 & 0.81 & \textbf{0.55} & 2.27 & 62.01 & 0.66 & 0.80 & \textbf{0.45} & 0.49 & 129.90 & 0.63 & 0.79 & 0.50 & 0.00 & 154.20 & 0.54 & 0.82 & \textbf{0.54} & 0.22 & 138.85 & 0.63 & 0.79 & 0.50 \\
            \( (0.5, 1.5) \) & 2.21 & 52.50 & 0.68 & 0.82 & \textbf{0.47} & 2.20 & 55.24 & 0.66 & 0.79 & 0.45 & 0.01 & \textbf{152.97} & \textbf{0.55} & 0.82 & \textbf{0.55} & 1.85 & 65.97 & 0.67 & 0.80 & \textbf{0.45} & 0.43 & 129.22 & \textbf{0.63} & 0.80 & 0.50 & 0.01 & \textbf{152.97} & 0.55 & 0.82 & \textbf{0.54} & 0.33 & \textbf{131.47} & \textbf{0.64} & 0.80 & \textbf{0.49} \\
            \( (0.8, 0.5) \) & 4.63 & 29.17 & 0.66 & 0.81 & \textbf{0.47} & 5.30 & 36.83 & 0.66 & 0.78 & 0.44 & 0.00 & 154.20 & 0.53 & 0.82 & \textbf{0.55} & 2.23 & 73.89 & 0.66 & 0.81 & 0.47 & 3.28 & 91.54 & 0.63 & 0.79 & 0.49 & 0.00 & 154.20 & 0.55 & 0.82 & \textbf{0.54} & 0.12 & 146.35 & 0.63 & 0.80 & 0.51 \\
            \( (0.8, 0.8) \) & 2.66 & 48.28 & \textbf{0.69} & 0.84 & \textbf{0.47} & 2.60 & 49.99 & 0.68 & 0.81 & 0.45 & 0.00 & 154.20 & 0.54 & 0.82 & \textbf{0.55} & 2.34 & 59.29 & \textbf{0.68} & 0.82 & 0.46 & 0.36 & 135.61 & 0.63 & 0.81 & 0.51 & 0.00 & 154.20 & 0.55 & 0.82 & \textbf{0.54} & 0.23 & 137.96 & 0.63 & 0.81 & 0.50 \\
            \rowcolor{gray!15} \( (0.8, 1.5) \) & 2.13 & 52.88 & 0.69 & 0.84 & \textbf{0.47} & 2.12 & 56.06 & 0.68 & 0.82 & 0.45 & 0.01 & \textbf{152.97} & \textbf{0.55} & 0.82 & \textbf{0.55} & 1.82 & 65.38 & 0.68 & 0.82 & 0.46 & 0.33 & 134.16 & \textbf{0.63} & 0.81 & 0.51 & 0.01 & \textbf{152.97} & 0.56 & 0.82 & \textbf{0.54} & 0.27 & 135.60 & 0.64 & 0.81 & 0.50 \\
            \( (1.5, 0.5) \) & 3.62 & 33.27 & 0.67 & 0.84 & 0.48 & 3.83 & 39.13 & 0.68 & 0.83 & 0.46 & 0.00 & 154.20 & 0.54 & 0.82 & \textbf{0.55} & 2.47 & 64.31 & 0.67 & 0.84 & 0.47 & 1.42 & 113.47 & 0.63 & 0.82 & 0.50 & 0.00 & 154.20 & 0.55 & 0.82 & \textbf{0.54} & 0.15 & 143.41 & 0.62 & 0.82 & 0.51 \\
            \( (1.5, 0.8) \) & 2.52 & 48.70 & 0.68 & 0.85 & \textbf{0.47} & 2.53 & 51.06 & \textbf{0.68} & 0.84 & 0.46 & 0.00 & 154.20 & 0.54 & 0.82 & \textbf{0.55} & 2.34 & \textbf{56.89} & \textbf{0.68} & \textbf{0.84} & 0.47 & 0.26 & 140.26 & 0.62 & 0.82 & 0.51 & 0.00 & 154.20 & 0.55 & 0.82 & \textbf{0.54} & 0.24 & 137.60 & 0.63 & 0.82 & 0.51 \\
            \( (1.5, 1.5) \) & 2.06 & 53.98 & 0.69 & \textbf{0.85} & 0.48 & 1.99 & 58.74 & 0.68 & \textbf{0.84} & 0.47 & 0.01 & \textbf{152.97} & \textbf{0.55} & \textbf{0.82} & \textbf{0.55} & 1.78 & 66.07 & 0.68 & \textbf{0.84} & 0.47 & 0.33 & 134.32 & 0.63 & \textbf{0.82} & 0.52 & 0.01 & \textbf{152.97} & \textbf{0.56} & \textbf{0.82} & \textbf{0.54} & 0.28 & 134.64 & 0.64 & \textbf{0.82} & 0.51 \\
            \midrule
            \multicolumn{1}{c|}{\textbf{Setting~(ii)}} & \multicolumn{5}{c|}{optimal} & \multicolumn{5}{c|}{lossval} & \multicolumn{5}{c|}{BIC} & \multicolumn{5}{c|}{HBIC} & \multicolumn{5}{c|}{HBIC$_c$} & \multicolumn{5}{c|}{HBICG} & \multicolumn{5}{c}{HBICG$(0,2)$} \\
            \( (\mu_1, \mu_2) \) & $nb$ & $d_h$ & $\text{F}_1$ & $\acc$ & MSE & $nb$ & $d_h$ & $\text{F}_1$ & $\acc$ & MSE & $nb$ & $d_h$ & $\text{F}_1$ & $\acc$ & MSE & $nb$ & $d_h$ & $\text{F}_1$ & $\acc$ & MSE & $nb$ & $d_h$ & $\text{F}_1$ & $\acc$ & MSE & $nb$ & $d_h$ & $\text{F}_1$ & $\acc$ & MSE & $nb$ & $d_h$ & $\text{F}_1$ & $\acc$ & MSE \\
            \midrule
            \( (0.5, 0.5) \) & 6.64 & \textbf{18.83} & 0.81 & 0.86 & 1.97 & 4.56 & \textbf{28.17} & 0.67 & 0.73 & \textbf{1.65} & 1.80 & 56.73 & 0.81 & 0.91 & 2.28 & 92.22 & 64.90 & 0.46 & 0.33 & 3.28 & 3.32 & 40.12 & 0.55 & 0.54 & \textbf{1.81} & 1.81 & 56.50 & 0.85 & 0.93 & 2.07 & 2.69 & 44.18 & 0.55 & 0.55 & 1.80 \\
            \( (0.5, 0.8) \) & 7.24 & 30.05 & 0.86 & 0.92 & 2.05 & 8.80 & 40.42 & 0.68 & 0.73 & \textbf{1.65} & 2.02 & 51.85 & 0.81 & 0.91 & 2.25 & 26.22 & 56.14 & 0.47 & 0.38 & 2.10 & 3.32 & 43.15 & 0.55 & 0.55 & 1.87 & 2.02 & 51.85 & 0.86 & 0.93 & 2.03 & 2.69 & 44.69 & 0.56 & 0.55 & 1.83 \\
            \( (0.5, 1.5) \) & 3.45 & 41.15 & \textbf{0.89} & \textbf{0.94} & 1.99 & 3.18 & 43.01 & 0.73 & 0.80 & 1.72 & 2.50 & 47.67 & 0.80 & 0.91 & 2.24 & 4.39 & \textbf{41.94} & 0.55 & 0.54 & 2.11 & 2.63 & 45.04 & 0.61 & 0.63 & 1.92 & 2.50 & 47.67 & 0.84 & 0.92 & 2.05 & 3.07 & \textbf{43.75} & 0.56 & 0.56 & 1.91 \\
            \( (0.8, 0.5) \) & 6.98 & 19.47 & 0.82 & 0.88 & 1.92 & 4.44 & 29.60 & 0.69 & 0.75 & 1.66 & 1.95 & 52.71 & 0.85 & 0.93 & 2.09 & 91.00 & 64.87 & 0.47 & 0.36 & 3.21 & 3.32 & \textbf{39.46} & 0.57 & 0.59 & \textbf{1.81} & 1.95 & 52.71 & 0.87 & 0.93 & 1.97 & 2.64 & 44.18 & 0.58 & 0.59 & \textbf{1.79} \\
            \( (0.8, 0.8) \) & 7.23 & 30.56 & 0.87 & 0.93 & 1.97 & 8.44 & 41.25 & 0.70 & 0.77 & 1.66 & 2.12 & 50.15 & 0.85 & 0.93 & 2.06 & 25.85 & 55.72 & 0.49 & 0.42 & 2.08 & 3.34 & 42.28 & 0.57 & 0.59 & 1.88 & 2.12 & 50.15 & \textbf{0.87} & \textbf{0.93} & 1.94 & 2.59 & 44.89 & 0.58 & 0.60 & 1.81 \\
            \rowcolor{gray!15} \( (0.8, 1.5) \) & 3.51 & 41.15 & 0.89 & 0.94 & 1.95 & 3.15 & 43.01 & 0.75 & 0.82 & 1.73 & 2.55 & 47.64 & 0.84 & 0.92 & 2.07 & 4.39 & \textbf{41.94} & 0.57 & 0.59 & 2.10 & 2.63 & 45.04 & 0.64 & 0.67 & 1.90 & 2.55 & 47.64 & 0.86 & 0.93 & 1.95 & 2.99 & 43.98 & 0.60 & 0.62 & 1.87 \\
            \( (1.5, 0.5) \) & 7.35 & 19.46 & 0.82 & 0.88 & 1.87 & 5.22 & 29.39 & 0.72 & 0.79 & 1.69 & 2.09 & 50.62 & \textbf{0.86} & 0.93 & 1.91 & 85.32 & 64.53 & 0.51 & 0.47 & 3.03 & 3.26 & 41.17 & 0.61 & 0.65 & 1.83 & 2.09 & 50.62 & 0.86 & 0.93 & 1.88 & 2.53 & 45.04 & 0.62 & 0.66 & 1.80 \\
            \( (1.5, 0.8) \) & 7.59 & 30.99 & 0.85 & 0.91 & \textbf{1.86} & 8.09 & 40.95 & 0.73 & 0.80 & 1.69 & 2.22 & 49.12 & 0.86 & 0.93 & \textbf{1.90} & 23.77 & 54.78 & 0.54 & 0.53 & \textbf{2.03} & 3.33 & 43.30 & 0.61 & 0.65 & 1.89 & 2.21 & 49.14 & 0.86 & 0.93 & \textbf{1.87} & 2.48 & 45.23 & 0.63 & 0.67 & 1.81 \\
            \( (1.5, 1.5) \) & 3.50 & 41.17 & 0.88 & 0.93 & 1.89 & 3.10 & 42.89 & \textbf{0.78} & \textbf{0.85} & 1.75 & 2.64 & \textbf{45.78} & 0.86 & \textbf{0.93} & \textbf{1.90} & 4.38 & \textbf{41.94} & \textbf{0.61} & \textbf{0.65} & 2.08 & 2.69 & 44.96 & \textbf{0.66} & \textbf{0.70} & 1.90 & 2.65 & \textbf{45.78} & 0.86 & 0.92 & 1.88 & 2.78 & 44.62 & \textbf{0.68} & \textbf{0.73} & 1.84 \\
            \midrule
            \multicolumn{1}{c|}{\textbf{Setting~(iii)}} & \multicolumn{5}{c|}{optimal} & \multicolumn{5}{c|}{lossval} & \multicolumn{5}{c|}{BIC} & \multicolumn{5}{c|}{HBIC} & \multicolumn{5}{c|}{HBIC$_c$} & \multicolumn{5}{c|}{HBICG} & \multicolumn{5}{c}{HBICG$(0,2)$} \\
            \( (\mu_1, \mu_2) \) & $nb$ & $d_h$ & $\text{F}_1$ & $\acc$ & MSE & $nb$ & $d_h$ & $\text{F}_1$ & $\acc$ & MSE & $nb$ & $d_h$ & $\text{F}_1$ & $\acc$ & MSE & $nb$ & $d_h$ & $\text{F}_1$ & $\acc$ & MSE & $nb$ & $d_h$ & $\text{F}_1$ & $\acc$ & MSE & $nb$ & $d_h$ & $\text{F}_1$ & $\acc$ & MSE & $nb$ & $d_h$ & $\text{F}_1$ & $\acc$ & MSE \\
            \midrule
            \( (0.5, 0.5) \) & 4.05 & \textbf{17.23} & 0.82 & 0.92 & 0.48 & 4.44 & 21.99 & 0.81 & 0.91 & \textbf{0.41} & 0.03 & 149.81 & 0.64 & 0.88 & 0.74 & 3.22 & 32.83 & 0.83 & 0.92 & \textbf{0.44} & 0.91 & 114.60 & 0.73 & 0.88 & 0.62 & 0.38 & 129.61 & 0.68 & 0.88 & 0.68 & 2.05 & 49.09 & 0.81 & 0.92 & 0.48 \\
            \( (0.5, 0.8) \) & 2.57 & 40.61 & 0.84 & 0.93 & \textbf{0.47} & 2.56 & 41.87 & 0.83 & 0.93 & 0.44 & 0.20 & 139.63 & 0.65 & 0.88 & 0.73 & 2.48 & 44.34 & 0.84 & 0.93 & 0.45 & 0.37 & 130.63 & 0.71 & 0.87 & 0.65 & 0.62 & 114.39 & 0.69 & 0.89 & 0.65 & 2.01 & 49.70 & 0.83 & 0.92 & \textbf{0.46} \\
            \( (0.5, 1.5) \) & 2.29 & 44.94 & \textbf{0.84} & 0.93 & 0.49 & 2.30 & 45.49 & 0.83 & 0.93 & 0.45 & 0.56 & 115.31 & 0.66 & 0.89 & 0.71 & 2.29 & 47.77 & 0.84 & 0.93 & 0.46 & 0.78 & 108.29 & 0.75 & 0.89 & \textbf{0.61} & 0.87 & 98.54 & 0.72 & 0.90 & 0.62 & 1.79 & 54.89 & 0.83 & 0.92 & 0.47 \\
            \( (0.8, 0.5) \) & 3.95 & 17.84 & 0.82 & 0.93 & \textbf{0.47} & 4.19 & \textbf{21.98} & 0.83 & 0.92 & 0.42 & 0.03 & 149.81 & 0.65 & 0.88 & 0.74 & 3.28 & 32.24 & 0.84 & 0.93 & \textbf{0.44} & 0.69 & 118.95 & 0.73 & 0.88 & 0.63 & 0.42 & 127.20 & 0.69 & 0.88 & 0.67 & 2.09 & \textbf{48.12} & 0.82 & 0.92 & 0.48 \\
            \( (0.8, 0.8) \) & 2.53 & 40.78 & 0.84 & 0.93 & \textbf{0.47} & 2.53 & 41.93 & \textbf{0.84} & 0.93 & 0.45 & 0.20 & 139.63 & 0.65 & 0.88 & 0.73 & 2.46 & 44.34 & \textbf{0.84} & \textbf{0.93} & 0.46 & 0.37 & 130.83 & 0.72 & 0.88 & 0.66 & 0.61 & 115.20 & 0.70 & 0.89 & 0.65 & 2.00 & 50.24 & \textbf{0.83} & \textbf{0.93} & 0.47 \\
            \rowcolor{gray!15} \( (0.8, 1.5) \) & 2.32 & 45.40 & \textbf{0.84} & \textbf{0.93} & 0.48 & 2.29 & 45.72 & \textbf{0.84} & \textbf{0.93} & 0.46 & 0.57 & 115.00 & 0.66 & 0.89 & 0.71 & 2.29 & 47.41 & 0.84 & \textbf{0.93} & 0.47 & 0.75 & 107.59 & \textbf{0.75} & 0.90 & \textbf{0.61} & 0.93 & 94.23 & 0.72 & 0.90 & 0.61 & 1.80 & 54.89 & 0.83 & 0.93 & 0.48 \\
            \( (1.5, 0.5) \) & 3.71 & 18.66 & 0.81 & 0.92 & 0.48 & 3.98 & 22.65 & 0.83 & 0.93 & 0.44 & 0.07 & 148.04 & 0.66 & 0.88 & 0.74 & 3.49 & \textbf{29.88} & 0.83 & 0.93 & 0.45 & 0.58 & 123.45 & 0.73 & 0.89 & 0.64 & 0.49 & 123.11 & 0.70 & 0.89 & 0.66 & 2.01 & 50.14 & 0.81 & 0.92 & 0.50 \\
            \( (1.5, 0.8) \) & 2.60 & 40.78 & 0.83 & 0.93 & 0.49 & 2.51 & 41.93 & 0.82 & 0.93 & 0.47 & 0.21 & 138.68 & 0.66 & 0.89 & 0.72 & 2.47 & 44.34 & 0.82 & 0.93 & 0.47 & 0.37 & 131.42 & 0.72 & 0.89 & 0.66 & 0.68 & 110.77 & 0.71 & 0.89 & 0.63 & 1.91 & 51.40 & 0.82 & 0.92 & 0.49 \\
            \( (1.5, 1.5) \) & 2.34 & 45.40 & 0.83 & 0.93 & 0.49 & 2.31 & 45.72 & 0.82 & 0.93 & 0.47 & 0.60 & \textbf{113.00} & \textbf{0.67} & \textbf{0.89} & \textbf{0.69} & 2.28 & 47.04 & 0.82 & 0.93 & 0.48 & 0.79 & \textbf{107.09} & 0.75 & \textbf{0.90} & \textbf{0.61} & 0.95 & \textbf{93.44} & \textbf{0.74} & \textbf{0.90} & \textbf{0.60} & 1.79 & 54.89 & 0.82 & 0.93 & 0.49 \\
            \bottomrule
        \end{tabular}%
    }
\end{table}

Tables~\ref{tab:sensitivity-m1} and~\ref{tab:sensitivity-m3} reveal that the two exponents play distinct and interpretable roles.
The fusion exponent \(\mu_2\) governs the degree to which consecutive covariance estimates are fused, and its effect is most visible when \(m^*=3\).
Indeed, in \textbf{Settings~(i)} and \textbf{(iii)}, a larger \(\mu_2\) at fixed \(\mu_1\) lowers the detected break count and raises \(d_h\), the effect being much weaker when \(m^*=1\) than when \(m^*=3\).
In \textbf{Setting~(ii)}, however, the fusion exponent is decisive, and the data-driven criteria require a large \(\mu_2\).
Specifically, at small \(\mu_2\), HBIC degenerates into massive over-segmentation, whereas at \(\mu_2=1.5\) it detects only a few breaks and localizes them far better, with HBIC$_c$ showing the same contrast in milder form.
This is the behaviour anticipated in Subsection~\ref{subsec:exp}: because the adaptive weights scale the fusion penalty by the magnitude of the preliminary temporal differences, a large \(\mu_2\) is what prevents the criteria from admitting spurious breaks in the banded design.

The exponent \(\mu_1\), in turn, controls the degree of sparsity-adaptive shrinkage applied to individual covariance entries.
Larger \(\mu_1\) strengthens the penalization of entries with small initial magnitude, which can improve support recovery up to a point but eventually over-shrinks genuinely non-zero entries.
Across the configurations of \textbf{Settings~(i)} and \textbf{(iii)}, the \(\text{F}_1\) scores and MSE are largely flat in \(\mu_1\), with the intermediate value \(\mu_1=0.8\) matching or slightly exceeding its neighbouring values in most cases; the differences are small, of the order of \(0.01\) in \(\text{F}_1\).
Increasing \(\mu_1\) to \(1.5\) does not improve support recovery or estimation accuracy and can marginally degrade them, so the intermediate value \(\mu_1=0.8\) is a reasonable and robust default.

By contrast, the data-driven criteria are largely insensitive to \((\mu_1,\mu_2)\) in \textbf{Setting~(ii)} once \(\mu_2\) is large enough, reflecting the regularity of the banded structure and its limited need for adaptive shrinkage.
BIC and HBICG are insensitive in \textbf{Settings~(i)} and~\textbf{(iii)} for a different reason.
Both select essentially no change point, so their metrics are governed by the design rather than by the exponents.
The one exception is the draw-sensitive plug-in calibration of HBIC$_c$, which under-selects in \textbf{Settings~(i)} and~\textbf{(iii)} but remains reliable in \textbf{Setting~(ii)}.
This matches the instability already reported for the non-adaptive estimator in Subsection~\ref{subsec:exp}, and it confirms that HBIC$_c$ is best read alongside HBICG$(0,2)$.

The default \((\mu_1,\mu_2)=(0.8,1.5)\) thus deliberately balances the two penalties: \(\mu_2=1.5\) prevents the over-segmentation of \textbf{Setting~(ii)}, while \(\mu_1=0.8\) avoids the over-penalization that arises at \(\mu_1=1.5\).
Although not the column-wise best for every metric, the shaded rows are close to the best among the data-driven criteria in all three settings and for both \(m^*=1\) and \(m^*=3\), which makes the default robust when the true change-point structure and the covariance sparsity are unknown.

\subsection{Empirical computational complexity}\label{subsec:complexity}

We empirically examine the computational complexity of GFlsL under \textbf{Setting~(i)} with \(m^*=1\), recording wall-clock time in seconds averaged over 100 replications for three estimator configurations:
(a)~the first-stage estimator~\eqref{stat_crit_1}, depending only on \(\lambda\);
(b)~the non-adaptive estimator, depending on \((\lambda_1,\lambda_2)\);
and (c)~the second-stage adaptive estimator~\eqref{eq:opt-prob-two-stage}, which takes as input a first-stage solution whose estimated number of change points is at least \(m^*\) and closest to it, then solves over the \((\lambda_1,\lambda_2)\) grid.
We first study how computation time varies with the tuning parameters at \((T,p)=(200,10)\), then examine the scaling behavior as \(T\) and \(p\) increase separately.
The tuning grids are \(\lambda_1 \in p\times\{10^{-5},2\times10^{-5},\ldots,10^{-4}\}\) and \(\lambda_2 \in p\times\{10^{-2},2\times10^{-2},\ldots,10^{-1}\}\), each containing 10 equally spaced values; the adaptive-weight parameters \(a_T\), \(\mu_1\), and \(\mu_2\) are set as in the synthetic experiments.

Figure~\ref{fig:comp-time-stage1} displays the first-stage computation time as a function of \(\lambda\).
The curve is non-monotone: the average time first decreases, reaches its minimum near the middle of the grid, and then rises again towards the largest values of \(\lambda\).
This pattern tracks the number of detected change points: at small \(\lambda\) the first-stage fit retains a highly non-constant path, whereas at large \(\lambda\) few change points remain but more iterations are required to converge.

\begin{figure}[ht!]
	\centering
	\includegraphics[width=0.48\textwidth]{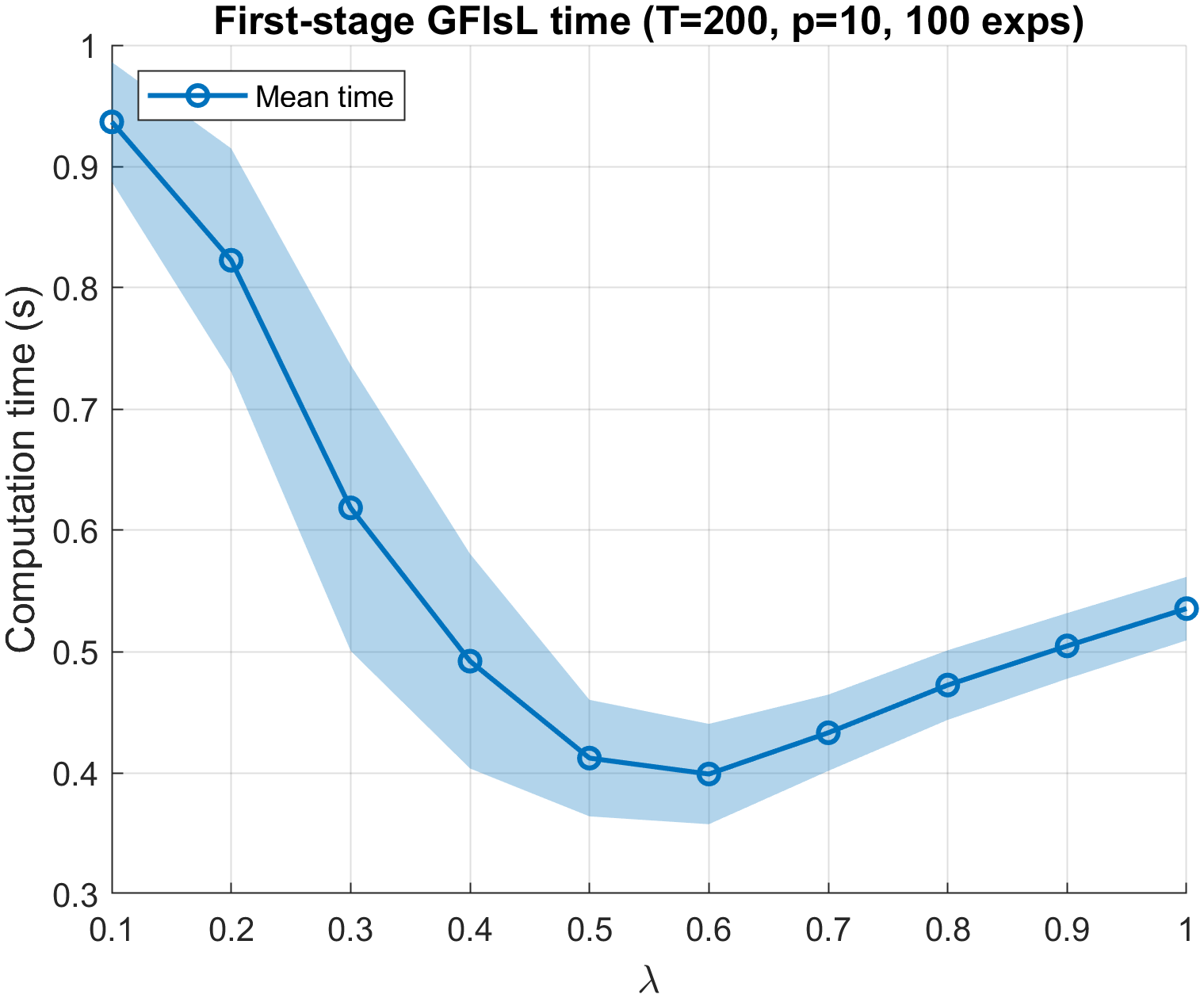}
	\caption{First-stage computation time in seconds as a function of \(\lambda\) with \(T = 200\), \(p = 10\), \(m^* = 1\) fixed, averaged over 100 replications. Shaded band: \(\pm 1\) standard deviation.}
	\label{fig:comp-time-stage1}
\end{figure}

Table~\ref{tab:comp-time} reports the computation times of the non-adaptive and second-stage adaptive estimators over the \((\lambda_1,\lambda_2)\) grid.
For the non-adaptive estimator, computation time is again largest at small \(\lambda_2\), smallest near the middle of the grid, and increasing once more as \(\lambda_2\) approaches its upper end, while the variation along the \(\lambda_1\) direction is mild.
By contrast, the second-stage adaptive estimator increases broadly monotonically in \(\lambda_2\) for each fixed \(\lambda_1\), again with mild variation along \(\lambda_1\).
Across the three configurations, every combination of tuning parameters requires about a second on average, so that GFlsL remains computationally tractable over the entire grid.

\begin{table}[ht!]
	\centering
	\caption{Computation time in seconds over the \((\lambda_1,\lambda_2)\) grid with \(T = 200\), \(p = 10\), \(m^* = 1\) fixed, averaged over 100 replications.
	Rows are \(\lambda_1\) and columns \(\lambda_2\); the left grid is the non-adaptive estimator and the right grid the second-stage (adaptive) estimator.}
	\label{tab:comp-time}
	\resizebox{0.465\textwidth}{!}{%
		\begin{tabular}{l*{10}{c}}
			\toprule
			\multicolumn{11}{c}{Non-adaptive} \\
			\(\lambda_1 \backslash \lambda_2\) & 0.1 & 0.2 & 0.3 & 0.4 & 0.5 & 0.6 & 0.7 & 0.8 & 0.9 & 1.0 \\
			\midrule
			0.0001 & 0.90 & 0.81 & 0.62 & 0.49 & 0.41 & 0.39 & 0.43 & 0.47 & 0.50 & 0.53 \\
			0.0002 & 0.88 & 0.81 & 0.61 & 0.48 & 0.40 & 0.39 & 0.42 & 0.46 & 0.49 & 0.52 \\
			0.0003 & 0.88 & 0.80 & 0.61 & 0.48 & 0.40 & 0.39 & 0.42 & 0.46 & 0.49 & 0.51 \\
			0.0004 & 0.87 & 0.80 & 0.60 & 0.47 & 0.40 & 0.39 & 0.42 & 0.45 & 0.48 & 0.51 \\
			0.0005 & 0.87 & 0.80 & 0.60 & 0.47 & 0.39 & 0.39 & 0.41 & 0.45 & 0.48 & 0.50 \\
			0.0006 & 0.86 & 0.80 & 0.60 & 0.46 & 0.39 & 0.39 & 0.41 & 0.45 & 0.48 & 0.50 \\
			0.0007 & 0.86 & 0.79 & 0.59 & 0.46 & 0.39 & 0.39 & 0.41 & 0.44 & 0.47 & 0.50 \\
			0.0008 & 0.86 & 0.79 & 0.59 & 0.46 & 0.39 & 0.38 & 0.41 & 0.44 & 0.47 & 0.50 \\
			0.0009 & 0.86 & 0.79 & 0.59 & 0.45 & 0.39 & 0.38 & 0.41 & 0.44 & 0.47 & 0.50 \\
			0.0010 & 0.85 & 0.79 & 0.59 & 0.45 & 0.38 & 0.38 & 0.41 & 0.44 & 0.47 & 0.49 \\
			\bottomrule
		\end{tabular}}
	\hfill
	\resizebox{0.465\textwidth}{!}{%
		\begin{tabular}{l*{10}{c}}
			\toprule
			\multicolumn{11}{c}{Second-stage (adaptive)} \\
			\(\lambda_1 \backslash \lambda_2\) & 0.1 & 0.2 & 0.3 & 0.4 & 0.5 & 0.6 & 0.7 & 0.8 & 0.9 & 1.0 \\
			\midrule
			0.0001 & 0.71 & 0.78 & 0.82 & 0.86 & 0.89 & 0.92 & 0.95 & 0.98 & 1.03 & 1.05 \\
			0.0002 & 0.70 & 0.76 & 0.81 & 0.85 & 0.89 & 0.93 & 0.96 & 1.02 & 1.07 & 1.10 \\
			0.0003 & 0.69 & 0.76 & 0.80 & 0.84 & 0.88 & 0.93 & 0.98 & 1.07 & 1.13 & 1.14 \\
			0.0004 & 0.69 & 0.75 & 0.79 & 0.84 & 0.88 & 0.93 & 0.99 & 1.09 & 1.14 & 1.16 \\
			0.0005 & 0.68 & 0.74 & 0.78 & 0.83 & 0.87 & 0.94 & 1.01 & 1.09 & 1.14 & 1.19 \\
			0.0006 & 0.67 & 0.73 & 0.77 & 0.83 & 0.88 & 0.94 & 1.02 & 1.10 & 1.16 & 1.16 \\
			0.0007 & 0.67 & 0.73 & 0.77 & 0.82 & 0.88 & 0.93 & 1.03 & 1.11 & 1.12 & 1.17 \\
			0.0008 & 0.66 & 0.72 & 0.77 & 0.82 & 0.87 & 0.94 & 1.03 & 1.11 & 1.12 & 1.18 \\
			0.0009 & 0.66 & 0.71 & 0.76 & 0.81 & 0.87 & 0.94 & 1.04 & 1.11 & 1.13 & 1.16 \\
			0.0010 & 0.65 & 0.71 & 0.76 & 0.80 & 0.86 & 0.93 & 1.03 & 1.09 & 1.10 & 1.16 \\
			\bottomrule
		\end{tabular}}
\end{table}

We next examine how computation time scales with problem size by fixing representative tuning parameters.
Specifically, we set \(\lambda=0.05p\) for the first stage, \(\lambda_1=5\times10^{-5}p\), \(\lambda_2=0.05p\) for the non-adaptive and adaptive stages, and varying \(T\) and \(p\) separately.
Figure~\ref{fig:comp-time-varyT} reports computation time as a function of \(T\in\{50,100,150,200,250,300\}\) with \(p=10\) fixed.
All three estimators become monotonically slower as \(T\) grows, the second-stage adaptive estimator more steeply than the first-stage and non-adaptive ones.
This scaling is consistent with the \(O(T)\) per-iteration cost of solving the tridiagonal system~\eqref{eq:linear-subsystem}, with the adaptive stage typically performing more ADMM iterations to converge.
Figure~\ref{fig:comp-time-varyp} reports computation time as a function of \(p\in\{5,10,15,20,25,30\}\) with \(T=200\) fixed.
The growth is sharply superlinear in \(p\): runtimes increase by more than an order of magnitude between \(p=5\) and \(p=30\), the adaptive stage being the slowest throughout.
This growth reflects the per-iteration cost of the \(p(p+1)/2\) independent tridiagonal subsystems, which is quadratic in \(p\), together with the \(p\times p\) factorizations that the positive-definite projection performs in the \(V_t\) update~\eqref{eq:V-update}.

\begin{figure}[ht!]
	\centering
	\begin{subfigure}[t]{0.3\textwidth}
		\centering
		\includegraphics[width=\textwidth]{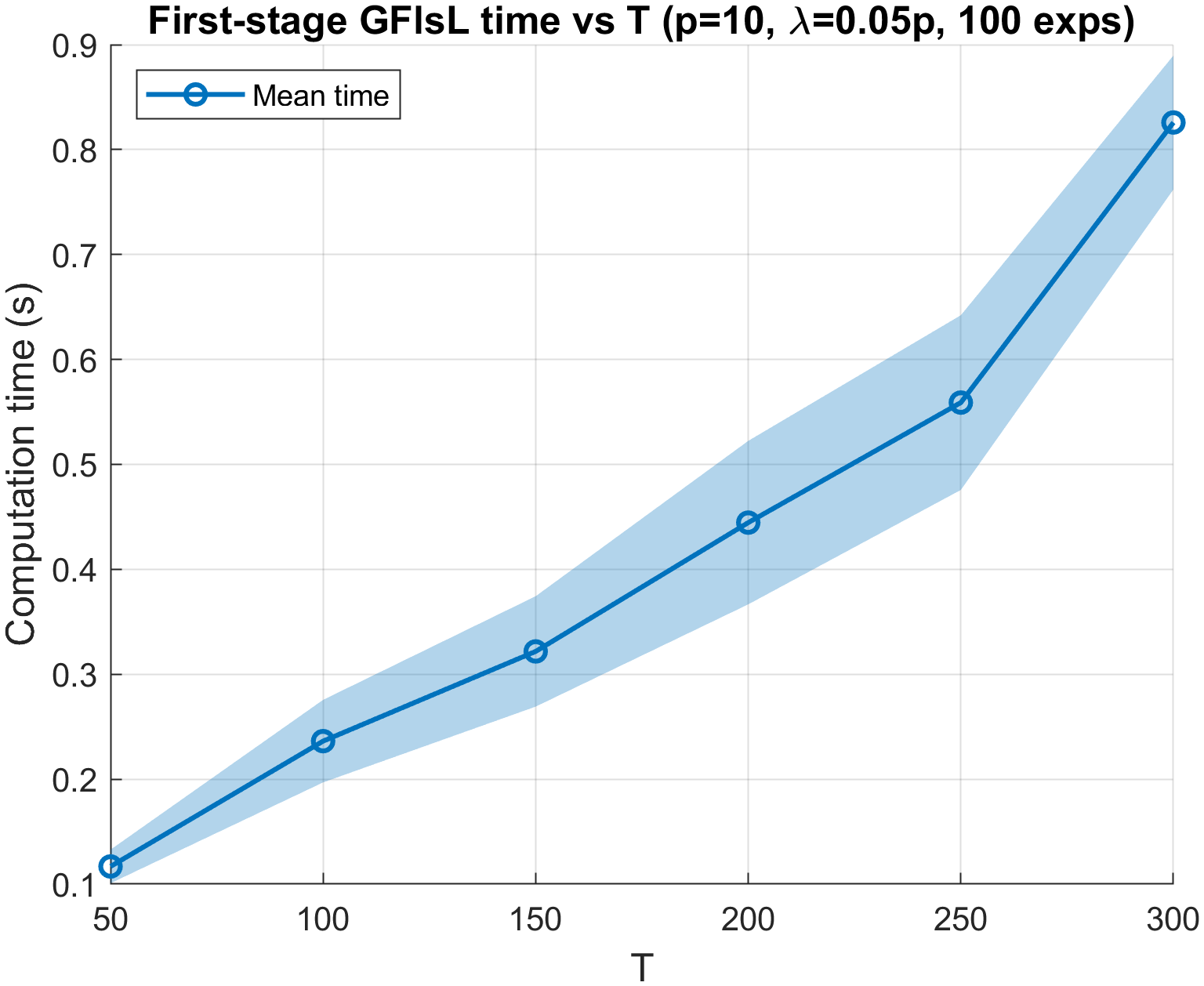}
		\caption{First-stage.}
	\end{subfigure}\hfill
	\begin{subfigure}[t]{0.32\textwidth}
		\centering
		\includegraphics[width=\textwidth]{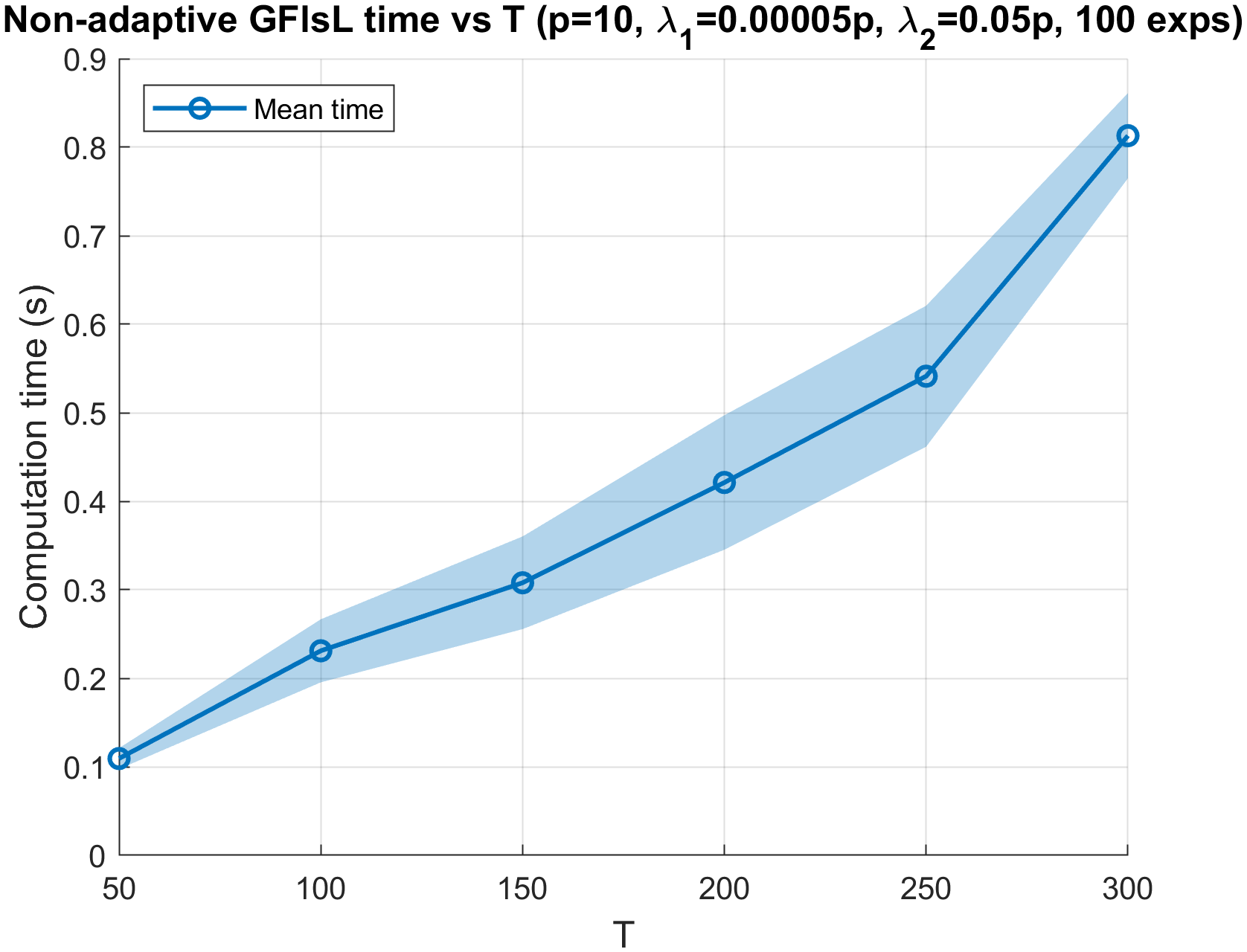}
		\caption{Non-adaptive.}
	\end{subfigure}\hfill
	\begin{subfigure}[t]{0.32\textwidth}
		\centering
		\includegraphics[width=\textwidth]{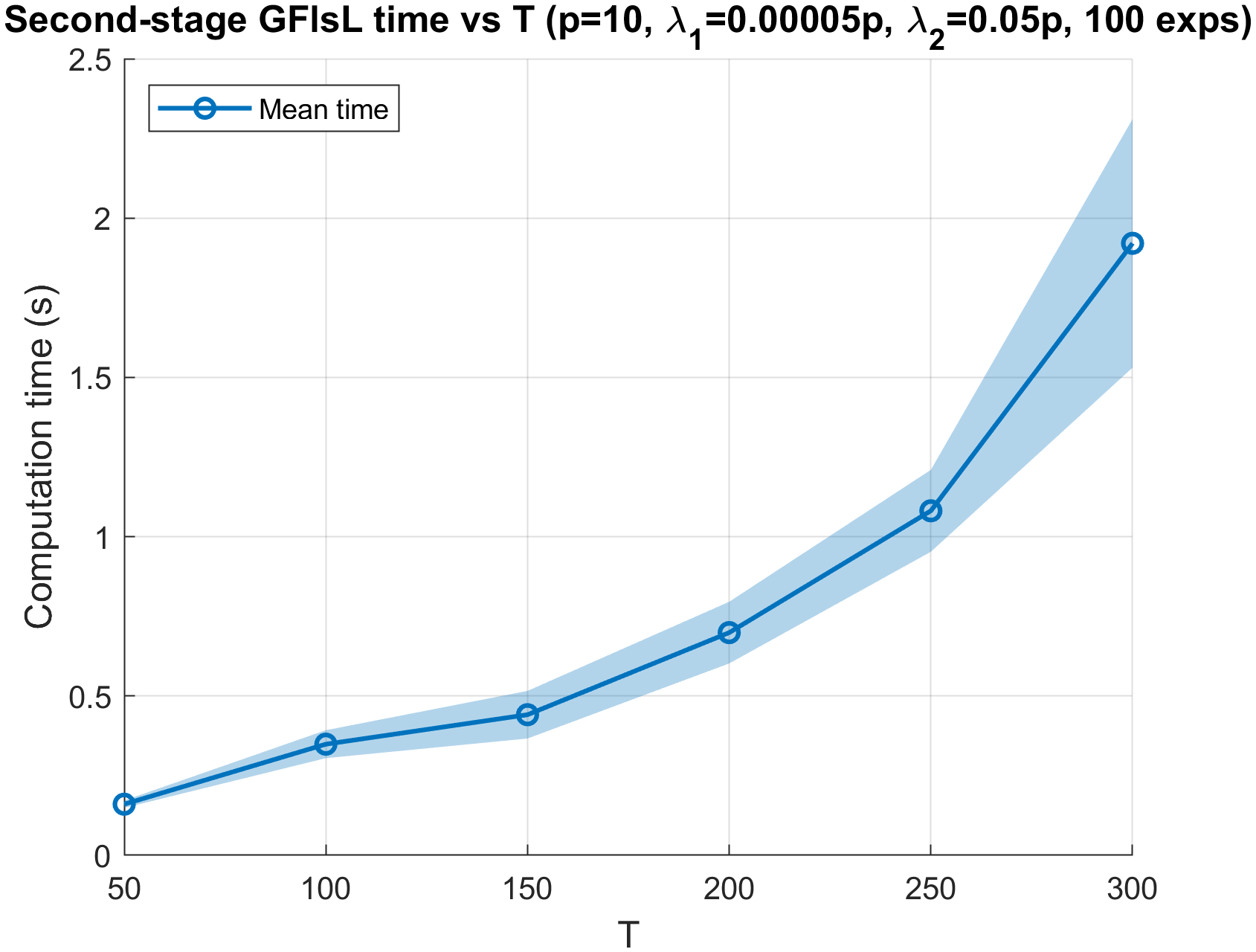}
		\caption{Second-stage adaptive.}
	\end{subfigure}
	\caption{Computation time in seconds as a function of \(T\) with \(p = 10\) fixed, averaged over 100 replications. Shaded bands: \(\pm 1\) standard deviation.}
	\label{fig:comp-time-varyT}
\end{figure}

\begin{figure}[ht!]
	\centering
	\begin{subfigure}[t]{0.29\textwidth}
		\centering
		\includegraphics[width=\textwidth]{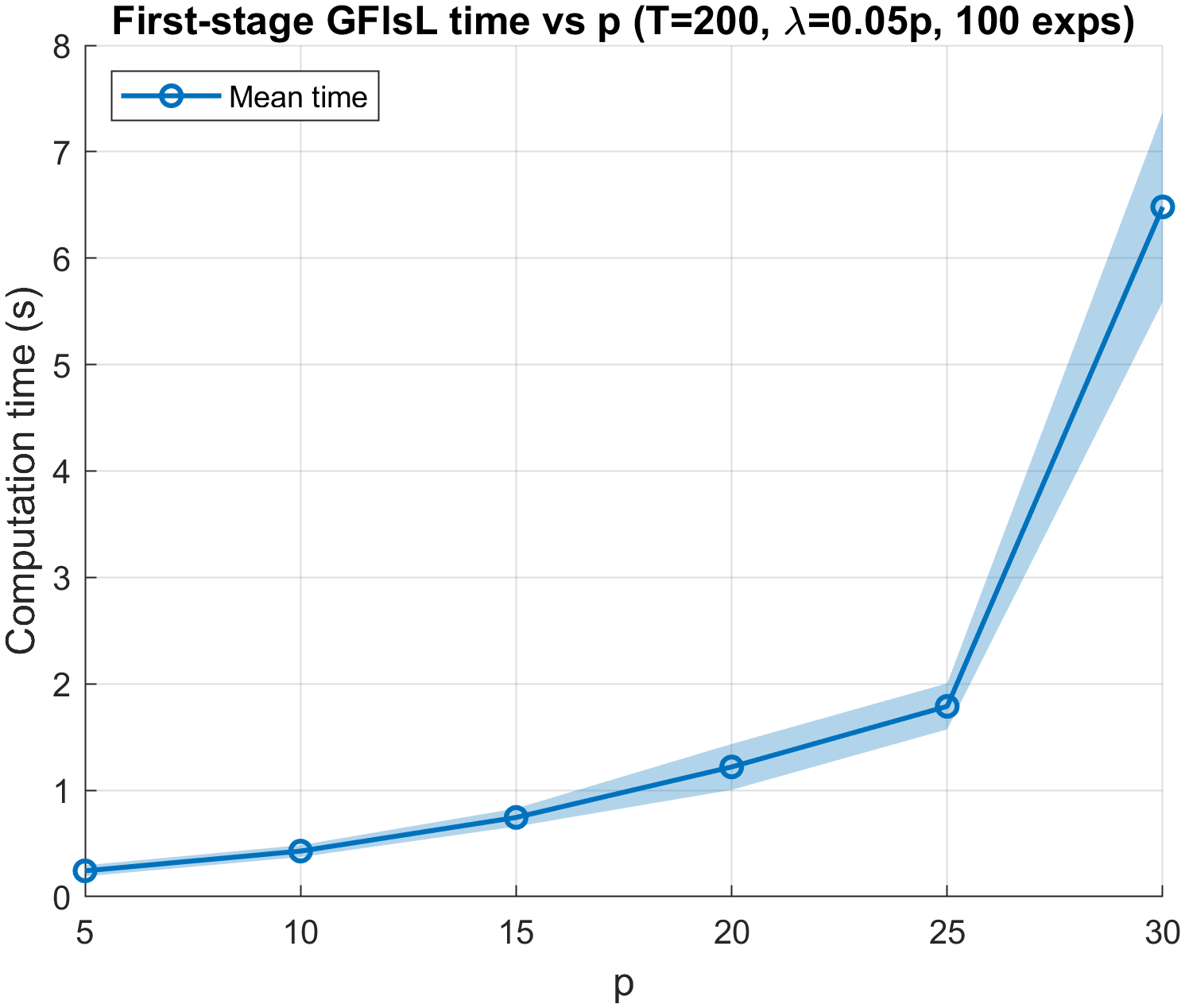}
		\caption{First-stage.}
	\end{subfigure}\hfill
	\begin{subfigure}[t]{0.32\textwidth}
		\centering
		\includegraphics[width=\textwidth]{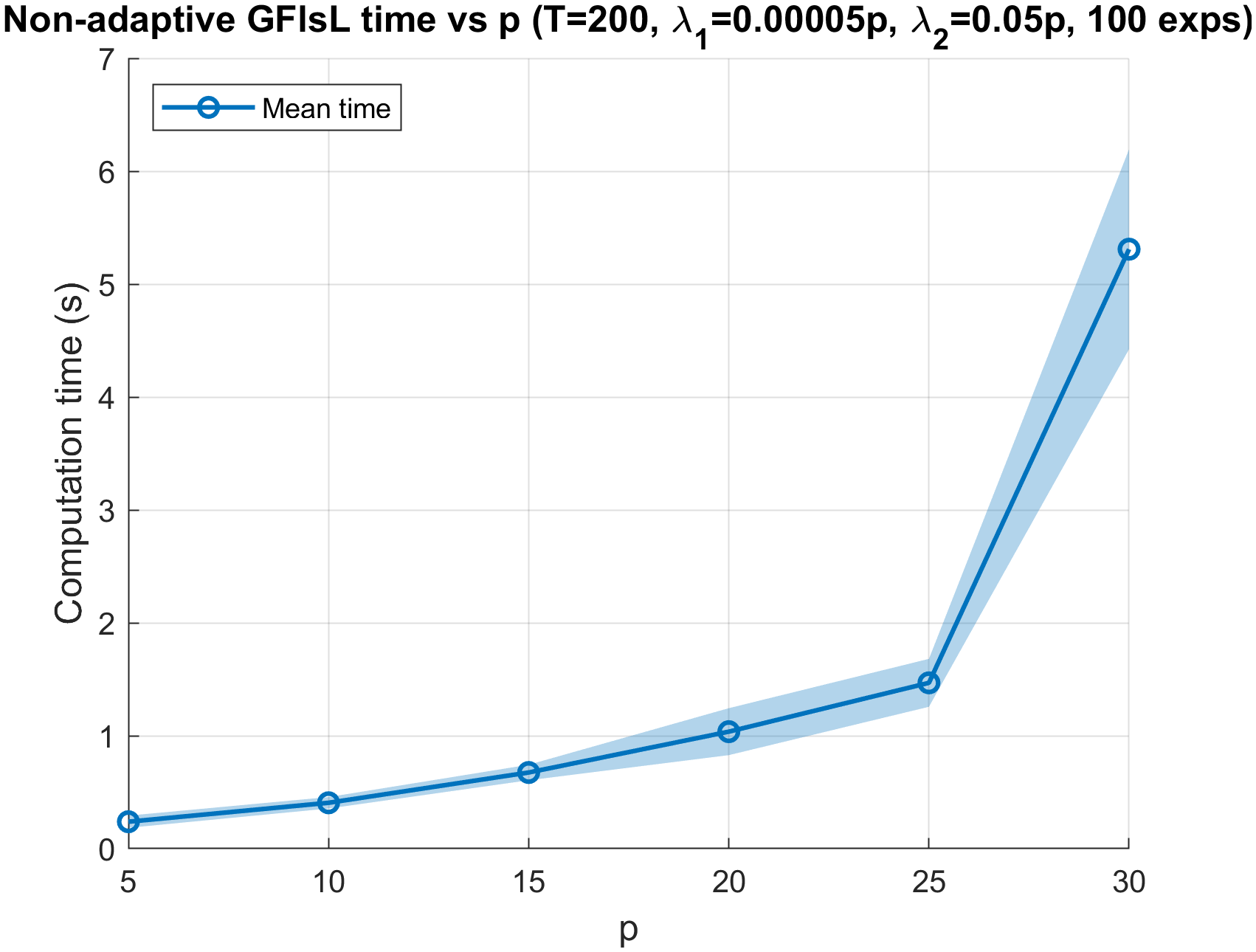}
		\caption{Non-adaptive.}
	\end{subfigure}\hfill
	\begin{subfigure}[t]{0.32\textwidth}
		\centering
		\includegraphics[width=\textwidth]{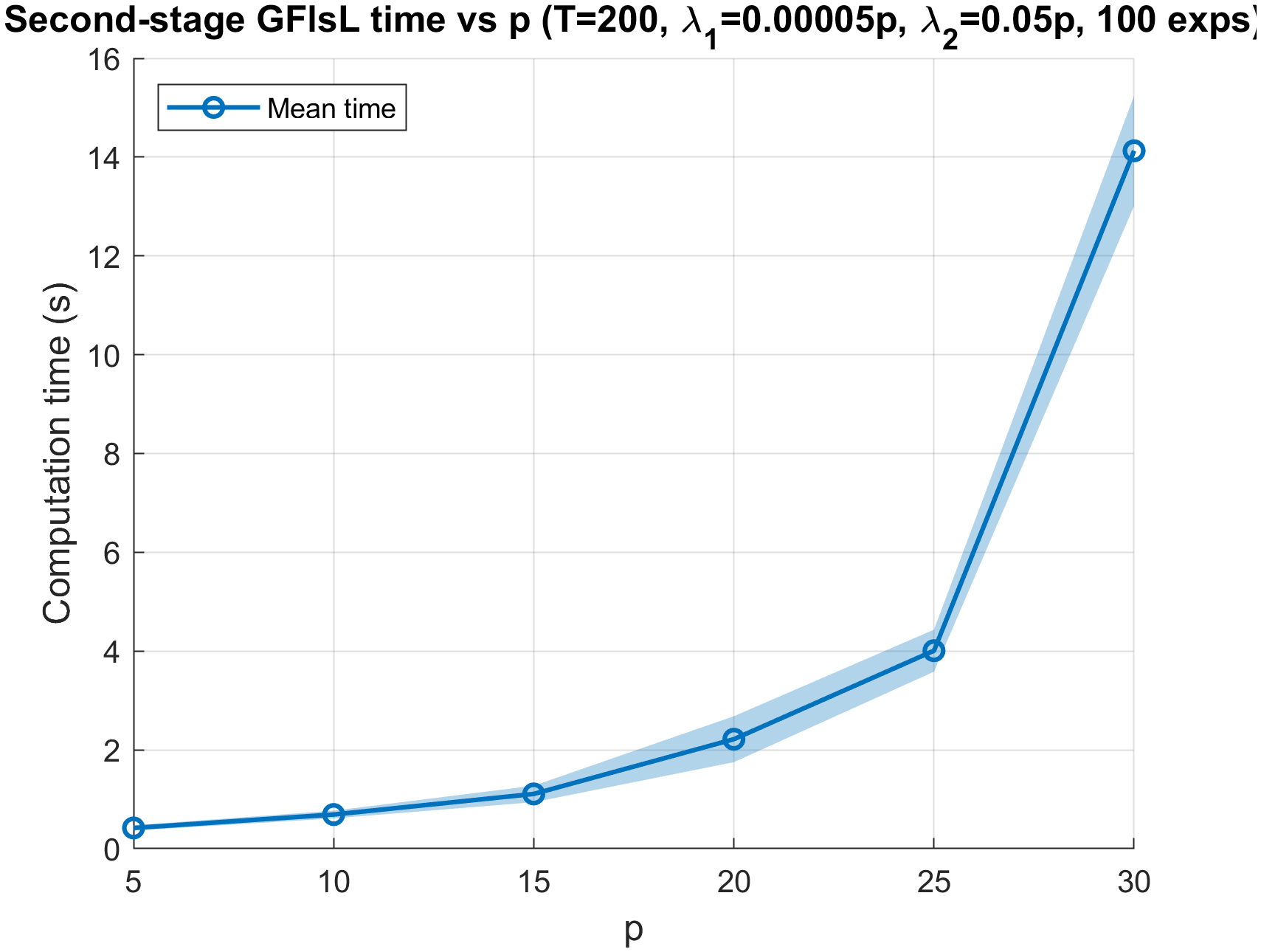}
		\caption{Second-stage adaptive.}
	\end{subfigure}
	\caption{Computation time in seconds as a function of \(p\) with \(T = 200\) fixed, averaged over 100 replications. Shaded bands: \(\pm 1\) standard deviation.}
	\label{fig:comp-time-varyp}
\end{figure}

\section{Real data experiment}\label{sec:real_data}

In this section, we consider change point detection in the variance-covariance of daily log-returns from \(p=94\) stocks selected from the S\&P~100.\footnote{S\&P 100 indices excluding AbbVie Inc., Dow Inc., General Motors, Kraft Heinz, Kinder Morgan and PayPal Holdings. The data are sourced from \url{https://finance.yahoo.com}.}
The price panel spans February~18, 2010 to January~23, 2020, and the resulting log-return series contains \(T=2499\) observations, so that the sample ends shortly before the market disruption caused by the COVID-19 pandemic in early 2020.
For each asset \(j\), the log-return is defined as \(X_{t,j}=100\times\log(P_{t,j}/P_{t-1,j})\), where \(P_{t,j}\) is the closing price on day~\(t\).

Throughout this section we report the non-adaptive and adaptive GFlsL estimators selected by the five fully data-driven criteria defined in Subsection~\ref{subsec:tuning-para}, together with those of the two competing detectors, BSOP and TBFL.
The ``optimal'' and ``lossval'' criteria and the WBSIP and DCDP algorithms require information beyond the training data and are therefore unavailable here.
In the absence of ground truth, we assess the estimators through two sources of information.
The first is the correspondence between the estimated change-point dates and the economic events of the sample period.
The second is the path of the Frobenius norms of consecutive covariance differences \(\|\widehat{\Theta}_t-\widehat{\Theta}_{t-1}\|_F\), where \(\widehat{\Theta}_t\) is the covariance estimated at time \(t\) by a specific model, together with the analogous path of a rolling blended proxy $\widehat{\Sigma}_t^{\mathrm{roll}} = (1-a) X_t X_t^{\top} + a \overline{\Sigma}_{\mathrm{roll},t}$,
where \(\overline{\Sigma}_{\mathrm{roll},t}\) is the sample covariance over a \(42\)-day, approximately two-month, rolling window and \(a=0.95\).
The rolling sample covariance is the standard nonparametric estimate of the local covariance, but it reacts to a regime shift only after the shift has entered the window, whereas the outer product \(X_tX_t^{\top}\) of a single observation reacts immediately but is degenerate.
Blending the two yields a proxy that is responsive and non-degenerate at the same time.

Figure~\ref{fig:realdata-breaks} reports both signals for the non-adaptive and adaptive GFlsL estimator under each of the five criteria, and for the two competing detectors.
Every panel traces \(\|\widehat{\Theta}_t-\widehat{\Theta}_{t-1}\|_F\) along one estimated covariance path on its left axis and the analogous quantity for the rolling proxy on its right axis, and the vertical dashed lines mark the detected change points.
The two competing detectors occupy one panel each, because every refit of a detector shares its change points and differs from the other refits only in the size of the jump realized at those dates; the marker placed at each detected date distinguishes the refits.
The panels are scaled to their own range, so only the movement over time, and not its magnitude, is comparable across panels.

\begin{figure}[p]
	\centering
	\setlength{\abovecaptionskip}{3pt}
	\setlength{\belowcaptionskip}{0pt}
	\begin{subfigure}[t]{0.48\textwidth}
		\centering
		\includegraphics[width=\textwidth]{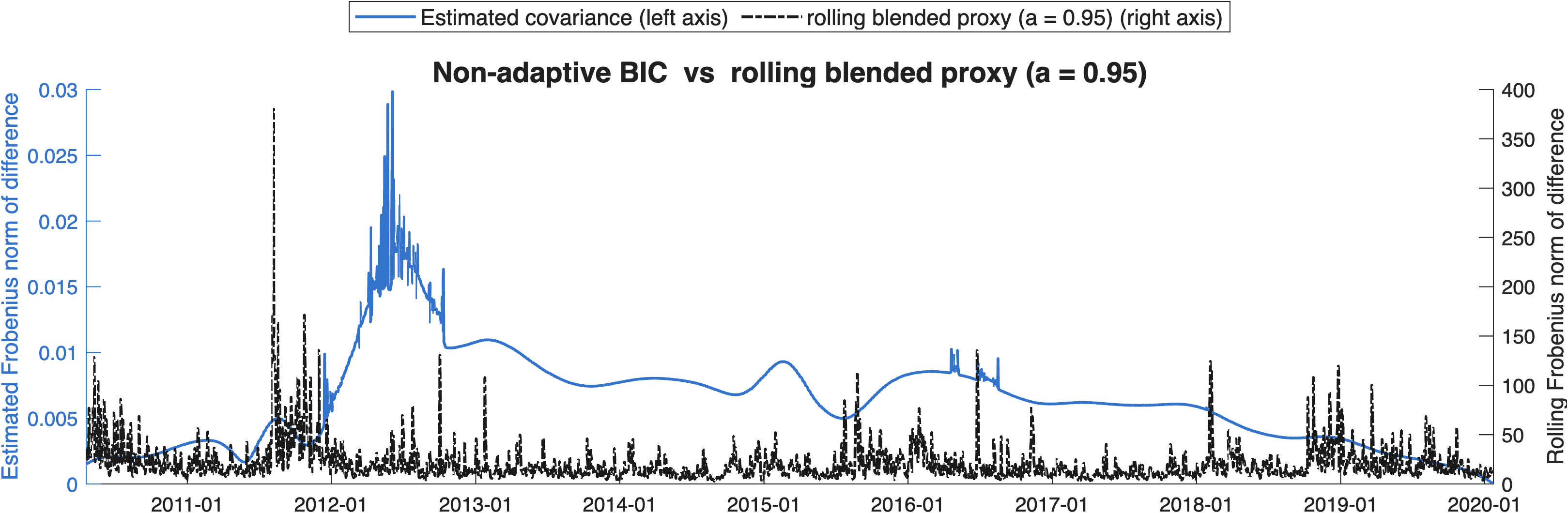}
		\caption{GFlsL: Non-adaptive BIC.}
		\label{fig:rd-na-bic}
	\end{subfigure}\hfill
	\begin{subfigure}[t]{0.48\textwidth}
		\centering
		\includegraphics[width=\textwidth]{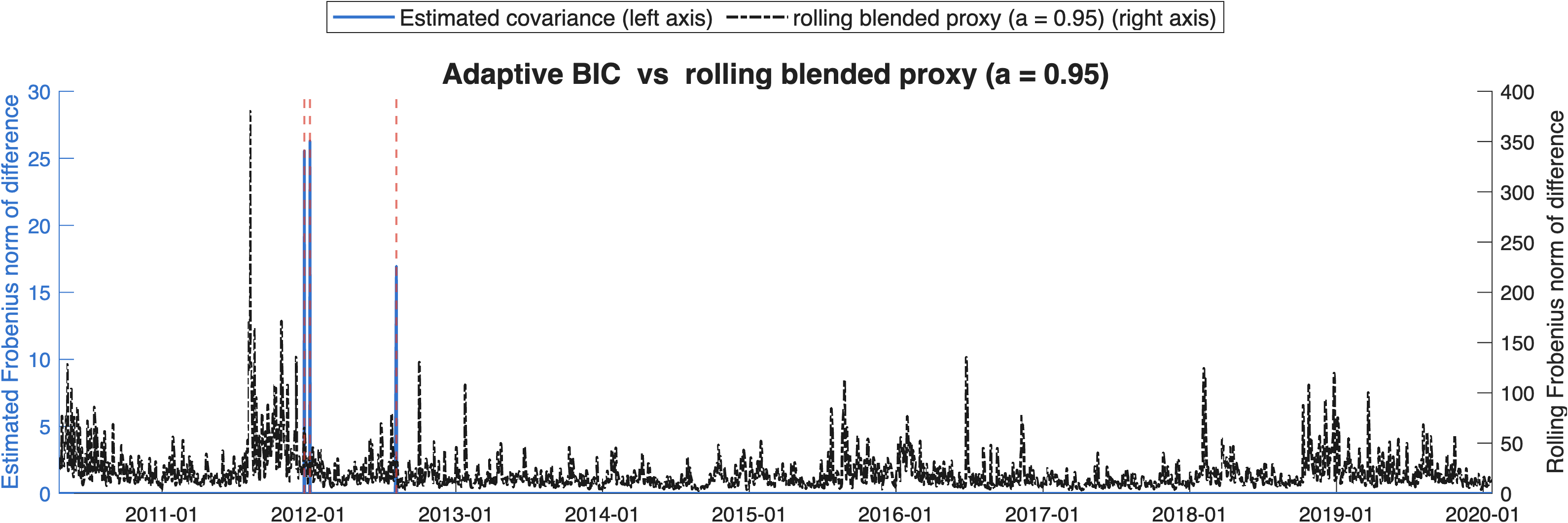}
		\caption{GFlsL: Adaptive BIC.}
		\label{fig:rd-ad-bic}
	\end{subfigure}\\
	\begin{subfigure}[t]{0.48\textwidth}
		\centering
		\includegraphics[width=\textwidth]{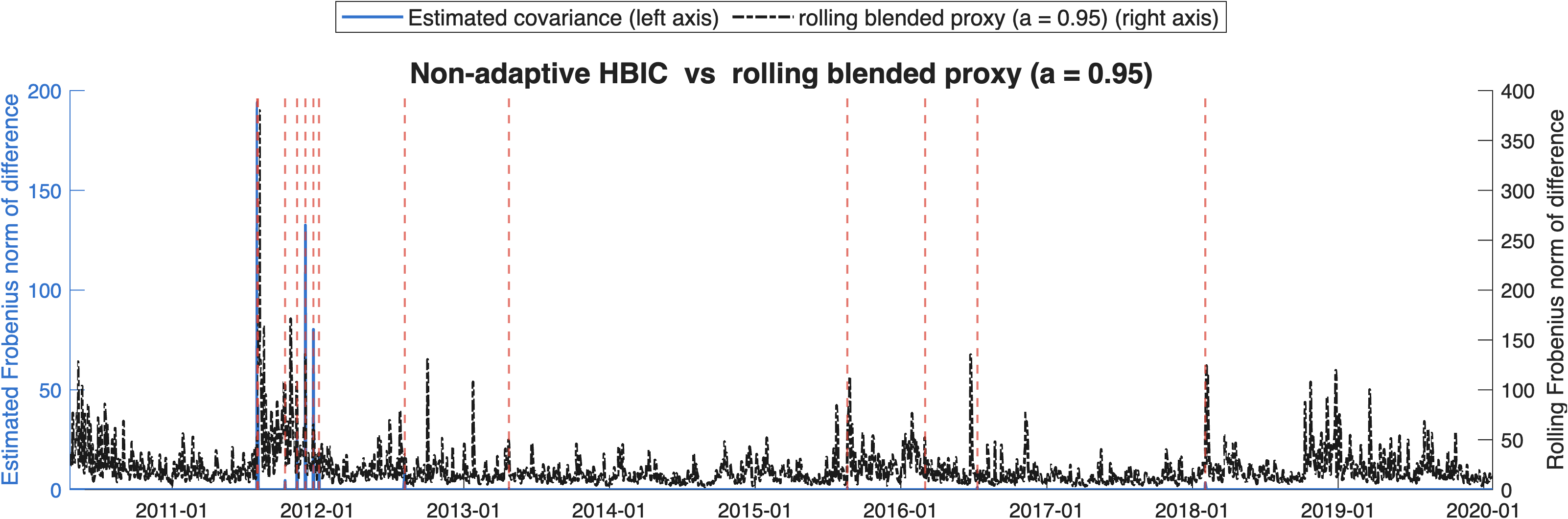}
		\caption{GFlsL: Non-adaptive HBIC.}
		\label{fig:rd-na-hbic}
	\end{subfigure}\hfill
	\begin{subfigure}[t]{0.48\textwidth}
		\centering
		\includegraphics[width=\textwidth]{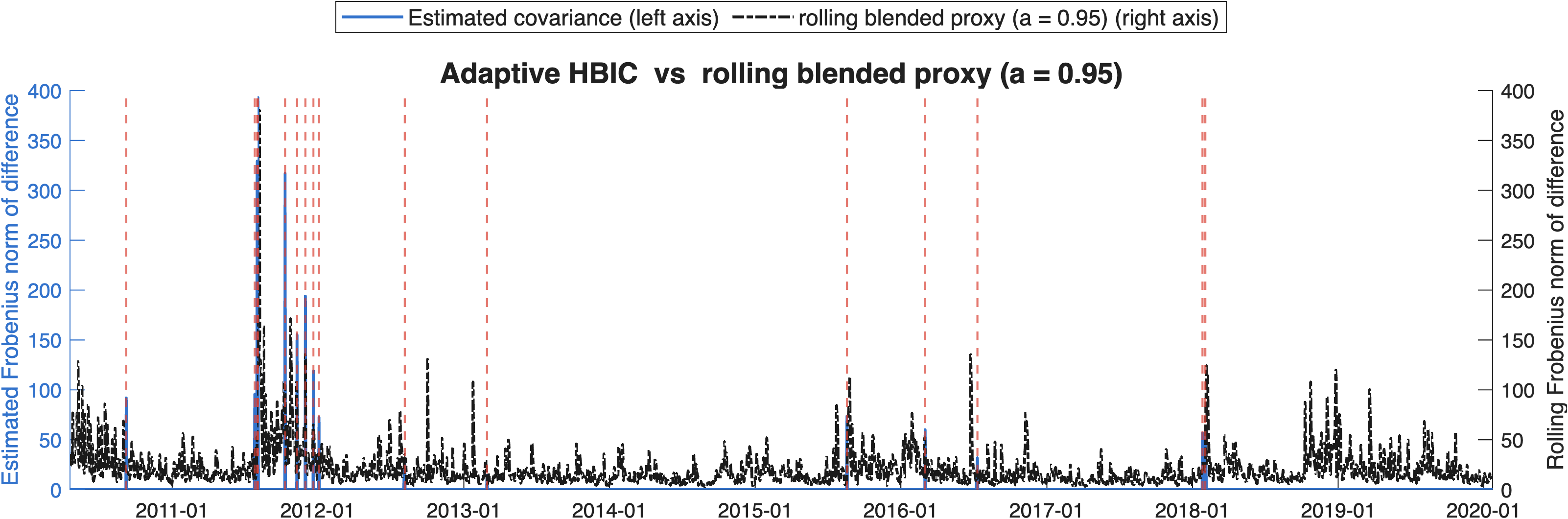}
		\caption{GFlsL: Adaptive HBIC.}
		\label{fig:rd-ad-hbic}
	\end{subfigure}\\
	\begin{subfigure}[t]{0.48\textwidth}
		\centering
		\includegraphics[width=\textwidth]{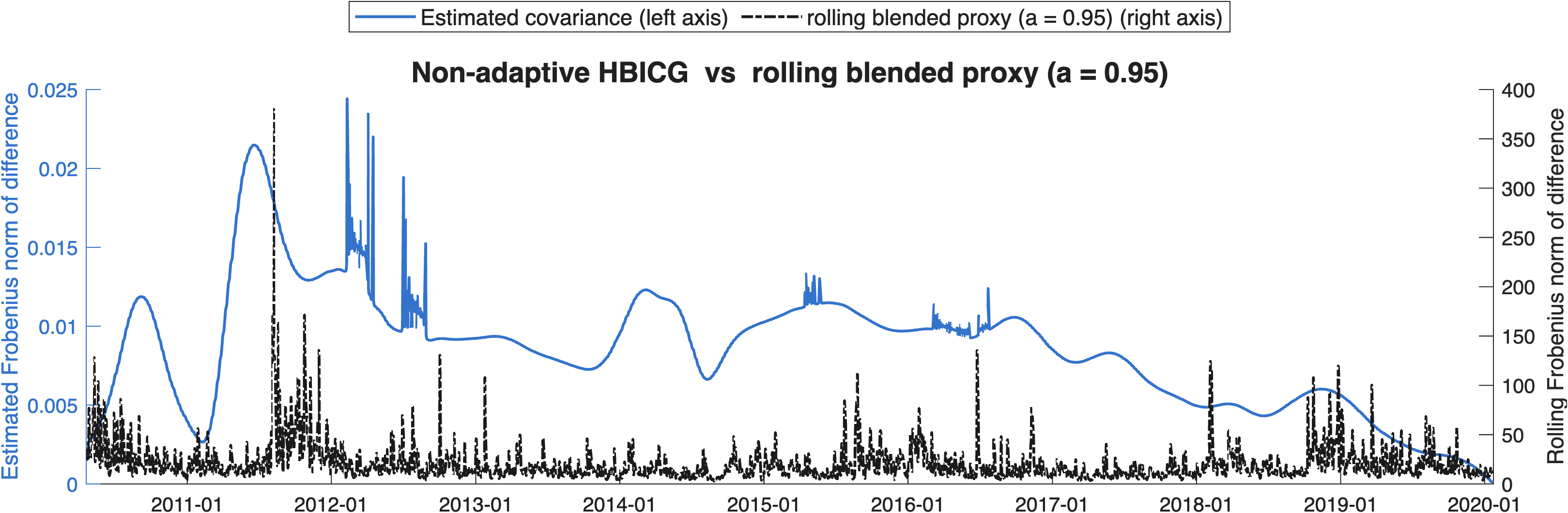}
		\caption{GFlsL: Non-adaptive HBICG.}
		\label{fig:rd-na-hbicg}
	\end{subfigure}\hfill
	\begin{subfigure}[t]{0.48\textwidth}
		\centering
		\includegraphics[width=\textwidth]{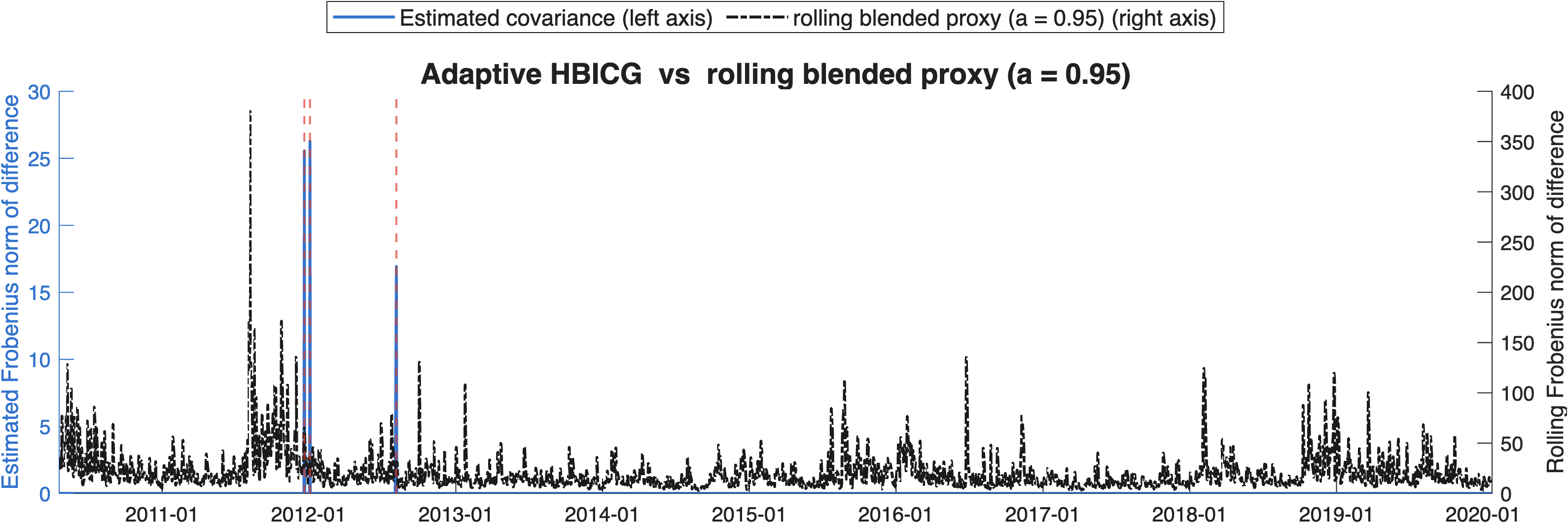}
		\caption{GFlsL: Adaptive HBICG.}
		\label{fig:rd-ad-hbicg}
	\end{subfigure}\\
	\begin{subfigure}[t]{0.48\textwidth}
		\centering
		\includegraphics[width=\textwidth]{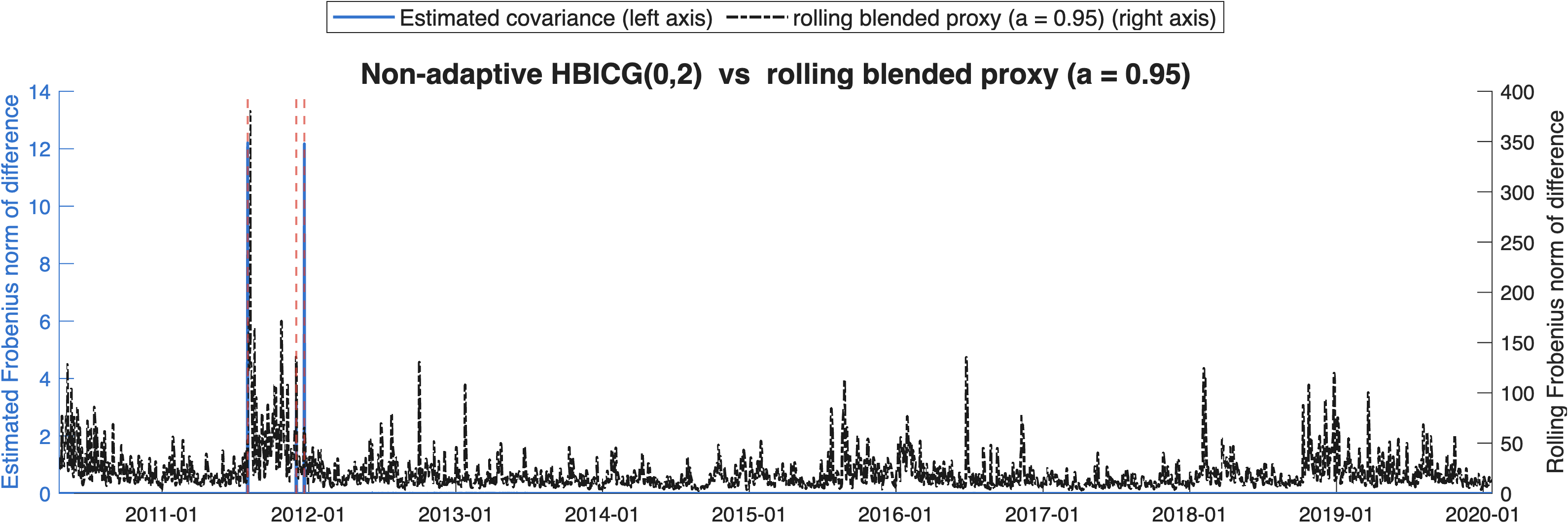}
		\caption{GFlsL: Non-adaptive HBICG\((0,2)\).}
		\label{fig:rd-na-hbicg02}
	\end{subfigure}\hfill
	\begin{subfigure}[t]{0.48\textwidth}
		\centering
		\includegraphics[width=\textwidth]{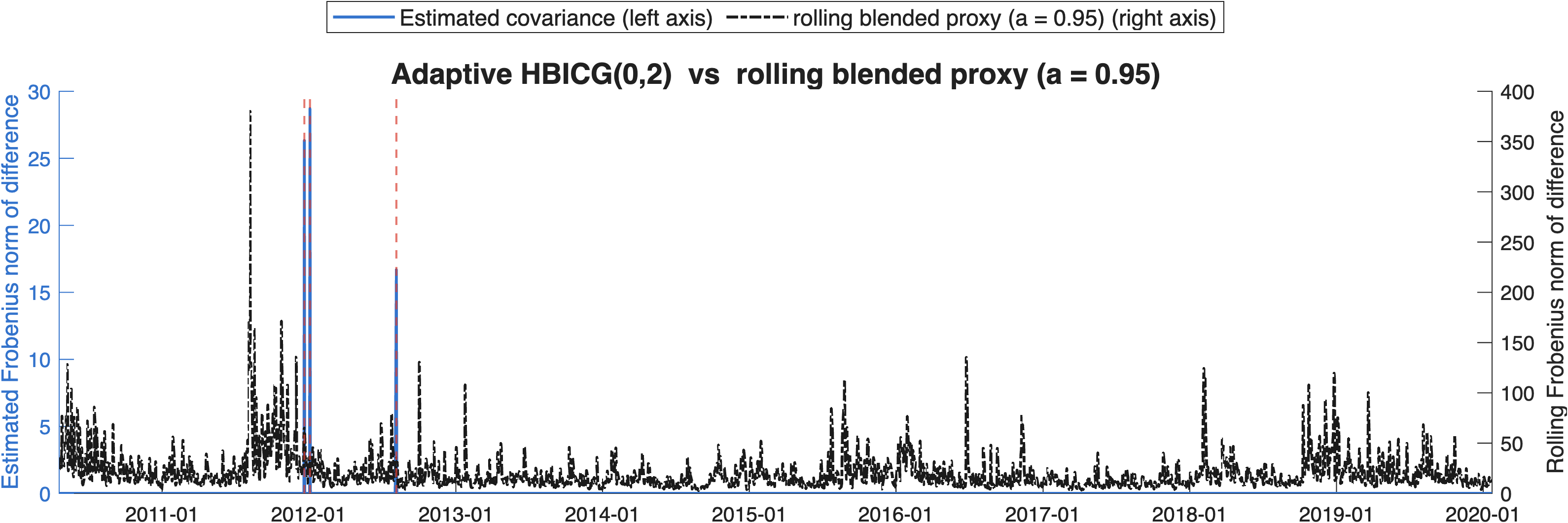}
		\caption{GFlsL: Adaptive HBICG\((0,2)\).}
		\label{fig:rd-ad-hbicg02}
	\end{subfigure}\\
	\begin{subfigure}[t]{0.48\textwidth}
		\centering
		\includegraphics[width=\textwidth]{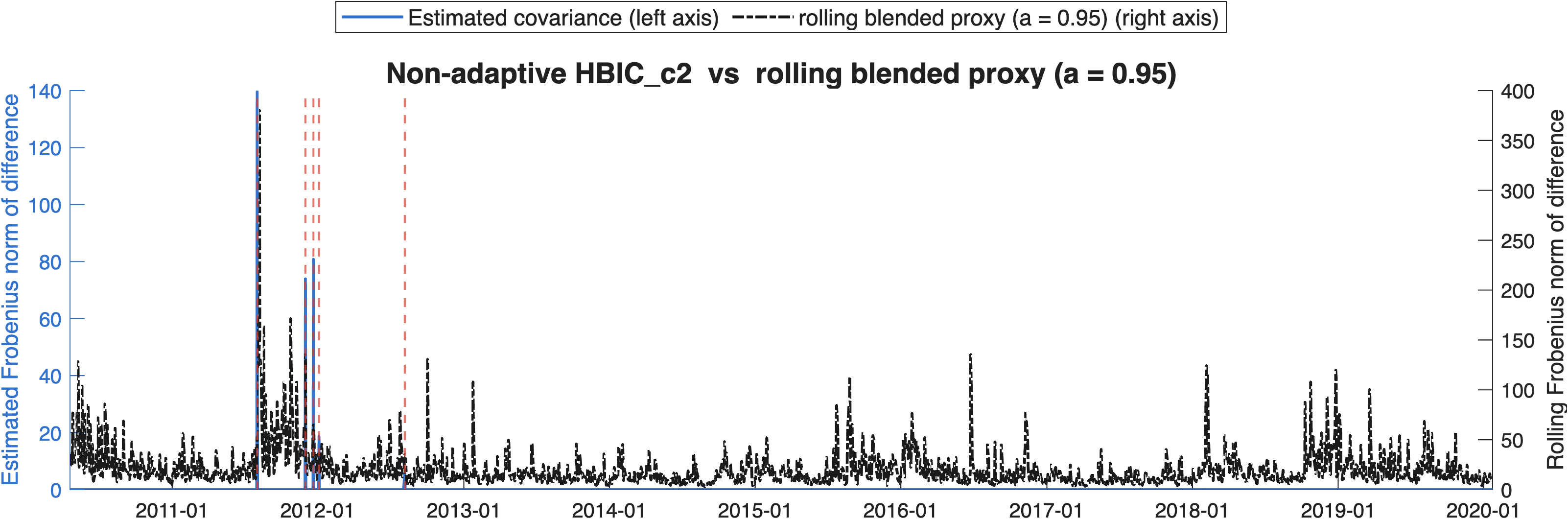}
		\caption{GFlsL: Non-adaptive \(\mathrm{HBIC}_c\).}
		\label{fig:rd-na-hbicc}
	\end{subfigure}\hfill
	\begin{subfigure}[t]{0.48\textwidth}
		\centering
		\includegraphics[width=\textwidth]{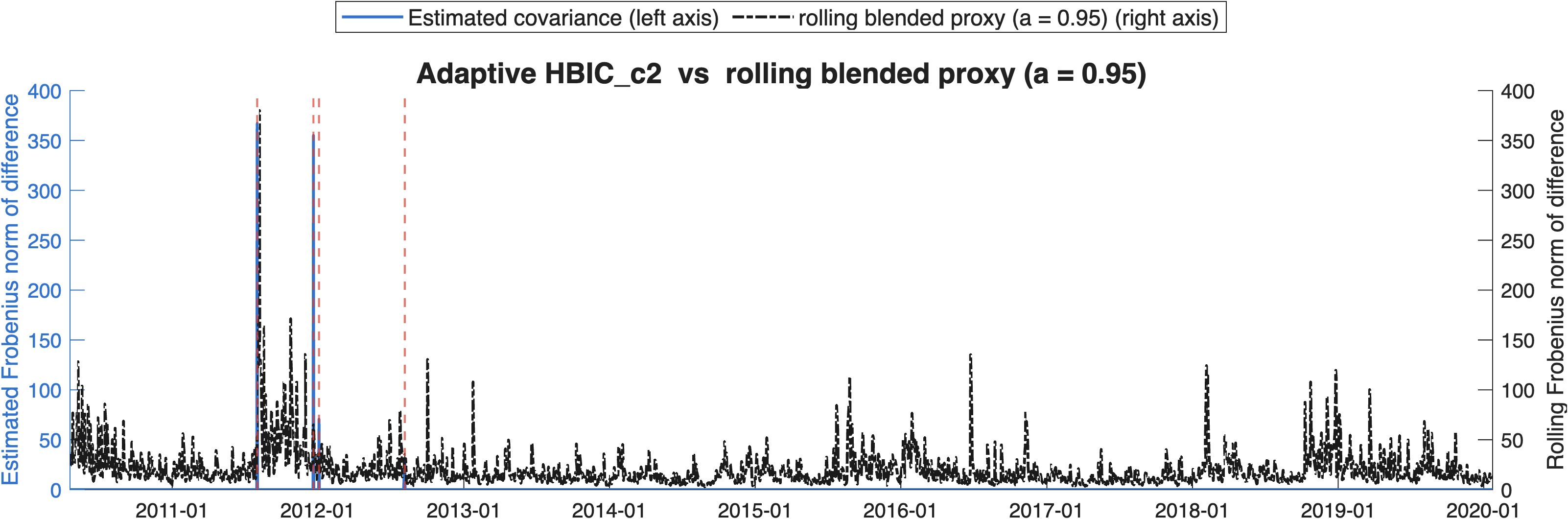}
		\caption{GFlsL: Adaptive \(\mathrm{HBIC}_c\).}
		\label{fig:rd-ad-hbicc}
	\end{subfigure}\\
	\begin{subfigure}[t]{0.48\textwidth}
		\centering
		\includegraphics[width=\textwidth]{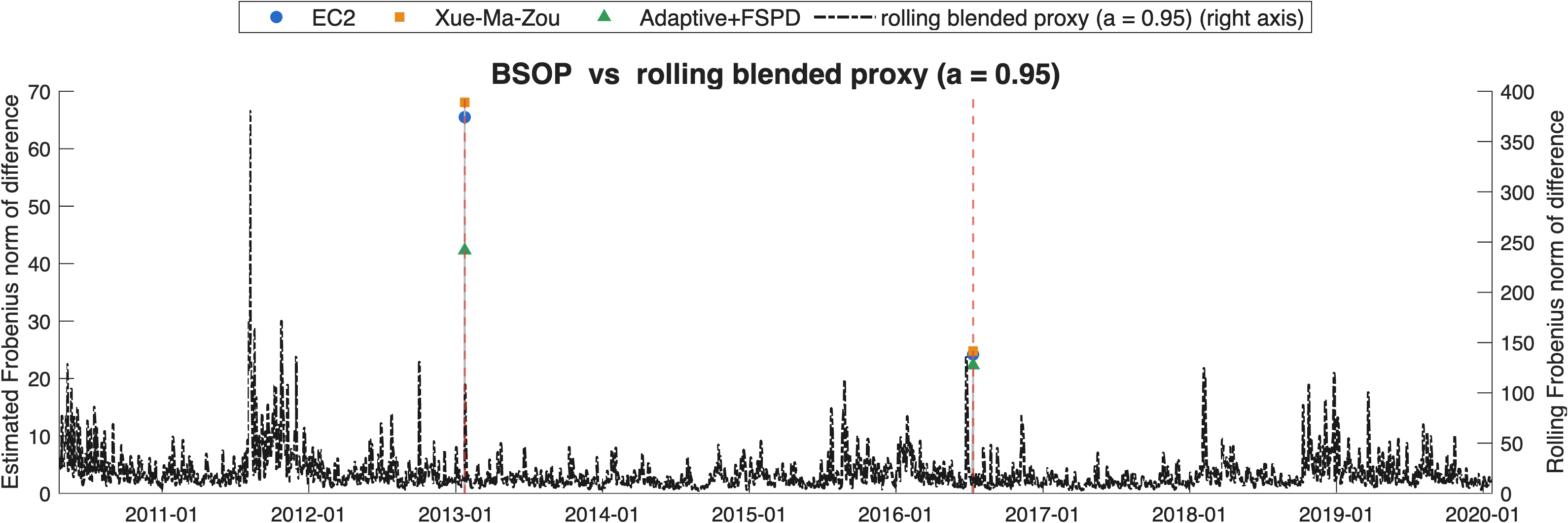}
		\caption{BSOP.}
		\label{fig:rd-bsop}
	\end{subfigure}\hfill
	\begin{subfigure}[t]{0.48\textwidth}
		\centering
		\includegraphics[width=\textwidth]{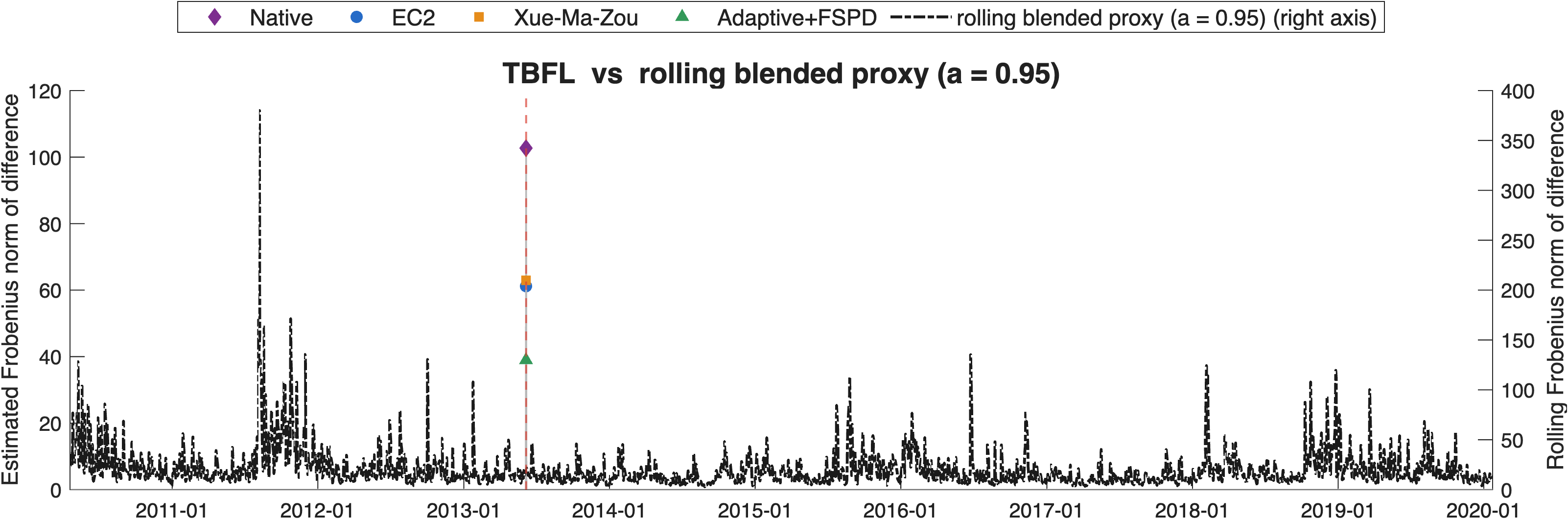}
		\caption{TBFL.}
		\label{fig:rd-tbfl}
	\end{subfigure}
	\caption{Break diagnostics for the S\&P~100 experiment.
	Each panel reports one method on a shared time axis: the left axis carries the Frobenius norm of the consecutive covariance differences of its estimated covariance path, and the right axis the same quantity for the rolling blended proxy with \(a=0.95\).
	The vertical dashed lines mark the detected change points.
	Panels~\subref{fig:rd-na-bic}--\subref{fig:rd-ad-hbicc} report the non-adaptive and the adaptive GFlsL estimator under each of the five criteria.
	Panels~\subref{fig:rd-bsop} and~\subref{fig:rd-tbfl} report the competing detectors, where every refit of a detector shares one panel with the native solution of TBFL, and the marker at each detected change point identifies the member.
	Every panel is scaled to its own range, so the magnitudes of the jumps are comparable only within a panel.
	\label{fig:realdata-breaks}}
\end{figure}

Adaptive \(\mathrm{HBICG}(0,2)\) and adaptive \(\mathrm{HBIC}_c\), the recommended criterion of Subsection~\ref{subsec:exp} and its competitive alternative, differ in this experiment on a single date.
Specifically, adaptive \(\mathrm{HBICG}(0,2)\) selects 20 December 2011, 3 January 2012 and 3 August 2012, identical to those dates detected by adaptive BIC and adaptive HBICG, while adaptive \(\mathrm{HBIC}_c\) selects the same three dates together with 1 August 2011.
A week later than this additional date, on 8 August 2011, the Frobenius norms of the consecutive covariance differences of the rolling proxy attain their maximum, in the days following the loss of the triple-A rating of the United States on 5 August 2011 and the global sell-off that accompanied it.
The three dates shared by the two criteria therefore lie later in the same crisis period and correspond to two responses to the euro-area sovereign debt crisis: 20 December 2011 and 3 January 2012 bracket the first three-year longer-term refinancing operation of the ECB, and 3 August 2012 falls in the week in which the ECB announced that it stood ready to purchase sovereign bonds.
We therefore read the two criteria as locating the same crisis, with adaptive \(\mathrm{HBIC}_c\) additionally dating its opening.

The remaining criteria behave as the simulation study suggests.
Non-adaptive BIC and non-adaptive HBICG detect no change point at all, the failure mode of an over-penalized criterion.
The two sparser non-adaptive criteria recover at most five dates each, and none of those dates is absent from the non-adaptive HBIC estimator.
Specifically, non-adaptive \(\mathrm{HBICG}(0,2)\) keeps 1 August 2011, 30 November 2011 and 20 December 2011, and non-adaptive \(\mathrm{HBIC}_c\) adds 3 January 2012 and 3 August 2012.
Both place their first change point on 1 August 2011 and both retain 30 November 2011, two dates that adaptive \(\mathrm{HBICG}(0,2)\) discards.

HBIC penalizes so lightly that it selects the most change points of all criteria, thirteen in the non-adaptive and sixteen in the adaptive case.
Its extra dates do not contradict the segmentation of adaptive \(\mathrm{HBICG}(0,2)\) and adaptive \(\mathrm{HBIC}_c\), since each of them also corresponds to an economic event, and the disagreement concerns how far the list is extended rather than where the crisis lies.
Taken together, the two HBIC solutions add 7 September 2010, when the Federal Reserve signalled further asset purchases; 26 July 2011, the debt-ceiling stand-off in the United States; 3 August 2011, immediately before the downgrade of the United States sovereign rating; 10 October 2011, the rescue of the Franco-Belgian bank Dexia; 9 November 2011, when Italian ten-year yields rose above \(7\%\) and the Italian government resigned; 27 February 2013, after the inconclusive Italian general election, which raised Italian borrowing costs and unsettled European equity markets; 23 April 2013, when the collapse of gold prices, weak Chinese data and the Boston Marathon bombing produced the steepest one-day decline of that year; 18--19 August 2015, the Chinese devaluation and the global sell-off culminating on 24 August 2015; 1 March 2016, the trough of the oil- and China-driven sell-off; 8 July 2016, after the Brexit referendum; and 25 January and 1 February 2018, at the end of the low-volatility period terminated by the spike of 5 February 2018.
Notice that the adaptive solution is not a superset of the non-adaptive one, as we have already seen for \(\mathrm{HBICG}(0,2)\).
Specifically, the two HBIC solutions disagree in 2013, where the non-adaptive solution places 23 April 2013 and the adaptive one 27 February 2013, and in August 2015, where the non-adaptive solution places 19 August 2015 and the adaptive one 18 August 2015, one trading day apart.

The competing detectors are far coarser.
BSOP identifies only two change points, on 22 January 2013 and 8 July 2016, while TBFL identifies a single one, on 5 June 2013.
The choice of refit changes the estimated covariance paths but not the detected dates, so Figure~\ref{fig:realdata-breaks} gives each detector a single panel.
Neither detector enters the 2011--2012 window, which is consistent with the failure of BSOP to detect a change point at \(p=20\) in the simulation experiments, and neither resolves the successive events of the euro-area crisis.
The date of 8 July 2016, in the aftermath of the Brexit referendum, is the only date on which a competing detector agrees with the HBIC segmentations, and it agrees with none of the dates selected by adaptive \(\mathrm{HBICG}(0,2)\) or adaptive \(\mathrm{HBIC}_c\).

\section{Concluding remarks}

We proposed a filtering procedure for change-point detection in variance-covariance matrix. We showed that the estimated change points are sufficiently close to the true change points and established the consistency of the estimated covariance matrix in each regime. We also proposed an ADMM procedure to solve the optimization problem with change points.
Alternatively to adaptive regularization, which requires a preliminary non-adaptive estimate to construct the weights, non-convex penalization with fusion may be considered, which would require different implementation strategies. Moreover, it would be interesting to integrate jumps in both the first order and the second order moments. We shall leave these topics for future research.

\vspace*{0.2cm}

\noindent\textbf{Acknowledgements}
Benjamin Poignard was supported by JSPS KAKENHI (25K16617) and RIKEN.
\appendix
\numberwithin{equation}{section}

\makeatletter 
% "activate" the preparatory code, but for section-level headers only
\newcommand{\section@cntformat}{Appendix \thesection\ }
\makeatother

\section{Intermediary results}\label{appendix:intermediary}

\begin{lemma}\label{bound_gradient}
Suppose Assumption \ref{assumption_dgp}, Assumption \ref{assumption_regularity} and Assumption \ref{assumption_rates}(i) are satisfied. For a sequence $\delta_T \rightarrow 0$ as $T\rightarrow\infty$, then $\normalfont\underset{r-s \geq T\delta_T}{\underset{1 \leq s<r\leq T+1}{\sup}} \left\|\frac{1}{\sqrt{r-s}}\overset{r-1}{\underset{t=s}{\sum}}\big(X_tX^\top_t-\text{Var}(X_t)\big)\right\|_{\max} = O_p(\sqrt{\log(pT)})$.
\end{lemma}
\begin{proof}
Recall that $\Sigma^*_{j}$ is the true variance-covariance of $X_t$, with $t \in \Bc^*_j$. Now take any $s \leq t \leq r-1$, with $s,r \in \{1,\ldots,T\}$ such that $r-s\geq T\delta_T$. By Assumption \ref*{assumption_dgp}, we can apply Lemma A.2 of \citet{Qian2016} on $\zeta_{kl,t}\coloneqq  X_{k,t}X_{l,t}$. Note that the latter result follows from Theorem 1 of \citet{Merlevede2011}, which states the mixing condition $\alpha(t) \leq \exp(-c_1t^{\gamma_1})$ for some $c_1,\gamma_1>0$; then for $c_{\alpha}=1,\gamma_1=1$, we may take $\rho = \exp(-2c_1)$, which allows us to apply Lemma A.2 of \citet{Qian2016}.
Thus, $\forall 1 \leq k,l \leq p$, there exist constants $C_1,C_2$ such that, for any $C>0$ large enough, with $r>s$:
{\footnotesize{\begin{align*}
\Pb\left(\!\big|\frac{1}{\sqrt{r\!-\!s}}\overset{r\!-\!1}{\underset{t=s}{\sum}}\big(X_{k,t}X_{l,t}\!-\!\text{Var}(X_t)_{kl}\big)\big|\!\geq\!C\sqrt{\log(pT)}\!\right) \leq (T\!+\!1) \exp\!\Big(\!-\,\frac{(C(r\!-\!s)\log(pT))^{\frac{\gamma}{2(1+\gamma)}}}{C_1}\!\Big) \!+\! \exp\!\Big(\!\!-\,\frac{C(r\!-\!s)\log(pT)}{C_2}\!\Big).
\end{align*}}}
Applying the previous inequality, by Bonferroni's inequality:
\begin{align*}
  & \Pb\left(\underset{r-s \geq T\delta_T}{\underset{1 \leq s<r\leq T+1}{\sup}}\|\frac{1}{\sqrt{r-s}}\overset{r-1}{\underset{t=s}{\sum}}\big(X_tX^\top_t-\text{Var}(X_t)\big)\|_{\max}\geq C\sqrt{\log(pT)}\right) \\
  \leq {} & T^2 \underset{r-s \geq T\delta_T}{\underset{1 \leq s<r\leq T+1}{\sup}} \Pb\Big(\underset{1 \leq k,l \leq p}{\max}\,|\frac{1}{\sqrt{r-s}}\overset{r-1}{\underset{t=s}{\sum}}\big(X_tX^\top_t-\text{Var}(X_t)\big)_{kl}|\geq C\sqrt{\log(pT)}\Big) \\
  \leq {} & T^2 \underset{r-s \geq T\delta_T}{\underset{1 \leq s<r\leq T+1}{\sup}} p^2 \Big\{ (T+1) \, \exp\Big(-\frac{(C(r-s)\log(pT))^{\frac{\gamma}{2(1+\gamma)}}}{C_1}\Big) + \exp\Big(-\frac{C(r-s)\log(pT)}{C_2}\Big)\Big\} \\
  \leq {} & \exp\Big(-\frac{(CT\delta_T\log(pT))^{\frac{\gamma}{2(1+\gamma)}}}{C_1}+4\log(pT)\Big) + \exp\Big(-\frac{CT\delta_T\log(pT)}{C_2}+2\log(pT)\Big), %\rightarrow 0 \; \text{as} \; T \rightarrow \infty,
\end{align*}
which goes to $0$ as $T \rightarrow \infty$ by Assumption \ref*{assumption_rates}(i) implying $(T\delta_T\log(pT))^{(\gamma/(2(1+\gamma)))} \propto \log(pT)$.
\end{proof}

\begin{lemma}\label{optimality_cond}
Consider problem (\ref{stat_crit_1}). Define $\Gamma_t = \Theta_{t}-\Theta_{t-1}$, $ t \geq 2$ and $\Gamma_1 = \Theta_1$. The GFlsL estimator $\{\widetilde\Theta_t\}^T_{t=1}$ satisfies the conditions $\forall t \in \{1,\ldots,T\}, \; \frac{1}{T}\overset{T}{\underset{r=t}{\sum}}\Big(\widetilde{\Theta}_r-X_rX^\top_r\Big) + \lambda \widetilde{E}_{t}=\mathbf{0}_{p \times p}$,
where $\widetilde{E}_{t}$ is the sub-gradient matrix defined by $\widetilde{E}_{1}=\mathbf{0}_{p \times p}$ and for $t=2,\ldots,T$, $\widetilde{E}_{t}$ satisfies
\begin{equation*}
 \widetilde{E}_{t} =\frac{\widetilde{\Gamma}_{t}}{\|\widetilde{\Gamma}_{t}\|_F}\;\; \text{if} \;\; \widetilde{\Gamma}_{t} \neq \mathbf{0}_{p \times p},\;\; \text{and} \;\;
\|\widetilde{E}_{t}\|_F \leq 1 \;\; \text{if} \;\; \|\widetilde{\Gamma}_{t}\|_F=0.
\end{equation*}
\end{lemma}
\begin{proof}
Defining $\Gamma_t = \Theta_{t}-\Theta_{t-1}$, $ t \geq 2$ and $\Gamma_1 = \Theta_1$, the problem stated in (\ref{stat_crit_1}) can be recast as a minimization of the function
\begin{equation*}
\Gb_{\lambda}(\{\Theta_t\}_{t=1}^T,\!\mathcal{X}_T)\!=\!\frac{1}{2T}\overset{T}{\underset{t=1}{\sum}}\text{tr}(\Big(\overset{t}{\underset{s=1}{\sum}}\Gamma_s\Big)^\top\Big(\overset{t}{\underset{s=1}{\sum}}\Gamma_s\Big)- (X_tX^\top_t)\Big(\overset{t}{\underset{s=1}{\sum}}\Gamma_s\Big) - \Big(\overset{t}{\underset{s=1}{\sum}}\Gamma_s\Big)^\top (X_tX^\top_t)) + \lambda\overset{T}{\underset{t=2}{\sum}} \|\Gamma_t\|_F.
\end{equation*}
Invoking subdifferential calculus, a necessary and sufficient condition for $(\widetilde\Gamma_t)_{1\leq t \leq T}$ to minimize $\Gb_{\lambda}(\cdot,\mathcal{X}_T)$ is that for all $t = 1,\ldots,T$, $\mathbf{0}_{p \times p} \in \Rb^{p \times p}$ belongs to the subdifferential of $\Gb_{\lambda}(\cdot,\mathcal{X}_T)$ with respect to $(\Gamma_t)_{1\leq t \leq T}$ at $(\widetilde\Gamma_t)_{1\leq t \leq T}$, that is
$\frac{1}{T}\overset{T}{\underset{r=t}{\sum}}\Big(\Big(\overset{r}{\underset{s=1}{\sum}}\widetilde\Gamma_s\Big)-X_rX^\top_r\Big) + \lambda \widetilde{E}_{t}=\mathbf{0}_{p \times p}$, with the subgradient matrices defined as: $\widetilde{E}_{1}=\mathbf{0}_{p \times p}$ and
$\widetilde{E}_{t} =\frac{\widetilde{\Gamma}_{t}}{\|\widetilde{\Gamma}_{t}\|_F}$ if $\widetilde{\Gamma}_{t} \neq \mathbf{0}_{p \times p}$, and $\|\widetilde{E}_{t}\|_F \leq 1$ if $\|\widetilde{\Gamma}_{t}\|_F=0$.
Now if $t=\widetilde{T}_j$ for $j \in \{1,\ldots,\widetilde{m}\}$ is one of the estimated change dates, then $\widetilde{\Gamma}_t \neq \mathbf{0}_{p \times p}$, and $\frac{1}{T}\overset{T}{\underset{r=\widetilde{T}_j}{\sum}}\Big(\widetilde{\Theta}_r - X_rX^\top_r \Big) +  \lambda\frac{\widetilde{\Gamma}_{\widetilde{T}_j}}{\|\widetilde{\Gamma}_{\widetilde{T}_j}\|_F}=\mathbf{0}_{p \times p}$,
since the change points cannot occur at $t=1$ and $\overset{r}{\underset{s=1}{\sum}}\Gamma_s= \widetilde{\Theta}_r$. When $t=1$, then the first order condition with respect to $\Gamma_t$ yields $\frac{1}{T}\overset{T}{\underset{r=1}{\sum}}\Big(\Big(\overset{r}{\underset{s=1}{\sum}}\widetilde\Gamma_s\Big)-X_rX^\top_r\Big) =\mathbf{0}_{p \times p}$,
so that $\frac{1}{T}\|\overset{T}{\underset{r=1}{\sum}}\Big(\Big(\overset{r}{\underset{s=1}{\sum}}\widetilde\Gamma_s\Big)-X_rX^\top_r\Big)\|_F\leq \lambda$.
\end{proof}

\begin{lemma}\label{optimality_cond_adaptive}
Consider problem (\ref{stat_crit_adaptive}). Define $\Gamma_t = \Theta_{t}-\Theta_{t-1}$, $ t \geq 2$ and $\Gamma_1 = \Theta_1$. The adaptive GFlsL estimator $\{\widehat\Theta_t\}^T_{t=1}$ satisfies the conditions
\begin{equation*}
\forall t \in \{1,\ldots,T\}, \; \frac{1}{T}\overset{T}{\underset{r=t}{\sum}}\Big(\widehat{\Theta}_r-X_rX^\top_r\Big) + \lambda_1\overset{T}{\underset{r=t}{\sum}} \underset{u\neq v}{\sum}\xi_{uv,1r}\widehat{E}_{uv,1r} + \lambda_2\xi_{2t}\widehat{E}_{2t}=\mathbf{0}_{p \times p},
\end{equation*}
where $\widehat{E}_{1t}, \widehat{E}_{2t}$ are the sub-gradient matrices defined by
\begin{equation*}
\forall u\neq v, \; \widehat{E}_{uv,1t} = \begin{cases}
\normalfont\text{sgn}\left(\overset{t}{\underset{s=1}{\sum}}\widehat{\Gamma}_{uv,s}\right)& \text{if} \; \overset{t}{\underset{s=1}{\sum}}\widehat{\Gamma}_{uv,s} \neq 0,\\
\in [-1,1] & \text{otherwise},
\end{cases}
\end{equation*}
and $\widehat{E}_{21}=\mathbf{0}_{p \times p}$ and for $t=2,\ldots,T$, $\widehat{E}_{2t}$ satisfies $\widehat{E}_{2t} =\frac{\widehat{\Gamma}_{t}}{\|\widehat{\Gamma}_{t}\|_F}$ if $\widehat{\Gamma}_{t} \neq \mathbf{0}_{p \times p}$, and $\|\widehat{E}_{2t}\|_F \leq 1$ if $\|\widehat{\Gamma}_{t}\|_F=0$.
\end{lemma}
\begin{proof}
Defining $\Gamma_t = \Theta_{t}-\Theta_{t-1}$, $ t \geq 2$ and $\Gamma_1 = \Theta_1$, the problem stated in (\ref{stat_crit_adaptive}) can be recast as a minimization of the function
\begin{align*}
\Gb_{\lambda_1,\lambda_2}(\{\Theta_t\}_{t=1}^T,\!\mathcal{X}_T)
=&\frac{1}{2T}\overset{T}{\underset{t=1}{\sum}}\text{tr}\left(\Big(\overset{t}{\underset{s=1}{\sum}}\Gamma_s\Big)^\top\Big(\overset{t}{\underset{s=1}{\sum}}\Gamma_s\Big)- (X_tX^\top_t)\Big(\overset{t}{\underset{s=1}{\sum}}\Gamma_s\Big) - \Big(\overset{t}{\underset{s=1}{\sum}}\Gamma_s\Big)^\top (X_tX^\top_t)\right) \\
& + \lambda_1\overset{T}{\underset{t=1}{\sum}}\underset{u\neq v}{\sum}\xi_{uv,1t}|\overset{t}{\underset{s=1}{\sum}}\Gamma_{uv,s}| + \lambda_2\overset{T}{\underset{t=2}{\sum}} \xi_{2t}\|\Gamma_t\|_F.
\end{align*}
Invoking subdifferential calculus, a necessary and sufficient condition for $(\widehat\Gamma_t)_{1\leq t \leq T}$ to minimize $\Gb_{\lambda_1,\lambda_2}(\cdot,\mathcal{X}_T)$ is that for all $t = 1,\ldots,T$, $\mathbf{0}_{p \times p} \in \Rb^{p \times p}$ belongs to the subdifferential of $\Gb_{\lambda_1, \lambda_2}(\cdot,\mathcal{X}_T)$ with respect to $(\Gamma_t)_{1\leq t \leq T}$ at $(\widehat\Gamma_t)_{1\leq t \leq T}$, that is
\begin{equation*}
\frac{1}{T}\overset{T}{\underset{r=t}{\sum}}\Big(\Big(\overset{r}{\underset{s=1}{\sum}}\widehat\Gamma_s\Big)-X_rX^\top_r\Big) + \lambda_1\overset{T}{\underset{r=t}{\sum}} \underset{u\neq v}{\sum}\xi_{uv,1r}\widehat{E}_{uv,1r} + \lambda_2\xi_{2t}\widehat{E}_{2t}=\mathbf{0}_{p \times p},
\end{equation*}
with the subgradient matrices defined as: $\widehat{E}_{21}=\mathbf{0}_{p \times p}$ and
$\widehat{E}_{2t} =\frac{\widehat{\Gamma}_{t}}{\|\widehat{\Gamma}_{t}\|_F}\; \text{if} \; \widehat{\Gamma}_{t} \neq \mathbf{0}_{p \times p},\;\; \text{and} \;\;
\|\widehat{E}_{2t}\|_F \leq 1 \; \text{if} \; \|\widehat{\Gamma}_{t}\|_F=0$, and
\begin{equation*}
\forall u\neq v, \; \widehat{E}_{uv,1t} = \begin{cases}
\normalfont\text{sgn}\Bigg(\overset{t}{\underset{s=1}{\sum}}\widehat{\Gamma}_{uv,s}\Bigg)& \text{if} \; \overset{t}{\underset{s=1}{\sum}}\widehat{\Gamma}_{uv,s} \neq 0,\\
\in [-1,1] & \text{otherwise}.
\end{cases}
\end{equation*}
Now if $t=\widehat{T}_j$ for $j \in \{1,\ldots,\widehat{m}\}$ is one of the estimated change-point locations, then $\widehat{\Gamma}_t \neq \mathbf{0}_{p \times p}$, and $\frac{1}{T}\overset{T}{\underset{r=\widehat{T}_j}{\sum}}\Big(\widehat{\Theta}_r - X_rX^\top_r \Big) +  \lambda_1\overset{T}{\underset{r=\widehat{T}_j}{\sum}} \underset{u \neq v}{\sum}\xi_{uv,1r}\widehat{E}_{uv,1r} + \lambda_2\xi_{2\widehat{T}_j}\frac{\widehat{\Gamma}_{\widehat{T}_j}}{\|\widehat{\Gamma}_{\widehat{T}_j}\|_F}=\mathbf{0}_{p \times p}$,
since the change points cannot occur at $t=1$ and $\overset{r}{\underset{s=1}{\sum}}\Gamma_s= \widehat{\Theta}_r$. When $t=1$, then the first order condition with respect to $\Gamma_t$ yields $\frac{1}{T}\overset{T}{\underset{r=1}{\sum}}\Big(\Big(\overset{r}{\underset{s=1}{\sum}}\widehat\Gamma_s\Big)-X_rX^\top_r\Big) +\lambda_1\overset{T}{\underset{r=1}{\sum}} \underset{u \neq v}{\sum}\xi_{uv,1r}\widehat{E}_{uv,1r}=\mathbf{0}_{p \times p}$,
so that 
$\frac{1}{T}\Bigg\|\overset{T}{\underset{r=1}{\sum}}\Big(\Big(\overset{r}{\underset{s=1}{\sum}}\widehat\Gamma_s\Big)-X_rX^\top_r\Big)+\lambda_1\overset{T}{\underset{r=1}{\sum}} \underset{u \neq v}{\sum}\xi_{uv,1r} \widehat{E}_{1r}\Bigg\|_F\leq \lambda_2\xi_{21}$.
\end{proof}

\begin{lemma}
\label{lemma_unique_inactive_pattern}
Define the estimated partition 
\(\widehat{\Tc}_{\widehat m}\) and consider the associated block-wise criterion
deduced from (\ref{stat_crit_adaptive}). Let $\overline\Sigma=(\overline\Sigma_1,\ldots,\overline\Sigma_{\widehat m+1})$
be a minimizer of the restricted problem over
$\mathcal O_{\widehat m+1}
=
\left\{
(\Sigma_1,\ldots,\Sigma_{\widehat m+1}):
\Sigma_{uv,j}=0
\ \text{for all } (u,v)\notin \Sc_j^*
\right\}$.
Assume that there exist subgradient matrices
$\overline E_{1t}$ and $\overline E_{2,\widehat T_j}$ satisfying the block-wise
KKT conditions and such that, for every $j$ and every
$(u,v)\notin \Sc_j^*$, $\left|
\frac{
\sum_{t\in\widehat{\Bc}_j}
\xi_{uv,1t}\overline E_{uv,1t}
}{
\sum_{t\in\widehat{\Bc}_j}
\xi_{uv,1t}
}
\right|
<1$.
Then every minimizer
$\widehat\Sigma=(\widehat\Sigma_1,\ldots,\widehat\Sigma_{\widehat{m}+1})$
of the block-wise criterion with partition $\widehat{\Tc}_{\widehat{m}}$
satisfies $\widehat\Sigma_{uv,j}=0,
\;
(u,v)\notin \Sc_j^*,\; j=1,\ldots,\widehat m+1$.
\end{lemma}

\begin{proof}
Let $\overline{\Sigma}$ be the restricted minimizer. By assumption, there exist
subgradient matrices $\overline E_{1t}$ and $\overline E_{2,\widehat T_j}$
satisfying the block-wise KKT conditions. Moreover, for every inactive
coordinate $(u,v)\notin\Sc_j^*$,
$\left|
\frac{
\sum_{t\in\widehat{\Bc}_j}
\xi_{uv,1t}\overline E_{uv,1t}
}{
\sum_{t\in\widehat{\Bc}_j}
\xi_{uv,1t}
}
\right|<1$.
So the KKT conditions are satisfied at $\overline{\Sigma}$ with strict dual
feasibility on the inactive coordinates. Assume that there exists a minimizer $\widehat{\Sigma}$ of the unrestricted
block-wise criterion such that, for some $(u,v)\notin\Sc_j^*$,
$\widehat{\Sigma}_{uv,j}\neq0$. Then the subgradient of the corresponding
$\ell_1$ penalty at this coordinate must have absolute value one. This
contradicts the strict dual feasibility condition above. Therefore every
unrestricted minimizer satisfies $\widehat\Sigma_{uv,j}=0$, for $(u,v)\notin\Sc_j^*$,
for all $j=1,\ldots,\widehat{m}+1$.
\end{proof}

\begin{proposition}[{Dual problem for~\eqref{eq:opt-prob-two-stage}}]\label{prop:dual}
	The dual problem of \eqref{eq:opt-prob-two-stage} is
		\begin{align}
			\max_{\mathbf{Y}} \quad & \left\{ \sum_{t=1}^T \left[-\frac{T}{2} \| W_t \|_F^2 - \left\langle W_t, X_tX_t^{\top} \right\rangle + \epsilon \tr\left(Z_{t+1} - Z_t + W_t - Y_{t, \mathrm{off}}\right) \right] \right\} \notag\\
			\text{s. t.} \quad      & Z_1 = Z_{T+1} = \mathbf{0}_{p \times p}; \notag\\
			                        & Z_{t+1} - Z_t + W_t - Y_{t, \mathrm{off}} \succeq 0 \,\,\,\, \forall t = 1, \dots ,T; \notag\\
			                        & \| Z_t \|_F \leq \lambda_2 \xi_{2t} \,\,\,\, \forall t = 2, \dots, T; \notag\\
			                        & | Y_{uv, t} | \leq \lambda_1 \xi_{uv, 1t} \,\,\,\, \forall t = 1, \dots , T, u, v = 1, \dots , p \text{ with } u \neq v, \label{eq:dual-prob}
		\end{align}
		where \( \mathbf{Y} \coloneqq \left\{ \{W_t\}_{t=1}^T, \{Y_{t, \mathrm{off}}\}_{t=1}^T, \{Z_t\}_{t=2}^T \right\} \) is the dual variable with \( W_t \in \mathcal{S}^p \), \( Y_{t, \mathrm{off}} \in \mathcal{S}_{\mathrm{off}}^p \), \( Z_t \in \mathcal{S}^p \) for all \( t \), and \( W_t \), \( Y_{t, \mathrm{off}} \), and \( Z_t \) are associated with the constraints \( U_t = \Theta_t - X_tX_t^{\top} \), \( \Upsilon_{t, \mathrm{off}} = \Theta_{t, \mathrm{off}} \), and \( D_t = \Theta_t - \Theta_{t-1} \), respectively.
		Moreover, the optimal values of \eqref{eq:opt-prob-two-stage} and \eqref{eq:dual-prob} are the same.
\end{proposition}

\begin{proof}
		Introduce the auxiliary variable \( U_t = \Theta_t - X_tX_t^{\top} \) and rewrite \eqref{eq:opt-prob-two-stage} as
			\begin{align}
				\min_{\mathbf{X}} \quad & \sum_{t=1}^T \left[ \frac{1}{2T} \| U_t \|_F^2 + \delta_{\cdot \succeq \epsilon I_p} (\Theta_t) \right] + \lambda_1 \sum_{t=1}^T \sum_{u\neq v} \xi_{uv, 1t} | \Upsilon_{uv, t} | + \lambda_2 \sum_{t=2}^T \xi_{2t} \| D_t \|_F \notag\\
				\text{s.t.} \quad & U_t = \Theta_t - X_tX_t^{\top}, \quad \Upsilon_{t, \mathrm{off}} = \Theta_{t, \mathrm{off}} \quad \forall t = 1, \dots, T, \notag\\
				                  & D_t = \Theta_t - \Theta_{t-1} \quad \forall t = 2, \dots, T, \label{eq:cons-opt-prob}
			\end{align}
		where \( \mathbf{X} \coloneqq \{ \{\Theta_t\}, \{U_t\}, \{\Upsilon_{t, \mathrm{off}}\}, \{D_t\} \} \), and \( \delta_{\cdot \succeq \epsilon I_p} \) is the indicator function of \( \{S \in \mathcal{S}^p \mid S \succeq \epsilon I_p\} \).
		The associated Lagrangian is
		\begin{align*}
			& L(\mathbf{X}; \mathbf{Y}) =  \sum_{t=1}^T \left[ \frac{1}{2T} \| U_t \|_F^2 + \delta_{\cdot \succeq \epsilon I_p} (\Theta_t) \right] + \lambda_1 \sum_{t=1}^T \sum_{u\neq v} \xi_{uv, 1t} | \Upsilon_{uv, t} | + \lambda_2 \sum_{t=2}^T \xi_{2t} \| D_t \|_F \\
			                        & - \sum_{t=1}^T \langle W_t, U_t - \Theta_t + X_tX_t^{\top} \rangle - \sum_{t=1}^T \langle Y_{t, \mathrm{off}}, \Theta_{t, \mathrm{off}} - \Upsilon_{t, \mathrm{off}} \rangle - \sum_{t=2}^T \langle Z_t, \Theta_t - \Theta_{t-1} - D_t \rangle.
		\end{align*}
		Set \( Z_1 = Z_{T+1} = \mathbf{0}_{p \times p} \) and \( \Delta_t \coloneqq Z_{t+1} - Z_t + W_t - Y_{t, \mathrm{off}} \). Then the dual function \( g(\mathbf{Y}) \coloneqq \inf_{\mathbf{X}} L(\mathbf{X}; \mathbf{Y}) \) is obtained by minimizing successively over the four primal blocks. For each \( t \),
		$\inf_{U_t \in \mathcal{S}^p} \left\{ \frac{1}{2T} \| U_t \|_F^2 - \langle W_t, U_t \rangle \right\} = - \frac{T}{2} \| W_t \|_F^2$.
		For each \( t \) and each off-diagonal pair \( u \neq v \),
		\[
			\inf_{\Upsilon_{uv,t} \in \mathbb{R}} \left\{ \lambda_1 \xi_{uv,1t} |\Upsilon_{uv,t}| + Y_{uv,t}\Upsilon_{uv,t} \right\} =
			\begin{cases}
				0,      & \text{if } |Y_{uv,t}| \leq \lambda_1 \xi_{uv,1t}, \\
				-\infty, & \text{otherwise}.
			\end{cases}
		\]
		Similarly, for each \( t = 2, \dots, T \),
		\[
			\inf_{D_t \in \mathcal{S}^p} \left\{ \lambda_2 \xi_{2t} \| D_t \|_F + \langle Z_t, D_t \rangle \right\} =
			\begin{cases}
				0,      & \text{if } \| Z_t \|_F \leq \lambda_2 \xi_{2t}, \\
				-\infty, & \text{otherwise}.
			\end{cases}
		\]
		Finally, for each \( t \),
		\[
			\inf_{\Theta_t \succeq \epsilon I_p} \langle \Delta_t, \Theta_t \rangle =
			\begin{cases}
				\epsilon \, \mathrm{tr}(\Delta_t), & \text{if } \Delta_t \succeq 0, \\
				-\infty,                            & \text{otherwise},
			\end{cases}
		\]
		where the first case follows by writing \( \Theta_t = \epsilon I_p + S_t \) with \( S_t \succeq 0 \). Collecting these terms yields
		\[
			g(\mathbf{Y}) = \sum_{t=1}^T \left[ -\frac{T}{2} \| W_t \|_F^2 - \left\langle W_t, X_tX_t^{\top} \right\rangle + \epsilon \, \mathrm{tr}(\Delta_t) \right]
		\]
		under the constraints \( \Delta_t \succeq 0 \), \( \| Z_t \|_F \leq \lambda_2 \xi_{2t} \), and \( |Y_{uv,t}| \leq \lambda_1 \xi_{uv,1t} \), and \( g(\mathbf{Y}) = -\infty \) otherwise. Since \( \Delta_t = Z_{t+1} - Z_t + W_t - Y_{t, \mathrm{off}} \), this is exactly the dual problem~\eqref{eq:dual-prob}.
		Strong duality follows from~\cite[Theorem 31.1]{R70}, because \eqref{eq:cons-opt-prob} satisfies Slater's condition; for instance, one may take \( \Theta_t = 2 \epsilon I_p \), \( U_t = 2 \epsilon I_p - X_tX_t^{\top} \), \( \Upsilon_{t, \mathrm{off}} = \mathbf{0}_{p \times p} \), and \( D_t = \mathbf{0}_{p \times p} \) for all \( t \).
\end{proof}

\section{Proofs}\label{appendix:proofs}

\subsection{Proof of Theorem \ref{theorem_consistency}}

\textbf{\emph{Proof of point (i).}}

The proof builds upon the works of \citet{Harchaoui2010}, Proposition 3, \citet{Qian2016}, Theorem 3.1 and \citet{Gibberd2021}, Theorem 1. We define: $A_{T,j} = \big\{|\widetilde{T}_j-T^*_j| \geq T\delta_T\big\}, \;\; C_T = \big\{\underset{1\leq j \leq m^*}{\max}|\widetilde{T}_j-T^*_j|<\mathcal{I}_{\min}/2\big\}$.
By union bound, $\Pb(\underset{1 \leq j \leq m^*}{\max}|\widetilde{T}_j-T^*_j|\geq T \delta_T) \leq \sum^{m^*}_{j=1}\Pb(A_{T,j})$, $m^*<\infty$. So we aim to show: \\
$(a) \; \sum^{m^*}_{j=1}\Pb(A_{T,j} \cap C_T) \rightarrow 0, \;\; (b) \; \sum^{m^*}_{j=1}\Pb(A_{T,j} \cap C^c_T) \rightarrow 0$, with $C^c_T$ the complement of $C_T$.

\noindent\emph{Proof of (a).} We show: $\sum^{m^*}_{j=1}\Pb(A^+_{T,j} \cap C_T) \rightarrow 0 \; \text{and} \; \sum^{m^*}_{j=1}\Pb(A^-_{T,j} \cap C_T) \rightarrow 0$,
where $A^+_{T,j}=\{T^*_j-\widetilde{T}_j \geq T \delta_T\}, A^-_{T,j}=\{\widetilde{T}_j-T^*_j \geq T \delta_T\}$.
We prove $\sum^{m^*}_{j=1}\Pb(A^+_{T,j} \cap C_T) \rightarrow 0$ as the other case follows in the same spirit. In light of $C_T$:
\begin{equation}\label{c_T_def}
\forall j \in \{1,\ldots,m^*\}, \; T^*_{j-1} < \widetilde{T}_j < T^*_{j+1}.
\end{equation}
By Lemma \ref{optimality_cond}, with $t = T^*_j$ and $t=\widetilde{T}_j$, in $\text{vec}(\cdot)$ form:
\begin{align*}
\lefteqn{\frac{1}{T}\overset{T}{\underset{r=\widetilde{T}_j}{\sum}}\text{vec}\big(\widetilde{\Theta}_r-X_rX^\top_r\big)+\text{vec}\left(\lambda \frac{\widetilde{\Gamma}_{\widetilde{T}_j}}{\|\widetilde{\Gamma}_{\widetilde{T}_j}\|_F}\right) }\\
& = \frac{1}{T}\overset{T}{\underset{r=\widetilde{T}_j}{\sum}}\Big(\text{vec}\big(\Theta^*_r-X_rX^\top_r\big)+(I_p\otimes I_p)\text{vec}\big(\widetilde{\Theta}_r-\Theta^*_r\big)\Big)+\text{vec}\big(\lambda \frac{\widetilde{\Gamma}_{\widetilde{T}_j}}{\|\widetilde{\Gamma}_{\widetilde{T}_j}\|_F}\big)= \mathbf{0}_{p^2 \times 1},
\end{align*}
and $\left\|\frac{1}{T}\!\overset{T}{\underset{r=T^*_j}{\sum}}
\Big(\text{vec}\big(\Theta^*_r-X_rX^\top_r\big)+(I_p\otimes I_p)\text{vec}\big(\widetilde{\Theta}_r-\Theta^*_r\big)\Big)\right\|_2\leq\lambda$.
Therefore, under $T^*_j>\widetilde{T}_j$, taking the differences, by the triangle inequality, we obtain:
\begin{align}
2\lambda & \geq \left\|\!\frac{1}{T}\!\overset{T^*_j-1}{\underset{r=\widetilde{T}_j}{\sum}}\Big(\text{vec}\big(\Theta^*_r-X_rX^\top_r\big)+(I_p\otimes I_p)\text{vec}\big(\widetilde{\Theta}_r-\Theta^*_r\big)\Big)\right\|_2 \nonumber\\
& \geq \left\|\!\frac{1}{T}\!\overset{T^*_j-1}{\underset{r=\widetilde{T}_j}{\sum}}\Big(\text{vec}\big(\Sigma^*_j-X_rX^\top_r\big)+(I_p\otimes I_p)\text{vec}\big(\widetilde{\Sigma}_{j+1}-\Sigma^*_j\big)\Big)\right\|_2 \nonumber\\
& \geq \left\|\!\frac{1}{T}\!\overset{T^*_j-1}{\underset{r=\widetilde{T}_j}{\sum}}\Big((I_p\otimes I_p)\text{vec}\big(\Sigma_{j+1}^{*}-\Sigma^*_j\big)\Big)\|_2 - \|\!\frac{1}{T}\!\overset{T^*_j-1}{\underset{r=\widetilde{T}_j}{\sum}}\Big((I_p\otimes I_p)\text{vec}\big(\widetilde{\Sigma}_{j+1}-\Sigma^*_{j+1}\big)\Big)\right\|_2 \nonumber\\&\quad  - \left\|\!\frac{1}{T}\!\overset{T^*_j-1}{\underset{r=\widetilde{T}_j}{\sum}}\Big(\text{vec}\big(\Sigma^*_{j}-X_rX^\top_r\big)\Big)\right\|_2 := R_{Tj,1}+R_{Tj,2}+R_{Tj,3}. \label{bound_optimality}
\end{align}
where the first equality holds since $\widetilde{\Theta}_r = \widetilde{\Sigma}_{j+1}$ and $\Theta^*_r=\Sigma^*_j$ for $r \in [\widetilde{T}_j,T^*_j-1]$ by (\ref{c_T_def}).
Let the event: $\overline{R}_{Tj} = \{2 \lambda \geq \frac{1}{3}R_{Tj,1}\}\cup \{R_{Tj,2}\geq \frac{1}{3}R_{Tj,1}\}\cup \{R_{Tj,3}\geq \frac{1}{3}R_{Tj,1}\}$.
Since inequality (\ref{bound_optimality}) holds with probability one, then $\Pb(\overline{R}_{Tj})=1$. Therefore, we have:
\begin{align*}
  \Pb(A^+_{T,j}\cap C_T) \leq {} & \Pb(A^+_{T,j}\cap C_T\cap \{2 \lambda\geq \tfrac{1}{3}R_{Tj,1}\}) + \Pb(A^+_{T,j}\cap C_T\cap \{R_{Tj,2}\geq \tfrac{1}{3}R_{Tj,1}\})\\
  &+ \Pb(A^+_{T,j}\cap C_T\cap \{R_{Tj,3}\geq \tfrac{1}{3}R_{Tj,1}\}) =:AC_{j,1}+AC_{j,2}+AC_{j,3}.
\end{align*}
Let us first bound $\sum^{m^*}_{j=1}AC_{j,1}$. Since $\|AB\|_F \geq \lambda_{\min}(A^\top A)^{1/2}\|B\|_F$, for $1 \leq j \leq m^*$:
\begin{align*}
		     AC_{j,1} \leq {} & \Pb(A^+_{T,j}\cap \{2 \lambda \geq \tfrac{1}{3}R_{Tj,1}\}) \\
		\leq {} & \Pb\left(\left\|\frac{1}{T^*_j-\widetilde{T}_j}\overset{T^*_j-1}{\underset{r=\widetilde{T}_j}{\sum}} \text{vec}(\Sigma^*_{j+1}-\Sigma^*_j)\right\|_2 \leq \frac{3T}{T^*_j-\widetilde{T}_j} 2 \lambda,T^*_j-\widetilde{T}_j \geq T\delta_T\right) \\
		\leq {} & \Pb\Big(\|\Sigma^*_{j+1}-\Sigma^*_j\|_F \leq \frac{6T\lambda}{T^*_j-\widetilde{T}_j},T^*_j-\widetilde{T}_j \geq T\delta_T\Big) 
		\leq \Pb\Big(1\leq \frac{6\lambda}{\eta_{\min}\delta_T},T^*_j-\widetilde{T}_j \geq T\delta_T\Big).
\end{align*}
with $\eta_{\min}=\underset{1 \leq j \leq m^*}{\min}\|\Sigma^*_{j+1}-\Sigma^*_j\|_F$. By $\lambda/(\eta_{\min}\delta_T)\rightarrow 0$ in Assumption \ref{assumption_rates}(iv), we deduce $\sum^{m^*}_{j=1}AC_{j,1} \rightarrow 0$. We now bound $\sum^{m^*}_{j=1}AC_{j,2}$. For any $j=1,\ldots,m^*$:
\begin{align*}
  AC_{j, 2} = {} & \Pb\Big(\!A^+_{T,j} \!\cap\! C_T \!\cap\! \Big\{\!\left\|\!\frac{1}{T^*_j\!-\!\widetilde{T}_j}\!\overset{T^*_j\!-\!1}{\underset{r\!=\!\widetilde{T}_j}{\sum}}\!(I_p \otimes I_p)\text{vec}(\widetilde{\Sigma}_{j+1}\!-\!\Sigma^*_{j+1})\!\right\|_2 \!\geq\! \frac{1}{3}\!\left\|\!\frac{1}{T^*_j\!-\!\widetilde{T}_j}\!\overset{T^*_j\!-\!1}{\underset{r\!=\!\widetilde{T}_j}{\sum}}\! \!(I_p \otimes I_p)\text{vec}(\Sigma^*_{j+1}\!-\!\Sigma^*_j)\!\right\|_2\!\Big\}\!\Big) \\
  \leq {} & \Pb\Big(A^+_{T,j}\cap C_T \cap \Big\{ \|\widetilde{\Sigma}_{j+1}-\Sigma^*_{j+1}\|_F \geq \frac{1}{3}\|\Sigma^*_{j+1}-\Sigma^*_j\|_F\}\Big).
\end{align*}
We now need to evaluate the bound for $\|\widetilde{\Sigma}_{j+1}-\Sigma^*_{j+1}\|_F$. To do so, we rely on the KKT conditions of Lemma \ref{optimality_cond}. We have $\widetilde{\Theta}_t = \widetilde{\Sigma}_{j+1}$ when $t \in [T^*_j,(T^*_j+T^*_{j+1})/2-1]$ as $\widetilde{T}_j<T^*_j$ given $A^+_{T,j}$ and $\widetilde{T}_{j+1}>(T^*_j+T^*_{j+1})/2$ given $C_T$. Therefore, by Lemma \ref{optimality_cond} with $l = (T^*_j+T^*_{j+1})/2$ and $l=T^*_j$, following the steps to obtain inequality (\ref{bound_optimality}), we get
\begin{align*}
2 \lambda 
& \geq\left\|\frac{1}{T}\overset{(T^*_j+T^*_{j+1})/2-1}{\underset{r=T^*_j}{\sum}} \text{vec}(\widetilde{\Sigma}_{j+1}-\Sigma^*_{j+1})\right\|_2 - \left\|\frac{1}{T}\sum^{(T^*_j+T^*_{j+1})/2-1}_{r=T^*_j}\big(\Sigma^*_{j} - X_rX^\top_r\big)\right\|_F.
\end{align*}
Therefore, conditional on $C_T$: $\|\widetilde{\Sigma}_{j+1}-\Sigma^*_{j+1}\|_F
\leq \frac{2 T\lambda}{T^*_{j+1}-T^*_{j}}  + \left\|\frac{2}{T^*_{j+1}-T^*_{j}}\sum^{(T^*_j+T^*_{j+1})/2-1}_{r=T^*_j}\big(\Sigma^*_{j}- X_rX^\top_r\big)\right\|_F$.
We deduce
\begin{align}
\lefteqn{\sum^{m^*}_{j=1} \Pb\Big(\Big\{\|\widetilde{\Sigma}_{j+1}-\Sigma^*_{j+1}\|_F \geq \|\Sigma^*_{j+1}-\Sigma^*_j\|_F/3\Big\}\cap C_T\Big)}\nonumber\\
& \leq \sum^{m^*}_{j=1} \Pb\Big(\frac{2 T\lambda}{T^*_{j+1}-T^*_{j}} \geq \|\Sigma^*_{j+1}-\Sigma^*_j\|_F/6\Big)\nonumber\\
&\quad  +  \sum^{m^*}_{j=1} \Pb\left(\left\|\frac{2}{T^*_{j+1}-T^*_{j}}\sum^{(T^*_j+T^*_{j+1})/2-1}_{r=T^*_j}\big(\Sigma^*_j- X_rX^\top_r\big)\right\|_F \geq \|\Sigma^*_{j+1}-\Sigma^*_j\|_F/6\right)\label{control_sum}.
\end{align}
The first term tends to zero since $T\lambda/(\mathcal{I}_{\min}\eta_{\min}) \rightarrow 0$ by Assumption \ref{assumption_rates}(ii) and (iv). As for the second term, for $C>0$ finite, applying Lemma \ref{bound_gradient}, we deduce that for any $j$:
\begin{equation*}
\Pb\left(\left\|\frac{1}{(T^*_{j+1}-T^*_{j})/2}\sum^{(T^*_j+T^*_{j+1})/2-1}_{r=T^*_j}\Big(\Sigma^*_{j} - X_rX^\top_r\Big)\right\|_F \geq  \eta_{\min}/6\right) \rightarrow 0,
\end{equation*}
since $(\eta_{\min}\mathcal{I}^{1/2}_{\min})^{-1}p\sqrt{\log(pT)}\rightarrow 0$. So $\sum^{m^*}_{j=1}AC_{j,2} \rightarrow 0$. We now consider $\sum^{m^*}_{j=1}AC_{j,3}$. Applying the same reasoning to show the convergence of the second summation on the right-hand side of (\ref{control_sum}), we get
$\left\|\frac{1}{T^*_{j}-\widetilde{T}_{j}}\sum^{T^*_j-1}_{r=\widetilde{T}_j}\big(\Sigma^*_{j} - X_rX^\top_r\big)\right\|_F = O_p\left(p\sqrt{\frac{\log(pT)}{T\delta_T}}\right) = o_p(\eta_{\min})$,
when $T^*_j-\widetilde{T}_j\geq T\delta_T$,
and
\begin{align*}
\lefteqn{\sum^{m^*}_{j=1}AC_{j,3} \leq \Pb(A^+_{T,j}\cap \{R_{Tj,3}\geq \frac{1}{3}R_{Tj,1}\})} \\
& \leq \sum^{m^*}_{j=1}\!\Pb\left(A^+_{T,j} \!\cap\! \left\{\left\|\frac{1}{T^*_{j}\!-\!\widetilde{T}_{j}}\sum^{T^*_j-1}_{r=\widetilde{T}_j}\!\big(\Sigma^*_{j} -X_rX^\top_r\big)\right\|_F\!\geq\! \frac{1}{3}\left\|\frac{1}{T^*_j-\widetilde{T}_j}\overset{T^*_j-1}{\underset{r=\widetilde{T}_j}{\sum}}\! \text{vec}(\Sigma^*_{j+1}\!-\!\Sigma^*_j)\right\|_2\right\}\right)\\
& \leq  \sum^{m^*}_{j=1} \Pb\left(A^+_{T,j} \cap \left\{\left\|\frac{1}{T^*_{j}-\widetilde{T}_{j}}\sum^{T^*_j-1}_{r=\widetilde{T}_j} \big(\Sigma^*_{j}- X_rX^\top_r\big)\right\|_F\geq \frac{1}{3}\eta_{\min}\right\}\right).
\end{align*}
Since $T\delta_T \leq T^*_j-\widetilde{T}_j$, then under $(\sqrt{T\delta_T}\eta_{\min})^{-1}\,p\sqrt{ T\log(pT)} \rightarrow 0$, we deduce $\sum^{m^*}_{j=1}AC_{j,3} \rightarrow 0$. Consequently, we proved $\sum^{m^*}_{j=1}\Pb(A_{T,j} \cap C_T) \rightarrow 0$.

\noindent\emph{Proof of (b).} We prove (b) by showing $\sum^{m^*}_{j=1}\Pb(A^+_{T,j} \cap C^c_T) \rightarrow 0$ and $\sum^{m^*}_{j=1}\Pb(A^-_{T,j} \cap C^c_T) \rightarrow 0$. As in the proof of (a), we simply show $\sum^{m^*}_{j=1}\Pb(A^+_{T,j} \cap C^c_T) \rightarrow 0$. To do so, we define:
\begin{align*}
D^{(l)}_T&\coloneqq \big\{\exists j \in \{1,\ldots,m^*\}, \widetilde{T}_j \leq T^*_{j-1}\big\} \cap C^c_T, \; 
D^{(m)}_T\coloneqq \big\{\forall j \in \{1,\ldots,m^*\}, T^*_{j-1}<\widetilde{T}_j < T^*_{j+1}\big\} \cap C^c_T,\\
D^{(r)}_T&\coloneqq \big\{\exists j \in \{1,\ldots,m^*\}, \widetilde{T}_j \geq T^*_{j+1}\big\} \cap C^c_T,
\end{align*}
where $C^c_T = \{\underset{1 \leq j \leq m^*}{\max}|\widetilde{T}_j-T^*_j|\geq \mathcal{I}_{\min}/2\}$. Then, we have:
\begin{equation*}
\sum^{m^*}_{j=1}\Pb(A^+_{T,j} \cap C^c_T) = \sum^{m^*}_{j=1}\Big[\Pb(A^+_{T,j} \cap D^{(l)}_T)+\Pb(A^+_{T,j} \cap D^{(m)}_T) +\Pb(A^+_{T,j} \cap D^{(r)}_T)\Big].
\end{equation*}
We first bound $\sum^{m^*}_{j=1}\Pb(A^+_{T,j} \cap D^{(m)}_T)$. For any $j$:
\begin{align*}
\Pb(A^+_{T,j} \cap D^{(m)}_T)
& \leq \Pb(A^+_{t,j} \cap \big\{\widetilde{T}_{j+1}-T^*_j \geq \frac{1}{2}\mathcal{I}_{\min}\big\}\cap D^{(m)}_T) + \Pb(A^+_{t,j} \cap \big\{\widetilde{T}_{j+1}-T^*_j < \frac{1}{2}\mathcal{I}_{\min}\big\}\cap D^{(m)}_T)\\
& \leq \Pb(A^+_{t,j} \cap \big\{\widetilde{T}_{j+1}-T^*_j \geq \frac{1}{2}\mathcal{I}_{\min}\big\}\cap D^{(m)}_T) + \Pb(A^+_{t,j} \cap \big\{T^*_{j+1}-\widetilde{T}_{j+1} \geq \frac{1}{2}\mathcal{I}_{\min}\big\}\cap D^{(m)}_T),
\end{align*}
since $0 \leq \widetilde{T}_{j+1}-T^*_j \leq \mathcal{I}_{\min}/2$ implies $T^*_{j+1}-\widetilde{T}_{j+1}=(T^*_{j+1}-T^*_j)-(\widetilde{T}_{j+1}-T^*_j)\geq \mathcal{I}_{\min}-\mathcal{I}_{\min}/2 = \mathcal{I}_{\min}/2$. Moreover, since
\begin{equation*}
\resizebox{\linewidth}{!}{$\displaystyle\Big\{A^+_{t,j} \cap \big\{T^*_{j+1}-\widetilde{T}_{j+1} \geq \frac{1}{2}\mathcal{I}_{\min}\big\}\cap D^{(m)}_T\Big\} \subset \overset{m^*-1}{\underset{k=j+1}{\cup}}\Big[ \big\{T^*_k-\widetilde{T}_k\geq \mathcal{I}_{\min}/2\big\} \cap \big\{\widetilde{T}_{k+1}-T^*_k \geq \mathcal{I}_{\min}/2\big\}\cap D^{(m)}_T\Big],$}
\end{equation*}
we deduce:
\begin{align}\label{control_Dm}
\lefteqn{\sum^{m^*}_{j=1}\Pb(A^+_{T,j} \cap D^{(m)}_T)  \leq \sum^{m^*}_{j=1} \Pb(A^+_{t,j} \cap \big\{\widetilde{T}_{j+1}-T^*_j \geq \frac{1}{2}\mathcal{I}_{\min}\big\}\cap D^{(m)}_T)  } \notag\\
&\quad  + \sum^{m^*}_{j=1} \sum^{m^*-1}_{k=j+1}  \Pb\Big(\big\{T^*_k-\widetilde{T}_k\geq \mathcal{I}_{\min}/2\big\} \cap \big\{\widetilde{T}_{k+1}-T^*_k \geq \mathcal{I}_{\min}/2\big\}\cap D^{(m)}_T\Big).
\end{align}
Let us treat the first term. By Lemma \ref{optimality_cond} with $t=\widetilde{T}_j$ and $t=T^*_j$, we obtain:
\begin{equation*}
\frac{1}{T}\overset{T}{\underset{r=\widetilde{T}_j}{\sum}}\Big(\text{vec}\big(\Theta^*_r-X_rX^\top_r\big)+(I_p \otimes I_p)\text{vec}\big(\widetilde{\Theta}_r-\Theta^*_r\big) \Big)= \lambda \text{vec}\left(\frac{\widetilde{\Gamma}_{\widetilde{T}_j}}{\|\widetilde{\Gamma}_{\widetilde{T}_j}\|_F}\right), 
\end{equation*}
and $\left\|\frac{1}{T}\overset{T}{\underset{r=T^*_j}{\sum}}\Big(\text{vec}\big(\Theta^*_r-X_rX^\top_r\big)+(I_p \otimes I_p)\text{vec}\big(\widetilde{\Theta}_r-\Theta^*_r\big) \Big)\right\|_2 \leq \lambda$.
We deduce
\begin{align*}
\lefteqn{\|\widetilde{\Sigma}_{j+1}-\Sigma^*_j\|_F - \left\|\frac{1}{T^*_j-\widetilde{T}_j}\overset{T^*_j-1}{\underset{r=\widetilde{T}_j}{\sum}}\big(\Sigma^*_{j} -X_rX^\top_r\big)\right\|_F }\\
& \leq \frac{1}{T^*_j-\widetilde{T}_j} \left\|\overset{T^*_j-1}{\underset{r=\widetilde{T}_j}{\sum}}\Big(\text{vec}\big(\Sigma^*_j-X_rX^\top_r\big)+(I_p \otimes I_p)\text{vec}\big(\widetilde{\Sigma}_{j+1}-\Sigma^*_j\big)\Big)\right\|_2\leq \frac{2\lambda T}{T^*_j-\widetilde{T}_j}.
\end{align*}
As a consequence:
\begin{equation}\label{bound_param_1}
\|\widetilde{\Sigma}_{j+1}-\Sigma^*_j\|_F \leq \frac{2\lambda T}{T^*_j-\widetilde{T}_j} + \left\|\frac{1}{T^*_j-\widetilde{T}_j}\overset{T^*_j-1}{\underset{r=\widetilde{T}_j}{\sum}}\big(\Sigma^*_{j} -X_rX^\top_r\big)\right\|_F.
\end{equation}
In the same vein, applying Lemma \ref{optimality_cond} with $t = \widetilde{T}_{j+1}$ and $t=T^*_j$, we obtain:
\begin{align}
\|\widetilde{\Sigma}_{j+1}-\Sigma^*_{j+1}\|_F  & \leq \frac{2\lambda T}{\widetilde{T}_{j+1}-T^*_j} + \left\|\frac{1}{\widetilde{T}_{j+1}-T^*_j}\overset{\widetilde{T}_{j+1}-1}{\underset{r=T^*_j}{\sum}}\big(\Sigma^*_{j} - X_rX^\top_r\big)\right\|_F.\label{bound_param_2}
\end{align}
Let the event:
\begin{align}
\lefteqn{E_{T,j} \coloneqq \Big\{\|\Sigma^*_{j+1}-\Sigma^*_j\|_F \leq 2\lambda \Big[\frac{T}{T^*_j-\widetilde{T}_j}+\frac{T}{\widetilde{T}_{j+1}-T^*_j}\Big]}\label{E_event}\\
&\quad  \left\|\frac{1}{T^*_j-\widetilde{T}_j}\overset{T^*_j-1}{\underset{r=\widetilde{T}_j}{\sum}}\big(\Sigma^*_{j} -X_rX^\top_r\big)\right\|_F + \left\|\frac{1}{\widetilde{T}_{j+1}-T^*_j}\overset{\widetilde{T}_{j+1}-1}{\underset{r=T^*_j}{\sum}}\big(\Sigma^*_{j} -X_rX^\top_r\big)\right\|_F\Big\}. \nonumber
\end{align}
Therefore, by the triangle inequality, (\ref{bound_param_1}) and (\ref{bound_param_2}) imply that the event $E_{T,j}$ holds with probability one. Hence:
{\footnotesize{\begin{align}
\lefteqn{\sum^{m^*}_{j=1} \Pb(A^+_{t,j} \cap \big\{\widetilde{T}_{j+1}-T^*_j \geq \frac{1}{2}\mathcal{I}_{\min}\big\}\cap D^{(m)}_T)}\nonumber\\
& = \sum^{m^*}_{j=1} \Pb(E_{T,j}\cap A^+_{t,j} \cap \big\{\widetilde{T}_{j+1}-T^*_j \geq \frac{1}{2}\mathcal{I}_{\min}\big\}\cap D^{(m)}_T)\nonumber\\
& \leq \sum^{m^*}_{j=1} \Pb(E_{T,j}\cap \big\{T^*_j - \widetilde{T}_j> T \delta_T\big\} \cap  \big\{\widetilde{T}_{j+1}-T^*_j \geq \frac{1}{2}\mathcal{I}_{\min}\big\}) \nonumber\\
& \leq \sum^{m^*}_{j=1} \Bigg\{\Pb(\frac{2\lambda}{\delta_T}+ \frac{4\lambda T}{\mathcal{I}_{\min}}\geq \|\Sigma^*_{j+1}-\Sigma^*_j\|_F/3) + \Pb\left(\left\{\left\|\frac{1}{T^*_j-\widetilde{T}_j}\overset{T^*_j-1}{\underset{r=\widetilde{T}_j}{\sum}}\big(\Sigma^*_{j} -X_rX^\top_r\big)\right\|_F \geq \|\Sigma^*_{j+1}-\Sigma^*_j\|_F/3\Big\}\cap \big\{T^*_j-\widetilde{T}_j > T \delta_T\right\}\right)\nonumber\\
&\quad  + \Pb\Bigg(\left\{\left\|\frac{1}{\widetilde{T}_{j+1}-T^*_j}\overset{\widetilde{T}_{j+1}-1}{\underset{r=T^*_j}{\sum}}\big(\Sigma^*_{j} -X_rX^\top_r\big)\right\|_F \geq \|\Sigma^*_{j+1}-\Sigma^*_j\|_F/3\right\}\cap \Big\{\widetilde{T}_{j+1}-T^*_j \geq \mathcal{I}_{\min}/2\Big\}\Bigg)\Bigg\}. \label{bound_A+}
\end{align}}}
The first term in (\ref{bound_A+}) tends to zero under $\lambda/(\eta_{\min}\delta_T)\rightarrow 0$, $\lambda T/(\mathcal{I}_{\min}\eta_{\min})\rightarrow 0$. Moreover, $\left\|\frac{1}{T^*_j-\widetilde{T}_j}\overset{T^*_j-1}{\underset{r=\widetilde{T}_j}{\sum}}\big(\Sigma^*_{j} -X_rX^\top_r\big)\right\|_F = O_p\left(p\sqrt{\frac{\log(pT)}{T\delta_T}}\right) = o_p(\eta_{\min})$, $\left\|\frac{1}{\widetilde{T}_{j+1}-T^*_j}\overset{\widetilde{T}_{j+1}-1}{\underset{r=T^*_j}{\sum}}\big(\Sigma^*_{j} - X_rX^\top_r\big)\right\|_F = O_p\left(p\sqrt{\frac{\log(pT)}{\mathcal{I}_{\min}}}\right) = o_p(\eta_{\min})$,
under Assumption \ref{assumption_rates}(ii)-(iii). In the same manner, we can show that the second term in (\ref{control_Dm}) tends to zero.

We now consider $\sum^{m^*}_{j=1}\Pb(A^+_{T,j} \cap D^{(l)}_T)$. The probability of the event $A^+_{T,j} \cap D^{(l)}_T$ is upper bounded by: $\Pb(D^{(l)}_T)\leq \sum^{m^*}_{j=1}2^{j-1}\Pb(\max(l \in \{1,\ldots,m^*\}:\widetilde{T}_l \leq T^*_{l-1})=j)$.
Now $\max(l \in \{1,\ldots,m^*\}:\widetilde{T}_l \leq T^*_{l-1})=j$ implies $\widetilde{T}_j\leq T^*_{j-1}$ and $\widetilde{T}_{l+1}>T^*_l$ for any $j \leq l \leq m^*$ and:
\begin{equation*}
\big\{\max(l \in \{1,\ldots,m^*\}:\widetilde{T}_l \leq T^*_{l-1})=j\big\}\subset \overset{m^*-1}{\underset{k=j}{\cup}} \Big(\big\{T^*_k-\widetilde{T}_k\geq \mathcal{I}_{\min}/2\big\}\cap\big\{\widetilde{T}_{k+1}-T^*_k\geq \mathcal{I}_{\min}/2\big\}\Big).
\end{equation*}
Therefore:
\begin{align}
\lefteqn{\sum^{m^*}_{j=1}\Pb(A^+_{T,j} \cap D^{(l)}_T)}\label{control_Dl}\\
  & \leq \resizebox{0.92\linewidth}{!}{$\displaystyle m^* \overset{m^*-1}{\underset{j=1}{\sum}}2^{j-1}\overset{m^*-1}{\underset{k=j}{\sum}} \Pb\Big(\big\{T^*_k-\widetilde{T}_k\geq \mathcal{I}_{\min}/2\big\}\cap\big\{\widetilde{T}_{k+1}-T^*_k\geq \mathcal{I}_{\min}/2\big\}\Big) + m^*2^{m^*-1}\Pb(T^*_{m^*}-\widetilde{T}_{m^*}\geq \mathcal{I}_{\min}/2).$}\nonumber
\end{align}
First, we consider the second term of the right-hand side of (\ref{control_Dl}). Let $j=m^*$ in (\ref{E_event}), then $E_{T,m^*}$ holds with probability one. Therefore:
\begin{align*}
\lefteqn{m^*2^{m^*-1}\Pb(T^*_{m^*}-\widetilde{T}_{m^*}\geq \mathcal{I}_{\min}/2)=m^*2^{m^*-1}\Pb(E_{T,m^*}\cap\big\{T^*_{m^*}-\widetilde{T}_{m^*}\geq \mathcal{I}_{\min}/2\big\})}\\
  & \leq  m^*2^{m^*-1}\Pb(\frac{2\lambda}{\delta_T} + \frac{4\lambda T}{\mathcal{I}_{\min}}\geq \|\Sigma^*_{m^*+1}-\Sigma^*_{m^*}\|_F/3) \\
  &\quad  + m^*2^{m^*-1} \Pb(\|\frac{1}{T^*_{m^*}\!-\!\widetilde{T}_{m^*}}\overset{T^*_{m^*}-1}{\underset{r=\widetilde{T}_{m^*}}{\sum}}\!\big(\Sigma^*_{m^*}- X_rX^\top_r\big)\|_F \!\geq\! \|\Sigma^*_{m^*+1}\!-\!\Sigma^*_{m^*}\|_F/3,T^*_{m^*}\!-\!\widetilde{T}_{m^*} \!\geq\! \mathcal{I}_{\min}/2) \\
  &\quad  + m^*2^{m^*-1}\Pb(\|\frac{1}{T-T^*_{m^*}}\overset{T}{\underset{r=T^*_{m^*}}{\sum}}\big(\Sigma^*_{m^*} - X_rX^\top_r\big)\|_F \geq \|\Sigma^*_{m^*+1}-\Sigma^*_{m^*}\|_F/3).
\end{align*}
Since $m^*2^{m^*-1} = O(T\log(T))$, then $\log(m^*2^{m^*-1}) = O(\log(T^{1+\eps/2}))$. So under the conditions $(\sqrt{T\delta_T}\eta_{\min})^{-1}\,p\sqrt{ \log(pT)} \rightarrow 0, (\mathcal{I}^{1/2}_{\min}\eta_{\min})^{-1}\,p\sqrt{ \log(pT)} \rightarrow 0$, the right-hand side of the previous inequality converges to zero. As for the first term of (\ref{control_Dl}), applying $j=k$ in (\ref{E_event}):
{\footnotesize{\begin{align*}
\lefteqn{m^* \overset{m^*-1}{\underset{j=1}{\sum}}2^{j-1}\overset{m^*-1}{\underset{k=j}{\sum}} \Pb\Big(\big\{T^*_k-\widetilde{T}_k\geq \mathcal{I}_{\min}/2\big\}\cap\big\{\widetilde{T}_{k+1}-T^*_k\geq \mathcal{I}_{\min}/2\big\}\Big)}\\
& \leq m^* 2^{m^*-1}\overset{m^*-1}{\underset{k=1}{\sum}} \Pb\Big(E_{T,k}\cap\big\{T^*_k-\widetilde{T}_k\geq \mathcal{I}_{\min}/2\big\}\cap\big\{\widetilde{T}_{k+1}-T^*_k\geq \mathcal{I}_{\min}/2\big\}\Big)\\
& \leq  m^* 2^{m^*-1}\overset{m^*-1}{\underset{k=1}{\sum}} \Big\{\Pb\left(\frac{2\lambda}{\delta_T} + \frac{4\lambda T}{\mathcal{I}_{\min}}\geq \|\Sigma^*_{k+1}-\Sigma^*_{k}\|_F/3\right) \\
&\quad  + \Pb\left(\left\|\frac{1}{T^*_{k}-\widetilde{T}_{k}}\overset{T^*_{k}-1}{\underset{r=\widetilde{T}_{k}}{\sum}}\big(\Sigma^*_{k} -X_rX^\top_r\big)\right\|_F \geq \|\Sigma^*_{k+1}-\Sigma^*_{k}\|_F/3,T^*_{k}-\widetilde{T}_{k} \geq \mathcal{I}_{\min}/2\right) \\
&\quad  + \Pb\left(\left\|\frac{1}{\widetilde{T}_{k+1}-T^*_{k}}\overset{\widetilde{T}_{k+1}-1}{\underset{r=T^*_{k}}{\sum}}\big(\Sigma^*_{k} -X_rX^\top_r\big)\right\|_F \geq \|\Sigma^*_{k+1}-\Sigma^*_{k}\|_F/3,\widetilde{T}_{k+1}-T^*_{k}\geq \mathcal{I}_{\min}/2\right)\Big\}.
\end{align*}}}
The right-hand side of the last inequality converges to zero under the same conditions. Finally, we can prove that $\sum^{m^*}_{j=1}\Pb(A^+_{T,j} \cap D^{(r)}_T) \rightarrow 0$.

\vskip 0.2in

\noindent\textbf{\emph{Proof of point (ii).}}\\
\noindent By point (i) and under Assumption \ref{assumption_rates}(ii), for any $j=1,\ldots,m^*$, $|\widetilde{T}_j-T^*_j|=O_p(T\delta_T)$, which is $|\widetilde{T}_j-T^*_j| = o_p(\mathcal{I}_{\min})$ under Assumption \ref{assumption_rates}(ii). Hence, $(T^*_{j-1}+T^*_j)/2 < \widetilde{T}_j < T^*_j$ or $T^*_j \leq \widetilde{T}_j < (T^*_j + T^*_{j+1})/2$ is satisfied for any $j$. Set $l = 1,\ldots,m^*$ and assume $(T^*_{l-1}+T^*_l)/2 < \widetilde{T}_l < T^*_l$ and consider two cases: (ii-a) $(T^*_{l}+T^*_{l+1})/2 < \widetilde{T}_{l+1} < T^*_{l+1}$ and (ii-b) $T^*_{l+1} \leq \widetilde{T}_{l+1}$. In case (ii-a), by Lemma \ref{optimality_cond} with change points $t = \widetilde{T}_l$ and $t = \widetilde{T}_{l+1}$:
{\footnotesize{\begin{align*}
\lefteqn{2\lambda  \geq \left\|\frac{1}{T}\overset{T}{\underset{r = \widetilde{T}_l}{\sum}} \text{vec}\big(\widetilde{\Theta}_r-X_rX^\top_r\big)-\frac{1}{T}\overset{T}{\underset{r = \widetilde{T}_{l+1}}{\sum}} \text{vec}\big(\widetilde{\Theta}_r-X_rX^\top_r\big)\right\|_2 = \left\|\frac{1}{T}\overset{\widetilde{T}_{l+1}-1}{\underset{r = \widetilde{T}_l}{\sum}} \text{vec}\big(\widetilde{\Theta}_r-X_rX^\top_r\big)\right\|_2}\\
& = \Big\|\frac{1}{T}\overset{T^*_l-1}{\underset{r = \widetilde{T}_l}{\sum}} \Big(\text{vec}\big(\Sigma^*_l-X_rX^\top_r\big) + (I_p\otimes I_p)\text{vec}\big(\widetilde{\Sigma}_{l+1}-\Sigma^*_l\big)\Big)+
\frac{1}{T}\overset{\widetilde{T}_{l+1}-1}{\underset{r = T^*_l}{\sum}} \Big(\text{vec}\big(\Sigma^*_{l+1}-X_rX^\top_r\big) + (I_p\otimes I_p)\text{vec}\big(\widetilde{\Sigma}_{l+1}-\Sigma^*_{l+1}\big)\Big)\Big\|_2\\
& \geq \left\|\frac{1}{T}\overset{\widetilde{T}_{l+1}-1}{\underset{r = T^*_l}{\sum}} \Big(\text{vec}\big(\Sigma^*_{l+1}-X_rX^\top_r\big) + (I_p\otimes I_p)\text{vec}\big(\widetilde{\Sigma}_{l+1}-\Sigma^*_{l+1}\big)\Big)\right\|_2 -\left\|\frac{1}{T}\overset{T^*_l-1}{\underset{r = \widetilde{T}_l}{\sum}} \Big(\text{vec}\big(\Sigma^*_l-X_rX^\top_r\big) + (I_p\otimes I_p)\text{vec}\big(\widetilde{\Sigma}_{l+1}-\Sigma^*_l\big)\Big)\right\|_2\\
& \geq \frac{\widetilde{T}_{l+1}-T^*_l}{T}\left\{\|\widetilde{\Sigma}_{l+1}-\Sigma^*_{l+1})\|_F - \left\|\frac{1}{\widetilde{T}_{l+1}-T^*_l}\overset{\widetilde{T}_{l+1}-1}{\underset{r = T^*_{l}}{\sum}}\big(\Sigma^*_{l+1}-X_rX^\top_r\big)\right\|_F\right\} \\
&\quad  - \frac{T^*_l-\widetilde{T}_l}{T}\left\|\frac{1}{T^*_l-\widetilde{T}_l}\overset{T^*_l-1}{\underset{r = \widetilde{T}_l}{\sum}} \Big(\text{vec}\big(\Sigma^*_l-X_rX^\top_r\big) + (I_p\otimes I_p)\text{vec}\big(\widetilde{\Sigma}_{l+1}-\Sigma^*_l\big)\Big)\right\|_2.
\end{align*}}}
Therefore, using part (i) of Theorem \ref{theorem_consistency}, and the bound on $\|\widetilde{\Sigma}_{l+1}-\Sigma^*_l\|_F$,
we obtain
\begin{align*}
2\lambda \geq \frac{\widetilde{T}_{l+1}-T^*_l}{T} \left\{\|\widetilde{\Sigma}_{l+1}-\Sigma^*_{l+1}\|_F-\left\|\frac{1}{\widetilde{T}_{l+1}-T^*_l}\overset{\widetilde{T}_{l+1}-1}{\underset{r = T^*_{l}}{\sum}}\big(\Sigma^*_{l+1}-X_rX^\top_r\big)\right\|_F\right\}- O_p\left(\frac{T^*_l-\widetilde{T}_l}{T}\right).
\end{align*}
We deduce $2\lambda
\geq \frac{\widetilde{T}_{l+1}-T^*_l}{T} \left\{\|\widetilde{\Sigma}_{l+1}-\Sigma^*_{l+1}\|_F-O_p\left( p\sqrt{\frac{\log(pT)}{I^*_{l+1}}}\right)\right\}- O_p\left(\frac{T^*_l-\widetilde{T}_l}{T}\right)$.
As a consequence, it can be deduced that
\begin{equation}\label{bound_prob_regime}
\|\widetilde{\Sigma}_{l+1}-\Sigma^*_{l+1}\|_F=O_p\left(\frac{\lambda \, T}{I^*_{l+1}}+\frac{T\delta_T}{I^*_{l+1}} + p\sqrt{ \frac{\log(pT)}{I^*_{l+1}}}\right).
\end{equation}
In case (ii-b), by Lemma \ref{optimality_cond}, with change points $t = \widetilde{T}_l$ and $t = \widetilde{T}_{l+1}$, we have
{\footnotesize{\begin{align*}
\lefteqn{2\lambda  \geq \left\|\frac{1}{T}\overset{\widetilde{T}_{l+1}-1}{\underset{r = \widetilde{T}_l}{\sum}} \text{vec}\big(\widetilde{\Theta}_r-X_rX^\top_r\big)\right\|_2}\\
& \geq \left\|\frac{1}{T}\overset{T^*_{l+1}-1}{\underset{r = T^*_l}{\sum}} \Big(\text{vec}\big(\Sigma^*_{l+1}-X_rX^\top_r\big) + (I_p\otimes I_p)\text{vec}\big(\widetilde{\Sigma}_{l+1}-\Sigma^*_{l+1}\big)\Big)\right\|_2 - \left\|\frac{1}{T}\overset{T^*_l-1}{\underset{r = \widetilde{T}_l}{\sum}} \Big(\text{vec}\big(\Sigma^*_l-X_rX^\top_r\big) + (I_p\otimes I_p)\text{vec}\big(\widetilde{\Sigma}_{l+1}-\Sigma^*_l\big)\Big)\right\|_2\\
&\quad  - \left\|\frac{1}{T}\overset{\widetilde{T}_{l+1}-1}{\underset{r = T^*_{l+1}}{\sum}} \Big(\text{vec}\big(\Sigma^*_{l+2}-X_rX^\top_r\big) + (I_p\otimes I_p)\text{vec}\big(\widetilde{\Sigma}_{l+1}-\Sigma^*_{l+2}\big)\Big)\right\|_2.
\end{align*}}}
We deduce $2\lambda
  \geq \frac{I^*_{l+1}}{T} \left\{\|\widetilde{\Sigma}_{l+1}-\Sigma^*_{l+1}\|_F-O_p\left(p \sqrt{ \frac{\log(pT)}{I^*_{l+1}}}\right)\right\} - O_p\left(\frac{T^*_l-\widetilde{T}_l}{T}\right)-O_p\left(\frac{\widetilde{T}_{l+1}-T^*_{l+1}}{T}\right)$.
Hence, (\ref{bound_prob_regime}) holds. Using similar arguments, we can show that the latter is satisfied when $T^*_l \leq \widetilde{T}_l < (T^*_l+T^*_{l+1})/2$.

\subsection{Proof of Theorem \ref{theorem_date_recov}}

Using the result of Theorem \ref{theorem_consistency},  we aim to show that:
\begin{equation}\label{proba_bound}
\Pb\big(\{h(\widetilde{\Tc}_{\widetilde{m}},\Tc^*_{m^*})> T\delta_T\}\cap \{m^* < \widetilde{m} \leq m_{\max}\}\big) \rightarrow 0 \;\; \text{as} \;\; T \rightarrow \infty.
\end{equation}
To so, we define: $L_{m,k,1} = \big\{\forall 1 \leq l \leq m, |\widetilde{T}_l - T^*_k|>T\delta_T \; \text{and} \; \widetilde{T}_l < T^*_k\big\}$,
$L_{m,k,2} = \big\{\forall 1 \leq l \leq m, |\widetilde{T}_l - T^*_k|>T\delta_T \; \text{and} \; \widetilde{T}_l > T^*_k\big\}$, 
$L_{m,k,3} = \big\{\exists 1 \leq l \leq m-1, |\widetilde{T}_l - T^*_k|>T\delta_T, |\widetilde{T}_{l+1} - T^*_k|>T\delta_T \; \text{and} \; \widetilde{T}_l < T^*_k < \widetilde{T}_{l+1}\big\}$.
The probability (\ref{proba_bound}) can be bounded as:
\begin{align*}
\lefteqn{\Pb\big(\{h(\widetilde{\Tc}_{\widetilde{m}},\Tc^*_{m^*})> T\delta_T\}\cap \{m^* < \widetilde{m} \leq m_{\max}\}\big) \leq \overset{m_{\max}}{\underset{m=m^*+1}{\sum}} \Pb\big(h(\widetilde{\Tc}_{\widetilde{m}},\Tc^*_{m^*})> T\delta_T\big)}\\
  & \leq \resizebox{0.92\linewidth}{!}{$\displaystyle\overset{m_{\max}}{\underset{m=m^*+1}{\sum}} \overset{m^*}{\underset{k=1}{\sum}}\Pb\big(\forall l \in \{1,\ldots,m\}, |\widetilde{T}_l-T^*_k|>T\delta_T\big) = \overset{m_{\max}}{\underset{m=m^*+1}{\sum}} \overset{m^*}{\underset{k=1}{\sum}} \Big[\Pb\big(L_{m,k,1}\big)+\Pb\big(L_{m,k,2}\big)+\Pb\big(L_{m,k,3}\big)\Big].$}
\end{align*}
We first focus on $\sum^{m_{\max}}_{m=m^*+1} \sum^{m^*}_{k=1}\Pb\big(L_{m,k,1}\big)$, which can be expressed as:
\begin{equation*}
\Pb\big(L_{m,k,1}\big) = \Pb\big(L_{m,k,1}\cap\{\widetilde{T}_m>T^*_{k-1}\}\big)+\Pb\big(L_{m,k,1}\cap\{\widetilde{T}_m\leq T^*_{k-1}\}\big).
\end{equation*}
By Lemma \ref{optimality_cond} with change points $t = \widetilde{T}_m$ and $t = T^*_k$, given the case $T^*_k\geq \widetilde{T}_m > T^*_{k-1}$:
\begin{align*}
\frac{1}{T}\overset{T}{\underset{r=\widetilde{T}_m}{\sum}}\text{vec}\big(\widetilde{\Theta}_r-X_rX^\top_r\big) + \text{vec}\left(\lambda\frac{\widetilde{\Gamma}_{\widetilde{T}_m}}{\|\widetilde{\Gamma}_{\widetilde{T}_m}\|_F}\right)= \mathbf{0}_{p^2 \times 1},
\end{align*}
and $\left\|\frac{1}{T}\overset{T}{\underset{r=T^*_k}{\sum}}\text{vec}\big(\widetilde{\Theta}_r)-X_rX^\top_r\big) \right\|_2 \leq  \lambda$.
Therefore, taking the differences, we get:
\begin{align*}
2 \lambda \geq \left\|\frac{1}{T}\overset{T^*_k-1}{\underset{r=\widetilde{T}_m}{\sum}}(I_p\otimes I_p)\text{vec}(\widetilde{\Sigma}_{m+1} - \Sigma^*_{k+1})+
\frac{1}{T}\overset{T^*_k-1}{\underset{r=\widetilde{T}_m}{\sum}}(I_p\otimes I_p)\text{vec}(\Sigma^*_{k+1} - \Sigma^*_{k})+\frac{1}{T}\overset{T^*_k-1}{\underset{r=\widetilde{T}_m}{\sum}}\text{vec}(\Sigma^*_k - X_rX^\top_r)\right\|_2.
\end{align*}
Therefore, the event $\Bc_T$ defined as
\begin{align*}
\lefteqn{\Bc_T \coloneqq  \Big\{\|\Sigma^*_{k+1}-\Sigma^*_k\|_F \leq \frac{2\lambda T}{T^*_k-\widetilde{T}_m}}\\[-.3cm]
  &\quad +\left\|\frac{1}{T^*_k-\widetilde{T}_m}\overset{T^*_k-1}{\underset{r=\widetilde{T}_m}{\sum}}(I_p\otimes I_p)\text{vec}(\widetilde{\Sigma}_{m+1} - \Sigma^*_{k+1})\right\|_2+\left\|\frac{1}{T^*_k-\widetilde{T}_m}\overset{T^*_k-1}{\underset{r=\widetilde{T}_m}{\sum}}\text{vec}\big(\Sigma^*_k-X_rX^\top_r\big) \right\|_2\Big] \Big\},
\end{align*}
holds with probability one. Hence, we deduce
\begin{align*}
\overset{m_{\max}}{\underset{m=m^*+1}{\sum}} \overset{m^*}{\underset{k=1}{\sum}} \Pb\big(L_{m,k,1} \cap \{\widetilde{T}_m>T^*_{k-1}\}\big) = \overset{m_{\max}}{\underset{m=m^*+1}{\sum}} \overset{m^*}{\underset{k=1}{\sum}} \Pb\big(\Bc_T \cap L_{m,k,1} \cap \{\widetilde{T}_m>T^*_{k-1}\}\big) \leq M_{1,1}+M_{1,2}+M_{1,3},
\end{align*}
with $M_{1,1} \coloneqq \overset{m_{\max}}{\underset{m=m^*+1}{\sum}} \overset{m^*}{\underset{k=1}{\sum}}\Pb\big(\|\Sigma^*_{k+1}-\Sigma^*_k\|_F/3\leq 2\lambda \delta^{-1}_T\big)$, and
$$
 M_{1,2} \!\coloneqq\! \overset{m_{\max}}{\underset{m=m^*+1}{\sum}} \overset{m^*}{\underset{k=1}{\sum}} \Pb\left(T^*_k\!-\!\widetilde{T}_m\!>\!T\delta_T,\|\Sigma^*_{k+1}\!-\!\Sigma^*_k\|_F/3 \!\leq\! \left\|\frac{1}{T^*_k-\widetilde{T}_m}\overset{T^*_k-1}{\underset{r=\widetilde{T}_m}{\sum}}\!(I_p\otimes I_p)\text{vec}(\widetilde{\Sigma}_{m+1} \!-\! \Sigma^*_{k+1})\right\|_2\right),
$$
\vspace{-.2cm}
$$
M_{1,3} \!\coloneqq\! \overset{m_{\max}}{\underset{m=m^*+1}{\sum}} \overset{m^*}{\underset{k=1}{\sum}} \Pb\left(T^*_k\!-\!\widetilde{T}_m\!>\!T\delta_T,\|\Sigma^*_{k+1}\!-\!\Sigma^*_k\|_F/3 \!\leq\! \left\|\frac{1}{T^*_k-\widetilde{T}_m}\overset{T^*_k-1}{\underset{r=\widetilde{T}_m}{\sum}}\!\text{vec}\big(\Sigma^*_k-X_rX^\top_r\big) \right\|_2\right).
$$
In the same vein as in the analysis of (\ref{bound_A+}), we can show that $M_{1,1}, M_{1,3} \rightarrow 0$ as $T \rightarrow \infty$. $M_{1,2}$ requires more arguments. By Lemma (\ref{optimality_cond}), with change points $t = T^*_k$ and $t = T^*_{k+1}$:
\begin{align*}
\left\|\frac{1}{T}\overset{T}{\underset{r=T^*_k}{\sum}}\text{vec}\big(\widetilde{\Sigma}_{m+1}-X_rX^\top_r) \right\|_2 \leq  \lambda, \; \text{and} \; 
\left\|\frac{1}{T}\overset{T}{\underset{r=T^*_{k+1}}{\sum}}\text{vec}\big(\widetilde{\Sigma}_{m+1}-X_rX^\top_r\big) \right\|_2 \leq  \lambda.
\end{align*}
Therefore $2\lambda\geq \left\|\frac{1}{T}\overset{T^*_{k+1}-1}{\underset{r=T^*_{k}}{\sum}}(I_p\otimes I_p)\text{vec}(\widetilde{\Sigma}_{m+1} - \Sigma^*_{k+1}) +\frac{1}{T}\overset{T^*_{k+1}-1}{\underset{r=T^*_k}{\sum}}\text{vec}\big(\Sigma^*_{k+1}-X_rX^\top_r\big)\right\|_2$, which implies
$
\|\widetilde{\Sigma}_{m+1}-\Sigma^*_{k+1}\|_F \leq \frac{2\lambda T}{T^*_{k+1}-T^*_{k}}+\left\|\frac{1}{T^*_{k+1}-T^*_{k}}\overset{T^*_{k+1}-1}{\underset{r=T^*_k}{\sum}}\text{vec}\big(\Sigma^*_{k+1}-X_rX^\top_r\big)\right\|_2$.
We deduce
{\footnotesize{\begin{align}
\lefteqn{M_{1,2} \leq \overset{m_{\max}}{\underset{m=m^*+1}{\sum}} \overset{m^*}{\underset{k=1}{\sum}} \Pb\big(\|\Sigma^*_{k+1}-\Sigma^*_k\|_F/3 \leq \|\widetilde{\Sigma}_{m+1}-\Sigma^*_{k+1}\|_F\big) } \label{proba_M2}\\
  & \leq\overset{m_{\max}}{\underset{m=m^*+1}{\sum}} \overset{m^*}{\underset{k=1}{\sum}}\Bigg[ \Pb\big( \|\Sigma^*_{k+1}-\Sigma^*_k\|_F/6 \leq \frac{2\lambda T}{\mathcal{I}_{\min}}\big)  + \Pb\left(\|\Sigma^*_{k+1}\!-\!\Sigma^*_k\|_F/6 \!\leq\! \left\|\frac{1}{T^*_{k+1}\!-\!T^*_{k}}\overset{T^*_{k+1}-1}{\underset{r=T^*_k}{\sum}}\!\text{vec}\Big(\Sigma^*_{k+1}-X_rX^\top_r\big)\right\|_2 \right)\Bigg].\nonumber
\end{align}}}
The first term in the second inequality of (\ref{proba_M2}) tends to zero under the conditions $\lambda T/(\mathcal{I}_{\min}\eta_{\min}) \rightarrow 0$. And under $(\eta_{\min}\mathcal{I}^{1/2}_{\min})^{-1}\,\sqrt{ T\log(T)}\rightarrow 0$, the second term tends to zero. Therefore, we conclude $\sum^{m_{\max}}_{m=m^*+1} \sum^{m^*}_{k=1} \Pb\big(L_{m,k,1} \cap \{\widetilde{T}_m>T^*_{k-1}\}\big) \rightarrow 0$ as $T \rightarrow \infty$. Based on similar arguments, we can show $\sum^{m_{\max}}_{m=m^*+1} \sum^{m^*}_{k=1} \Pb\big(L_{m,k,1} \cap \{\widetilde{T}_m \leq T^*_{k-1}\}\big) \rightarrow 0$ as $T \rightarrow \infty$. Therefore, $\sum^{m_{\max}}_{m=m^*+1} \sum^{m^*}_{k=1}\Pb\big(L_{m,k,1}\big) \rightarrow 0$ as $T \rightarrow \infty$. Similarly, it can be proved that $\sum^{m_{\max}}_{m=m^*+1} \sum^{m^*}_{k=1}\Pb\big(L_{m,k,2}\big) \rightarrow 0$ as $T \rightarrow \infty$. \\
We now consider $\sum^{m_{\max}}_{m=m^*+1} \sum^{m^*}_{k=1}\Pb\big(L_{m,k,3}\big)$. Define
\begin{equation*}
\begin{array}{llll}
\resizebox{\linewidth}{!}{$\displaystyle L^{(1)}_{m,k,3} \coloneqq L_{m,k,3}\cap \{T^*_{k-1}<\widetilde{T}_l < \widetilde{T}_{l+1}<T^*_{k+1}\}, L^{(2)}_{m,k,3} \coloneqq L_{m,k,3}\cap \{T^*_{k-1}<\widetilde{T}_l <T^*_{k+1}, \widetilde{T}_{l+1}\geq T^*_{k+1}\},$} &&\\
\resizebox{\linewidth}{!}{$\displaystyle L^{(3)}_{m,k,3} \coloneqq L_{m,k,3}\cap \{\widetilde{T}_l \leq T^*_{k-1},T^*_{k-1} <\widetilde{T}_{l+1}<T^*_{k+1}\}, L^{(4)}_{m,k,3} \coloneqq L_{m,k,3}\cap \{\widetilde{T}_l \leq T^*_{k-1},T^*_{k+1} <\widetilde{T}_{l+1}\}.$} &&
\end{array}
\end{equation*}
First, we consider $L^{(1)}_{m,k,3}$. By Lemma (\ref{optimality_cond}), for the change points $t = T^*_k$ and $t = \widetilde{T}_l$, we obtain
\begin{align}
2\lambda 
& \geq \left\|\frac{1}{T}\overset{T^*_{k}-1}{\underset{r=\widetilde{T}_{l}}{\sum}}(I_p\otimes I_p)\text{vec}\big(\widetilde{\Sigma}_{l+1} - \Sigma^*_{k}) +\frac{1}{T}\overset{T^*_{k}-1}{\underset{r=\widetilde{T}_{l}}{\sum}}\text{vec}\big(\Sigma^*_{k}-X_rX^\top_r)\right\|_2,\label{changepoint_1}
\end{align}
and for the change points $t = T^*_k$ and $t = \widetilde{T}_{l+1}$, we get
\begin{align}
2\lambda
& \geq \left\|\frac{1}{T}\overset{\widetilde{T}_{l+1}-1}{\underset{r=T^*_k}{\sum}}(I_p\otimes I_p)\text{vec}(\widetilde{\Sigma}_{l+1} - \Sigma^*_{k+1}) +\frac{1}{T}\overset{\widetilde{T}_{l+1}-1}{\underset{r=T^*_{k}}{\sum}}\text{vec}\big(\Sigma^*_{k+1}-X_rX^\top_r\big)\right\|_2.\label{changepoint_2}
\end{align}
Moreover, by the triangle inequality, we have
\begin{align*}
\lefteqn{\|\Sigma^*_{k+1}-\Sigma^*_k\|_F \leq \|\widetilde{\Sigma}_{l+1}-\Sigma^*_k\|_F + \|\widetilde{\Sigma}_{l+1}-\Sigma^*_{k+1}\|_F}\\
  & \leq \frac{2\lambda T}{T^*_k-\widetilde{T}_l} + \left\|\frac{1}{T^*_k-\widetilde{T}_l}\overset{T^*_{k}-1}{\underset{r=\widetilde{T}_{l}}{\sum}}\text{vec}\big(\Sigma^*_{k}-X_rX^\top_r\big)\right\|_2 +\frac{2\lambda T}{\widetilde{T}_{l+1}-T^*_k} + \left\|\frac{1}{\widetilde{T}_{l+1}-T^*_k}\overset{\widetilde{T}_{l+1}-1}{\underset{r=T^*_k}{\sum}}\text{vec}\big(\Sigma^*_{k}-X_rX^\top_r\big)\right\|_2.
\end{align*}
So $\sum^{m_{\max}}_{m=m^*+1} \sum^{m^*}_{k=1}\Pb\big(L^{(1)}_{m,k,3}\big)$ is upper bounded as follows
\begin{align*}
\lefteqn{\sum^{m_{\max}}_{m=m^*+1} \sum^{m^*}_{k=1}\Pb\big(L^{(1)}_{m,k,3}\big)
\leq \sum^{m_{\max}}_{m=m^*+1} \sum^{m^*}_{k=1}\Pb\big(\|\Sigma^*_{k+1}-\Sigma^*_{k}\|_F/3 \leq 2\lambda\delta^{-1}_T\big)}
  \\
&\quad  + \sum^{m_{\max}}_{m=m^*+1} \sum^{m^*}_{k=1}\Pb\left(\|\Sigma^*_{k+1}-\Sigma^*_{k}\|_F/3 \leq \left\|\frac{1}{T^*_k-\widetilde{T}_l}
\overset{T^*_{k}-1}{\underset{r=\widetilde{T}_{l}}{\sum}}\text{vec}\big(\Sigma^*_{k}-X_rX^\top_r)\right\|_2,T^*_{k}-\widetilde{T}_l\geq T \delta_T\right)\\
&\quad + \sum^{m_{\max}}_{m=m^*+1} \sum^{m^*}_{k=1}\Pb\left(\|\Sigma^*_{k+1}-\Sigma^*_{k}\|_F/3 \leq\left \|\frac{1}{\widetilde{T}_{l+1}-T^*_k}\overset{\widetilde{T}_{l+1}-1}{\underset{r=T^*_k}{\sum}}\text{vec}\big(\Sigma^*_{k}-X_rX^\top_r)\right\|_2, \widetilde{T}_{l+1}-T^*_{k}\geq T \delta_T\right),
\end{align*}
which tends to zero in the spirit as in (\ref{bound_A+}). For $L^{(2)}_{m,k,3}$, by Lemma \ref{optimality_cond} with change points $t = T^*_k$ and $t = \widetilde{T}_l$ to obtain (\ref{changepoint_1}) and with change points $t=T^*_{k}$, $t = T^*_{k+1}$, we get
\begin{align}
2\lambda
& \geq\left\|\frac{1}{T}\overset{T^*_{k+1}-1}{\underset{r=T^*_{k}}{\sum}}(I_p\otimes I_p)\text{vec}(\widetilde{\Sigma}_{l+1} - \Sigma^*_{k+1}) +\frac{1}{T}\overset{T^*_{k+1}-1}{\underset{r=T^*_{k}}{\sum}}\text{vec}\big(\Sigma^*_{k+1}-X_rX^\top_r\big)\right\|_2.\label{changepoint_3}
\end{align}
By the triangle inequality, we have
\begin{align*}
\lefteqn{\|\Sigma^*_{k+1}-\Sigma^*_{k}\|_F\leq \|\widetilde{\Sigma}_{l+1}-\Sigma^*_{k}\|_F+\|\widetilde{\Sigma}_{l+1}-\Sigma^*_{k+1}\|_F}\\
  & \leq  \frac{2\lambda T}{T^*_k-\widetilde{T}_l} + \left\|\frac{1}{T^*_k-\widetilde{T}_l}\overset{T^*_{k}-1}{\underset{r=\widetilde{T}_{l}}{\sum}}\text{vec}\big(\Sigma^*_{k}-X_rX^\top_r\big)\right\|_2 + \frac{2\lambda T}{T^*_{k+1}-T^*_k} + \left\|\frac{1}{T^*_{k+1}-T^*_k}\overset{T^*_{k+1}-1}{\underset{r=T^*_k}{\sum}}\text{vec}\big(\Sigma^*_{k+1}-X_rX^\top_r\big)\right\|_2.
\end{align*}
Therefore, we obtain
\begin{align*}
\lefteqn{\sum^{m_{\max}}_{m=m^*+1} \sum^{m^*}_{k=1}\Pb\big(L^{(2)}_{m,k,3}\big)
 \leq \sum^{m_{\max}}_{m=m^*+1} \sum^{m^*}_{k=1}\Pb\left(\|\Sigma^*_{k+1}-\Sigma^*_{k}\|_F/3 \leq 2\lambda \delta^{-1}_T+\frac{2\lambda T}{\mathcal{I}_{\min}}\right)}
  \\
  &\quad  + \sum^{m_{\max}}_{m=m^*+1} \sum^{m^*}_{k=1}\Pb\left(\|\Sigma^*_{k+1}-\Sigma^*_{k}\|_F/3 \leq  \left\|\frac{1}{T^*_k-\widetilde{T}_l}\overset{T^*_{k}-1}{\underset{r=\widetilde{T}_{l}}{\sum}}\text{vec}\big(\Sigma^*_{k}-X_rX^\top_r\big)\right\|_2,T^*_{k}-\widetilde{T}_l\geq T \delta_T\right) \\
  &\quad  + \sum^{m_{\max}}_{m=m^*+1} \sum^{m^*}_{k=1}\Pb\left(\|\Sigma^*_{k+1}-\Sigma^*_{k}\|_F/3 \leq \left\|\frac{1}{T^*_{k+1}-T^*_k}\overset{T^*_{k+1}-1}{\underset{r=T^*_k}{\sum}}\text{vec}\big(\Sigma^*_{k}-X_rX^\top_r\big)\right\|_2\right),
\end{align*}
which will tend to zero based on similar arguments as in the convergence of (\ref{bound_A+}). For $L^{(3)}_{m,k,3}$, by Lemma \ref{optimality_cond} with change points $t = T^*_{k-1}$ and $t = T^*_k$, we have
\begin{align}
2\lambda
& \geq \left\|\frac{1}{T}\overset{T^*_{k}-1}{\underset{r=T^*_{k-1}}{\sum}}(I_p\otimes I_p)\text{vec}\big(\widetilde{\Sigma}_{l+1} - \Sigma^*_{k}\big) +\frac{1}{T}\overset{T^*_{k}-1}{\underset{r=T^*_{k-1}}{\sum}}\text{vec}\big(\Sigma^*_{k}-X_rX^\top_r\big)\right\|_2,\label{changepoint_4}
\end{align}
and with change points $t=T^*_{k}$, $t = \widetilde{T}_{l+1}$, we get
\begin{align}
2\lambda
& \geq \left\|\frac{1}{T}\overset{\widetilde{T}_{l+1}-1}{\underset{r=T^*_{k}}{\sum}}(I_p\otimes I_p)\text{vec}\big(\widetilde{\Sigma}_{l+1} - \Sigma^*_{k+1}\big) +\frac{1}{T}\overset{\widetilde{T}_{l+1}-1}{\underset{r=T^*_{k}}{\sum}}\text{vec}\big(\Sigma^*_{k+1}-X_rX^\top_r\big)\right\|_2.\label{changepoint_5}
\end{align}
By the triangle inequality, we deduce
{\small{\begin{align*}
\lefteqn{\|\Sigma^*_{k+1}-\Sigma^*_{k}\|_F\leq \|\widetilde{\Sigma}_{l+1}-\Sigma^*_{k}\|_F+\|\widetilde{\Sigma}_{l+1}-\Sigma^*_{k+1}\|_F}\\
  & \leq \frac{2\lambda T}{T^*_k-T^*_{k-1}} + \left\|\frac{1}{T^*_k-T^*_{k-1}}\overset{T^*_{k}-1}{\underset{r=T^*_{k-1}}{\sum}}\text{vec}\big(\Sigma^*_{k}-X_rX^\top_r\big)\right\|_2 +\frac{2\lambda T}{\widetilde{T}_{l+1}-T^*_k} + \left\|\frac{1}{\widetilde{T}_{l+1}-T^*_k}\overset{\widetilde{T}_{l+1}-1}{\underset{r=T^*_k}{\sum}}\text{vec}\big(\Sigma^*_{k+1}-X_rX^\top_r)\right\|_2.
\end{align*}}}
We deduce
\begin{align*}
\lefteqn{\sum^{m_{\max}}_{m=m^*+1} \sum^{m^*}_{k=1}\Pb\big(L^{(3)}_{m,k,3}\big)\leq \sum^{m_{\max}}_{m=m^*+1} \sum^{m^*}_{k=1}\Pb\left(\|\Sigma^*_{k+1}-\Sigma^*_{k}\|_F/3 \leq \frac{2\lambda T}{\mathcal{I}_{\min}}+ 2\lambda \delta^{-1}_T\right)}
  \\
  &\quad  + \sum^{m_{\max}}_{m=m^*+1} \sum^{m^*}_{k=1}\Pb\left(\|\Sigma^*_{k+1}-\Sigma^*_{k}\|_F/3 \leq \left\|\frac{1}{T^*_k-T^*_{k-1}}\overset{T^*_{k}-1}{\underset{r=T^*_{k-1}}{\sum}}\text{vec}\big(\Sigma^*_{k}-X_rX^\top_r)\right\|_2\right) \\
  &\quad  + \sum^{m_{\max}}_{m=m^*+1} \sum^{m^*}_{k=1}\Pb\left(\|\Sigma^*_{k+1}-\Sigma^*_{k}\|_F/3 \leq \left\|\frac{1}{\widetilde{T}_{l+1}-T^*_{k+1}}\overset{\widetilde{T}_{l+1}-1}{\underset{r=T^*_k}{\sum}}\text{vec}\big(\Sigma^*_{k+1}-X_rX^\top_r)\right\|_2, \widetilde{T}_{l+1}-T^*_k\geq T\delta_T\right),
\end{align*}
which tends to zero based on the same arguments as in the convergence of (\ref{bound_A+}). Finally, to analyze $L^{(4)}_{m,k,3}$, applying Lemma \ref{optimality_cond} with $t = T^*_{k-1},t = T^*_k$ to obtain (\ref{changepoint_4}) and with $t = T^*_{k},t = T^*_{k+1}$ to obtain (\ref{changepoint_3}). By the triangle inequality, we have
{\footnotesize{\begin{align*}
\lefteqn{\|\Sigma^*_{k+1}-\Sigma^*_{k}\|_F\leq \|\widetilde{\Sigma}_{l+1}-\Sigma^*_{k}\|_F+\|\widetilde{\Sigma}_{l+1}-\Sigma^*_{k+1}\|_F}\\
& \leq\frac{2\lambda T}{T^*_k-T^*_{k-1}} + \left\|\frac{1}{T^*_k-T^*_{k-1}}\overset{T^*_{k}-1}{\underset{r=T^*_{k-1}}{\sum}}\text{vec}\big(\Sigma^*_{k}-X_rX^\top_r\big)\right\|_2 + \frac{2\lambda T}{T^*_{k+1}-T^*_k} + \left\|\frac{1}{T^*_{k+1}-T^*_k}\overset{T^*_{k+1}-1}{\underset{r=T^*_k}{\sum}}\text{vec}\big(\Sigma^*_{k+1}-X_rX^\top_r\big)\right\|_2,
\end{align*}}}
we deduce
\begin{align*}
\sum^{m_{\max}}_{m=m^*+1} \sum^{m^*}_{k=1}\Pb\big(L^{(4)}_{m,k,3}\big)
& \leq  \sum^{m_{\max}}_{m=m^*+1} \sum^{m^*}_{k=1}\Pb\left(\|\Sigma^*_{k+1}-\Sigma^*_{k}\|_F/3 \leq \frac{2\lambda T}{\mathcal{I}_{\min}}\right)\\
& + \sum^{m_{\max}}_{m=m^*+1} \sum^{m^*}_{k=1}\Pb\left(\|\Sigma^*_{k+1}-\Sigma^*_{k}\|_F/3 \leq  \left\|\frac{1}{T^*_k-T^*_{k-1}}\overset{T^*_{k}-1}{\underset{r=T^*_{k-1}}{\sum}}\text{vec}\big(\Sigma^*_{k}-X_rX^\top_r\big)\right\|_2\right)\\
& + \sum^{m_{\max}}_{m=m^*+1} \sum^{m^*}_{k=1}\Pb\left(\|\Sigma^*_{k+1}-\Sigma^*_{k}\|_F/3 \leq \left\|\frac{1}{T^*_{k+1}-T^*_k}\overset{T^*_{k+1}-1}{\underset{r=T^*_k}{\sum}}\text{vec}\big(\Sigma^*_{k}-X_rX^\top_r\big)\right\|_2\right),
\end{align*}
which also tends to zero, as in the proof of  the convergence to zero of (\ref{bound_A+}). We conclude that $\Pb\big(\{h(\widetilde{\Tc}_{\widetilde{m}},\Tc^*_{m^*})> T\delta_T\}\cap \{m^* < \widetilde{m} \leq m_{\max}\}\big) \rightarrow 0$ as $T \rightarrow \infty$.

\subsection{Proof of Theorem \ref{theorem_understimation_neg}}

The proof follows the proof of Theorem 3.3 of \citet{Qian2016}. Let $\breve{m}$, $\breve{\Tc}_{\breve{m}} = \{\breve{T}_1,\ldots,\breve{T}_{\breve{m}}\}$ and $\underline{\widetilde{\Sigma}}_{\breve{m}}(\breve{\Tc}_{\breve{m}})=(\widetilde{\Sigma}_1(\breve{\Tc}_{\breve{m}}),\ldots,\widetilde{\Sigma}_{\breve{m}}(\breve{\Tc}_{\breve{m}}), \widetilde{\Sigma}_{\breve{m}+1}(\breve{\Tc}_{\breve{m}}))$ be the hypothesized Group Fused LASSO estimator of the number of change points, set of change points and set of estimated variance-covariance matrices, respectively. We define $\underline{\widetilde{\Sigma}}_{\Tc_m} = \arg\min_{\underline{\Sigma}_m} \; \Lb(\underline{\Sigma}_m,\mathcal{X}_T,\Tc_m)$, 
with \\$\Lb(\underline{\Sigma}_m,\mathcal{X}_T,\Tc_m) = \frac{1}{2T}\overset{m+1}{\underset{j=1}{\sum}} \overset{T_j-1}{\underset{t=T_{j-1}}{\sum}}\text{tr}\big(\Sigma^\top_j\Sigma_j-(X_tX^\top_t)^\top\Sigma_j-\Sigma^\top_j(X_tX^\top_t)\big)$,
as the estimator of $\underline{\Sigma}_m = (\Sigma_1,\ldots,\Sigma_{m+1})$ for the given set of change-point locations $\Tc_m$. Define $\Lb_\lambda(\underline{\Sigma}_m,\mathcal{X}_T,\Tc_m) = \Lb(\underline{\Sigma}_m,\mathcal{X}_T,\Tc_m) + \lambda \sum^m_{j=1}\|\Sigma_{j+1}-\Sigma_j\|_F$.
We aim to show $\Pb(\Lb_\lambda(\underline{\widetilde{\Sigma}}_{\breve{m}}(\breve{\Tc}_{\breve{m}}),\mathcal{X}_T,\breve{\Tc}_{\breve{m}})>\Lb_\lambda(\underline{\widetilde{\Sigma}}_{m^*}(\widetilde{\Tc}_{m^*}),\mathcal{X}_T,\widetilde{\Tc}_{m^*})) \rightarrow 1$.
Under Assumption \ref{assumption_rates}(v), we have
\begin{align*}
\lefteqn{\frac{T}{\Ic_{\min}\eta^2_{\min}}\left[\Lb_\lambda(\underline{\widetilde{\Sigma}}_{\breve{m}}(\breve{\Tc}_{\breve{m}}),\mathcal{X}_T,\breve{\Tc}_{\breve{m}})-\Lb_\lambda(\underline{\widetilde{\Sigma}}_{m^*}(\widetilde{\Tc}_{m^*}),\mathcal{X}_T,\widetilde{\Tc}_{m^*})\right]}\\
& = \frac{T}{\Ic_{\min}\eta^2_{\min}}\Bigg[\frac{1}{2T}\overset{\breve{m}+1}{\underset{j=1}{\sum}}\overset{\breve{T}_j-1}{\underset{t=\breve{T}_{j-1}}{\sum}}\text{tr}\left(\widetilde{\Sigma}_j(\breve{\Tc}_{\breve{m}})^\top\widetilde{\Sigma}_j(\breve{\Tc}_{\breve{m}})-(X_tX^\top_t)^\top\widetilde{\Sigma}_j(\breve{\Tc}_{\breve{m}})-\widetilde{\Sigma}_j(\breve{\Tc}_{\breve{m}})^\top(X_tX^\top_t)\right) \\
&\quad  - \frac{1}{2T}\overset{m^*+1}{\underset{j=1}{\sum}}\overset{\widetilde{T}_j-1}{\underset{t=\widetilde{T}_{j-1}}{\sum}}\text{tr}\left(\widetilde{\Sigma}_j(\widetilde{\Tc}_{m^*})^\top\widetilde{\Sigma}_j(\widetilde{\Tc}_{m^*})- (X_tX^\top_t)^\top\widetilde{\Sigma}_j(\widetilde{\Tc}_{m^*})-\widetilde{\Sigma}_j(\widetilde{\Tc}_{m^*})^\top(X_tX^\top_t)\right)\Bigg]\\
&\quad  + \frac{\lambda T}{\Ic_{\min}\eta^2_{\min}}\left[\overset{\breve{m}}{\underset{j=1}{\sum}} \|\widetilde{\Sigma}_{j+1}(\breve{\Tc}_{\breve{m}})-\widetilde{\Sigma}_j(\breve{\Tc}_{\breve{m}})\|_F-\overset{m^*}{\underset{j=1}{\sum}} \|\widetilde{\Sigma}_{j+1}(\widetilde{\Tc}_{m^*})-\widetilde{\Sigma}_j(\widetilde{\Tc}_{m^*})\|_F\right]\\
& \geq \frac{T}{\Ic_{\min}\eta^2_{\min}}\left[\Lb(\underline{\widetilde{\Sigma}}_{\breve{m}}(\breve{\Tc}_{\breve{m}}),\mathcal{X}_T,\breve{\Tc}_{\breve{m}})-\Lb(\underline{\widetilde{\Sigma}}_{m^*}(\widetilde{\Tc}_{m^*}),\mathcal{X}_T,\widetilde{\Tc}_{m^*})\right]+o_p(1).
\end{align*}
It is sufficient to show
\begin{equation}\label{conv_obj}
\Pb\left(\underset{0 \leq m < m^*}{\inf}\;\frac{T}{\Ic_{\min}\eta^2_{\min}}\left[\Lb(\underline{\widetilde{\Sigma}}_{m}(\Tc_{m}),\mathcal{X}_T,\Tc_{m})-\Lb(\underline{\widetilde{\Sigma}}_{m^*}(\widetilde{\Tc}_{m}),\mathcal{X}_T,\widetilde{\Tc}_{m})\right]>c+o_p(1)\right) \rightarrow 1,
\end{equation}
for some $c >0$, with $\Tc_m =\{T_1,\ldots,T_{m}\}$, $1<T_1<\ldots<T_m<T$ an arbitrary $m$-dimensional set of potential change-point locations. To prove (\ref{conv_obj}), we show:
\begin{equation}\label{conv_obj_1}
\frac{T}{\Ic_{\min}\eta^2_{\min}}\Big[\Lb(\underline{\widetilde{\Sigma}}_{m^*}(\widetilde{\Tc}_{m^*}),\mathcal{X}_T,\widetilde{\Tc}_{m^*})-\Lb(\underline{\Sigma}^*_{m^*}(\Tc^*_{m^*}),\mathcal{X}_T,\Tc^*_{m^*})\Big] = o_p(1),
\end{equation}
and
\begin{equation}\label{conv_obj_2}
\Pb\left(\underset{0 \leq m < m^*}{\inf}\;\frac{T}{\Ic_{\min}\eta^2_{\min}}\Big[\Lb(\underline{\widetilde{\Sigma}}_{m}(\Tc_{m}),\mathcal{X}_T,\Tc_{m})-\Lb(\underline{\Sigma}^*_{m^*}(\Tc^*_{m^*}),\mathcal{X}_T,\Tc^*_{m^*})\Big]>c+o_p(1)\right) \rightarrow 1,    
\end{equation}
where $\Lb(\underline{\Sigma}^*_{m^*}(\Tc^*_{m^*}),\mathcal{X}_T,\Tc^*_{m^*}) = \frac{1}{2T}\overset{m^*+1}{\underset{j=1}{\sum}} \overset{T^*_j-1}{\underset{t=T^*_{j-1}}{\sum}}\text{tr}\left(\Sigma^{*\top}_j\Sigma^*_j-(X_tX^\top_t)^\top\Sigma^*_j-\Sigma^{*\top}_j(X_tX^\top_t)\right)$. We first show (\ref{conv_obj_1}). We have $\Lb(\underline{\widetilde{\Sigma}}_{m^*}(\widetilde{\Tc}_{m^*}),\mathcal{X}_T,\widetilde{\Tc}_{m^*})-\Lb(\underline{\Sigma}^*_{m^*}(\Tc^*_{m^*}),\mathcal{X}_T,\Tc^*_{m^*}) =: \overset{m^*+1}{\underset{j=1}{\sum}} \Lb_j$.
The analysis of $\Lb_j$ depends on four cases:
(a) $\widetilde{T}_{j-1} < T^*_{j-1}$ and $\widetilde{T}_j<T^*_j$; (b) $\widetilde{T}_{j-1} < T^*_{j-1}$ and $\widetilde{T}_j\geq T^*_j$; (c) $\widetilde{T}_{j-1} \geq T^*_{j-1}$ and $\widetilde{T}_j<T^*_j$; (d) $\widetilde{T}_{j-1} \geq T^*_{j-1}$ and $\widetilde{T}_j \geq T^*_j$. In case (a), we have
\begin{align*}
\lefteqn{\Lb_j = \frac{1}{2T} \overset{T^*_j-1}{\underset{t=T^*_{j-1}}{\sum}} \text{tr}\left(\widetilde{\Sigma}^{\top}_j\widetilde{\Sigma}_j-(X_tX^\top_t)^\top\widetilde{\Sigma}_j-\widetilde{\Sigma}^{\top}_j(X_tX^\top_t)-\big[\Sigma^{*\top}_j\Sigma^*_j-(X_tX^\top_t)^\top\Sigma^*_j-\Sigma^{*\top}_j(X_tX^\top_t)\big]\right)}\\
&\quad  + \frac{1}{2T} \overset{T^*_{j-1}}{\underset{t=\widetilde{T}_{j-1}}{\sum}} \text{tr}\left(\widetilde{\Sigma}^{\top}_j\widetilde{\Sigma}_j-(X_tX^\top_t)^\top\widetilde{\Sigma}_j-\widetilde{\Sigma}^{\top}_j(X_tX^\top_t)-\big[\Sigma^{*\top}_j\Sigma^*_j-(X_tX^\top_t)^\top\Sigma^*_j-\Sigma^{*\top}_j(X_tX^\top_t)\big]\right)\\
&\quad  - \frac{1}{2T} \overset{T^*_j-1}{\underset{t=\widetilde{T}_{j}}{\sum}} \text{tr}\left(\widetilde{\Sigma}^{\top}_j\widetilde{\Sigma}_j-(X_tX^\top_t)^\top\widetilde{\Sigma}_j-\widetilde{\Sigma}^{\top}_j(X_tX^\top_t)-\big[\Sigma^{*\top}_j\Sigma^*_j-(X_tX^\top_t)^\top\Sigma^*_j-\Sigma^{*\top}_j(X_tX^\top_t)\big]\right)\\
& =: \Lb^{(1)}_j+\Lb^{(2)}_j-\Lb^{(3)}_j.
\end{align*}
Since $\Lb^{(1)}_j = \text{vec}\big(\widetilde{\Sigma}_j-\Sigma^*_j\big)^\top\frac{1}{T}\overset{T^*_j-1}{\underset{t=T^*_{j-1}}{\sum}} \text{vec}\big(\Sigma^*_j-X_tX^\top_t\big)+ \text{vec}\big(\widetilde{\Sigma}_j-\Sigma^*_j\big)^\top\frac{1}{T}\overset{T^*_j-1}{\underset{t=T^*_{j-1}}{\sum}} (I_p \otimes I_p) \text{vec}\big(\widetilde{\Sigma}_j-\Sigma^*_j\big)$, then, uniformly in $j$: $|\Lb^{(1)}_j| \leq \|\widetilde{\Sigma}_j-\Sigma^*_j\|_FO_p\left(p\sqrt{\frac{\log(p\Ic_{\min})}{T}}\right) +    \|\widetilde{\Sigma}_j-\Sigma^*_j\|^2_F T^{-1}\max_j(\Ic^*_j)$.
As for $\Lb^{(2)}_j$, we have
\begin{align*}
\begin{array}{llll}
\Lb^{(2)}_j 
 & =\text{vec}\big(\widetilde{\Sigma}_j-\Sigma^*_j\big)^\top\frac{1}{T}\overset{T^*_j-1}{\underset{t=\widetilde{T}_{j-1}}{\sum}} \text{vec}\big(\Sigma^*_j-X_tX^\top_t\big)+ \text{vec}\big(\widetilde{\Sigma}_j-\Sigma^*_j\big)^\top\frac{1}{T}\overset{T^*_j-1}{\underset{t=\widetilde{T}_{j-1}}{\sum}} (I_p \otimes I_p) \text{vec}\big(\widetilde{\Sigma}_j-\Sigma^*_j\big)\\
& =: \Lb^{(2,1)}_j + \Lb^{(2,2)}_j. 
 \end{array}
\end{align*}
By Theorem \ref{theorem_consistency}(i) and Assumption \ref{assumption_dgp}:
\begin{align*}
\lefteqn{|\Lb^{(2,1)}_j|\leq\|\widetilde{\Sigma}_j-\Sigma^*_j\|_F \frac{1}{T}\overset{T^*_j-1}{\underset{t=\widetilde{T}_{j-1}}{\sum}} \|\Sigma^*_j-X_tX^\top_t\|_F}\\
& \leq \delta_T\|\widetilde{\Sigma}_j-\Sigma^*_j\|_F p\frac{1}{T\delta_T}\overset{T^*_j-1}{\underset{t=T^*_{j}-T\delta_T}{\sum}} \|\Sigma^*_j-X_tX^\top_t\|_{\max} = \delta_T\|\widetilde{\Sigma}_j-\Sigma^*_j\|_FO_p\left(p\sqrt{\frac{\log(pT)}{\delta_T T}}\right), 
\end{align*}
and $|\Lb^{(2,2)}_j|\leq \delta_T\|\widetilde{\Sigma}_j-\Sigma^*_j\|^2_F$. Therefore, $\Lb^{(2)}_j = \delta_T\big(\|\widetilde{\Sigma}_j-\Sigma^*_j\|_FO_p\left(p\sqrt{\frac{\log(pT)}{\delta_T T}}\right)+\|\widetilde{\Sigma}_j-\Sigma^*_j\|^2_F\big)$ uniformly in $j$. In the same manner, $\Lb^{(3)}_j = \delta_T\big(\|\widetilde{\Sigma}_j-\Sigma^*_j\|_FO_p\left(p\sqrt{\frac{\log(pT)}{\delta_T T}}\right)+\|\widetilde{\Sigma}_j-\Sigma^*_j\|^2_F\big)$. We deduce 
$\Lb_j = \left[\big(p\sqrt{\frac{\log(pT)}{T}}+\delta_T\big)\|\widetilde{\Sigma}_j-\Sigma^*_j\|_F+\|\widetilde{\Sigma}_j-\Sigma^*_j\|^2_F\right]O_p(1)$.
It can be shown that the same result holds for cases (b)-(d). By Theorem \ref{theorem_consistency}(ii) and Assumption \ref{assumption_rates}(iv)-(v), $\|\widetilde{\Sigma}_j-\Sigma^*_j\|_F=O_p(p\sqrt{\frac{\log(pT)}{\Ic_{\min}}})$, $p\sqrt{\log(p\Ic_{\min})/T}+\delta_T=O(p\sqrt{\log(p\Ic_{\min})/T})$. By Assumption \ref{assumption_rates}(v), we deduce (\ref{conv_obj_1}):
\begin{equation*}
\frac{T}{\Ic_{\min}\eta^2_{\min}}\overset{m^*+1}{\underset{j=1}{\sum}}\Lb_j = \frac{T(m^*+1)}{\Ic_{\min}\eta^2_{\min}}\left[O_p\left(p\sqrt{\frac{\log(pT)}{T}}\right)\,O_p\left(p\sqrt{\frac{\log(pT)}{\Ic_{\min}}}\right)+O_p\left(p^2\frac{\log(pT)}{\Ic_{\min}}\right)\right]=o_p(1).
\end{equation*}
We now show (\ref{conv_obj_2}). Hereafter, We assume that $m^*=1$ and $\widetilde{T}_1<T^*_1$ because the other cases can be studied in the same manner. Then $m=0$ and $\Tc_0$ is empty. Therefore, $\underline{\widetilde{\Sigma}}_{\Tc_0}$ is the minimum of the squared Frobenius distance $\|\widehat{S}-\Sigma\|^2_F$ denoted by $\widetilde{\Sigma}$, that is
 $\underline{\widetilde{\Sigma}}_{\Tc_0} = \widetilde{\Sigma} = \frac{1}{T}\overset{T}{\underset{t=1}{\sum}}X_tX^\top_t$. Therefore, with $m_0=1$, we have
\begin{equation*}
\widetilde{\Sigma} = \frac{1}{T}\overset{T^*_1-1}{\underset{t=1}{\sum}}\big(X_tX^\top_t-\Sigma^*_1\big)+\frac{I^*_1}{T}\Sigma^*_1+\frac{1}{T}\overset{T}{\underset{t=T^*_1}{\sum}}\big(X_tX^\top_t-\Sigma^*_1\big)+\frac{I^*_2}{T}\Sigma^*_2 = \Sigma^*_T + O_p\left(p\sqrt{\frac{\log(pT)}{T}}\right),
\end{equation*}
with $\Sigma^*_T = \frac{I^*_1}{T}\Sigma^*_1+\frac{I^*_2}{T}\Sigma^*_2=O_p(1)$. We then have $\Sigma^*_T-\Sigma^*_1= \frac{I^*_2}{T}(\Sigma^*_2-\Sigma^*_1)$, $\Sigma^*_T-\Sigma^*_2= -\frac{I^*_1}{T}(\Sigma^*_2-\Sigma^*_1)$. We obtain
\begin{align*}
\lefteqn{\Lb(\underline{\widetilde{\Sigma}}_{m^*}(\widetilde{\Tc}_{m^*}),\mathcal{X}_T,\widetilde{\Tc}_{m^*})-\Lb(\underline{\Sigma}^*_{m^*}(\Tc^*_{m^*}),\mathcal{X}_T,\Tc^*_{m^*}) }\\
& = \text{vec}(\widetilde{\Sigma}-\Sigma^*_1)^\top\frac{1}{T}\overset{T^*_1-1}{\underset{t=1}{\sum}}(I_p\otimes I_p)\text{vec}(\widetilde{\Sigma}-\Sigma^*_1) + \text{vec}(\widetilde{\Sigma}-\Sigma^*_1)^\top\frac{1}{T}\overset{T^*_1-1}{\underset{t=1}{\sum}}\text{vec}(\Sigma^*_1-X_tX^\top_t)  \\
&\quad  + \text{vec}(\widetilde{\Sigma}-\Sigma^*_2)^\top\frac{1}{T}\overset{T}{\underset{t=T^*_1}{\sum}}(I_p\otimes I_p)\text{vec}(\widetilde{\Sigma}-\Sigma^*_2) + \text{vec}(\widetilde{\Sigma}-\Sigma^*_2)^\top\frac{1}{T}\overset{T}{\underset{t=T^*_1}{\sum}}\text{vec}(\Sigma^*_2-X_tX^\top_t)\\
& = c_T + O_p\left(p\sqrt{\frac{\log(pT)}{T}}\Big(p\sqrt{\frac{\log(pT)}{T}}+\|\Sigma^*_2-\Sigma^*_1\|_F\Big)\right) = c_T+O_p\left(p\sqrt{\frac{\log(pT)}{T}}\right),
\end{align*}
where 
\begin{align*}
\lefteqn{c_T = \frac{I^*_1}{T}\text{vec}(\Sigma^*_T-\Sigma^*_1)^\top\text{vec}(\Sigma^*_T-\Sigma^*_1) + \frac{I^*_2}{T}\text{vec}(\Sigma^*_T-\Sigma^*_2)^\top\text{vec}(\Sigma^*_T-\Sigma^*_2)}\\
& = \frac{I^*_1I^*_2}{T}\left[\frac{I^*_2}{T}\text{vec}(\Sigma^*_2-\Sigma^*_1)^\top\text{vec}(\Sigma^*_2-\Sigma^*_1)+\frac{I^*_1}{T}\text{vec}(\Sigma^*_2-\Sigma^*_1)^\top\text{vec}(\Sigma^*_2-\Sigma^*_1)\right]\geq c\frac{\Ic_{\min}\eta^2_{\min}}{T},
\end{align*}
for some constant $c>0$. Therefore, 
\begin{equation*}
\frac{T}{\Ic_{\min}\eta^2_{\min}}\left[\Lb(\underline{\widetilde{\Sigma}}_{m}(\Tc_{m}),\mathcal{X}_T,\Tc_{m})-\Lb(\underline{\Sigma}^*_{m^*}(\Tc^*_{m^*}),\mathcal{X}_T,\Tc^*_{m^*})\right]\geq c+\frac{T}{\Ic_{\min}\eta^2_{\min}}O_p\left(p\sqrt{\frac{\log(pT)}{T}}\right)=c+o_p(1),
\end{equation*}
by Assumption \ref{assumption_rates}(v). This proves the case $m^*=1$. By similar arguments, we can show the general case $m^*\geq 2$.

\subsection{Proof of Theorem \ref{theorem_consistency_adaptive}}

\textbf{\emph{Proof of point (i).}}\\
We define: $A_{T,j} = \big\{|\widehat{T}_j-T^*_j| \geq T\delta_T\big\}, \;\; C_T = \big\{\underset{1\leq j \leq m^*}{\max}|\widehat{T}_j-T^*_j|<\mathcal{I}_{\min}/2\big\}$.
By union bound, $\Pb(\underset{1 \leq j \leq m^*}{\max}|\widehat{T}_j-T^*_j|\geq T \delta_T) \leq \sum^{m^*}_{j=1}\Pb(A_{T,j})$, $m^*<\infty$. So we aim to show: \\
$(a) \; \sum^{m^*}_{j=1}\Pb(A_{T,j} \cap C_T) \rightarrow 0, \;\; (b) \; \sum^{m^*}_{j=1}\Pb(A_{T,j} \cap C^c_T) \rightarrow 0$,
with $C^c_T$ the complement of $C_T$.

\noindent\emph{Proof of (a).} We show: $\sum^{m^*}_{j=1}\Pb(A^+_{T,j} \cap C_T) \rightarrow 0 \; \text{and} \; \sum^{m^*}_{j=1}\Pb(A^-_{T,j} \cap C_T) \rightarrow 0$,
where $A^+_{T,j}=\{T^*_j-\widehat{T}_j \geq T \delta_T\}, A^-_{T,j}=\{\widehat{T}_j-T^*_j \geq T \delta_T\}$.
We prove $\sum^{m^*}_{j=1}\Pb(A^+_{T,j} \cap C_T) \rightarrow 0$ as the other case follows in the same spirit. In light of $C_T$:
\begin{equation}\label{c_T_def_adaptive}
\forall j \in \{1,\ldots,m^*\}, \; T^*_{j-1} < \widehat{T}_j < T^*_{j+1}.
\end{equation}
By Lemma \ref{optimality_cond_adaptive}, with $t = T^*_j$ and $t=\widehat{T}_j$, in $\text{vec}(\cdot)$ form:
\begin{align*}
\frac{1}{T}\overset{T}{\underset{r=\widehat{T}_j}{\sum}}\left(\text{vec}\big(\Theta^*_r-X_rX^\top_r\big)+(I_p\otimes I_p)\text{vec}\big(\widehat{\Theta}_r-\Theta^*_r\big)\right)+\text{vec}\left(\lambda_1\overset{T}{\underset{r=\widehat{T}_j}{\sum}} \underset{u \neq v}{\sum} \xi_{uv,1r}\widehat{E}_{uv,1r}+ \lambda_2 \xi_{2\widehat{T}_j}\frac{\widehat{\Gamma}_{\widehat{T}_j}}{\|\widehat{\Gamma}_{\widehat{T}_j}\|_F}\right)= \mathbf{0}_{p^2 \times 1},
\end{align*}
and $\left\|\frac{1}{T}\!\overset{T}{\underset{r=T^*_j}{\sum}}
\Big(\text{vec}\big(\Theta^*_r-X_rX^\top_r\big)+(I_p\otimes I_p)\text{vec}\big(\widehat{\Theta}_r-\Theta^*_r\big)\Big)+\lambda_1\!\text{vec}\left(\!\overset{T}{\underset{r=T^*_j}{\sum}}\underset{u\neq v}{\sum}\xi_{uv,1r}\!\widehat{E}_{uv,1r}\right)\right\|_2\leq\lambda_2\xi_{2T^*_j}$.
Therefore, under $T^*_j>\widehat{T}_j$, taking the differences, by the triangle inequality, we obtain:
\begin{align*}
2\lambda_2(\xi_{2\widehat{T}_j}+\xi_{2T^*_j})\geq \left\|\!\frac{1}{T}\!\overset{T^*_j-1}{\underset{r=\widehat{T}_j}{\sum}}\Big(\text{vec}\big(\Theta^*_r-X_rX^\top_r\big)+(I_p\otimes I_p)\text{vec}\big(\widehat{\Theta}_r-\Theta^*_r\big)+\lambda_1\!\text{vec}\left(\!\overset{T^*_j-1}{\underset{r=T^*_j}{\sum}}\underset{u\neq v}{\sum}\xi_{uv,1r}\!\widehat{E}_{uv,1r}\right)\Big)\right\|_2 
\end{align*}
Each component of $\lambda_1\sum^{T^*_j-1}_{r=\widehat{T}_{j}}\underset{u \neq v}{\sum}\xi_{uv,1r}\widehat{E}_{uv,1r}$ is bounded by $\pm \lambda_1\big(\underset{\widehat{T}_j \leq r \leq T^*_j}{\sup}\underset{1 \leq u \neq v \leq p}{\max}\xi_{uv,1r}\big)(T^*_j-\widehat{T}_j)$, we deduce
\begin{align}
& 2\lambda_2(\xi_{2\widehat{T}_j}+\xi_{2T^*_j})+\lambda_1\sqrt{p(p-1)}\big(\sup_{\widehat{T}_j \leq r \leq T^*_j}\underset{1 \leq u \neq v \leq p}{\max}\xi_{uv,1r}\big)(T^*_j-T^*_{j-1})\nonumber\\
 \geq & \left\|\!\frac{1}{T}\!\overset{T^*_j-1}{\underset{r=\widehat{T}_j}{\sum}}\Big(\text{vec}\big(\Sigma^*_j-X_rX^\top_r\big)+(I_p\otimes I_p)\text{vec}\big(\widehat{\Sigma}_{j+1}-\Sigma^*_j\big)\Big)\right\|_2 \nonumber\\
 \geq & \left\|\!\frac{1}{T}\!\overset{T^*_j-1}{\underset{r=\widehat{T}_j}{\sum}}\Big((I_p\otimes I_p)\text{vec}\big(\Sigma_{j+1}-\Sigma^*_j\big)\Big)\right\|_2 - \left\|\!\frac{1}{T}\!\overset{T^*_j-1}{\underset{r=\widehat{T}_j}{\sum}}\Big((I_p\otimes I_p)\text{vec}\big(\widehat{\Sigma}_{j+1}^{*}-\Sigma^*_{j+1}\big)\Big)\right\|_2 \nonumber\\&\quad  - \left\|\!\frac{1}{T}\!\overset{T^*_j-1}{\underset{r=\widehat{T}_j}{\sum}}\text{vec}\big(\Sigma^*_{j}-X_rX^\top_r\big)\right\|_2 := R_{Tj,1}+R_{Tj,2}+R_{Tj,3}. \label{bound_optimality_adaptive}
\end{align}
where the first equality holds since $\widehat{\Theta}_r = \widehat{\Sigma}_{j+1}$ and $\Theta^*_r=\Sigma^*_j$ for $r \in [\widehat{T}_j,T^*_j-1]$ by (\ref{c_T_def_adaptive}).
Let: $\overline{R}_{Tj} = \{2\lambda_2(\xi_{2\widehat{T}_j}+\xi_{2T^*_j})+\lambda_1\sqrt{p(p-1)}\big(\sup_{\widehat{T}_j \leq r \leq T^*_j}\underset{1 \leq u \neq v \leq p}{\max}\xi_{uv,1r}\big)(T^*_j-T^*_{j-1}) \geq \frac{1}{3}R_{Tj,1}\}\cup \{R_{Tj,2}\geq \frac{1}{3}R_{Tj,1}\}\cup \{R_{Tj,3}\geq \frac{1}{3}R_{Tj,1}\}$.
Since inequality (\ref{bound_optimality_adaptive}) holds with probability one, then $\Pb(\overline{R}_{Tj})=1$. Therefore, we have:
\begin{align*}
  \Pb(A^+_{T,j}\cap C_T) \leq {} & \Pb(A^+_{T,j}\cap C_T\cap \{2\lambda_2(\xi_{2\widehat{T}_j}+\xi_{2T^*_j})+\lambda_1\sqrt{p(p-1)}\big(\sup_{\widehat{T}_j \leq r \leq T^*_j}\underset{1 \leq u \neq v \leq p}{\max}\xi_{uv,1r}\big)(T^*_j-T^*_{j-1})\geq \tfrac{1}{3}R_{Tj,1}\}) \\
                                 & + \Pb(A^+_{T,j}\cap C_T\cap \{R_{Tj,2}\geq \tfrac{1}{3}R_{Tj,1}\}) + \Pb(A^+_{T,j}\cap C_T\cap \{R_{Tj,3}\geq \tfrac{1}{3}R_{Tj,1}\}) \\
  \eqqcolon {}                   & AC_{j,1}+AC_{j,2}+AC_{j,3}.
\end{align*}
Let us first bound $\sum^{m^*}_{j=1}AC_{j,1}$. Since $\|AB\|_F \geq \lambda_{\min}(A^\top A)^{1/2}\|B\|_F$, for $1 \leq j \leq m^*$:
\begin{align*}
		     AC_{j,1} \leq {} & \Pb(A^+_{T,j}\cap \{2\lambda_2(\xi_{2\widehat{T}_j}+\xi_{2T^*_j})+\lambda_1\sqrt{p(p-1)}\big(\sup_{\widehat{T}_j \leq r \leq T^*_j}\underset{1 \leq u \neq v \leq p}{\max}\xi_{uv,1r}\big)(T^*_j-T^*_{j-1})\geq \tfrac{1}{3}R_{Tj,1}\}) \\
		\leq {} & \Pb\Big(\left\|\frac{1}{T^*_j-\widehat{T}_j}\overset{T^*_j-1}{\underset{r=\widehat{T}_j}{\sum}} \text{vec}(\Sigma^*_{j+1}-\Sigma^*_j)\right\|_2 \leq \frac{3T}{T^*_j-\widehat{T}_j} \big[2 \lambda_2(\xi_{2\widehat{T}_j}+\xi_{2T^*_j})\\
        {} & +\lambda_1\sqrt{p(p-1)}\big(\sup_{\widehat{T}_j \leq r \leq T^*_j}\underset{1 \leq u \neq v \leq p}{\max}\xi_{uv,1r}\big)(T^*_j-T^*_{j-1})\big],T^*_j-\widehat{T}_j \geq T\delta_T\Big) \\
		\leq {} & \Pb\Big(\|\Sigma^*_{j+1}-\Sigma^*_j\|_F \leq \frac{6T\lambda_2}{T^*_j-\widehat{T}_j}(\xi_{2\widehat{T}_j}+\xi_{2T^*_j})+3T\lambda_1\sqrt{p(p-1)}\big(\sup_{\widehat{T}_j \leq r \leq T^*_j}\underset{1 \leq u \neq v \leq p}{\max}\xi_{uv,1r}\big),T^*_j-\widehat{T}_j \geq T\delta_T\Big) \\
		\leq {} & \Pb\Big(1\leq \frac{6\lambda_2}{\eta_{\min}\delta_T}(\xi_{2\widehat{T}_j}+\xi_{2T^*_j})+\frac{3T\lambda_1}{\eta_{\min}}\sqrt{p(p-1)}\big(\sup_{\widehat{T}_j \leq r \leq T^*_j}\underset{1 \leq u \neq v \leq p}{\max}\xi_{uv,1r}\big),T^*_j-\widehat{T}_j \geq T\delta_T\Big).
\end{align*}
Under the conditions of Theorem \ref{theorem_consistency_adaptive}, under the conditions $\lambda_2\max(\alpha_{p,T},a_T)^{-\mu_2}/(\eta_{\min}\delta_T)\rightarrow 0$, $T\lambda_1\sqrt{p(p-1)}\max(\alpha_{p,T},a_T)^{-\mu_1}/\eta_{\min} \rightarrow 0$, we deduce $\sum^{m^*}_{j=1}AC_{j,1} \rightarrow 0$. We now bound $\sum^{m^*}_{j=1}AC_{j,2}$. For any $j=1,\ldots,m^*$:
\begin{align*}
  AC_{j, 2} = {} & \Pb\left(\!A^+_{T,j} \!\cap\! C_T \!\cap\! \left\{\!\left\|\!\frac{1}{T^*_j\!-\!\widehat{T}_j}\!\overset{T^*_j\!-\!1}{\underset{r\!=\!\widehat{T}_j}{\sum}}\!(I_p \otimes I_p)\text{vec}(\widehat{\Sigma}_{j+1}\!-\!\Sigma^*_{j+1})\!\right\|_2 \!\geq\! \frac{1}{3}\!\left\|\!\frac{1}{T^*_j\!-\!\widehat{T}_j}\!\overset{T^*_j\!-\!1}{\underset{r\!=\!\widehat{T}_j}{\sum}}\! \!(I_p \otimes I_p)\text{vec}(\Sigma^*_{j+1}\!-\!\Sigma^*_j)\!\right\|_2\!\right\}\!\right) \\
  \leq {} & \Pb\left(A^+_{T,j}\cap C_T \cap \Big\{ \|\widehat{\Sigma}_{j+1}-\Sigma^*_{j+1}\|_F \geq \frac{1}{3}\|\Sigma^*_{j+1}-\Sigma^*_j\|_F\}\right).
\end{align*}
We now need to evaluate the bound for $\|\widehat{\Sigma}_{j+1}-\Sigma^*_{j+1}\|_F$. To do so, we rely on the KKT conditions of Lemma \ref{optimality_cond_adaptive}. We have $\widehat{\Theta}_t = \widehat{\Sigma}_{j+1}$ when $t \in [T^*_j,(T^*_j+T^*_{j+1})/2-1]$ as $\widehat{T}_j<T^*_j$ given $A^+_{T,j}$ and $\widehat{T}_{j+1}>(T^*_j+T^*_{j+1})/2$ given $C_T$. Therefore, by Lemma \ref{optimality_cond_adaptive} with $l = (T^*_j+T^*_{j+1})/2$ and $l=T^*_j$, following the steps to obtain inequality (\ref{bound_optimality_adaptive}), we get
{\small{\begin{align*}
& 2 \lambda_2(\xi_{2(T^*_j+T^*_{j+1})/2-1}+\xi_{2T^*_j}) + \lambda_1\sqrt{p(p-1)}\big(\sup_{T^*_j \leq r \leq (T^*_j+T^*_{j+1})/2-1}\underset{1 \leq u \neq v \leq p}{\max}\xi_{uv,1r}\big)((T^*_j+T^*_{j+1})/2-T^*_j) \\
\geq & \left\|\frac{1}{T}\overset{(T^*_j+T^*_{j+1})/2-1}{\underset{r=T^*_j}{\sum}} \text{vec}(\widehat{\Sigma}_{j+1}-\Sigma^*_{j+1})\right\|_2 - \left\|\frac{1}{T}\sum^{(T^*_j+T^*_{j+1})/2-1}_{r=T^*_j}\big(\Sigma^*_{j} - X_rX^\top_r\big)\right\|_F.
\end{align*}}}
Therefore, conditional on $C_T$:
{\footnotesize{\begin{align*}
\lefteqn{\|\widehat{\Sigma}_{j+1}-\Sigma^*_{j+1}\|_F \leq \frac{2 T\lambda_2}{T^*_{j+1}-T^*_{j}}(\xi_{2(T^*_j+T^*_{j+1})/2-1}+\xi_{2T^*_j}) }\\
&\quad   +
\frac{T\lambda_1\sqrt{p(p-1)}}{T^*_{j+1}-T^*_{j}}\big(\sup_{T^*_j \leq r \leq (T^*_j+T^*_{j+1})/2-1}\underset{1 \leq u \neq v \leq p}{\max}\xi_{uv,1r}\big)((T^*_j+T^*_{j+1})/2-T^*_j)+\left\|\frac{2}{T^*_{j+1}-T^*_{j}}\sum^{(T^*_j+T^*_{j+1})/2-1}_{r=T^*_j}\big(\Sigma^*_{j}- X_rX^\top_r\big)\right\|_F.
\end{align*}}}
We deduce
{\footnotesize{\begin{align}
\lefteqn{\sum^{m^*}_{j=1} \Pb\Big(\Big\{\|\widehat{\Sigma}_{j+1}-\Sigma^*_{j+1}\|_F \geq \|\Sigma^*_{j+1}-\Sigma^*_j\|_F/3\Big\}\cap C_T\Big) \leq \sum^{m^*}_{j=1} \Pb\Big(\frac{2 T\lambda_2}{T^*_{j+1}-T^*_{j}}(\xi_{2(T^*_j+T^*_{j+1})/2-1}+\xi_{2T^*_j}) \nonumber}\\
&\quad  +
\frac{T\lambda_1\sqrt{p(p-1)}}{T^*_{j+1}-T^*_{j}}\big(\sup_{T^*_j \leq r \leq (T^*_j+T^*_{j+1})/2-1}\underset{1 \leq u \neq v \leq p}{\max}\xi_{uv,1r}\big)((T^*_j+T^*_{j+1})/2-T^*_j) \geq \|\Sigma^*_{j+1}-\Sigma^*_j\|_F/6\Big)\nonumber\\
&\quad  +  \sum^{m^*}_{j=1} \Pb\Bigg(\left\|\frac{2}{T^*_{j+1}-T^*_{j}}\sum^{(T^*_j+T^*_{j+1})/2-1}_{r=T^*_j}\big(\Sigma^*_j- X_rX^\top_r\big)\right\|_F \geq \|\Sigma^*_{j+1}-\Sigma^*_j\|_F/6\Bigg)\nonumber\\
& \leq \sum^{m^*}_{j=1} \Pb\Big(\frac{2 T\lambda_2}{\Ic_{\min}}(\xi_{2(T^*_j+T^*_{j+1})/2-1}+\xi_{2T^*_j})+ T\lambda_1\sqrt{p(p-1)}\big(\sup_{T^*_j \leq r \leq (T^*_j+T^*_{j+1})/2-1}\underset{1 \leq u \neq v \leq p}{\max}\xi_{uv,1r}\big) \geq \|\Sigma^*_{j+1}-\Sigma^*_j\|_F/6\Big)\nonumber\\
&\quad  + \sum^{m^*}_{j=1} \Pb\left(\left\|\frac{1}{(T^*_{j+1}-T^*_{j})/2}\sum^{(T^*_j+T^*_{j+1})/2-1}_{r=T^*_j}\big(\Sigma^*_j- X_rX^\top_r\big)\right\|_F \geq \|\Sigma^*_{j+1}-\Sigma^*_j\|_F/6\right)
\label{control_sum_adaptive}.
\end{align}}}
Since $T\lambda_2\max(\alpha_{p,T},a_T)^{-\mu_2}/(\mathcal{I}_{\min}\eta_{\min}) \rightarrow 0$, $T\lambda_1\sqrt{p(p-1)} \max(\alpha_{p,T},a_T)^{-\mu_1}/\eta_{\min} \rightarrow 0$, the first term tends to zero as \(T \to \infty\). As for the second term, for $C>0$ finite, applying Lemma \ref{bound_gradient}, we deduce that for any $j$: $\Pb\left(\left\|\frac{1}{(T^*_{j+1}-T^*_{j})/2}\sum^{(T^*_j+T^*_{j+1})/2-1}_{r=T^*_j}\Big(\Sigma^*_{j} - X_rX^\top_r\Big)\right\|_F \geq  \eta_{\min}/6\right) \rightarrow 0$,
since $(\eta_{\min}\mathcal{I}^{1/2}_{\min})^{-1}p\sqrt{\log(pT)}\rightarrow 0$. So $\sum^{m^*}_{j=1}AC_{j,2} \rightarrow 0$. We now consider $\sum^{m^*}_{j=1}AC_{j,3}$. Applying the same reasoning to show the convergence of the second summation on the right-hand side of (\ref{control_sum_adaptive}), we get $\left\|\frac{1}{T^*_{j}-\widehat{T}_{j}}\sum^{T^*_j-1}_{r=\widehat{T}_j}\big(\Sigma^*_{j} - X_rX^\top_r\big)\right\|_F = O_p\left(p\sqrt{\frac{\log(pT)}{T\delta_T}}\right) = o_p(\eta_{\min})$,
when $T^*_j-\widehat{T}_j\geq T\delta_T$,
and
\begin{align*}
\lefteqn{\sum^{m^*}_{j=1}AC_{j,3} \leq \Pb(A^+_{T,j}\cap \{R_{Tj,3}\geq \frac{1}{3}R_{Tj,1}\})} \\
& \leq \sum^{m^*}_{j=1}\!\Pb\left(A^+_{T,j} \!\cap\! \left\{\|\frac{1}{T^*_{j}\!-\!\widehat{T}_{j}}\sum^{T^*_j-1}_{r=\widehat{T}_j}\!\big(\Sigma^*_{j} -X_rX^\top_r\big)\|_F\!\geq\! \frac{1}{3}\left\|\frac{1}{T^*_j-\widehat{T}_j}\overset{T^*_j-1}{\underset{r=\widehat{T}_j}{\sum}}\! \text{vec}(\Sigma^*_{j+1}\!-\!\Sigma^*_j)\right\|_2\right\}\right)\\
& \leq  \sum^{m^*}_{j=1} \Pb\left(A^+_{T,j} \cap \left\{\left\|\frac{1}{T^*_{j}-\widehat{T}_{j}}\sum^{T^*_j-1}_{r=\widehat{T}_j} \big(\Sigma^*_{j}- X_rX^\top_r\big)\right\|_F\geq \frac{1}{3}\eta_{\min}\right\}\right).
\end{align*}
Since $T\delta_T \leq T^*_j-\widehat{T}_j$, then under $(\sqrt{T\delta_T}\eta_{\min})^{-1}\,p\sqrt{ \log(pT)} \rightarrow 0$, we deduce $\sum^{m^*}_{j=1}AC_{j,3} \rightarrow 0$. Consequently, we proved $\sum^{m^*}_{j=1}\Pb(A_{T,j} \cap C_T) \rightarrow 0$.

\noindent\emph{Proof of (b).} We prove (b) by showing $\sum^{m^*}_{j=1}\Pb(A^+_{T,j} \cap C^c_T) \rightarrow 0$ and $\sum^{m^*}_{j=1}\Pb(A^-_{T,j} \cap C^c_T) \rightarrow 0$. As in the proof of (a), we simply show $\sum^{m^*}_{j=1}\Pb(A^+_{T,j} \cap C^c_T) \rightarrow 0$. To do so, we define:
\begin{align*}
D^{(l)}_T&\coloneqq \big\{\exists j \in \{1,\ldots,m^*\}, \widehat{T}_j \leq T^*_{j-1}\big\} \cap C^c_T,\;
D^{(m)}_T\coloneqq \big\{\forall j \in \{1,\ldots,m^*\}, T^*_{j-1}<\widehat{T}_j < T^*_{j+1}\big\} \cap C^c_T,\\
D^{(r)}_T&\coloneqq \big\{\exists j \in \{1,\ldots,m^*\}, \widehat{T}_j \geq T^*_{j+1}\big\} \cap C^c_T,
\end{align*}
where $C^c_T = \{\underset{1 \leq j \leq m^*}{\max}|\widehat{T}_j-T^*_j|\geq \mathcal{I}_{\min}/2\}$. Then, we have:
\begin{equation*}
\sum^{m^*}_{j=1}\Pb(A^+_{T,j} \cap C^c_T) = \sum^{m^*}_{j=1}\Big[\Pb(A^+_{T,j} \cap D^{(l)}_T)+\Pb(A^+_{T,j} \cap D^{(m)}_T) +\Pb(A^+_{T,j} \cap D^{(r)}_T)\Big].
\end{equation*}
We first bound $\sum^{m^*}_{j=1}\Pb(A^+_{T,j} \cap D^{(m)}_T)$. For any $j$:
\begin{align*}
\Pb(A^+_{T,j} \cap D^{(m)}_T)& \leq \Pb(A^+_{t,j} \cap \big\{\widehat{T}_{j+1}-T^*_j \geq \frac{1}{2}\mathcal{I}_{\min}\big\}\cap D^{(m)}_T) + \Pb(A^+_{t,j} \cap \big\{\widehat{T}_{j+1}-T^*_j < \frac{1}{2}\mathcal{I}_{\min}\big\}\cap D^{(m)}_T)\\
& \leq \Pb(A^+_{t,j} \cap \big\{\widehat{T}_{j+1}-T^*_j \geq \frac{1}{2}\mathcal{I}_{\min}\big\}\cap D^{(m)}_T) + \Pb(A^+_{t,j} \cap \big\{T^*_{j+1}-\widehat{T}_{j+1} \geq \frac{1}{2}\mathcal{I}_{\min}\big\}\cap D^{(m)}_T),
\end{align*}
since $0 \leq \widehat{T}_{j+1}-T^*_j \leq \mathcal{I}_{\min}/2$ implies $T^*_{j+1}-\widehat{T}_{j+1}=(T^*_{j+1}-T^*_j)-(\widehat{T}_{j+1}-T^*_j)\geq \mathcal{I}_{\min}-\mathcal{I}_{\min}/2 = \mathcal{I}_{\min}/2$. Moreover, since
\begin{equation*}
\resizebox{\linewidth}{!}{$\displaystyle\Big\{A^+_{t,j} \cap \big\{T^*_{j+1}-\widehat{T}_{j+1} \geq \frac{1}{2}\mathcal{I}_{\min}\big\}\cap D^{(m)}_T\Big\} \subset \overset{m^*-1}{\underset{k=j+1}{\cup}}\Big[ \big\{T^*_k-\widehat{T}_k\geq \mathcal{I}_{\min}/2\big\} \cap \big\{\widehat{T}_{k+1}-T^*_k \geq \mathcal{I}_{\min}/2\big\}\cap D^{(m)}_T\Big],$}
\end{equation*}
we deduce:
\begin{align}\label{control_Dm_adaptive}
\lefteqn{\sum^{m^*}_{j=1}\Pb(A^+_{T,j} \cap D^{(m)}_T)  \leq \sum^{m^*}_{j=1} \Pb(A^+_{t,j} \cap \big\{\widehat{T}_{j+1}-T^*_j \geq \frac{1}{2}\mathcal{I}_{\min}\big\}\cap D^{(m)}_T)  } \notag\\
&\quad  + \sum^{m^*}_{j=1} \sum^{m^*-1}_{k=j+1}  \Pb\Big(\big\{T^*_k-\widehat{T}_k\geq \mathcal{I}_{\min}/2\big\} \cap \big\{\widehat{T}_{k+1}-T^*_k \geq \mathcal{I}_{\min}/2\big\}\cap D^{(m)}_T\Big).
\end{align}
Let us treat the first term. By Lemma \ref{optimality_cond_adaptive} with $t=\widehat{T}_j$ and $t=T^*_j$, we obtain:
\begin{equation*}
\frac{1}{T}\overset{T}{\underset{r=\widehat{T}_j}{\sum}}\Big(\text{vec}\big(\Theta^*_r-X_rX^\top_r\big)+(I_p \otimes I_p)\text{vec}\big(\widehat{\Theta}_r-\Theta^*_r\big) \Big) + \lambda_1\text{vec}\left(\overset{T}{\underset{r=\widehat{T}_j}{\sum}}\underset{u \neq v}{\sum}\xi_{uv,1r}\widehat{E}_{uv,1r}\right)= \lambda_2 \xi_{2\widehat{T}_j}\text{vec}\left(\frac{\widehat{\Gamma}_{\widehat{T}_j}}{\|\widehat{\Gamma}_{\widehat{T}_j}\|_F}\right), 
\end{equation*}
and $\left\|\frac{1}{T}\overset{T}{\underset{r=T^*_j}{\sum}}\Big(\text{vec}\big(\Theta^*_r-X_rX^\top_r\big)+(I_p \otimes I_p)\text{vec}\big(\widehat{\Theta}_r-\Theta^*_r\big) \Big)+\lambda_1\text{vec}\left(\overset{T}{\underset{r=T^*_j}{\sum}}\underset{u \neq v}{\sum}\xi_{uv,1r}\widehat{E}_{uv,1r}\right)\right\|_2 \leq \lambda_2 \xi_{2T^*_j}$.
We deduce
{\footnotesize{\begin{align*}
\|\widehat{\Sigma}_{j+1}-\Sigma^*_j\|_F - \left\|\frac{1}{T^*_j-\widehat{T}_j}\overset{T^*_j-1}{\underset{r=\widehat{T}_j}{\sum}}\big(\Sigma^*_{j} -X_rX^\top_r\big)\right\|_F 
& \leq  \frac{1}{T^*_j-\widehat{T}_j} \left\|\overset{T^*_j-1}{\underset{r=\widehat{T}_j}{\sum}}\Big(\text{vec}\big(\Sigma^*_j-X_rX^\top_r\big)+(I_p \otimes I_p)\text{vec}\big(\widehat{\Sigma}_{j+1}-\Sigma^*_j\big)\Big)\right\|_2\\
& \leq \frac{2\lambda_2 T}{T^*_j-\widehat{T}_j}(\xi_{2T^*_j}+\xi_{2\widehat{T}_j})+\lambda_1T \sqrt{p(p-1)}\big(\sup_{\widehat{T}_j \leq r \leq T^*_j-1}\underset{1 \leq u \neq v \leq p}{\max}\xi_{uv,1r}\big).
\end{align*}}}
As a consequence:
{\small{\begin{align}
\lefteqn{\|\widehat{\Sigma}_{j+1}-\Sigma^*_j\|_F \leq \frac{2\lambda_2 T}{T^*_j-\widehat{T}_j}(\xi_{2T^*_j}+\xi_{2\widehat{T}_j})}\nonumber\\
&\quad  +\lambda_1T \sqrt{p(p-1)}\big(\sup_{\widehat{T}_j \leq r \leq T^*_j-1}\underset{1 \leq u \neq v \leq p}{\max}\xi_{uv,1r}\big) + \left\|\frac{1}{T^*_j-\widehat{T}_j}\overset{T^*_j-1}{\underset{r=\widehat{T}_j}{\sum}}\big(\Sigma^*_{j} -X_rX^\top_r\big)\right\|_F.\label{bound_param_1_adaptive}
\end{align}}}
In the same vein, applying Lemma \ref{optimality_cond_adaptive} with $t = \widehat{T}_{j+1}$ and $t=T^*_j$, we obtain:
\begin{align}
\lefteqn{\|\widehat{\Sigma}_{j+1}-\Sigma^*_{j+1}\|_F  \leq \frac{2\lambda_2 T}{\widehat{T}_{j+1}-T^*_j}(\xi_{2\widehat{T}_{j+1}}+\xi_{2T^*_j}) }\label{bound_param_2_adaptive}\\
&\quad  +\lambda_1T \sqrt{p(p-1)}\big(\sup_{T^*_j \leq r \leq \widehat{T}_{j+1}-1}\underset{1 \leq u \neq v \leq p}{\max}\xi_{uv,1r}\big) + \left\|\frac{1}{\widehat{T}_{j+1}-T^*_j}\overset{\widehat{T}_{j+1}-1}{\underset{r=T^*_j}{\sum}}\big(\Sigma^*_{j} - X_rX^\top_r\big)\right\|_F.\nonumber
\end{align}
Let the event:
\begin{align}
\lefteqn{E_{T,j} \coloneqq \Bigg\{\|\Sigma^*_{j+1}-\Sigma^*_j\|_F \leq 2\lambda_2 T\Big[\frac{\xi_{2T^*_j}+\xi_{2\widehat{T}_j}}{T^*_j-\widehat{T}_j} + \frac{\xi_{2\widehat{T}_{j+1}}+\xi_{2T^*_j}}{\widehat{T}_{j+1}-T^*_j} \Big]}\label{E_event_adaptive}\\
&\quad  + \lambda_1T \sqrt{p(p-1)}\Big[\big(\sup_{\widehat{T}_j \leq r \leq T^*_j-1}\underset{1 \leq u \neq v \leq p}{\max}\xi_{uv,1r}\big) + \big(\sup_{T^*_j \leq r \leq \widehat{T}_{j+1}-1}\underset{1 \leq u \neq v \leq p}{\max}\xi_{uv,1r}\big)\Big]\nonumber\\
&\quad  \left\|\frac{1}{T^*_j-\widehat{T}_j}\overset{T^*_j-1}{\underset{r=\widehat{T}_j}{\sum}}\big(\Sigma^*_{j} -X_rX^\top_r\big)\right\|_F + \left\|\frac{1}{\widehat{T}_{j+1}-T^*_j}\overset{\widehat{T}_{j+1}-1}{\underset{r=T^*_j}{\sum}}\big(\Sigma^*_{j} -X_rX^\top_r\big)\right\|_F\Bigg\}. \nonumber
\end{align}
Therefore, by the triangle inequality, (\ref{bound_param_1_adaptive}) and (\ref{bound_param_2_adaptive}) imply that the event $E_{T,j}$ holds with probability one. Hence:
{\footnotesize{\begin{align}
\lefteqn{\sum^{m^*}_{j=1} \Pb(A^+_{t,j} \cap \big\{\widehat{T}_{j+1}-T^*_j \geq \frac{1}{2}\mathcal{I}_{\min}\big\}\cap D^{(m)}_T)
 \leq \sum^{m^*}_{j=1} \Pb(E_{T,j}\cap \big\{T^*_j - \widehat{T}_j> T \delta_T\big\} \cap  \big\{\widehat{T}_{j+1}-T^*_j \geq \frac{1}{2}\mathcal{I}_{\min}\big\})}\nonumber\\
& \leq \resizebox{0.98\linewidth}{!}{$\displaystyle\sum^{m^*}_{j=1} \Pb\Big(2\lambda_2\frac{\xi_{2T^*_j}+\xi_{2\widehat{T}_j}}{\delta_T} + 4\lambda_2T\frac{\xi_{2\widehat{T}_{j+1}}+\xi_{2T^*_j}}{\Ic_{\min}} + \lambda_1T \sqrt{p(p-1)}\Big[\big(\sup_{\widehat{T}_j \leq r \leq T^*_j-1}\underset{u \neq v }{\max}\,\xi_{uv,1r}\big) + \big(\sup_{T^*_j \leq r \leq \widehat{T}_{j+1}-1}\underset{ u \neq v}{\max}\,\xi_{uv,1r}\big)\Big]\geq \|\Sigma^*_{j+1}-\Sigma^*_j\|_F/3\Big) $}\nonumber\\
&\quad  + \sum^{m^*}_{j=1} \Pb\left(\Big\{\|\frac{1}{T^*_j-\widehat{T}_j}\overset{T^*_j-1}{\underset{r=\widehat{T}_j}{\sum}}\big(\Sigma^*_{j} -X_rX^\top_r\big)\|_F \geq \|\Sigma^*_{j+1}-\Sigma^*_j\|_F/3\Big\}\cap \big\{T^*_j-\widehat{T}_j > T \delta_T\big\}\right)\nonumber\\
&\quad  + \sum^{m^*}_{j=1} \Pb\left(\left\{\left\|\frac{1}{\widehat{T}_{j+1}-T^*_j}\overset{\widehat{T}_{j+1}-1}{\underset{r=T^*_j}{\sum}}\big(\Sigma^*_{j} -X_rX^\top_r\big)\right\|_F \geq \|\Sigma^*_{j+1}-\Sigma^*_j\|_F/3\right\}\cap \big\{\widehat{T}_{j+1}-T^*_j \geq \mathcal{I}_{\min}/2\big\}\right). \label{bound_A+_adaptive}
\end{align}}}
Under $\lambda_2/(\max(\alpha_{p,T},a_T)^{\mu_2}\eta_{\min}\delta_T)\rightarrow 0$, $\lambda_2  T/(\max(\alpha_{p,T},a_T)^{\mu_2}\mathcal{I}_{\min}\eta_{\min})\rightarrow 0$, and under the condition $\lambda_1T\sqrt{p(p-1)}/(\max(\alpha_{p,T},a_T)^{\mu_1}\eta_{\min}) \rightarrow 0$, the first term in (\ref{bound_A+_adaptive}) tends to zero. Moreover, 
$\left\|\frac{1}{T^*_j-\widehat{T}_j}\overset{T^*_j-1}{\underset{r=\widehat{T}_j}{\sum}}\big(\Sigma^*_{j} -X_rX^\top_r\big)\right\|_F = O_p\left(p\sqrt{\frac{\log(pT)}{T\delta_T}}\right) = o_p(\eta_{\min})$, $\left\|\frac{1}{\widehat{T}_{j+1}-T^*_j}\overset{\widehat{T}_{j+1}-1}{\underset{r=T^*_j}{\sum}}\big(\Sigma^*_{j} - X_rX^\top_r\big)\right\|_F = O_p\left(p\sqrt{\frac{\log(pT)}{\mathcal{I}_{\min}}}\right) = o_p(\eta_{\min})$.
In the same manner, we can show that the second term in (\ref{control_Dm_adaptive}) tends to zero. We now consider $\sum^{m^*}_{j=1}\Pb(A^+_{T,j} \cap D^{(l)}_T)$. The probability of the event $A^+_{T,j} \cap D^{(l)}_T$ is upper bounded by: $\Pb(D^{(l)}_T)\leq \sum^{m^*}_{j=1}2^{j-1}\Pb(\max(l \in \{1,\ldots,m^*\}:\widehat{T}_l \leq T^*_{l-1})=j)$.
Now $\max(l \in \{1,\ldots,m^*\}:\widehat{T}_l \leq T^*_{l-1})=j$ implies $\widehat{T}_j\leq T^*_{j-1}$ and $\widehat{T}_{l+1}>T^*_l$ for any $j \leq l \leq m^*$ and:
\begin{equation*}
\big\{\max(l \in \{1,\ldots,m^*\}:\widehat{T}_l \leq T^*_{l-1})=j\big\}\subset \overset{m^*-1}{\underset{k=j}{\cup}} \Big(\big\{T^*_k-\widehat{T}_k\geq \mathcal{I}_{\min}/2\big\}\cap\big\{\widehat{T}_{k+1}-T^*_k\geq \mathcal{I}_{\min}/2\big\}\Big).
\end{equation*}
Therefore:
\begin{align}
\lefteqn{\sum^{m^*}_{j=1}\Pb(A^+_{T,j} \cap D^{(l)}_T)}\label{control_Dl_adaptive}\\
  & \leq \resizebox{0.92\linewidth}{!}{$\displaystyle m^* \overset{m^*-1}{\underset{j=1}{\sum}}2^{j-1}\overset{m^*-1}{\underset{k=j}{\sum}} \Pb\Big(\big\{T^*_k-\widehat{T}_k\geq \mathcal{I}_{\min}/2\big\}\cap\big\{\widehat{T}_{k+1}-T^*_k\geq \mathcal{I}_{\min}/2\big\}\Big) + m^*2^{m^*-1}\Pb(T^*_{m^*}-\widehat{T}_{m^*}\geq \mathcal{I}_{\min}/2).$}\nonumber
\end{align}
First, we consider the second term of the right-hand side of (\ref{control_Dl_adaptive}). Let $j=m^*$ in (\ref{E_event_adaptive}), then $E_{T,m^*}$ holds with probability one. Therefore:
{\small{\begin{align*}
\lefteqn{m^*2^{m^*-1}\Pb(T^*_{m^*}-\widehat{T}_{m^*}\geq \mathcal{I}_{\min}/2)=m^*2^{m^*-1}\Pb(E_{T,m^*}\cap\big\{T^*_{m^*}-\widehat{T}_{m^*}\geq \mathcal{I}_{\min}/2\big\})}\\
  & \leq  m^*2^{m^*-1}\Pb\Bigg(\frac{2\lambda_2}{\delta_T}(\xi_{2\widehat{T}_{m^*}}+\xi_{2T^*_{m^*}}) + \frac{4\lambda_2 T}{\mathcal{I}_{\min}}(\xi_{2\widehat{T}_{m^*}} + \xi_{2T^*_{m^*}})\\
  &\quad  + \lambda_1T\sqrt{p(p-1)}\Big[\big(\sup_{\widehat{T}_{m^*} \leq r \leq T^*_{m^*}-1}\underset{1 \leq u \neq v \leq p}{\max}\xi_{uv,1r}\big) + \big(\sup_{T^*_{m^*} \leq r \leq \widehat{T}_{m^*+1}-1}\underset{1 \leq u \neq v \leq p}{\max}\xi_{uv,1r}\big)\Big]
  \geq \|\Sigma^*_{m^*+1}-\Sigma^*_{m^*}\|_F/3\Bigg) \\
  &\quad  + m^*2^{m^*-1} \Pb\left(\left\|\frac{1}{T^*_{m^*}\!-\!\widehat{T}_{m^*}}\overset{T^*_{m^*}-1}{\underset{r=\widehat{T}_{m^*}}{\sum}}\!\big(\Sigma^*_{m^*}- X_rX^\top_r\big)\right\|_F \!\geq\! \|\Sigma^*_{m^*+1}\!-\!\Sigma^*_{m^*}\|_F/3,T^*_{m^*}\!-\!\widehat{T}_{m^*} \!\geq\! \mathcal{I}_{\min}/2\right) \\
  &\quad  + m^*2^{m^*-1}\Pb\left(\left\|\frac{1}{T-T^*_{m^*}}\overset{T}{\underset{r=T^*_{m^*}}{\sum}}\big(\Sigma^*_{m^*} - X_rX^\top_r\big)\right\|_F \geq \|\Sigma^*_{m^*+1}-\Sigma^*_{m^*}\|_F/3\right).
\end{align*}}}
Since $m^*2^{m^*-1} = O(T\log(T))$, then $\log(m^*2^{m^*-1}) = O(\log(T^{1+\eps/2}))$. So under the conditions $(\sqrt{T\delta_T}\eta_{\min})^{-1}\,p\sqrt{ \log(pT)} \rightarrow 0, (\mathcal{I}^{1/2}_{\min}\eta_{\min})^{-1}\,p\sqrt{ \log(pT)} \rightarrow 0$, the right-hand side of the previous inequality converges to zero. As for the first term of (\ref{control_Dl_adaptive}), applying $j=k$ in (\ref{E_event_adaptive}):
\begin{align*}
\lefteqn{m^* \overset{m^*-1}{\underset{j=1}{\sum}}2^{j-1}\overset{m^*-1}{\underset{k=j}{\sum}} \Pb\Big(\big\{T^*_k-\widehat{T}_k\geq \mathcal{I}_{\min}/2\big\}\cap\big\{\widehat{T}_{k+1}-T^*_k\geq \mathcal{I}_{\min}/2\big\}\Big)}\\
& \leq  m^* 2^{m^*-1}\overset{m^*-1}{\underset{k=1}{\sum}} \Bigg\{\Pb\Bigg(\frac{2\lambda_2}{\delta_T}(\xi_{2\widehat{T}_{k}}+\xi_{2T^*_{k}}) + \frac{4\lambda_2 T}{\mathcal{I}_{\min}} (\xi_{2\widehat{T}_{k+1}}+\xi_{2T^*_{k}})\\
&\quad  + \lambda_1T\sqrt{p(p-1)}\Big[\big(\sup_{\widehat{T}_{k} \leq r \leq T^*_{k}-1}\underset{1 \leq u \neq v \leq p}{\max}\xi_{uv,1r}\big) + \big(\sup_{T^*_{k} \leq r \leq \widehat{T}_{k+1}-1}\underset{1 \leq u \neq v \leq p}{\max}\xi_{uv,1r}\big)\Big]\geq \|\Sigma^*_{k+1}-\Sigma^*_{k}\|_F/3\Bigg) \\
&\quad  + \Pb\left(\left\|\frac{1}{T^*_{k}-\widehat{T}_{k}}\overset{T^*_{k}-1}{\underset{r=\widehat{T}_{k}}{\sum}}\big(\Sigma^*_{k} -X_rX^\top_r\big)\right\|_F \geq \|\Sigma^*_{k+1}-\Sigma^*_{k}\|_F/3,T^*_{k}-\widehat{T}_{k} \geq \mathcal{I}_{\min}/2\right) \\
&\quad  + \Pb\left(\left\|\frac{1}{\widehat{T}_{k+1}-T^*_{k}}\overset{\widehat{T}_{k+1}-1}{\underset{r=T^*_{k}}{\sum}}\big(\Sigma^*_{k} -X_rX^\top_r\big)\right\|_F \geq \|\Sigma^*_{k+1}-\Sigma^*_{k}\|_F/3,\widehat{T}_{k+1}-T^*_{k}\geq \mathcal{I}_{\min}/2\right)\Bigg\}.
\end{align*}
The right-hand side of the last inequality converges to zero under the same conditions. Finally, we can prove that $\sum^{m^*}_{j=1}\Pb(A^+_{T,j} \cap D^{(r)}_T) \rightarrow 0$.

\vspace*{0.3cm}

\noindent\textbf{\emph{Proof of point (ii).}}\\
\noindent By point (i), for any $j=1,\ldots,m^*$, $|\widehat{T}_j-T^*_j|=O_p(T\delta_T)$, which is $|\widehat{T}_j-T^*_j| = o_p(\mathcal{I}_{\min})$. Hence, $(T^*_{j-1}+T^*_j)/2 < \widehat{T}_j < T^*_j$ or $T^*_j \leq \widehat{T}_j < (T^*_j + T^*_{j+1})/2$ is satisfied for any $j$. Set $l = 1,\ldots,m^*$ and assume $(T^*_{l-1}+T^*_l)/2 < \widehat{T}_l < T^*_l$ and consider two cases: (ii-a) $(T^*_{l}+T^*_{l+1})/2 < \widehat{T}_{l+1} < T^*_{l+1}$ and (ii-b) $T^*_{l+1} \leq \widehat{T}_{l+1}$. In case (ii-a), by Lemma \ref{optimality_cond} with change points $t = \widehat{T}_l$ and $t = \widehat{T}_{l+1}$:
\begin{align*}
2\lambda_2 (\xi_{2\widehat{T}_l}+\xi_{2\widehat{T}_{l+1}}) \geq  \left\|\frac{1}{T}\overset{\widehat{T}_{l+1}-1}{\underset{r = \widehat{T}_l}{\sum}} \text{vec}\big(\widehat{\Theta}_r-X_rX^\top_r\big)+\lambda_1\text{vec}\left(\overset{\widehat{T}_{l+1}-1}{\underset{r=\widehat{T}_l}{\sum}}\underset{u\neq v}{\sum}\xi_{uv,1r}\widehat{E}_{uv,1r}\right)\right\|_2.
\end{align*}
Therefore, we have
{\footnotesize{\begin{align*}
\lefteqn{2\lambda_2 (\xi_{2\widehat{T}_l}+\xi_{2\widehat{T}_{l+1}})+\lambda_1\sqrt{p(p-1)}\big(\sup_{\widehat{T}_l \leq r \leq \widehat{T}_{l+1}-1}\underset{1 \leq u \neq v \leq p}{\max}\xi_{uv,1r}\big) (\widehat{T}_{l+1}-\widehat{T_l})}\\
& \geq \left\|\frac{1}{T}\overset{\widehat{T}_{l+1}-1}{\underset{r = T^*_l}{\sum}} \Big(\text{vec}\big(\Sigma^*_{l+1}-X_rX^\top_r\big) + (I_p\otimes I_p)\text{vec}\big(\widehat{\Sigma}_{l+1}-\Sigma^*_{l+1}\big)\Big)\right\|_2 -\left\|\frac{1}{T}\overset{T^*_l-1}{\underset{r = \widehat{T}_l}{\sum}} \Big(\text{vec}\big(\Sigma^*_l-X_rX^\top_r\big) + (I_p\otimes I_p)\text{vec}\big(\widehat{\Sigma}_{l+1}-\Sigma^*_l\big)\Big)\right\|_2\\
& \geq \frac{\widehat{T}_{l+1}-T^*_l}{T}\Bigg\{\|\widehat{\Sigma}_{l+1}-\Sigma^*_{l+1})\|_F - \left\|\frac{1}{\widehat{T}_{l+1}-T^*_l}\overset{\widehat{T}_{l+1}-1}{\underset{r = T^*_{l}}{\sum}}\big(\Sigma^*_{l+1}-X_rX^\top_r\big)\right\|_F\Bigg\} \\
&\quad  - \frac{T^*_l-\widehat{T}_l}{T}\left\|\frac{1}{T^*_l-\widehat{T}_l}\overset{T^*_l-1}{\underset{r = \widehat{T}_l}{\sum}} \Big(\text{vec}\big(\Sigma^*_l-X_rX^\top_r\big) + (I_p\otimes I_p)\text{vec}\big(\widehat{\Sigma}_{l+1}-\Sigma^*_l\big)\Big)\right\|_2.
\end{align*}}}
Therefore, using part (i) of Theorem \ref{theorem_consistency_adaptive}, and the bound on $\|\widehat{\Sigma}_{l+1}-\Sigma^*_l\|_F$,  we obtain
\begin{align*}
\lefteqn{2\lambda_2(\xi_{2\widehat{T}_l}+\xi_{2\widehat{T}_{l+1}})+\lambda_1\sqrt{p(p-1)}\big(\sup_{\widehat{T}_l \leq r \leq \widehat{T}_{l+1}-1}\underset{1 \leq u \neq v \leq p}{\max}\xi_{uv,1r}\big) (\widehat{T}_{l+1}-\widehat{T_l})}\\
& \geq\frac{\widehat{T}_{l+1}-T^*_l}{T} \Big\{\|\widehat{\Sigma}_{l+1}-\Sigma^*_{l+1}\|_F-\left\|\frac{1}{\widehat{T}_{l+1}-T^*_l}\overset{\widehat{T}_{l+1}-1}{\underset{r = T^*_{l}}{\sum}}\big(\Sigma^*_{l+1}-X_rX^\top_r\big)\right\|_F\Big\}- O_p\left(\frac{T^*_l-\widehat{T}_l}{T}\right).
\end{align*}
We deduce
\begin{align*}
\lefteqn{2\lambda_2(\xi_{2\widehat{T}_l}+\xi_{2\widehat{T}_{l+1}})+\lambda_1\sqrt{p(p-1)}\big(\sup_{\widehat{T}_l \leq r \leq \widehat{T}_{l+1}-1}\underset{1 \leq u \neq v \leq p}{\max}\xi_{uv,1r}\big) (\widehat{T}_{l+1}-\widehat{T_l})}\\
& \geq \frac{\widehat{T}_{l+1}-T^*_l}{T} \Big\{\|\widehat{\Sigma}_{l+1}-\Sigma^*_{l+1}\|_F-O_p\left( p\sqrt{\frac{\log(pT)}{I^*_{l+1}}}\right)\Big\}- O_p\left(\frac{T^*_l-\widehat{T}_l}{T}\right).
\end{align*}
Therefore, we get
\begin{equation}\label{bound_prob_regime_adaptive}
\|\widehat{\Sigma}_{l+1}-\Sigma^*_{l+1}\|_F=O_p\left(\frac{\lambda_2 \, T}{I^*_{l+1}}\max(\alpha_{p,T},a_T)^{-\mu_2}+\lambda_1Tp(1+\frac{T\delta_T}{I_{l+1}})\max(\alpha_{p,T},a_T)^{-\mu_1}+\frac{T\delta_T}{I^*_{l+1}} + p\sqrt{ \frac{\log(pT)}{I^*_{l+1}}}\right).
\end{equation}
In case (ii-b), by Lemma \ref{optimality_cond_adaptive}, with change points $t = \widehat{T}_l$ and $t = \widehat{T}_{l+1}$, we have
{\footnotesize{\begin{align*}
\lefteqn{2\lambda_2 (\xi_{2\widehat{T}_l}+\xi_{2\widehat{T}_{l+1}})+\lambda_1\sqrt{p(p-1)}\big(\sup_{\widehat{T}_l \leq r \leq \widehat{T}_{l+1}-1}\underset{1 \leq u \neq v \leq p}{\max}\xi_{uv,1r}\big) (\widehat{T}_{l+1}-\widehat{T_l})}\\
& \geq \left\|\frac{1}{T}\overset{T^*_{l+1}-1}{\underset{r = T^*_l}{\sum}} \Big(\text{vec}\big(\Sigma^*_{l+1}-X_rX^\top_r\big) + (I_p\otimes I_p)\text{vec}\big(\widehat{\Sigma}_{l+1}-\Sigma^*_{l+1}\big)\Big)\right\|_2  - \left\|\frac{1}{T}\overset{T^*_l-1}{\underset{r = \widehat{T}_l}{\sum}} \Big(\text{vec}\big(\Sigma^*_l-X_rX^\top_r\big) + (I_p\otimes I_p)\text{vec}\big(\widehat{\Sigma}_{l+1}-\Sigma^*_l\big)\Big)\right\|_2\\
&\quad  - \left\|\frac{1}{T}\overset{\widehat{T}_{l+1}-1}{\underset{r = T^*_{l+1}}{\sum}} \Big(\text{vec}\big(\Sigma^*_{l+2}-X_rX^\top_r\big) + (I_p\otimes I_p)\text{vec}\big(\widehat{\Sigma}_{l+1}-\Sigma^*_{l+2}\big)\Big)\right\|_2.
\end{align*}}}
We deduce
\begin{align*}
\lefteqn{2\lambda_2 (\xi_{2\widehat{T}_l}+\xi_{2\widehat{T}_{l+1}})+\lambda_1\sqrt{p(p-1)}\big(\sup_{\widehat{T}_l \leq r \leq \widehat{T}_{l+1}-1}\underset{1 \leq u \neq v \leq p}{\max}\xi_{uv,1r}\big) (\widehat{T}_{l+1}-\widehat{T_l})}\\
  & \geq \frac{I^*_{l+1}}{T} \left\{\|\widehat{\Sigma}_{l+1}-\Sigma^*_{l+1}\|_F-O_p\left(p \sqrt{ \frac{\log(pT)}{I^*_{l+1}}}\right)\right\} - O_p\left(\frac{T^*_l-\widehat{T}_l}{T}\right)-O_p\left(\frac{\widehat{T}_{l+1}-T^*_{l+1}}{T}\right).
\end{align*}
Hence, (\ref{bound_prob_regime_adaptive}) holds. Using similar arguments, we can show that the latter is satisfied when $T^*_l \leq \widehat{T}_l < (T^*_l+T^*_{l+1})/2$.

\subsection{Proof of Theorem \ref{theorem_date_recov_adaptive}}

Using the result of Theorem \ref{theorem_consistency_adaptive},  we aim to show that:
\begin{equation}\label{proba_bound_adaptive}
\Pb\big(\{h(\widehat{\Tc}_{\widehat{m}},\Tc^*_{m^*})> T\delta_T\}\cap \{m^* < \widehat{m} \leq m_{\max}\}\big) \rightarrow 0 \;\; \text{as} \;\; T \rightarrow \infty.
\end{equation}
To do so, we define
$L_{m,k,1} = \big\{\forall 1 \leq l \leq m, |\widehat{T}_l - T^*_k|>T\delta_T \; \text{and} \; \widehat{T}_l < T^*_k\big\}$,
$L_{m,k,2} = \big\{\forall 1 \leq l \leq m, |\widehat{T}_l - T^*_k|>T\delta_T \; \text{and} \; \widehat{T}_l > T^*_k\big\}$, $L_{m,k,3} = \big\{\exists 1 \leq l \leq m-1, |\widehat{T}_l - T^*_k|>T\delta_T, |\widehat{T}_{l+1} - T^*_k|>T\delta_T \; \text{and} \; \widehat{T}_l < T^*_k < \widehat{T}_{l+1}\big\}$.
The probability (\ref{proba_bound_adaptive}) can be bounded as:
\begin{align*}
\lefteqn{\Pb\big(\{h(\widehat{\Tc}_{\widehat{m}},\Tc^*_{m^*})> T\delta_T\}\cap \{m^* < \widehat{m} \leq m_{\max}\}\big) \leq \overset{m_{\max}}{\underset{m=m^*+1}{\sum}} \Pb\big(h(\widehat{\Tc}_{\widehat{m}},\Tc^*_{m^*})> T\delta_T\big)}\\
  & \leq \resizebox{0.92\linewidth}{!}{$\displaystyle\overset{m_{\max}}{\underset{m=m^*+1}{\sum}} \overset{m^*}{\underset{k=1}{\sum}}\Pb\big(\forall l \in \{1,\ldots,m\}, |\widehat{T}_l-T^*_k|>T\delta_T\big) = \overset{m_{\max}}{\underset{m=m^*+1}{\sum}} \overset{m^*}{\underset{k=1}{\sum}} \Big[\Pb\big(L_{m,k,1}\big)+\Pb\big(L_{m,k,2}\big)+\Pb\big(L_{m,k,3}\big)\Big].$}
\end{align*}
We first focus on $\sum^{m_{\max}}_{m=m^*+1} \sum^{m^*}_{k=1}\Pb\big(L_{m,k,1}\big)$, which can be expressed as:
\begin{equation*}
\Pb\big(L_{m,k,1}\big) = \Pb\big(L_{m,k,1}\cap\{\widehat{T}_m>T^*_{k-1}\}\big)+\Pb\big(L_{m,k,1}\cap\{\widehat{T}_m\leq T^*_{k-1}\}\big).
\end{equation*}
By Lemma \ref{optimality_cond} with change points $t = \widehat{T}_m$ and $t = T^*_k$, given the case $T^*_k\geq \widehat{T}_m > T^*_{k-1}$:
\begin{align*}
\frac{1}{T}\overset{T}{\underset{r=\widehat{T}_m}{\sum}}\text{vec}\big(\widehat{\Theta}_r-X_rX^\top_r\big) + \lambda_1\text{vec}\left(\overset{T}{\underset{r=\widehat{T}_m}{\sum}}\underset{u \neq v}{\sum}\xi_{uv,1r}\widehat{E}_{uv,1r}\right)+\lambda_2\text{vec}\left(\xi_{2\widehat{T}_m}\frac{\widehat{\Gamma}_{\widehat{T}_m}}{\|\widehat{\Gamma}_{\widehat{T}_m}\|_F}\right)= \mathbf{0}_{p^2 \times 1},
\end{align*}
and $\left\|\frac{1}{T}\overset{T}{\underset{r=T^*_k}{\sum}}\text{vec}\big(\widehat{\Theta}_r)-X_rX^\top_r\big) + \lambda_1\text{vec}\left(\overset{T}{\underset{r=T^*_k}{\sum}}\underset{u \neq v}{\sum}\xi_{uv,1r}\widehat{E}_{uv,1r}\right)\right\|_2 \leq  \lambda_2\xi_{2T^*_k}$.
Therefore, taking the differences, we get:
\begin{align*}
\lefteqn{2 \lambda_2(\xi_{2T^*_k}+\xi_{2\widehat{T}_m})+\lambda_1\sqrt{p(p-1)}\big(\underset{\widehat{T}_m \leq r \leq T^*_k}{\sup}\underset{1 \leq u \neq v \leq p}{\max}\xi_{uv,1r}\big)(T^*_k-\widehat{T}_m) }\\
& \geq \left\|\frac{1}{T}\overset{T^*_k-1}{\underset{r=\widehat{T}_m}{\sum}}(I_p\otimes I_p)\text{vec}(\widehat{\Sigma}_{m+1} - \Sigma^*_{k+1})+
\frac{1}{T}\overset{T^*_k-1}{\underset{r=\widehat{T}_m}{\sum}}(I_p\otimes I_p)\text{vec}(\Sigma^*_{k+1} - \Sigma^*_{k})+\frac{1}{T}\overset{T^*_k-1}{\underset{r=\widehat{T}_m}{\sum}}\text{vec}(\Sigma^*_k - X_rX^\top_r)\right\|_2.
\end{align*}
Therefore, the event $\Bc_T$ defined as
\begin{align*}
\lefteqn{\Bc_T \coloneqq  \Big\{\|\Sigma^*_{k+1}-\Sigma^*_k\|_F \leq \frac{2\lambda_2 T}{T^*_k-\widehat{T}_m}(\xi_{2T^*_k}+\xi_{2\widehat{T}_m})+\lambda_1\sqrt{p(p-1)}\big(\underset{\widehat{T}_m \leq r \leq T^*_k}{\sup}\underset{1 \leq u \neq v \leq p}{\max}\xi_{uv,1r}\big)}\\
  &\quad +\left\|\frac{1}{T^*_k-\widehat{T}_m}\overset{T^*_k-1}{\underset{r=\widehat{T}_m}{\sum}}(I_p\otimes I_p)\text{vec}(\widehat{\Sigma}_{m+1} - \Sigma^*_{k+1})\right\|_2+\left\|\frac{1}{T^*_k-\widehat{T}_m}\overset{T^*_k-1}{\underset{r=\widehat{T}_m}{\sum}}\text{vec}\big(\Sigma^*_k-X_rX^\top_r\big) \right\|_2\Big] \Big\},
\end{align*}
holds with probability one. Hence, we deduce
\begin{align*}
\resizebox{\linewidth}{!}{$\displaystyle\overset{m_{\max}}{\underset{m=m^*+1}{\sum}} \overset{m^*}{\underset{k=1}{\sum}} \Pb\big(L_{m,k,1} \cap \{\widehat{T}_m>T^*_{k-1}\}\big) = \overset{m_{\max}}{\underset{m=m^*+1}{\sum}} \overset{m^*}{\underset{k=1}{\sum}} \Pb\big(\Bc_T \cap L_{m,k,1} \cap \{\widehat{T}_m>T^*_{k-1}\}\big) \leq M_{1,1}+M_{1,2}+M_{1,3},$}
\end{align*}
with
\begin{align*}
M_{1,1} \coloneqq \overset{m_{\max}}{\underset{m=m^*+1}{\sum}} \overset{m^*}{\underset{k=1}{\sum}}\Pb\left(\|\Sigma^*_{k+1}-\Sigma^*_k\|_F/3\leq 2\lambda_2 \delta^{-1}_T(\xi_{2T^*_k}+\xi_{2\widehat{T}_m})+\lambda_1\sqrt{p(p-1)}\big(\underset{\widehat{T}_m \leq r \leq T^*_k}{\sup}\underset{1 \leq u \neq v \leq p}{\max}\xi_{uv,1r}\right)
\big),
\end{align*}
$$
 M_{1,2} \!\coloneqq\! \overset{m_{\max}}{\underset{m=m^*+1}{\sum}} \overset{m^*}{\underset{k=1}{\sum}} \Pb\left(T^*_k\!-\!\widehat{T}_m\!>\!T\delta_T,\|\Sigma^*_{k+1}\!-\!\Sigma^*_k\|_F/3 \!\leq\! \left\|\frac{1}{T^*_k-\widehat{T}_m}\overset{T^*_k-1}{\underset{r=\widehat{T}_m}{\sum}}\!(I_p\otimes I_p)\text{vec}(\widehat{\Sigma}_{m+1} \!-\! \Sigma^*_{k+1})\right\|_2\right),
$$
$$
M_{1,3} \!\coloneqq\! \overset{m_{\max}}{\underset{m=m^*+1}{\sum}} \overset{m^*}{\underset{k=1}{\sum}} \Pb\left(T^*_k\!-\!\widehat{T}_m\!>\!T\delta_T,\|\Sigma^*_{k+1}\!-\!\Sigma^*_k\|_F/3 \!\leq\! \left\|\frac{1}{T^*_k-\widehat{T}_m}\overset{T^*_k-1}{\underset{r=\widehat{T}_m}{\sum}}\!\text{vec}\big(\Sigma^*_k-X_rX^\top_r\big) \right\|_2\right).
$$
In the same vein as in the analysis of (\ref{bound_A+_adaptive}), we can show that $M_{1,1}, M_{1,3} \rightarrow 0$ as $T \rightarrow \infty$. $M_{1,2}$ requires more arguments. By Lemma (\ref{optimality_cond_adaptive}), with change points $t = T^*_k$ and $t = T^*_{k+1}$:
\begin{equation*}
\left\|\frac{1}{T}\overset{T}{\underset{r=T^*_k}{\sum}}\text{vec}\big(\widehat{\Sigma}_{m+1}-X_rX^\top_r) +\lambda_1\text{vec}\left(\overset{T}{\underset{r=T^*_k}{\sum}}\underset{u \neq v}{\sum}\xi_{uv,1r}\widehat{E}_{uv,1r}\right)\right\|_2 \leq  \lambda_2\xi_{2T^*_k}, 
\end{equation*}
and $\left\|\frac{1}{T}\overset{T}{\underset{r=T^*_{k+1}}{\sum}}\text{vec}\big(\widehat{\Sigma}_{m+1}-X_rX^\top_r\big) +\lambda_1\text{vec}\left(\overset{T}{\underset{r=T^*_{k+1}}{\sum}}\underset{u \neq v}{\sum}\xi_{uv,1r}\widehat{E}_{uv,1r}\right)\right\|_2 \leq  \lambda_2 \xi_{2T^*_{k+1}}$.
Therefore
\begin{align*}
\lefteqn{2\lambda_2(\xi_{2T^*_{k}}+\xi_{2T^*_{k+1}})+\lambda_1\sqrt{p(p-1)}\big(\underset{T^*_k \leq r \leq T^*_{k+1}-1}{\sup}\underset{1 \leq u \neq v \leq p}{\max}\xi_{uv,1r}\big)(T^*_{k+1}-T^*_k)}\\
& \geq \left\|\frac{1}{T}\overset{T^*_{k+1}-1}{\underset{r=T^*_{k}}{\sum}}(I_p\otimes I_p)\text{vec}(\widehat{\Sigma}_{m+1} - \Sigma^*_{k+1}) +\frac{1}{T}\overset{T^*_{k+1}-1}{\underset{r=T^*_k}{\sum}}\text{vec}\big(\Sigma^*_{k+1}-X_rX^\top_r\big)\right\|_2,
\end{align*}
which implies
\begin{align*}
\resizebox{1.1\linewidth}{!}{$\|\widehat{\Sigma}_{m+1}-\Sigma^*_{k+1}\|_F \leq \left\|\frac{1}{T^*_{k+1}-T^*_{k}}\overset{T^*_{k+1}-1}{\underset{r=T^*_k}{\sum}}\text{vec}\big(\Sigma^*_{k+1}-X_rX^\top_r\big)\right\|_2+ \frac{2\lambda_2 T(\xi_{2T^*_{k}}+\xi_{2T^*_{k+1}})}{T^*_{k+1}-T^*_{k}}+\lambda_1T\big(\underset{T^*_k \leq r \leq T^*_{k+1}-1}{\sup}\underset{1 \leq u \neq v \leq p}{\max}\xi_{uv,1r}\big)\sqrt{p(p-1)}$}.
\end{align*}
We deduce
{\small{\begin{align}
\lefteqn{M_{1,2} \leq \overset{m_{\max}}{\underset{m=m^*+1}{\sum}} \overset{m^*}{\underset{k=1}{\sum}} \Pb\big(\|\Sigma^*_{k+1}-\Sigma^*_k\|_F/3 \leq \|\widehat{\Sigma}_{m+1}-\Sigma^*_{k+1}\|_F\big) } \label{proba_M2_adaptive}\\
  & \leq\overset{m_{\max}}{\underset{m=m^*+1}{\sum}} \overset{m^*}{\underset{k=1}{\sum}}\Big[ \Pb\left( \|\Sigma^*_{k+1}-\Sigma^*_k\|_F/6 \leq \frac{2\lambda_2 T}{\mathcal{I}_{\min}}(\xi_{2T^*_{k}}+\xi_{2T^*_{k+1}})+\lambda_1T\big(\underset{T^*_k \leq r \leq T^*_{k+1}-1}{\sup}\underset{1 \leq u \neq v \leq p}{\max}\xi_{uv,1r}\big)\sqrt{p(p-1)}\right) \nonumber\\
&\quad  + \Pb\left(\|\Sigma^*_{k+1}\!-\!\Sigma^*_k\|_F/6 \!\leq\! \left\|\frac{1}{T^*_{k+1}\!-\!T^*_{k}}\overset{T^*_{k+1}-1}{\underset{r=T^*_k}{\sum}}\!\text{vec}\Big(\Sigma^*_{k+1}-X_rX^\top_r\Big)\right\|_2 \right)\Big].\nonumber
\end{align}}}
The first term in the second inequality of (\ref{proba_M2_adaptive}) tends to zero under $\lambda_2 T\max(\alpha_{p,T},a_T)^{-\mu_2}/(\mathcal{I}_{\min}\eta_{\min}) \rightarrow 0$, $\lambda_1T\max(\alpha_{p,T},a_T)^{-\mu_1}/\eta_{\min} \to 0$. Moreover, under the condition $(\eta_{\min}\mathcal{I}^{1/2}_{\min})^{-1}\,\sqrt{ T\log(T)}\rightarrow 0$, the second term tends to zero. Therefore, we conclude $\sum^{m_{\max}}_{m=m^*+1} \sum^{m^*}_{k=1} \Pb\big(L_{m,k,1} \cap \{\widehat{T}_m>T^*_{k-1}\}\big) \rightarrow 0$ as $T \rightarrow \infty$. Based on similar arguments, we can show $\sum^{m_{\max}}_{m=m^*+1} \sum^{m^*}_{k=1} \Pb\big(L_{m,k,1} \cap \{\widehat{T}_m \leq T^*_{k-1}\}\big) \rightarrow 0$. As $T \rightarrow \infty$, $\sum^{m_{\max}}_{m=m^*+1} \sum^{m^*}_{k=1}\Pb\big(L_{m,k,1}\big) \rightarrow 0$.
Similarly, we can show $\sum^{m_{\max}}_{m=m^*+1} \sum^{m^*}_{k=1}\Pb\big(L_{m,k,2}\big) \rightarrow 0$ as $T \rightarrow \infty$. \\
We now consider $\sum^{m_{\max}}_{m=m^*+1} \sum^{m^*}_{k=1}\Pb\big(L_{m,k,3}\big)$. Define
\begin{equation*}
\begin{array}{llll}
\resizebox{\linewidth}{!}{$\displaystyle L^{(1)}_{m,k,3} \coloneqq L_{m,k,3}\cap \{T^*_{k-1}<\widehat{T}_l < \widehat{T}_{l+1}<T^*_{k+1}\}, L^{(2)}_{m,k,3} \coloneqq L_{m,k,3}\cap \{T^*_{k-1}<\widehat{T}_l <T^*_{k+1}, \widehat{T}_{l+1}\geq T^*_{k+1}\},$} &&\\
\resizebox{\linewidth}{!}{$\displaystyle L^{(3)}_{m,k,3} \coloneqq L_{m,k,3}\cap \{\widehat{T}_l \leq T^*_{k-1},T^*_{k-1} <\widehat{T}_{l+1}<T^*_{k+1}\}, L^{(4)}_{m,k,3} \coloneqq L_{m,k,3}\cap \{\widehat{T}_l \leq T^*_{k-1},T^*_{k+1} <\widehat{T}_{l+1}\}.$} &&
\end{array}
\end{equation*}
First, we consider $L^{(1)}_{m,k,3}$. By Lemma (\ref{optimality_cond_adaptive}), for the change points $t = T^*_k$ and $t = \widehat{T}_l$, we obtain
\begin{align}
\lefteqn{2\lambda_2(\xi_{2\widehat{T}_l}+\xi_{2T^*_k}) + \lambda_1 \big(\underset{\widehat{T}_l \leq r \leq T^*_{k}-1}{\sup}\underset{1 \leq u \neq v \leq p}{\max}\xi_{uv,1r}\big) \sqrt{p(p-1)} (T^*_{k}-\widehat{T}_l)}\nonumber\\
& \geq \left\|\frac{1}{T}\overset{T^*_{k}-1}{\underset{r=\widehat{T}_{l}}{\sum}}(I_p\otimes I_p)\text{vec}\big(\widehat{\Sigma}_{l+1} - \Sigma^*_{k}) +\frac{1}{T}\overset{T^*_{k}-1}{\underset{r=\widehat{T}_{l}}{\sum}}\text{vec}\big(\Sigma^*_{k}-X_rX^\top_r)\right\|_2,\label{changepoint_1_adaptive}
\end{align}
and for the change points $t = T^*_k$ and $t = \widehat{T}_{l+1}$, we get
\begin{align}
\lefteqn{2\lambda_2(\xi_{2T^*_k}+\xi_{2\widehat{T}_{l+1}}) + \lambda_1 \big(\underset{T^*_k \leq r \leq \widehat{T}_{l+1}-1}{\sup}\underset{1 \leq u \neq v \leq p}{\max}\xi_{uv,1r}\big) \sqrt{p(p-1)} (\widehat{T}_{l+1}-T^*_k)}\nonumber\\ 
& \geq \left\|\frac{1}{T}\overset{\widehat{T}_{l+1}-1}{\underset{r=T^*_k}{\sum}}(I_p\otimes I_p)\text{vec}(\widehat{\Sigma}_{l+1} - \Sigma^*_{k+1}) +\frac{1}{T}\overset{\widehat{T}_{l+1}-1}{\underset{r=T^*_{k}}{\sum}}\text{vec}\big(\Sigma^*_{k+1}-X_rX^\top_r\big)\right\|_2.\label{changepoint_2_adaptive}
\end{align}
Moreover, by the triangle inequality, we have
{\small{\begin{align*}
\lefteqn{\|\Sigma^*_{k+1}-\Sigma^*_k\|_F \leq \|\widehat{\Sigma}_{l+1}-\Sigma^*_k\|_F + \|\widehat{\Sigma}_{l+1}-\Sigma^*_{k+1}\|_F}\\
  & \leq \frac{2\lambda_2 T (\xi_{2\widehat{T}_l}+\xi_{2T^*_k})}{T^*_k-\widehat{T}_l}  + \lambda_1 \big(\underset{\widehat{T}_l \leq r \leq T^*_{k}-1}{\sup}\underset{1 \leq u \neq v \leq p}{\max}\xi_{uv,1r}\big) \sqrt{p(p-1)}+ \left\|\frac{1}{T^*_k-\widehat{T}_l}\overset{T^*_{k}-1}{\underset{r=\widehat{T}_{l}}{\sum}}\text{vec}\big(\Sigma^*_{k}-X_rX^\top_r\big)\right\|_2 \\
  &\quad  +\frac{2\lambda_2 T(\xi_{2T^*_k}+\xi_{2\widehat{T}_{l+1}})}{\widehat{T}_{l+1}-T^*_k} + \lambda_1 \big(\underset{T^*_k \leq r \leq \widehat{T}_{l+1}-1}{\sup}\underset{1 \leq u \neq v \leq p}{\max}\xi_{uv,1r}\big) \sqrt{p(p-1)}
  + \left\|\frac{1}{\widehat{T}_{l+1}-T^*_k}\overset{\widehat{T}_{l+1}-1}{\underset{r=T^*_k}{\sum}}\text{vec}\big(\Sigma^*_{k}-X_rX^\top_r\big)\right\|_2.
\end{align*}}}
So $\sum^{m_{\max}}_{m=m^*+1} \sum^{m^*}_{k=1}\Pb\big(L^{(1)}_{m,k,3}\big)$ is upper bounded as follows
{\footnotesize{\begin{align*}
\sum^{m_{\max}}_{m=m^*+1} \sum^{m^*}_{k=1}\Pb\big(L^{(1)}_{m,k,3}\big)
  & \leq  \sum^{m_{\max}}_{m=m^*+1} \sum^{m^*}_{k=1}\Pb\Bigg(\|\Sigma^*_{k+1}-\Sigma^*_{k}\|_F/3 \leq 2\lambda_2\delta^{-1}_T\big[(\xi_{2\widehat{T}_l}+\xi_{2T^*_k})+(\xi_{2T^*_k}+\xi_{2\widehat{T}_{l+1}})\big]\\
  &+ \lambda_1T\sqrt{p(p-1)}\big[\big(\underset{\widehat{T}_l \leq r \leq T^*_{k}-1}{\sup}\underset{1 \leq u \neq v \leq p}{\max}\xi_{uv,1r}\big)+\big(\underset{T^*_k \leq r \leq \widehat{T}_{l+1}-1}{\sup}\underset{1 \leq u \neq v \leq p}{\max}\xi_{uv,1r}\big)\big]\Bigg)
  \\
& + \sum^{m_{\max}}_{m=m^*+1} \sum^{m^*}_{k=1}\Pb\left(\|\Sigma^*_{k+1}-\Sigma^*_{k}\|_F/3 \leq \left\|\frac{1}{T^*_k-\widehat{T}_l}
\overset{T^*_{k}-1}{\underset{r=\widehat{T}_{l}}{\sum}}\text{vec}\big(\Sigma^*_{k}-X_rX^\top_r)\right\|_2,T^*_{k}-\widehat{T}_l\geq T \delta_T\right)\\
& + \sum^{m_{\max}}_{m=m^*+1} \sum^{m^*}_{k=1}\Pb\left(\|\Sigma^*_{k+1}-\Sigma^*_{k}\|_F/3 \leq \left\|\frac{1}{\widehat{T}_{l+1}-T^*_k}\overset{\widehat{T}_{l+1}-1}{\underset{r=T^*_k}{\sum}}\text{vec}\big(\Sigma^*_{k}-X_rX^\top_r)\right\|_2, \widehat{T}_{l+1}-T^*_{k}\geq T \delta_T\right),
\end{align*}}}
which tends to zero in the spirit as in (\ref{bound_A+_adaptive}). For $L^{(2)}_{m,k,3}$, by Lemma \ref{optimality_cond_adaptive} with change points $t = T^*_k$ and $t = \widehat{T}_l$ to obtain (\ref{changepoint_1_adaptive}) and with change points $t=T^*_{k}$, $t = T^*_{k+1}$, we get
\begin{align}
\lefteqn{2\lambda_2 (\xi_{2T^*_k}+\xi_{2T^*_{k+1}}) + \lambda_1\big(\underset{T^*_k \leq r \leq T^*_{k+1}-1}{\sup}\underset{1 \leq u \neq v \leq p}{\max}\xi_{uv,1r}\big)\sqrt{p(p-1)}(T^*_{k+1}-T^*_k)}\nonumber\\
& \geq\left\|\frac{1}{T}\overset{T^*_{k+1}-1}{\underset{r=T^*_{k}}{\sum}}(I_p\otimes I_p)\text{vec}(\widehat{\Sigma}_{l+1} - \Sigma^*_{k+1}) +\frac{1}{T}\overset{T^*_{k+1}-1}{\underset{r=T^*_{k}}{\sum}}\text{vec}\big(\Sigma^*_{k+1}-X_rX^\top_r\big)\right\|_2.\label{changepoint_3_adaptive}
\end{align}
By the triangle inequality, we have
{\footnotesize{\begin{align*}
\lefteqn{\|\Sigma^*_{k+1}-\Sigma^*_{k}\|_F\leq \|\widehat{\Sigma}_{l+1}-\Sigma^*_{k}\|_F+\|\widehat{\Sigma}_{l+1}-\Sigma^*_{k+1}\|_F}\\
  & \leq  \frac{2\lambda_2 T(\xi_{2\widehat{T}_l}+\xi_{2T^*_{k}})}{T^*_k-\widehat{T}_l} + \lambda_1 \big(\underset{\widehat{T}_l \leq r \leq T^*_{k}-1}{\sup}\underset{1 \leq u \neq v \leq p}{\max}\xi_{uv,1r}\big) \sqrt{p(p-1)}
  + \left\|\frac{1}{T^*_k-\widehat{T}_l}\overset{T^*_{k}-1}{\underset{r=\widehat{T}_{l}}{\sum}}\text{vec}\big(\Sigma^*_{k}-X_rX^\top_r\big)\right\|_2 \\
  &\quad  + \frac{2\lambda_2 T (\xi_{2T^*_k}+\xi_{2T^*_{k+1}})}{T^*_{k+1}-T^*_k} + \lambda_1 \big(\underset{T^*_k \leq r \leq T^*_{k+1}-1}{\sup}\underset{1 \leq u \neq v \leq p}{\max}\xi_{uv,1r}\big) \sqrt{p(p-1)} + \left\|\frac{1}{T^*_{k+1}-T^*_k}\overset{T^*_{k+1}-1}{\underset{r=T^*_k}{\sum}}\text{vec}\big(\Sigma^*_{k+1}-X_rX^\top_r\big)\right\|_2.
\end{align*}}}
Therefore, we obtain
{\footnotesize{\begin{align*}
\sum^{m_{\max}}_{m=m^*+1} \sum^{m^*}_{k=1}\Pb\big(L^{(2)}_{m,k,3}\big)
  & \leq  \sum^{m_{\max}}_{m=m^*+1} \sum^{m^*}_{k=1}\Pb\Bigg(\|\Sigma^*_{k+1}-\Sigma^*_{k}\|_F/3 \leq 2\lambda_2 \delta^{-1}_T(\xi_{2\widehat{T}_l}+\xi_{2T^*_{k}})+\frac{2\lambda_2 T}{\mathcal{I}_{\min}}(\xi_{2T^*_k}+\xi_{2T^*_{k+1}})\\
  & + \lambda_1\sqrt{p(p-1)}\big[ \big(\underset{\widehat{T}_l \leq r \leq T^*_{k}-1}{\sup}\underset{1 \leq u \neq v \leq p}{\max}\xi_{uv,1r}\big) + \big(\underset{T^*_k \leq r \leq T^*_{k+1}-1}{\sup}\underset{1 \leq u \neq v \leq p}{\max}\xi_{uv,1r}\big) \big]\Bigg)
  \\
  &  + \sum^{m_{\max}}_{m=m^*+1} \sum^{m^*}_{k=1}\Pb\left(\|\Sigma^*_{k+1}-\Sigma^*_{k}\|_F/3 \leq  \left\|\frac{1}{T^*_k-\widehat{T}_l}\overset{T^*_{k}-1}{\underset{r=\widehat{T}_{l}}{\sum}}\text{vec}\big(\Sigma^*_{k}-X_rX^\top_r\big)\right\|_2,T^*_{k}-\widehat{T}_l\geq T \delta_T\right) \\
  & + \sum^{m_{\max}}_{m=m^*+1} \sum^{m^*}_{k=1}\Pb\left(\|\Sigma^*_{k+1}-\Sigma^*_{k}\|_F/3 \leq \left\|\frac{1}{T^*_{k+1}-T^*_k}\overset{T^*_{k+1}-1}{\underset{r=T^*_k}{\sum}}\text{vec}\big(\Sigma^*_{k}-X_rX^\top_r\big)\right\|_2\right),
\end{align*}}}
which will tend to zero based on similar arguments as in the convergence of (\ref{bound_A+_adaptive}). For $L^{(3)}_{m,k,3}$, by Lemma \ref{optimality_cond_adaptive} with change points $t = T^*_{k-1}$ and $t = T^*_k$, we have
\begin{align}
\lefteqn{2\lambda_2(\xi_{2T^*_{k-1}}+\xi_{2T^*_{k}}) + \lambda_1 \big(\underset{T^*_{k-1} \leq r \leq T^*_k-1}{\sup}\underset{1 \leq u \neq v \leq p}{\max}\xi_{uv,1r}\big) \sqrt{p(p-1)} (T^*_{k}-T^*_{k-1})}\nonumber\\
& \geq \left\|\frac{1}{T}\overset{T^*_{k}-1}{\underset{r=T^*_{k-1}}{\sum}}(I_p\otimes I_p)\text{vec}\big(\widehat{\Sigma}_{l+1} - \Sigma^*_{k}\big) +\frac{1}{T}\overset{T^*_{k}-1}{\underset{r=T^*_{k-1}}{\sum}}\text{vec}\big(\Sigma^*_{k}-X_rX^\top_r\big)\right\|_2,\label{changepoint_4_adaptive}
\end{align}
and with change points $t=T^*_{k}$, $t = \widehat{T}_{l+1}$, we get
\begin{align}
\lefteqn{2\lambda_2(\xi_{2T^*_k}+\xi_{2\widehat{T}_{l+1}}) + \lambda_1 \big(\underset{T^*_k \leq r \leq \widehat{T}_{l+1}-1}{\sup}\underset{1 \leq u \neq v \leq p}{\max}\xi_{uv,1r}\big) \sqrt{p(p-1)} (\widehat{T}_{l+1}-T^*_k)}\nonumber\\
& \geq \left\|\frac{1}{T}\overset{\widehat{T}_{l+1}-1}{\underset{r=T^*_{k}}{\sum}}(I_p\otimes I_p)\text{vec}\big(\widehat{\Sigma}_{l+1} - \Sigma^*_{k+1}\big) +\frac{1}{T}\overset{\widehat{T}_{l+1}-1}{\underset{r=T^*_{k}}{\sum}}\text{vec}\big(\Sigma^*_{k+1}-X_rX^\top_r\big)\right\|_2.\label{changepoint_5_adaptive}
\end{align}
By the triangle inequality, we deduce
{\small{\begin{align*}
\lefteqn{\|\Sigma^*_{k+1}-\Sigma^*_{k}\|_F\leq \|\widehat{\Sigma}_{l+1}-\Sigma^*_{k}\|_F+\|\widehat{\Sigma}_{l+1}-\Sigma^*_{k+1}\|_F}\\
  & \leq \frac{2\lambda_2 T(\xi_{2T^*_{k-1}}+\xi_{2T^*_{k}})}{T^*_k-T^*_{k-1}}+\lambda_1 \big(\underset{T^*_{k-1} \leq r \leq T^*_k-1}{\sup}\underset{1 \leq u \neq v \leq p}{\max}\xi_{uv,1r}\big) \sqrt{p(p-1)} + \left\|\frac{1}{T^*_k-T^*_{k-1}}\overset{T^*_{k}-1}{\underset{r=T^*_{k-1}}{\sum}}\text{vec}\big(\Sigma^*_{k}-X_rX^\top_r\big)\right\|_2 \\
  &\quad  +\frac{2\lambda_2 T(\xi_{2\widehat{T}_l}+\xi_{2T^*_{k}})}{T^*_k-\widehat{T}_l} + \lambda_1 \big(\underset{\widehat{T}_l \leq r \leq T^*_{k}-1}{\sup}\underset{1 \leq u \neq v \leq p}{\max}\xi_{uv,1r}\big) \sqrt{p(p-1)} + \left\|\frac{1}{\widehat{T}_{l+1}-T^*_k}\overset{\widehat{T}_{l+1}-1}{\underset{r=T^*_k}{\sum}}\text{vec}\big(\Sigma^*_{k+1}-X_rX^\top_r)\right\|_2.
\end{align*}}}
We deduce
{\footnotesize{\begin{align*}
\sum^{m_{\max}}_{m=m^*+1} \sum^{m^*}_{k=1}\Pb\big(L^{(3)}_{m,k,3}\big)
  & \leq  \sum^{m_{\max}}_{m=m^*+1} \sum^{m^*}_{k=1}\Pb\Bigg(\|\Sigma^*_{k+1}-\Sigma^*_{k}\|_F/3 \leq \frac{2\lambda_2 T}{\mathcal{I}_{\min}}(\xi_{2T^*_{k-1}}+\xi_{2T^*_{k}}) + 2\lambda_2 \delta^{-1}_T(\xi_{2\widehat{T}_l}+\xi_{2T^*_{k}})\\
  & + \lambda_1 \sqrt{p(p-1)} \big[\big(\underset{T^*_{k-1} \leq r \leq T^*_k-1}{\sup}\underset{1 \leq u \neq v \leq p}{\max}\xi_{uv,1r}\big) + \big(\underset{\widehat{T}_l \leq r \leq T^*_{k}-1}{\sup}\underset{1 \leq u \neq v \leq p}{\max}\xi_{uv,1r}\big) \big]\Bigg)
  \\
  & + \sum^{m_{\max}}_{m=m^*+1} \sum^{m^*}_{k=1}\Pb\left(\|\Sigma^*_{k+1}-\Sigma^*_{k}\|_F/3 \leq \left\|\frac{1}{T^*_k-T^*_{k-1}}\overset{T^*_{k}-1}{\underset{r=T^*_{k-1}}{\sum}}\text{vec}\big(\Sigma^*_{k}-X_rX^\top_r)\right\|_2\right) \\
  & + \sum^{m_{\max}}_{m=m^*+1} \sum^{m^*}_{k=1}\Pb\left(\|\Sigma^*_{k+1}-\Sigma^*_{k}\|_F/3 \leq \left\|\frac{1}{\widehat{T}_{l+1}-T^*_{k+1}}\overset{\widehat{T}_{l+1}-1}{\underset{r=T^*_k}{\sum}}\text{vec}\big(\Sigma^*_{k+1}-X_rX^\top_r)\right\|_2, \widehat{T}_{l+1}-T^*_k\geq T\delta_T\right),
\end{align*}}}
which tends to zero based on the same arguments as in the convergence of (\ref{bound_A+_adaptive}). Finally, to analyze $L^{(4)}_{m,k,3}$, applying Lemma \ref{optimality_cond_adaptive} with $t = T^*_{k-1},t = T^*_k$ to obtain (\ref{changepoint_4_adaptive}) and with $t = T^*_{k},t = T^*_{k+1}$ to obtain (\ref{changepoint_3_adaptive}). By the triangle inequality, we have
{\footnotesize{\begin{align*}
\lefteqn{\|\Sigma^*_{k+1}-\Sigma^*_{k}\|_F\leq \|\widehat{\Sigma}_{l+1}-\Sigma^*_{k}\|_F+\|\widehat{\Sigma}_{l+1}-\Sigma^*_{k+1}\|_F}\\
  & \leq 
  \frac{2\lambda_2 T(\xi_{2T^*_{k-1}}+\xi_{2T^*_{k}})}{T^*_k-T^*_{k-1}} + \lambda_1 \big(\underset{T^*_{k-1} \leq r \leq T^*_k-1}{\sup}\underset{1 \leq u \neq v \leq p}{\max}\xi_{uv,1r}\big) \sqrt{p(p-1)} + \left\|\frac{1}{T^*_k-T^*_{k-1}}\overset{T^*_{k}-1}{\underset{r=T^*_{k-1}}{\sum}}\text{vec}\big(\Sigma^*_{k}-X_rX^\top_r\big)\right\|_2 \\
  &\quad  + \frac{2\lambda_2 T(\xi_{2T^*_{k}}+\xi_{2T^*_{k+1}})}{T^*_{k+1}-T^*_k} + \lambda_1 \big(\underset{T^*_{k} \leq r \leq T^*_{k+1}-1}{\sup}\underset{1 \leq u \neq v \leq p}{\max}\xi_{uv,1r}\big) \sqrt{p(p-1)}  + \left\|\frac{1}{T^*_{k+1}-T^*_k}\overset{T^*_{k+1}-1}{\underset{r=T^*_k}{\sum}}\text{vec}\big(\Sigma^*_{k+1}-X_rX^\top_r\big)\right\|_2,
\end{align*}}}
we deduce
{\footnotesize{\begin{align*}
\sum^{m_{\max}}_{m=m^*+1} \sum^{m^*}_{k=1}\Pb\big(L^{(4)}_{m,k,3}\big)
& \leq \sum^{m_{\max}}_{m=m^*+1} \sum^{m^*}_{k=1}\Bigg\{\Pb\Bigg(\|\Sigma^*_{k+1}-\Sigma^*_{k}\|_F/3 \leq \frac{2\lambda_2 T}{\mathcal{I}_{\min}}(\xi_{2T^*_{k-1}}+\xi_{2T^*_{k}})+\frac{2\lambda_2 T}{\mathcal{I}_{\min}}(\xi_{2T^*_{k}}+\xi_{2T^*_{k+1}})\\
& +\lambda_1 \big[\big(\underset{T^*_{k-1} \leq r \leq T^*_k-1}{\sup}\underset{1 \leq u \neq v \leq p}{\max}\xi_{uv,1r}\big) +\big(\underset{T^*_{k} \leq r \leq T^*_{k+1}-1}{\sup}\underset{1 \leq u \neq v \leq p}{\max}\xi_{uv,1r}\big) \big]\Bigg)\\
& + \Pb\left(\|\Sigma^*_{k+1}-\Sigma^*_{k}\|_F/3 \leq  \left\|\frac{1}{T^*_k-T^*_{k-1}}\overset{T^*_{k}-1}{\underset{r=T^*_{k-1}}{\sum}}\text{vec}\big(\Sigma^*_{k}-X_rX^\top_r\big)\right\|_2\right)\\
& + \Pb\left(\|\Sigma^*_{k+1}-\Sigma^*_{k}\|_F/3 \leq \left\|\frac{1}{T^*_{k+1}-T^*_k}\overset{T^*_{k+1}-1}{\underset{r=T^*_k}{\sum}}\text{vec}\big(\Sigma^*_{k}-X_rX^\top_r\big)\right\|_2\right)\Bigg\},
\end{align*}}}
which also tends to zero, as in the proof of  the convergence to zero of (\ref{bound_A+}). We conclude that $\Pb\big(\{h(\widehat{\Tc}_{\widehat{m}},\Tc^*_{m^*})> T\delta_T\}\cap \{m^* < \widehat{m} \leq m_{\max}\}\big) \rightarrow 0$ as $T \rightarrow \infty$.

\subsection{Proof of Theorem \ref{theorem_understimation_neg_adaptive}}

The proof follows the same steps as in the proof of Theorem \ref{theorem_understimation_neg}. Let $\breve{m}$, $\breve{\Tc}_{\breve{m}} = \{\breve{T}_1,\ldots,\breve{T}_{\breve{m}}\}$ and $\underline{\widehat{\Sigma}}_{\breve{m}}(\breve{\Tc}_{\breve{m}})=(\widehat{\Sigma}_1(\breve{\Tc}_{\breve{m}}),\ldots,\widehat{\Sigma}_{\breve{m}}(\breve{\Tc}_{\breve{m}}), \widehat{\Sigma}_{\breve{m}+1}(\breve{\Tc}_{\breve{m}}))$ be the hypothesized Group Fused LASSO estimator of the number of change points, set of change points and set of estimated variance-covariance matrices, respectively. Define $\underline{\widehat{\Sigma}}_{\Tc_m} = \arg\min_{\underline{\Sigma}_m} \; \Lb(\underline{\Sigma}_m,\mathcal{X}_T,\Tc_m)$, with \\
$\Lb(\underline{\Sigma}_m,\mathcal{X}_T,\Tc_m) = \frac{1}{2T}\overset{m+1}{\underset{j=1}{\sum}} \overset{T_j-1}{\underset{t=T_{j-1}}{\sum}}\text{tr}\left(\Sigma^\top_j\Sigma_j-(X_tX^\top_t)^\top\Sigma_j-\Sigma^\top_j(X_tX^\top_t)\right)$,
as the estimator of $\underline{\Sigma}_m = (\Sigma_1,\ldots,\Sigma_{m+1})$ for the given set of change-point locations $\Tc_m$. Define $\Lb_{\lambda_1\lambda_2}(\underline{\Sigma}_m,\mathcal{X}_T,\Tc_m) = \Lb(\underline{\Sigma}_m,\mathcal{X}_T,\Tc_m) + \lambda_1 \sum^{m+1}_{j=1}\sum_{u \neq v} \xi_{uv,1j}|\Sigma_{uv,j}|+ \lambda_2 \sum^m_{j=1}\xi_{2j}\|\Sigma_{j+1}-\Sigma_j\|_F$.
We aim to show that $\Pb\left(\Lb_{\lambda_1\lambda_2}(\underline{\widehat{\Sigma}}_{\breve{m}}(\breve{\Tc}_{\breve{m}}),\mathcal{X}_T,\breve{\Tc}_{\breve{m}})>\Lb_{\lambda_1\lambda_2}(\underline{\widehat{\Sigma}}_{m^*}(\widehat{\Tc}_{m^*}),\mathcal{X}_T,\widehat{\Tc}_{m^*})\right) \rightarrow 1$.
By definition of the adaptive weights, $\xi_{uv,1j}\leq \max(\alpha_{p,T},a_T)^{-\mu_1}$. Recall that $r_{p,T} := \frac{\lambda_2 \, T }{\max(\alpha_{p,T},a_T)^{\mu_2}\Ic_{\min}}+\frac{\lambda_1Tp}{\max(\alpha_{p,T},a_T)^{\mu_1}}\frac{T\delta_T}{\Ic_{\min}}+\frac{T\delta_T}{\Ic_{\min}} + p\sqrt{ \frac{\log(pT)}{\Ic_{\min}}}$. We have $\sum^{m^*+1}_{j=1}\sum_{u \neq v}|\Sigma^*_{uv,j}|=O(m^* s^*_{\max})$, and under $\|\widehat{\Sigma}_j-\Sigma^*_j\|_F=O_p(r_{p,T})$, we get $\sum^{m^*+1}_{j=1}\sum_{u \neq v}|\widehat{\Sigma}_{uv,j}(\widehat{\Tc}_{m^*})|=O_p(m^*s^*_{\max}+m^*pr_{p,T})$. Under Assumption \ref{assumption_rates_adaptive} and $m^* s^*_{\max}/(p\Ic_{\min}\eta_{\min}) \rightarrow 0$, $m^*T\delta_T/(\Ic_{\min}\eta_{\min})\rightarrow 0$, we have
{\small{\begin{align*}
\lefteqn{\frac{T}{\Ic_{\min}\eta^2_{\min}}\Big[\Lb_{\lambda_1\lambda_2}(\underline{\widehat{\Sigma}}_{\breve{m}}(\breve{\Tc}_{\breve{m}}),\mathcal{X}_T,\breve{\Tc}_{\breve{m}})-\Lb_{\lambda_1\lambda_2}(\underline{\widehat{\Sigma}}_{m^*}(\widehat{\Tc}_{m^*}),\mathcal{X}_T,\widehat{\Tc}_{m^*})\Big]}\\
& = \frac{T}{\Ic_{\min}\eta^2_{\min}}\Bigg[\frac{1}{2T}\overset{\breve{m}+1}{\underset{j=1}{\sum}}\overset{\breve{T}_j-1}{\underset{t=\breve{T}_{j-1}}{\sum}}\text{tr}(\widehat{\Sigma}_j(\breve{\Tc}_{\breve{m}})^\top\widehat{\Sigma}_j(\breve{\Tc}_{\breve{m}})-(X_tX^\top_t)^\top\widehat{\Sigma}_j(\breve{\Tc}_{\breve{m}})-\widehat{\Sigma}_j(\breve{\Tc}_{\breve{m}})^\top(X_tX^\top_t)) \\
&\quad  - \frac{1}{2T}\overset{m^*+1}{\underset{j=1}{\sum}}\overset{\widehat{T}_j-1}{\underset{t=\widehat{T}_{j-1}}{\sum}}\text{tr}(\widehat{\Sigma}_j(\widehat{\Tc}_{m^*})^\top\widehat{\Sigma}_j(\widehat{\Tc}_{m^*})- (X_tX^\top_t)^\top\widehat{\Sigma}_j(\widehat{\Tc}_{m^*})-\widehat{\Sigma}_j(\widehat{\Tc}_{m^*})^\top(X_tX^\top_t))\Bigg]\\
&\quad  + \frac{\lambda_1 T}{\Ic_{\min}\eta^2_{\min}}\Bigg[\sum^{\breve{m}+1}_{j=1}\sum_{u \neq v} \xi_{uv,1j}|\widehat{\Sigma}_{uv,j}(\breve{\Tc}_{\breve{m}})| - \sum^{m^*+1}_{j=1}\sum_{u \neq v} \xi_{uv,1j}|\widehat{\Sigma}_{uv,j}(\widehat{\Tc}_{m^*})| \Bigg]\\ 
&\quad  + \frac{\lambda_2 T}{\Ic_{\min}\eta^2_{\min}} \Bigg[\overset{\breve{m}}{\underset{j=1}{\sum}} \xi_{2j}\|\widehat{\Sigma}_{j+1}(\breve{\Tc}_{\breve{m}})-\widehat{\Sigma}_j(\breve{\Tc}_{\breve{m}})\|_F-\overset{m^*}{\underset{j=1}{\sum}}\xi_{2j} \|\widehat{\Sigma}_{j+1}(\widehat{\Tc}_{m^*})-\widehat{\Sigma}_j(\widehat{\Tc}_{m^*})\|_F\Bigg]\\
& \geq \frac{T}{\Ic_{\min}\eta^2_{\min}}\Bigg[\Lb(\underline{\widehat{\Sigma}}_{\breve{m}}(\breve{\Tc}_{\breve{m}}),\mathcal{X}_T,\breve{\Tc}_{\breve{m}})-\Lb(\underline{\widehat{\Sigma}}_{m^*}(\widehat{\Tc}_{m^*}),\mathcal{X}_T,\widehat{\Tc}_{m^*})\Bigg]+o_p(1).
\end{align*}}}
The rest of the proof follows the same steps as in the proof of Theorem \ref{theorem_understimation_neg}.

\subsection{Proof of Theorem \ref{theorem_recovery}}

The proof follows the same steps as in Theorem 5 of \citet{Kolar2012}. Under the conditions of Theorem \ref{theorem_date_recov_adaptive}, we work on the event $\Ec := \{\max_{1 \leq j \leq m^*}|\widehat{T}_j-T^*_j|\leq T \delta_T\}$. We show $\widehat{\Sc}_j \subset \Sc^*_j$ with probability tending to one. For block $j$, since $\widehat{\Sigma}_j$ is an optimal solution to (\ref{stat_crit_adaptive}), it satisfies
\begin{equation}\label{optim_block}
\frac{1}{T}|\widehat{\Bc}_j|\Big(\widehat{\Sigma}_j-\frac{1}{|\widehat{\Bc}_j|}\sum_{t \in \widehat{\Bc}_j}X_tX^\top_t\Big) + \lambda_1\sum_{t \in \widehat{\Bc}_j}\Xi_{1t}\odot \widehat{E}_{1t} + \lambda_2\Big(\xi_{2,\widehat{T}_{j-1}}\widehat{E}_{2\widehat{T}_{j-1}} - \xi_{2,\widehat{T}_{j}}\widehat{E}_{2\widehat{T}_{j}}\Big)=\mathbf{0}_{p \times p},
\end{equation}
where $\Xi_{1t}=\left(\xi_{uv,1t}\right)_{uv}$. We construct the matrices $\overline{\Sigma}_j,\widehat{E}_{2\widehat{T}_{j}}, \widehat{E}_{2\widehat{T}_{j-1}}, \widehat{E}_{1t}$ that satisfy (\ref{optim_block}) and verify that the sub-differential matrices are dual feasible. We now construct the oracle estimator $\overline{\Sigma} \in \mathcal{O}_{\widehat{m}+1}:=\left\{\Sigma_1,\ldots,\Sigma_{\widehat{m}+1}: \Sigma_{j,uv}=0,\, \forall (u,v)\notin \Sc^{*}_j\right\}$. It is defined
as the solution to the restricted problem
{\small{\begin{align*}
\left\{\overline{\Sigma}\right\}^{\widehat{m}+1}_{j=1} := \underset{\Sigma \in \mathcal{O}_{\widehat{m}+1}}{\arg\;\min}\left\{\sum^{\widehat{m}+1}_{j=1}\frac{|\widehat{\Bc}_j|}{2T}\left\|\Sigma_j-\frac{1}{|\widehat{\Bc}_j|}\sum_{t \in \widehat{\Bc}_j}X_tX^\top_t\right\|^2_F + \lambda_1 \sum^{\widehat{m}+1}_{j=1} \sum_{u \neq v}\Big(\sum_{t\in \widehat{\Bc}_j}\xi_{uv,1t}\Big)|\Sigma_{j,uv}| + \lambda_2\sum^{\widehat{m}}_{j=1}\xi_{2,\widehat{T}_{j}}\|\Sigma_{j+1}-\Sigma_{j}\|_F \right\}.
\end{align*}}}
Let $\widehat{S}_j:=\frac{1}{|\widehat{\Bc}_j|}\sum_{t \in \widehat{\Bc}_j}X_tX^\top_t$ and $\Xi_{1j}=\sum_{t \in \widehat{\Bc}_j}\Xi_{1t}$.
From (\ref{optim_block}), for a given block $\widehat{\Bc}_j$, we have:
\begin{equation*}
\overline{\Sigma}_{j,\Sc^*_j} = 
\widehat{S}_{j,\Sc^*_j} - T\lambda_1\frac{1}{|\widehat{\Bc}_j|}\Xi_{1j,\Sc^*_j}\odot \text{sgn}(\overline{\Sigma}_{j,\Sc^*_j}) - \frac{T}{|\widehat{\Bc}_j|}\lambda_2\Big(\xi_{2,\widehat{T}_{j}}\overline{E}_{2\widehat{T}_{j},\Sc^*_j} - \xi_{2,\widehat{T}_{j-1}}\overline{E}_{2\widehat{T}_{j-1},\Sc^*_j}\Big).
\end{equation*}
We deduce
\begin{equation}\label{optim_block_active}
\overline{\Sigma}_{j,\Sc^*_j} - \Sigma^*_{j,\Sc^*_j}= 
\widehat{S}_{j,\Sc^*_j}-\Sigma^*_{j,\Sc^*_j} - T\lambda_1\frac{1}{|\widehat{\Bc}_j|}\Xi_{1j,\Sc^*_j}\odot \text{sgn}(\overline{\Sigma}_{j,\Sc^*_j}) - \frac{T}{|\widehat{\Bc}_j|}\lambda_2\Big(\xi_{2,\widehat{T}_{j}}\overline{E}_{2\widehat{T}_{j},\Sc^*_j} - \xi_{2,\widehat{T}_{j-1}}\overline{E}_{2\widehat{T}_{j-1},\Sc^*_j}\Big).
\end{equation}
The term $\widehat{S}_{j,\Sc^*_j}-\Sigma^*_{j,\Sc^*_j}$ can be written as
\begin{equation*}
\widehat{S}_{j,\Sc^*_j}-\Sigma^*_{j,\Sc^*_j} = \frac{1}{|\widehat{\Bc}_j|}\sum_{k:\widehat{\Bc}_j \cap \Bc^*_k \neq \emptyset} \sum_{t \in \widehat{\Bc}_j \cap \Bc^*_k} (X_tX^\top_t - \Sigma^*_k) + \frac{1}{|\widehat{\Bc}_j|}  \sum_{k:\widehat{\Bc}_j \cap \Bc^*_k \neq \emptyset} |\widehat{\Bc}_j\cap \Bc^*_k|(\Sigma^*_k-\Sigma^*_j)=:A_{1,j}+A_{2,j}.
\end{equation*}
On $\Ec$, $\sum_{k \neq j}|\widehat{\Bc}_j\cap \Bc^*_k| \leq |\widehat{T}_{j-1}-T^*_{j-1}|+|\widehat{T}_j-T^*_j| \leq 2T\delta_T$. 
Therefore, under $\max_k\|\Sigma^*_k\|_{\max}=O(1)$, $\|A_{2,j}\|_{\max} \leq C \frac{T\delta_T}{|\widehat{\Bc}_j|}$ with $C>0$ a large finite constant.
Furthermore, as $|\widehat{\Bc}_j|\geq |\Bc^*_j|-2T\delta_T\geq \Ic_{\min}-2T\delta_T$, under $T\delta_T=o(\Ic_{\min})$, by Lemma \ref{bound_gradient}, $\|A_{1,j}\|_{\max}=O_p(\sqrt{\log(pT)/\Ic_{\min}})$ and $\|\widehat{S}_j-\Sigma^*_j\|_{\max}=O_p(\sqrt{\log(pT)/\Ic_{\min}}+T\delta_T/\Ic_{\min})$.
Since, by construction, $\overline{\Sigma}_{uv,j}=0$ for any $(u,v)\notin \Sc^*_j$, from (\ref{optim_block}), we have
\begin{equation*}
-\frac{1}{T}|\widehat{\Bc}_j|\widehat{S}_{uv,j} + \lambda_1\sum_{t \in \widehat{\Bc}_j}\xi_{uv,1t}\overline{E}_{uv,1t} + \lambda_2\Big(\xi_{2,\widehat{T}_{j-1}}\overline{E}_{uv,2\widehat{T}_{j-1}} - \xi_{2,\widehat{T}_{j}}\overline{E}_{uv,2\widehat{T}_{j}}\Big)=0.
\end{equation*}
Then, we get
\begin{equation}\label{optim_block_inactive}
\frac{\sum_{t \in \widehat{\Bc}_j}\xi_{uv,1t}\overline{E}_{uv,1t}}{\sum_{t \in \widehat{\Bc}_j}\xi_{uv,1t}}=\frac{\frac{1}{T}|\widehat{\Bc}_j|\widehat{S}_{uv,j}-\lambda_2\Big(\xi_{2,\widehat{T}_{j-1}}\overline{E}_{uv,2\widehat{T}_{j-1}} - \xi_{2,\widehat{T}_{j}}\overline{E}_{uv,2\widehat{T}_{j}}\Big)}{\lambda_1\sum_{t \in \widehat{\Bc}_j}\xi_{uv,1t}}.
\end{equation}
Since $\Sigma^*_{uv,j}=0$ for any $(u,v)\notin \Sc^*_j$, we have $\widehat{S}_{uv,j}=O_p(\sqrt{\log(pT)/\Ic_{\min}}+T\delta_T/\Ic_{\min})$. Moreover, $|\Big(\xi_{2,\widehat{T}_{j-1}}\overline{E}_{uv,2\widehat{T}_{j-1}} - \xi_{2,\widehat{T}_{j}}\overline{E}_{uv,2\widehat{T}_{j}}\Big)| \leq 2 \max_t \xi_{2,t}$. As for the adaptive weights in the LASSO penalty, for inactive entries, $|\widehat{\Bc}_j|^{-1}\sum_{t \in \widehat{\Bc}_j}\xi_{uv,1t} \gtrsim \max(\alpha_{p,T},a_T)^{-\mu_1}$. As a consequence, the right hand side of (\ref{optim_block_inactive}) becomes:
\begin{equation*}
\frac{\frac{1}{T}|\widehat{\Bc}_j|\widehat{S}_{uv,j}-\lambda_2\Big(\xi_{2,\widehat{T}_{j-1}}\overline{E}_{uv,2\widehat{T}_{j-1}} - \xi_{2,\widehat{T}_{j}}\overline{E}_{uv,2\widehat{T}_{j}}\Big)}{\lambda_1\sum_{t \in \widehat{\Bc}_j}\xi_{uv,1t}} = o_p(1),
\end{equation*}
under $\frac{1}{T}\Big(\sqrt{\frac{\log(pT)}{\Ic_{\min}}}+\frac{T\delta_T}{\Ic_{\min}}\Big) + \Ic^{-1}_{\min}\lambda_2\max_t\xi_{2,t}=o(\lambda_1\max(\alpha_{p,T},a_T)^{-\mu_1})$.
Therefore, with probability tending to one, we have 
$\max_j \max_{(u,v)\notin \Sc^*_j}\left|\frac{\sum_{t \in \widehat{\Bc}_j}\xi_{uv,1t}\overline{E}_{uv,1t}}{\sum_{t \in \widehat{\Bc}_j}\xi_{uv,1t}}\right|<1$.
By Lemma \ref{lemma_unique_inactive_pattern}, every minimizer of the unrestricted block-wise criterion with partition $\widehat{\Tc}_{\widehat{m}}$, in particular $\widehat{\Sigma}$, satisfies $\widehat{\Sigma}_{uv,j}=0$, for $(u,v)\notin \Sc^*_j$, for all $j=1,\ldots,\widehat{m}+1$. Therefore, $\widehat{\Sc}_j \subset \Sc^*_j$
with probability tending to one.

\medskip

Let us show that $\Sc^*_j\subset \widehat{\Sc}_j$ with probability tending to one. Taking the $\|\cdot\|_{\max}$ in (\ref{optim_block_active}), we have
\begin{equation*}
\|\overline{\Sigma}_{j,\Sc^*_j} - \Sigma^*_{j,\Sc^*_j}\|_{\max} \leq 
\|\widehat{S}_{j}-\Sigma^*_{j}\|_{\max} + T\lambda_1\frac{1}{|\widehat{\Bc}_j|}\|\Xi_{1j,\Sc^*_j}\|_{\max} + \frac{2T}{|\widehat{\Bc}_j|}\lambda_2\max_t\xi_{2,t}.
\end{equation*}
Under the condition $\min_j \min_{(u,v)\in \Sc^*_j}|\Sigma^*_{uv,j}| \geq L\left(p\sqrt{\frac{\log(pT)}{\Ic_{\min}}}+\frac{T\delta_T}{\Ic_{\min}}+T\lambda_1+\frac{T}{\Ic_{\min}}\lambda_2\max_t\xi_{2,t}\right)$,
with $L>0$ sufficiently large, we have $|\overline{\Sigma}_{uv,j}| \geq |\Sigma^*_{uv,j}|-|\overline{\Sigma}_{uv,j}-\Sigma^*_{uv,j}| >0$ with probability tending to one. We deduce $\frac{1}{|\widehat{\Bc}_j|}\|\Xi_{1j,\Sc^*_j}\|_{\max}=O_p(1)$, and
\begin{equation*}
\max_j\|\overline{\Sigma}_{j,\Sc^*_j} - \Sigma^*_{j,\Sc^*_j}\|_{\max} = O_p\left(\sqrt{\frac{\log(pT)}{\Ic_{\min}}}+\frac{T\delta_T}{\Ic_{\min}}+T\lambda_1+\frac{T}{\Ic_{\min}}\lambda_2\max_t\xi_{2,t}\right).
\end{equation*}
Moreover, $|\overline{\Sigma}_{uv,j}|\geq |\Sigma^*_{uv,j}|-|\overline{\Sigma}_{uv,j}-\Sigma^*_{uv,j}|>0$ with probability tending to one. Since $\widehat{\Sigma}$ coincides with a minimizer of the restricted problem on the event of strict dual feasibility, we deduce that $(u,v)\in \widehat{\Sc}_j$ for $(u,v)\in \Sc^*_j$, with probability tending to one.
Therefore, $\Sc^*_j\subset \widehat{\Sc}_j$ with probability tending to one. Combining with the first part, we conclude $\Sc^*_j=\widehat{\Sc}_j$ with probability tending to one.

\bibliography{biblio}

\end{document}